\documentclass[a4paper,11pt]{article}
\pdfoutput=1
\usepackage[margin=1in]{geometry}
\usepackage[utf8]{inputenc}
\usepackage[T1]{fontenc}
\usepackage{microtype}
\usepackage{amsfonts}
\usepackage{amssymb}
\usepackage{amsmath}
\usepackage{amsthm}
\usepackage{comment}
\usepackage{fullpage}
\usepackage{booktabs}
\usepackage{enumitem}
\usepackage[hidelinks,bookmarksopen]{hyperref}
\usepackage[capitalize]{cleveref}
\usepackage{mathtools}
\usepackage{tikz}
\usepackage[font=small,labelfont=bf]{caption}
\usepackage{titlesec}
\usepackage{thmtools}
\usepackage{wrapfig}
\usepackage{tablefootnote}
\usepackage{thm-restate}
\usepackage{soul}
\usepackage[norefs,nocites,nomsgs]{refcheck}  %

\makeatletter
\newcommand{\refcheckize}[1]{%
  \expandafter\let\csname @@\string#1\endcsname#1%
  \expandafter\DeclareRobustCommand\csname relax\string#1\endcsname[1]{%
    \csname @@\string#1\endcsname{##1}\@for\@temp:=##1\do{\wrtusdrf{\@temp}\wrtusdrf{{\@temp}}}}%
  \expandafter\let\expandafter#1\csname relax\string#1\endcsname
}
\newcommand{\refcheckizetwo}[1]{%
  \expandafter\let\csname @@\string#1\endcsname#1%
  \expandafter\DeclareRobustCommand\csname relax\string#1\endcsname[2]{%
    \csname @@\string#1\endcsname{##1}{##2}\wrtusdrf{##1}\wrtusdrf{{##1}}\wrtusdrf{##2}\wrtusdrf{{##2}}}%
  \expandafter\let\expandafter#1\csname relax\string#1\endcsname
}
\makeatother

\refcheckize{\cref}
\refcheckize{\Cref}
\refcheckizetwo{\crefrange}
\refcheckizetwo{\Crefrange}

\definecolor{mypink}{RGB}{199, 21 133}
\hypersetup{
  colorlinks,
  urlcolor = blue,
  linkcolor = mypink,
  citecolor = blue
}

\newtheorem{theorem}{Theorem}[section]

\newtheorem{lemma}[theorem]{Lemma}

\newtheorem{observation}[theorem]{Observation}
\newtheorem{proposition}[theorem]{Proposition}
\theoremstyle{definition}
\newtheorem{definition}[theorem]{Definition}
\theoremstyle{remark}
\newtheorem{remark}[theorem]{Remark}

\titleformat{\paragraph}[runin]
{\normalfont\itshape}{\theparagraph}{1em}{}
\titlespacing{\paragraph}{0pt}{0.55\baselineskip}{1em}

\makeatletter
\newcommand{\bfparagraph}{%
  \@startsection{paragraph}%
    {4}{\z@ }{2.25ex \@plus 1ex \@minus .2ex}%
    {-1em}{\normalfont\normalsize\bf}%
}

\setlist[enumerate]{nosep, topsep=1ex}
\setlist[itemize]{nosep, topsep=1ex}
\setlist[description]{nosep,topsep=1ex}

\allowdisplaybreaks

\let\oldabstract\abstract
\let\oldendabstract\endabstract
\makeatletter
\renewenvironment{abstract}
{%
  {\list{}{\addtolength{\leftmargin}{-1.0em}%
    \listparindent 1.5em%
     \itemindent    \listparindent%
     \rightmargin   \leftmargin%
     \parsep        \z@ \@plus\p@}%
     \item\relax}%
  {\endlist}%
\oldabstract}
{\oldendabstract}
\makeatother

\newcommand{\bigO}{\mathcal{O}}

\newcommand{\LCE}{\mathrm{LCE}}
\newcommand{\lcp}{\mathrm{lcp}}
\newcommand{\per}{\mathrm{per}}
\newcommand{\revstr}[1]{\overline{#1}}
\newcommand{\Lroot}{\mathrm{L\mbox{-}root}}
\newcommand{\Lhead}{\mathrm{L\mbox{-}head}}
\newcommand{\Ltail}{\mathrm{L\mbox{-}tail}}
\newcommand{\Lexp}{\mathrm{L\mbox{-}exp}}
\newcommand{\type}{\mathrm{type}}
\newcommand{\core}{\mathrm{C}}
\newcommand{\SACore}{\core_{\SA}}
\newcommand{\CountCore}{\mathrm{C}_{\mathrm{PM}}}
\newcommand{\STCore}{\mathrm{C}_{\mathrm{ST}}}
\newcommand{\Int}{\mathrm{int}}

\newcommand{\Intpad}{\mathrm{mstr}}
\newcommand{\Intpadc}{\mathrm{mstr}'}
\newcommand{\lrank}{\mathrm{lrank}}
\newcommand{\rrank}{\mathrm{rrank}}
\newcommand{\Pos}{\mathrm{Pos}}
\newcommand{\Posa}{\mathrm{Pos}^{\mathsf{a}}}
\newcommand{\Poss}{\mathrm{Pos}^{\mathsf{s}}}
\newcommand{\deltaa}{\delta^{\mathsf{a}}}
\newcommand{\deltas}{\delta^{\mathsf{s}}}
\newcommand{\Occ}{\mathrm{Occ}}
\newcommand{\Occa}{\mathrm{Occ}^{\mathsf{a}}}
\newcommand{\Occs}{\mathrm{Occ}^{\mathsf{s}}}
\newcommand{\Occam}{\mathrm{Occ}^{\mathsf{a}-}}
\newcommand{\Occsm}{\mathrm{Occ}^{\mathsf{s}-}}
\newcommand{\Occap}{\mathrm{Occ}^{\mathsf{a}+}}
\newcommand{\Occsp}{\mathrm{Occ}^{\mathsf{s}+}}
\newcommand{\rend}[1]{e(#1)}
\newcommand{\rendfull}[1]{e^{\rm full}(#1)}
\newcommand{\pend}[1]{e(#1)}
\newcommand{\pendfull}[1]{e^{\rm full}(#1)}
\newcommand{\LB}{\mathrm{RangeBeg}}
\newcommand{\UB}{\mathrm{RangeEnd}}
\newcommand{\unary}{\mathrm{unary}}
\newcommand{\Pow}{\mathrm{pow}}
\newcommand{\Successor}{\mathrm{succ}}
\newcommand{\rank}[3]{\mathsf{rank}_{#1,#2}(#3)}
\newcommand{\select}[3]{\mathsf{select}_{#1,#2}(#3)}
\newcommand{\fbeg}{s}
\newcommand{\fend}{t}
\newcommand{\Pref}{\mathrm{Pref}}

\newcommand{\stext}{s^{\rm text}}
\newcommand{\slex}{s^{\rm lex}}
\newcommand{\rtextm}{r^{\rm text}}
\newcommand{\rlexm}{r^{\rm lex}}

\newcommand{\T}{T}
\newcommand{\Trev}{\T^{\rm rev}}
\newcommand{\Pat}{P}
\newcommand{\Qsuf}{Q_{\rm suf}}

\renewcommand{\S}{\mathsf{S}}
\newcommand{\R}{\mathsf{R}}
\newcommand{\D}{\mathcal{D}}
\newcommand{\Su}{\mathcal{S}}
\newcommand{\Alphabet}{[0 \dd \sigma)}
\newcommand{\Z}{\mathbb{Z}}
\newcommand{\Zz}{\mathbb{Z}_{\ge 0}}
\newcommand{\Zp}{\mathbb{Z}_{+}}
\newcommand{\N}{\mathbb{N}}
\newcommand{\Lroots}{\mathrm{Roots}}
\newcommand{\Zset}{\mathsf{Z}}

\newcommand{\nil}{\bot}
\newcommand{\LCA}{\mathrm{LCA}}
\newcommand{\WA}{\mathrm{WA}}
\newcommand{\parent}{\mathrm{parent}}
\newcommand{\Root}{\mathrm{root}}

\newcommand{\sdepth}{\mathrm{sdepth}}
\newcommand{\str}{\mathrm{str}}
\newcommand{\child}{\mathrm{child}}
\newcommand{\pred}{\mathrm{pred}}

\newcommand{\findleaf}{\mathrm{findleaf}}
\newcommand{\isleaf}{\mathrm{isleaf}}
\newcommand{\isancestor}{\mathrm{isancestor}}
\newcommand{\cnt}{\mathrm{count}}
\newcommand{\slink}{\mathrm{slink}}
\newcommand{\wlink}{\mathrm{wlink}}
\newcommand{\wlinkprim}{\mathrm{wlink}'}
\newcommand{\ind}{\mathrm{index}}
\newcommand{\repr}{\mathrm{repr}}
\newcommand{\leftsibling}{\mathrm{leftsibling}}
\newcommand{\rightsibling}{\mathrm{rightsibling}}
\newcommand{\firstchild}{\mathrm{firstchild}}
\newcommand{\lastchild}{\mathrm{lastchild}}
\newcommand{\letter}{\mathrm{letter}}

\newcommand{\bpre}{b_{\rm pre}}
\newcommand{\epre}{e_{\rm pre}}
\newcommand{\deltatext}{\delta_{\rm text}}

\newcommand{\LTD}{L_{\D}}
\newcommand{\LTrange}{L_{\rm range}}
\newcommand{\LTper}{L_{\rm per}}
\newcommand{\LTpref}{L_{\rm pref}}
\newcommand{\LTrev}{L_{\rm rev}}
\newcommand{\LTroot}{L_{\rm root}}
\newcommand{\LTruns}{L_{\rm runs}}
\newcommand{\LTminexp}{L_{\rm minexp}}
\newcommand{\LTnode}{L_{\rm node}}
\newcommand{\LTchild}{L_{\rm child}}
\newcommand{\LTwa}{L_{\rm WA}}

\newcommand{\BVshort}{B_{3\tau-1}}
\newcommand{\BVS}{B_{\S}}
\newcommand{\BVRprim}{B_{\R'}}
\newcommand{\BVexp}{B_{\rm exp}}

\newcommand{\SA}{\mathrm{SA}}
\newcommand{\ISA}{\mathrm{ISA}}
\newcommand{\ARRshort}{A_{\rm short}}
\newcommand{\ARRsmap}{A_{\rm smap}}
\newcommand{\ARRsinvmap}{A^{-1}_{\rm smap}}
\newcommand{\ARRslex}{A_{\S}}
\newcommand{\ARRrmap}{A_{\rm rmap}}
\newcommand{\ARRrinvmap}{A^{-1}_{\rm rmap}}
\newcommand{\ARRnontail}{A_{\rm len}}
\newcommand{\ARRzlex}{A_{\Zset}}

\newcommand{\TSSS}{\mathcal{T}_{\S}}
\newcommand{\TZ}{\mathcal{T}_{\Zset}}
\newcommand{\Tshort}{\mathcal{T}_{3\tau-1}}
\newcommand{\ST}{\mathcal{T}_{\rm st}}

\newcommand{\mapToShort}{\mathrm{map}_{\ST, \Tshort}}
\newcommand{\mapToTSSS}{\mathrm{map}_{\ST, \TSSS}}
\newcommand{\mapToTZ}{\mathrm{map}_{\ST, \TZ}}
\newcommand{\pseudoInvTSSS}[2]{\mathrm{pseudoinv}_{\TSSS}(#1,#2)}
\newcommand{\pseudoInvTZ}[2]{\mathrm{pseudoinv}_{\TZ}(#1,#2)}

\newcommand{\sm}{\setminus}
\newcommand{\sub}{\subseteq}
\newcommand{\dd}{\mathinner{.\,.}}
\newcommand{\emptystring}{\varepsilon}

\newcommand{\probname}[1]{#1}

\begin{document}

\title{Breaking the $\bigO(n)$-Barrier in the
  Construction\\ of Compressed Suffix Arrays and Suffix Trees}

\author{
  \normalsize Dominik Kempa\thanks{Work in part done while
    at University of California, Berkeley and Johns Hopkins University.
    Supported by NIH HG011392, NSF DBI-2029552, 1652303, 1934846 grants,
    an Alfred P. Sloan Fellowship, and a Simons Foundation Junior
    Faculty Fellowship.}\\[-0.2ex]
  \normalsize Stony Brook University\\[-0.2ex]
  \normalsize \texttt{kempa@cs.stonybrook.edu}
  \and
  \normalsize Tomasz Kociumaka\thanks{Work mostly done while
    at the University of California, Berkeley, partly supported
    by NSF 1652303, 1909046, and HDR TRIPODS 1934846 grants,
    and an Alfred P. Sloan Fellowship.}\\[-0.2ex]
  \normalsize Max Planck Institute for Informatics\\[-0.2ex]
  \normalsize \texttt{tomasz.kociumaka@mpi-inf.mpg.de}
}

\date{\vspace{-0.5cm}}
\maketitle

\begin{abstract}
The suffix array, describing the lexicographical order of suffixes of a given text, and the suffix tree, a path-compressed trie of all suffixes, are the two most fundamental data structures for string processing, with plethora of applications in data compression, bioinformatics, and information retrieval. For a length-$n$ text, however, they use $\Theta(n \log n)$ bits of space, which is often too costly.  To address this, Grossi and Vitter [STOC 2000] and, independently, Ferragina and Manzini [FOCS 2000] introduced space-efficient versions of the suffix array, known as the \emph{compressed suffix array} (CSA) and the \emph{FM-index}. Sadakane [SODA 2002] then showed how to augment them to obtain the \emph{compressed suffix tree} (CST).  For a length-$n$ text over an alphabet of size $\sigma$, these structures use only $\mathcal{O}(n \log \sigma)$ bits. Nowadays, these structures are part of the standard toolbox: modern textbooks spend dozens of pages describing their applications, and they almost completely replaced suffix arrays and suffix trees in space-critical applications. The biggest remaining open question is how efficiently they can be constructed.  After two decades, the fastest algorithms still run in $\mathcal{O}(n)$ time [Hon et al., FOCS 2003], which is $\Theta(\log_{\sigma} n)$ factor away from the lower bound of $\Omega(n / \log_{\sigma} n)$ (following from the necessity to read the input).

In this paper, we make the first in 20 years improvement in $n$ for this problem by proposing a new compressed suffix array and a new compressed suffix tree which admit $o(n)$-time construction algorithms while matching the space bounds and the query times of the original CSA/CST and the FM-index.  More precisely, our structures take $\mathcal{O}(n\log \sigma)$ bits, support SA queries and full suffix tree functionality in $\mathcal{O}(\log^{\epsilon} n)$ time per operation, and can be constructed in $\mathcal{O}(n \min(1, \log \sigma / \sqrt{\log n}))$ time using $\mathcal{O}(n\log \sigma)$ bits of working space. (For example, if $\sigma=2$, the construction time is $\mathcal{O}(n / \sqrt{\log n}) = o(n)$.)  We derive this result as a corollary from a much more general reduction: We prove that all parameters of a compressed suffix array/tree (query time, space, construction time, and construction working space) can essentially be reduced to those of a data structure answering new query types that we call \emph{prefix rank} and \emph{prefix selection}. Using the novel techniques, we also develop a new index for pattern matching.
\end{abstract}

\thispagestyle{empty}
\clearpage
\pagenumbering{arabic}

\section{Introduction}

Let $\T$ be a text of length $n$. A \emph{suffix tree}~\cite{Weiner73}
of $\T$ is a trie of all suffixes of $\T$, in which
every unary path has been replaced with a single edge labeled by a
text substring. The resulting tree has less than $2n$ nodes and thus
can be encoded in $\bigO(n \log n)$ bits. Related to suffix trees are
\emph{suffix arrays}~\cite{mm1993}. The suffix array $\SA[1 \dd n]$ of
$\T$ stores the permutation of $\{1, \ldots, n\}$ such that $\SA[i]$
is the starting position of the $i$th lexicographically smallest
suffix of $\T$.  Consider now the following problem: Construct a data
structure that, given any length-$m$ pattern $\Pat$, counts the number
of occurrences of $\Pat$ in $\T$.  To solve it using a suffix tree, it
suffices to descend the tree in $\bigO(m)$ time and report the
precomputed number of leaves below the reached node. Using a
suffix array, it suffices to perform an $\bigO(m \log n)$-time binary
search resulting in the range $\SA[b \dd e)$ of suffixes of $\T$
having $\Pat$ as a prefix. Then, $e - b$ is the number of occurrences
of $\Pat$ in $\T$ (and $\SA[b \dd e)$ contains their starting
positions). The advantage of suffix array is that it is more space
efficient: it only needs $n \lceil \log n \rceil$ bits.  The queries,
however, are usually slightly slower.

The above is a canonical application of suffix arrays/trees. It is,
however, only the tip of the iceberg. Suffix trees and suffix arrays
are widely considered to be the two most fundamental data structures
for string processing. As written by Gusfield in his classical
textbook~\cite{gusfield}: \emph{``Suffix trees can be used to solve
the exact matching problem in linear time ($\ldots$), but their real
virtue comes from their use in linear-time solutions to many string
problems more complex than exact matching''}. This includes
well-studied problems like \probname{Maximal Repeats},
\probname{Longest Repeated Factor},
\probname{Minimal Absent Word}, \probname{Longest Common Substring},
\probname{Matching Statistics}, \probname{Maximal Unique Matches},
\probname{LZ77 Factorization}, \probname{BWT Compression}, and
many more (see, 
e.g.,~\cite{bwtbook,gusfield,MBCT2015,navarrobook,ohl2013}).

With the increasing size of datasets that need processing, plain
suffix arrays and suffix trees, however, have become expensive to use,
particularly in applications where the input text is over a small
alphabet $[0 \dd \sigma)$. Such text requires $n \lceil \log \sigma
\rceil$ bits, whereas the suffix array/tree uses at least $n \lceil
\log n \rceil$ bits of space. In some applications, the gap
$\tfrac{\log n}{\log \sigma}$ can be quite large, e.g., in
computational biology, where we usually have $\sigma = 4$, the gap is
typically between 16 and 32. This shortcoming was addressed by Grossi
and Vitter and, independently, Ferragina and Manzini at the turn of
the millennium. They introduced space-efficient versions of the suffix
array, known as the \emph{compressed suffix array (CSA)}~\cite{csa,
GrossiV05} and the \emph{FM-index}~\cite{fm, FerraginaM05}. For a
length-$n$ text over an alphabet of size $\sigma$, these data
structures use only $\bigO(n \log \sigma)$ bits, and they can answer
SA queries (asking for $\SA[i]$ given $i \in [1 \dd n]$) in
$\bigO(\log^{\epsilon} n)$ time, where $\epsilon > 0$ is an arbitrary
predefined constant. With such data structure, one can execute any
algorithm that uses the suffix array, but consuming less space and
only incurring a factor of $\bigO(\log^{\epsilon} n)$ penalty in the
runtime.\footnote{This is often acceptable: a slower algorithm remains
usable, but insufficient memory can thwart it \mbox{entirely}.}
Shortly after these discoveries, Sadakane~\cite{cst} extended
CSA/FM-index into a \emph{compressed suffix tree (CST)}, supporting
all suffix tree operations in $\bigO(\log^{\epsilon} n)$ time (while
still using $\bigO(n \log \sigma)$ bits of space).  This powerful
structure can be plugged into an even larger set of
algorithms~\cite{Gog11}.

Nowadays, CSAs and CSTs are widely used in practice. Modern string
algorithms textbooks focus on the use and applications of CSAs/CSTs
and related data structures~\cite{bwtbook,MBCT2015}, or even entirely
on the emerging notion of \emph{compressed data
structures}~\cite{navarrobook}. The FM-index occupies the central
role in some of the most commonly used bioinformatics tools, like
${\tt Bowtie}$~\cite{bowtie}, ${\tt BWA}$~\cite{bwa}, and ${\tt
Soap2}$~\cite{soap2}, and mature and highly engineered
implementations of CSAs and CSTs are available through the
\href{https://github.com/simongog/sdsl-lite}{\texttt{sdsl}} library%
\footnote{The library is available at
\url{https://github.com/simongog/sdsl-lite}.} of Gog et
al.~\cite{Gog11,sdsl}.  Despite these developments in functionality
and practical adoption of CSAs/CSTs, the time complexity of their
construction remains an open problem.  The original paper of Grossi
and Vitter~\cite{csa}, describes a method that, given a length-$n$
text over alphabet $\Sigma = [0 \dd \sigma)$, constructs the CSA in
$\bigO(n \log \sigma)$ time and using $\bigO(n \log n)$ bits of
working space. In 2003, a celebrated result of Hon et
al.~\cite{HonSS03} lowered the time complexity to $\bigO(n \log \log
\sigma)$ and the space to the optimal $\bigO(n \log \sigma)$
bits. Note, however, that, e.g., for $\sigma = 2$, this algorithm
still runs in $\Theta(n)$ time, which is slower by a $\Theta(\log n)$
factor than the lower bound of $\Omega(n / \log n)$, following simply
from the necessity to read the entire input. Recently,
Belazzougui~\cite{Belazzougui14} improved the time complexity of the
CSA/CST construction to randomized $\bigO(n)$ (while using the optimal
space of $\bigO(n \log \sigma)$ bits), making it independent of the
alphabet size $\sigma$. Shortly after, Munro, Navarro, and
Nekrich~\cite{MunroNN17} proposed a deterministic solution. Despite
these advances, 20 years after the result of
Hon et al.~\cite{HonSS03}, the bound of $\Omega(n)$ still stands on
the construction of CSAs/CSTs. Given the fundamental role of CSAs and
CSTs, we thus ask:

\vspace{-0.2ex}
\begin{center}
  \emph{Given a text over alphabet $\Sigma = [0 \dd \sigma)$
      represented using $\bigO(n \log \sigma)$ bits,\\ can we
      construct a compressed suffix array/tree of $\T$ in $o(n)$
      time?}
\end{center}

\bfparagraph{Our Results}

We answer the above question affirmatively by describing a new data
structure that takes $\bigO(n\log \sigma)$ bits, supports all operations
of CSA and CST in $\bigO(\log^{\epsilon} n)$ time, and can be
constructed in $\bigO(n\min(1, \log \sigma / \sqrt{\log n}))$ time using
$\bigO(n\log \sigma)$ bits of space (\cref{th:sa,th:st}). Thus, our
solution matches the size
and the query time of~\cite{fm,csa,cst} (as well as more recent
CSTs~\cite{BoucherCGHMNR21,CaceresN22,FischerMN09,Gagie2020,%
PrezzaR21,RussoNO11}) but, unlike those, admits a sublinear-time
construction for small $\sigma$. In particular, we achieve
$\bigO(n / \sqrt{\log n}) = o(n)$ time for $\sigma = 2$, constituting
the first improvement in $n$ since~2003~\cite{HonSS03}.

In addition to a new CSA/CST, we also present a new pattern matching
index. We show (in \cref{th:pm}) how, given a length-$n$ text $\T$
stored using $\bigO(n \log \sigma)$ bits, to construct in $\bigO(n
\min(1, \log \sigma / \sqrt{\log n}))$ time an index of size $\bigO(n
\log \sigma)$ bits that, given the packed representation (i.e., using
$\bigO(m \log \sigma)$ bits) of any pattern $\Pat[1 \dd m]$, counts
the occurrences of $\Pat$ in $\T$ in $\bigO(m / \log_{\sigma} n +
\log^{\epsilon} n)$ time (where $\epsilon >0$ is an arbitrary
predefined constant). The best previous solutions using compact space
(i.e., $\bigO(n \log \sigma)$ bits) achieve $\bigO(n \log \sigma /
\sqrt{\log n})$-time\footnote{Although CSA lets us implement pattern
counting queries, an index implementing pattern counting queries does
not let us implement SA queries; thus, although built in $o(n)$
time,~\cite{MunroNN20b} cannot be used to answer SA queries.}
construction and $\bigO(m / \log_{\sigma} n + \log n \cdot
\log_{\sigma} n)$-time queries~\cite{MunroNN20b}, or $\bigO(n)$-time
construction and $\bigO(m / \log_{\sigma} n + \log_{\sigma}^{\epsilon}
n)$-time queries~\cite{MunroNN20a}. Thus, for the most difficult case
of $\sigma = 2^{\bigO(\sqrt{\log n})}$, our construction subsumes both
these indexes in both aspects.\footnote{Note that if $\log \sigma =
\bigO(\sqrt{\log n})$, then $\log_{\sigma}^{\epsilon}n =
\Theta(\log^{\epsilon'}n)$ holds for $\epsilon' =
\tfrac{\epsilon}{2}$.} Since our pattern-matching index not only
returns the number of occurrences but also the range in $\SA$
containing all suffixes prefixed with $\Pat$, combining the above
result with our CSA yields the structure that can additionally
\emph{report} all occurrences of $\Pat$ in $\bigO(\log^{\epsilon} n)$
time per occurrence (\cref{th:pm-reporting}).

The query times of all our data structures (i.e., both the CSA/CST and
the pattern matching index) are worst-case,
and all our algorithms are deterministic.

Our data structures differ significantly from the CSA of Grossi and
Vitter~\cite{csa}, the FM-index of Ferragina and Manzini~\cite{fm},
and the CST of Sadakane~\cite{cst}, which are based on the so-called
$\Psi$ function~\cite{csa} or the Burrows--Wheeler
transform~\cite{bwt}. We instead rely on the combination of the
recently developed notion of \emph{string synchronizing sets
(SSS)}~\cite{sss} and the new type of queries we call \emph{prefix
rank} and \emph{prefix selection} queries.  Although the prior work on
SSS~\cite{sss,dynsa} laid out its basic properties, it cannot be
turned into an efficient CSA, because it heavily relies on
\emph{orthogonal range counting} queries~\cite{ChanP10}, which are
\emph{provably incapable} of supporting SA queries as fast as
the CSA or FM-index: P\v{a}tra\c{s}cu~\cite{Patrascu07} showed a lower bound
$\Omega(\tfrac{\log n}{\log \log n})$ on the query time of any
structure using near-linear space. On the other hand,
the $\bigO(n \log \sigma / \sqrt{\log n})$-time BWT construction
from~\cite{sss} is not sufficient to obtain an implementation of CSA
since the classical BWT-based CSA, in addition to BWT, requires $\SA$
samples, i.e., a set containing all pairs $(\SA^{-1}[j], j)$ such that $j$
is a multiple of $\log n$, and it is not known how to obtain such a
sequence using prior techniques. The key difficulty is computing the
(global) rank $\SA^{-1}[j]$ of each sampled suffix; an easy application of
sparse suffix sorting gives, in $\bigO(n / \log_{\sigma} n)$ time, the
lexicographic order of the sampled suffixes, but this is insufficient
for placing each sampled suffix among the $n$ suffixes of the original
string.

We sidestep these obstacles and demonstrate that \emph{general}
orthogonal range counting queries~\cite{chazelle,ChanP10} are in fact
not needed at all, and each of their uses can either be: (1)
eliminated completely (see the proof in \cref{sec:sa-periodic-construction}),
(2) replaced with prefix rank/selection queries (see
\cref{sec:prefix-queries}), or (3) improved, utilizing the fact that
the instances arising in our construction have properties that permit
a fast custom solution (see \cref{sec:range-queries}). More details
are provided in the Technical Overview
(\cref{sec:technical-overview}). As a result, we obtain a general set
of reductions for the construction of CSA/CST and pattern-matching
indexes, stated in \cref{th:sa-general,th:pm-general,th:st-general}.
In a single theorem, we can summarize it as follows; note that
our reduction achieves \emph{near-perfect efficiency}, i.e., it incurs
no overhead (compared to the optimal solution) in space,
preprocessing time, and preprocessing space, and only has an extra
$\bigO(\log \log n)$ term in the query time.
Everything else depends entirely on prefix rank and selection queries.

\begin{theorem}[Main result of this paper]\label{th:main}
  Consider a data structure answering prefix rank and selection queries
  (\cref{sec:rksel}) that, for any string of length $m$ over alphabet
  $\Alphabet^\ell$ (or equivalently, a sequence of $m$ length-$\ell$
  strings over alphabet $\Alphabet$), achieves:
  \begin{enumerate}[leftmargin=3.5ex]
  \item Space usage $S(m, \ell, \sigma)$ (measured in $\Theta(\log
    m)$-bit machine words),
  \item Preprocessing time $P_t(m, \ell, \sigma)$,
  \item Preprocessing space $P_s(m, \ell, \sigma)$,
  \item Query time $Q(m, \ell, \sigma)$.
  \end{enumerate}
  For every $\T \in \Alphabet^n$ with $2 \leq \sigma < n^{1/7}$, there
  exist $m=\bigO(n/\log_{\sigma} n)$ and $\ell=\bigO(\log_{\sigma} n)$
  such that, given the packed representation of $\T$, we can in
  $\bigO(n / \log_{\sigma} n + P_t(m, \ell, \sigma))$ time and
  $\bigO(n / \log_{\sigma} n + P_s(m, \ell, \sigma))$ working space
  build a structure of size $\bigO(n/\log_{\sigma} n + S(m, \ell,
  \sigma))$~that:
  \begin{itemize}[leftmargin=2.5ex]
  \item Supports $\SA$ and inverse $\SA$ queries in $\bigO(\log \log n
    + Q(m, \ell, \sigma))$ time;
  \item Supports all suffix tree operations (\cref{tab:st-operations})
    in $\bigO(\log \log n + Q(m, \ell, \sigma))$ time;
  \item Given the packed representation of any pattern $\Pat \in \Alphabet^p$,
    returns:
    \begin{itemize}[leftmargin=1ex]
    \item The range $\SA[b\dd e)$ of suffixes of $\T$ with prefix $P$
      in $\bigO(p / \log_{\sigma} n + \log \log n + Q(m, \ell,
      \sigma))$~time;
    \item All $occ$ starting positions of $P$ in $\T$ in $\bigO(p /
      \log_{\sigma} n + (occ + 1)(\log \log n + Q(m, \ell,
      \sigma)))$~time.
    \end{itemize}
  \end{itemize}
\end{theorem}

Using this general reduction, we obtain the specific tradeoffs
for CSA/CST and pattern matching queries we announced earlier by
plugging in the data structure for prefix rank/selection queries from
\cref{th:wavelet-tree}.

\bfparagraph{Related Work}

In parallel to efforts to improve the complexity of CSA/CST
construction, were the efforts to make it more
practical~\cite{sdsl,GogKKPP19,GogMP17,KarkkainenKP14a,OhlebuschFG10,PrezzaR21}.
This resulted in libraries of compressed data structures such as
\href{https://github.com/simongog/sdsl-lite}{\texttt{sdsl}}~\cite{sdsl},
\href{https://github.com/vigna/sux}{\texttt{sux}}, and
\href{https://github.com/fclaude/libcds}{\texttt{libcds}}.  More
recently, some of these data structures have been extended to the
dynamic setting, e.g., in the
\href{https://github.com/xxsds/DYNAMIC}{\texttt{DYNAMIC}}~\cite{Prezza17}
library.

In addition to CSA/CST and indexes using the optimal space of $\bigO(n
\log \sigma)$ bits, previous work addressed the problem of
designing structures using $\omega(n \log \sigma)$ but still
$o(n \log n)$ bits~\cite{FerraginaM05,Gao0N20,GrossiV05}. We
formulated our main result (\cref{th:main}) as a general reduction so
that techniques from these and similar future studies could be easily
combined with ours, potentially yielding new tradeoffs for
pattern matching and CSA/CST queries.

In recent years, there has also been progress in the query time of
$\bigO(n \log n)$-bit pattern matching indexes.  The suffix trees
support $\bigO(m)$-time pattern search after $\bigO(n)$-time randomized or
$\bigO(n \log\log \sigma)$-time deterministic
construction~\cite{FarachFM00,Ruzic08}.  Fischer and
Gawrychowski~\cite{wexp} achieved $\bigO(m + \log \log \sigma)$-time
queries after $\bigO(n)$-time deterministic construction, improving
upon~\cite{CKL15,mm1993}.  If the pattern is given using $\bigO(m \log
\sigma)$ bits, Bille et al.~\cite{BilleGS17} achieved $\bigO(m /
\log_{\sigma} n + \log m + \log \log \sigma)$ time, which Navarro and
Nekrich~\cite{NavarroN17} improved to~\mbox{$\bigO(m/\log_{\sigma} n +
1)$}.

Surprisingly, the size of some CSAs, CSTs, and compact indexes can be
reduced \emph{below} $n \lceil \log \sigma \rceil$ bits for
statistically compressible texts.  For example, already the original
FM-index~\cite{fm} takes only $\bigO(nH_k(\T)) + o(n \log \sigma)$
bits, where $H_k(\T)$ denotes the \emph{empirical $k$th-order entropy}
of the text $\T$~\cite{cover2006elements}. Currently, the smallest indexes
reach $nH_k(\T) + o(n(H_k(\T) + 1))$ bits~\cite{BarbayCGNN14,
BelazzouguiN15}.  Navarro and M{\"a}kinen~\cite{NavarroM07}, and
Belazzougui and Navarro~\cite{BelazzouguiN14} survey the achievable
tradeoffs for such \emph{fully compressed} indexes. Chan et
al.~\cite{ChanHLS07}, and M{\"a}kinen and Navarro~\cite{MakinenN08}
describe \emph{dynamic} compressed pattern-matching indexes
maintaining a collection of texts supporting insertions/deletions.

Compressed indexes based on LZ77~\cite{LZ77} and run-length
BWT~\cite{bwt} rapidly gain popularity. The early
indexes~\cite{DCC2015, BilleEGV18, BLRSRW15, GagieGKNP12,GagieGKNP14}
support only pattern search and random-access operations. Subsequent
works generalized them to other dictionary
compressors~\cite{ChristiansenEKN21,attractors,navarro201941} and
added dynamism~\cite{Gawrychowski2015, NishimotoDAM}. Support for SA
queries is a recent addition of Gagie et al.~\cite{Gagie2020}. Navarro
surveys these indexes~\cite{NavarroIndexes} and the intricate network
of the underlying compressibility measures~\cite{NavarroMeasures}.
Interestingly, some of these pattern matching indexes can be
constructed in compressed time.  For example, the index
of~\cite{Gawrychowski2015} can be constructed in $\bigO(z\log^3 n)$ time
from the LZ77 representation of $\T$ (with $z$ phrases), and then it
locates pattern occurrences in $\bigO(m + occ \log n)$ time.  On
the other hand, the only compressed index supporting SA queries~\cite{Gagie2020}
is only constructible in $\Omega(n)$ time, but it can be built in
compressed space $\bigO(r\log(n/r))$ given the run-length BWT of $\T$
(with $r$ runs).

\bfparagraph{Organization of the Paper}

After introducing the basic notation and tools in \cref{sec:prelim},
we give a technical overview of the paper in
\cref{sec:technical-overview}.  In \cref{sec:tools}, we then introduce
some auxiliary tools utilized in our data structures. \Cref{sec:sa}
describes our data structure answering $\SA$ and $\SA^{-1}$ queries. In
\cref{sec:pm}, we present our index for counting and reporting
occurrences of patterns given using packed representation. Finally, in
\cref{sec:st}, we extend the functionality of our CSA into that of a
CST.

\section{Preliminaries}\label{sec:prelim}

\begin{wrapfigure}{r}{0.35\textwidth}
  \vspace{-.55cm}
  \begin{tikzpicture}[xscale=0.8,yscale=0.37]
    \foreach \x [count=\i] in {, a, aababa, aababababaababa, aba,
      abaababa, abaababababaababa, ababa, ababaababa, abababaababa,
      ababababaababa, ba, baababa, baababababaababa, baba, babaababa,
      bababaababa, babababaababa}
        \draw (1.9, -\i) node[right]
          {$\texttt{\x\$}$};
    \draw(1.9,0) node[right] {\scriptsize $\T[\SA[i]\dd n]$};
    \foreach \x [count=\i] in {b, b, b, b, b, a, b, b,
                               a, a, a, a, a, a, b, a, a,\$}
      \draw (0.7, -\i) node {\footnotesize $\i$};
    \draw(0.7,0) node{\scriptsize $i$};
    \foreach \x [count=\i] in {18,17,12,3,15,10,1,13,8,
                               6,4,16,11,2,14,9,7,5}
      \draw (1.4, -\i) node {$\x\vphantom{\textbf{\underline{7}}}$};
    \draw(1.4,0) node{\scriptsize $\SA[i]$};
  \end{tikzpicture}

  \vspace{-0.4ex}
  \caption{A list of all sorted suffixes of $\T=
    \texttt{abaababababaababa\$}$ along with
    the suffix array of $\T$.}\label{fig:example}
  \vspace{-.6cm}
\end{wrapfigure}

A \emph{string} is a finite sequence of characters from a given
\emph{alphabet}.  The length of a string $S$ is denoted $|S|$. For
$i\in [1\dd |S|]$,\footnote{ For $i,j\in \Z$, denote $[i\dd j]=\{k
\mkern 1.5mu {\in} \mkern 1.5mu \Z : i \le k \le j\}$, $[i\dd j)=\{k
\mkern 1.5mu {\in} \mkern 1.5mu \Z : i \le k < j\}$, and $(i\dd
j]={\{k \mkern 1.5mu {\in} \mkern 1.5mu \Z: i < k \le j\}}$.  }
the $i$th character of $S$ is denoted $S[i]$.  A~\emph{substring} of
$S$ is a string of the form $S[i\dd j)=S[i]S[i+1]\cdots S[{j-1}]$ for
some $1\le i \le j \le |S|+1$. \emph{Prefixes} and \emph{suffixes} are
substrings of the form $S[1\dd j)$ and $S[i\dd |S|]$, respectively.
We use $\revstr{S}$ to denote the \emph{reverse} of $S$, i.e.,
$S[|S|]\cdots S[2]S[1]$.  We denote the \emph{concatenation} of two
strings $U$ and $V$, that is, $U[1]\cdots U[|U|]V[1]\cdots V[|V|]$, by
$UV$ or $U\cdot V$.  Furthermore, $S^k=\bigodot_{i=1}^k S$ is the
concatenation of $k \in \Zz$ copies of $S$; note that
$S^0=\emptystring$ is the \emph{empty string}. For a non-empty string $S \in
\Sigma^{+}$, we define the special infinite string $S^{\infty}$ such that
$S^{\infty}[i] = S[1 + (i-1) \bmod |S|]$ holds for every $i \in \Z$;
in particular, $S^{\infty}[1 \dd |S|] = S[1 \dd |S|]$. An
integer $p\in [1\dd |S|]$ is a \emph{period} of $S$ if $S[i] = S[i +
p]$ holds for every $i \in [1 \dd |S|-p]$. We denote the shortest
period of $S$ as $\per(S)$.

Throughout the paper, we consider a string (called the \emph{text})
$\T$ of length $n\geq 2$ over an integer alphabet $\Sigma = [0 \dd
\sigma)$, where $\sigma = n^{\bigO(1)}$.  We assume $\T[n] = 0$, and
that $0$ (also denoted with $\texttt{\$}$) does not appear elsewhere
in $\T$. We use $\preceq$ to denote the order on $\Sigma$, extended to
the \emph{lexicographic} order on $\Sigma^*$ (the set of strings over
$\Sigma$) so that $U,V\in \Sigma^*$ satisfy $U\preceq V$ if and only
if either $U$ is a prefix of $V$, or $U[1\dd i)=V[1\dd i)$ and
$U[i]\prec V[i]$ holds for some $i\in [1\dd \min(|U|,|V|)]$.  The
\emph{suffix array} $\SA[1\dd n]$ of $\T$ is a permutation of $[1\dd
n]$ such that $\T[\SA[1]\dd n] \prec \T[\SA[2]\dd n] \prec \cdots
\prec \T[\SA[n]\dd n]$, i.e., $\SA[i]$ is the starting position of the
lexicographically $i$th suffix of $\T$; see \cref{fig:example} for an
example.  The \emph{inverse suffix array} $\ISA[1 \dd n]$ (also denoted $\SA^{-1}[1\dd n]$) is the
inverse permutation of $\SA$, i.e., $\ISA[j] = i$ holds if and only if
$\SA[i] = j$. Intuitively, $\ISA[j]$ stores the lexicographic \emph{rank} of a
suffix $\T[j \dd n]$ among the suffixes of $\T$.  By $\lcp(U,V)$ we denote
the length of the longest common prefix of $U$ and $V$. For $j_1, j_2
\in [1 \dd n]$, we let $\LCE(j_1,j_2) = \lcp(\T[j_1 \dd ], \T[j_2 \dd
])$.  For any $\Pat, S \in \Sigma^{*}$, we let\vspace{-.5cm}
\begin{align*}
\Occ(\Pat, S) &= \{j
\in [1 \dd |S|] : j + |\Pat| \leq |S| + 1\text{ and }S[j \dd j {+}
|\Pat|) = \Pat\},\\ \LB(\Pat, S) &= |\{i \in [1 \dd |S|] : S[i \dd
|S|] \prec \Pat\}|,\\\UB(\Pat, S) &= \LB(\Pat, S) +
|\Occ(\Pat, S)|.\end{align*}  Observe that the following equality holds for every $\Pat \in \Sigma^{*}$:
\[\Occ(\Pat, \T) = \{\SA[i] : i \in (\LB(\Pat, \T) \dd
\UB(\Pat, \T)]\}.\]

We use the word RAM model of computation~\cite{Hagerup98} with $w$-bit
\emph{machine words}, where $w \ge \log n$.  In this model, strings
are typically represented as arrays, with each character occupying
one memory cell.  A single character, however, only needs $\lceil
\log \sigma \rceil$ bits, which might be much less than~$w$.  We can
therefore store (the \emph{packed representation} of) a text $\T\in
[0\dd \sigma)^n$ using $\bigO\big(\big\lceil{\frac{n\log
\sigma}{w}}\big\rceil\big)$ memory cells.

\subsection{(Prefix) Rank and Selection Queries}\label{sec:rksel}

Let us recall the (ordinary) rank and selection queries on a string
$S\in \Sigma^n$:
\begin{description}[style=sameline,itemsep=1ex]
\item[Rank query $\rank{S}{a}{j}$:] Given $a\in \Sigma$ and $j\in
  [0\dd n]$, compute $|\{i\in [1\dd j]: S[i]=a\}|$.
\item[Selection query $\select{S}{a}{r}$:] Given $a\in \Sigma$ and
  $r\in [1\dd \rank{S}{a}{n}]$, find the $r$th smallest element of
  $\{i\in [1\dd n] : S[i]=a\}$.
\end{description}

\begin{theorem}[Rank and selection queries in
    bitvectors~\cite{WaveletSuffixTree,Clark98,Jac89,MunroNV16}]\label{th:binrksel}
  For every string $S\in \{0,1\}^*$, there exists a data structure of
  $\bigO(|S|)$ bits answering rank and selection queries in $\bigO(1)$
  time.  Moreover, given the packed representations of $m$ binary
  strings of total length $n$, the data structures for all these
  strings can be constructed in $\bigO(m + n/\log n)$ time.
\end{theorem}

Next, we provide a generalization of rank and selection queries
specific to sequences of strings (strings whose characters are strings
themselves).  Let $W\in (\Sigma^*)^m$ be a sequence of $m$ strings.
\begin{description}[style=sameline,itemsep=1ex]
\item[Prefix rank query $\rank{W}{X}{j}$:] Given $X\in \Sigma^*$ and
  $j\in [0\dd m]$, compute $|\{i\in [1\dd j]: X\text{ is a prefix of
  }W[i]\}|$.
\item[Prefix selection query $\select{W}{X}{r}$:] Given $X\in
  \Sigma^*$ and $r\in [1\dd \rank{W}{X}{m}]$, find the $r$th smallest
  element of $\{i\in [1\dd m] : X\text{ is a prefix of }W[i]\}$.
\end{description}

The following result, proved in \cref{sec:prefix-queries} by building
on the results of Belazzougui and Puglisi~\cite{Belazzougui2016},
provides an efficient implementation of prefix rank and selection
queries.  Note that we require $W$ to consist of same-length strings
over an integer alphabet.

\begin{theorem}\label{th:wavelet-tree}
  For all integers $m,\ell,\sigma\in \Z_{\ge 1}$ satisfying $m\ge
  \sigma^\ell\ge 2$, every constant $\epsilon > 0$, and every string
  $W\in ([0\dd \sigma)^\ell)^{\le m}$, there exists a data structure
  of size $\bigO(m)$ answering prefix rank queries in
  $\bigO(\ell^{\epsilon/2}\log \log m)= \bigO(\log^\epsilon m)$ time and
  prefix selection queries in $\bigO(\ell^{\epsilon/2}) =
  \bigO(\log^\epsilon m)$ time.  Moreover, it can be constructed in
  $\bigO(m\min(\ell,\sqrt{\log m}))$ time using $\bigO(m)$ working space
  given the packed representation of $W$ and the constant parameter
  $\epsilon > 0$.
\end{theorem}

\subsection{Range Counting and Selection}

\newcommand{\rcount}[3]{\mathsf{rcount}_{#1}(#2,#3)}
\newcommand{\rselect}[3]{\mathsf{rselect}_{#1}(#2,#3)}

Let $A[1 \dd m]$ be an array of
nonnegative integers. We define the following queries on $A$:
\begin{description}[style=sameline,itemsep=1ex]
\item[Range counting query $\rcount{A}{v}{j}$:] Given an integer $v
  \geq 0$ and a position $j \in [0 \dd m]$, compute $|\{i \in [1 \dd
  j]: A[i] \geq v\}|$.
\item[Range selection query $\rselect{A}{v}{r}$:] Given integers $v
  \geq 0$ and $r \in [1 \dd \rcount{A}{v}{m}]$, find the $r$th smallest
  element of $\{i \in [1 \dd m] : A[i] \geq v\}$.
\end{description}

The currently fastest general-purpose data structure for range
counting/selection queries is described in~\cite[Theorems~2.3 and~3.3]{ChanP10}. The instances in our construction, however,
satisfy an additional property, namely, that the sum $\sum_{i=1}^{m}
A[i]$ is bounded. This lets us obtain a solution with faster queries and
smaller construction time; see \cref{sec:range-queries}.

\begin{restatable}{proposition}{prrangequeries}\label{pr:range-queries}
  An array $A[1{\dd}m']$ of $m'\,{\in}\,[2{\dd}m]$ nonnegative
  integers satisfying $\sum_{i=1}^{m'}A[i] = \bigO(m \log m)$ can
  be preprocessed in $\bigO(m)$ time so that range counting and
  selection queries can be answered in $\bigO(\log \log m)$ time
  and $\bigO(1)$ time, respectively.
\end{restatable}

\subsection{String Synchronizing Sets}

\begin{definition}[$\tau$-synchronizing set~\cite{sss}]\label{def:sss}
  Let $\T\in \Sigma^n$ be a string and let $\tau \in [1\dd
  \lfloor\frac{n}{2}\rfloor]$ be a parameter. A set $\S \subseteq [1
  \dd n - 2\tau + 1]$ is called a \emph{$\tau$-synchronizing set} of
  $\T$ if it satisfies the following \emph{consistency} and
  \emph{density} conditions:
  \begin{enumerate}
  \item If $\T[i \dd i + 2\tau) = \T[j\dd j + 2\tau)$, then $i \in \S$
    holds if and only if $j \in \S$ (for $i, j \in [1 \dd n - 2\tau
    + 1]$),
  \item $\S\cap[i \dd i + \tau)=\emptyset$ if and only if
    $i \in \R(\tau, \T)$ (for $i \in [1 \dd n - 3\tau + 2]$), where
    \[
      \R(\tau, \T) := \{i \in [1 \dd |\T| - 3\tau + 2] :
      \per(\T[i \dd i + 3\tau - 2]) \leq \tfrac13\tau\}.
    \]
  \end{enumerate}
\end{definition}

In most applications, we want to minimize $|\S|$. Note, however, that
the density condition imposes a lower bound
$|\S|=\Omega(\frac{n}{\tau})$ for strings of length $n\ge 3\tau-1$
that do not contain substrings of length $3\tau-1$ which are periodic
with period $\leq \frac13\tau$.  Thus, we cannot hope to achieve an
upper bound improving in the worst case upon the following one.

\begin{theorem}[{\cite[Proposition~8.10]
      {sss}}]\label{th:sss-existence-and-construction}
  For any string $\T$ of length $n$ and parameter $\tau \in [1\dd
  \lfloor\frac{n}{2}\rfloor]$, there exists a $\tau$-synchronizing
  set $\S$ of size $|\S| = \bigO \left( \frac{n}{\tau}
  \right)$. Moreover, if $\T \in \Alphabet^n$, where $\sigma =
  n^{\bigO(1)}$, such $\S$ can be deterministically constructed in
  $\bigO(n)$ time.
\end{theorem}

Note that when $\tau = \omega(1) \cap \bigO(\log_{\sigma} n)$ and
$\T \in \Alphabet^{n}$ is given in the packed representation, the
first part of \cref{th:sss-existence-and-construction} opens the
possibility of an algorithm
running in $\bigO(\tfrac{n}{\tau}) = o(n)$ time.  In~\cite{sss}, it
was shown that this lower bound is achievable (the upper bound $\tau =
\bigO(\log_{\sigma} n)$ follows from the fact that every algorithm needs
to at least read the input, which takes $\Theta(n / \log_{\sigma} n)$
time; thus, for larger $\tau$, the algorithm cannot run in
$\bigO(\tfrac{n}{\tau})$ time).

\begin{theorem}[{\cite[Theorem~8.11]
      {sss}}]\label{th:sss-packed-construction}
  For every constant $\mu<\frac{1}{5}$, given the packed
  representation of a text $\T \in \Alphabet^n$ and a positive integer
  $\tau \le \mu \log_\sigma n$, one can deterministically construct
  in $\bigO(\frac{n}{\tau})$ time a $\tau$-synchronizing set of size
  $\bigO(\frac{n}{\tau})$.
\end{theorem}

\section{Technical Overview}\label{sec:technical-overview}

In this section, we give an overview of our data structures to answer
$\SA$ and $\ISA$ queries (\cref{sec:to-sa}), pattern matching
queries (\cref{sec:to-pm}), and suffix tree queries
(\cref{sec:to-st}). Each subsection contains a summary of the key new
techniques.

\subsection{SA and ISA Queries}\label{sec:to-sa}

Let $\epsilon \in (0, 1)$ and $\T \in \Alphabet^n$, where $2 \leq
\sigma < n^{1/7}$.  In this section, we give an overview of our data
structure to compute the value of $\SA[i]$ (resp.\ $\ISA[j]$) given
any $i \in [1 \dd n]$ (resp.\ $j \in [1 \dd n]$) in
$\bigO(\log^{\epsilon} n)$ time. The data structure uses $\bigO(n /
\log_{\sigma} n)$ space.  We assume $\sigma < n^{1/7}$ since for
larger $\sigma$ the plain representations of $\SA$ and $\ISA$ use
$\bigO(n \log n) = \bigO(n \log \sigma)$ bits and can be constructed
in $\bigO(n)$ time~\cite{KarkkainenSB06}.

Let $\tau = \lfloor \mu \log_{\sigma} n \rfloor$, where $\mu < \frac16$
is a positive constant chosen so that $\tau \geq 1$ (such
$\mu$ exists by $\sigma < n^{1/7}$).  We use $\R$ as a shorthand for
$\R(\tau, \T)$ (see \cref{def:sss}). Our data structure to compute
$\SA[i]$ (resp.\ $\ISA[j]$) works differently depending on whether
$\SA[i] \in \R$ (resp.\ $j \in \R$). To check if $\SA[i] \in \R$, we
store a bitvector $\BVshort$ marking boundaries between the blocks of
suffixes in $\SA$ sharing the length-$(3\tau{-}1)$ prefix.  We also
store the sequence $\ARRshort$ of those prefixes (by $\mu < \frac16$,
it needs $\bigO(\sigma^{3\tau - 1}) = \bigO(n^{3\mu}) = o(n /
\log_{\sigma} n)$ space).  Given any $i \in [1 \dd n]$, we can then
check if $\SA[i] \in \R$ by first computing the block $k$ containing
position $i$ using a rank query on $\BVshort$, and then checking
(using a lookup table) if $X = \ARRshort[k]$ satisfies $\per(X)
\leq \tfrac{1}{3}\tau$. As for an $\ISA$ query, checking if $j \in \R$
only needs the lookup table (since we store $\T$).

\bfparagraph{The Nonperiodic Positions}

We first focus on computing $\ISA[j]$. Let $j \not\in \R$ and let $\S$
be a $\tau$-synchronizing set of $\T$ of size $n' := |\S| = \bigO(n /
\tau)$ (such $\S$ exists and can be quickly constructed using
\cref{th:sss-packed-construction}). The query algorithm relies on the
following two observations:

\begin{description}[style=sameline,itemsep=1ex,font={\normalfont\itshape}]
\item[Observation 1: $\S$ induces a partitioning of $\SA$ into blocks.]
  The density condition of $\S$ implies $\S \cap [j \dd j + \tau) \neq
  \emptyset$, i.e., the successor of $j$ in $\S$, denoted $s:=
  \Successor_{\S}(j)$, satisfies $s<j+\tau$.  Hence, the string $X :=
  \T[j \dd s + 2\tau)$, called the \emph{distinguishing prefix} of
  $\T[j \dd n]$, is of length $|X|\le 3\tau-1$. By the local
  consistency of $\S$, the set $\D$ of distinguishing prefixes of all
  suffixes $\T[j\dd n]$ with $j\notin \R$ is prefix-free (i.e., no
  string in $\D$ is a prefix of another).  All positions in $[1 \dd n]
  \setminus \R$ in the $\SA$ of $\T$ can thus be partitioned into
  disjoint blocks according to distinguishing prefixes. Since $\D
  \subseteq \Alphabet^{\leq 3\tau - 1}$, the number of blocks is
  $\bigO(\sigma^{3\tau - 1})$, so we can store their boundaries in a
  lookup table of size $\bigO(\sigma^{3\tau - 1}) = \bigO(n^{3\mu}) =
  o(n / \log_{\sigma} n)$.  To efficiently determine
  $\Successor_{\S}(j)$, we store a bitvector marking positions in $\S$,
  augmented with $\bigO(1)$-time rank and select queries. Once the
  block $\SA(b \dd e]$ for $X = \T[j \dd \Successor_{\S}(j) + 2\tau)$
  is found, it remains to locate $j$ within that block.
\item[Observation 2: The order in each block is consistent with $\S$.]
  Assume $\SA(b \dd e]$ represents all suffixes of $\T$ having $X =
  \T[j \dd \Successor_{\S}(j) + 2\tau)$ as a prefix. By the
  consistency condition, letting $\deltatext = |X| - 2\tau$, for every
  $i \in (b \dd e]$, we have $\Successor_{\S}(\SA[i]) = \SA[i] +
  \deltatext$.  Thus, letting $(\slex_i)_{i \in [1 \dd n']}$ contain
  $\S$ sorted by the corresponding suffixes $\T[\slex_i\dd n]$,
  positions in $\SA(b \dd e]$ increased by $\deltatext$ occur in
  $(\slex_i)_{i \in [1 \dd n']}$ in the same relative order. Hence, if
  we define $W[i] = \revstr{X_i}$, where $X_i = \T^{\infty}[\slex_i -
  \tau \dd \slex_i + 2\tau)$, and select $W[y]$ as the $k$th string in
  $W$ having $\revstr{X}$ as a prefix, then $\slex_y -
  \deltatext=\SA[b+k]$ is the $k$th position in $\SA(b \dd e]$. Thus,
  to obtain $\ISA[j]$, it suffices to find the index $y$ such that
  $\slex_y = \Successor_{\S}(j)$. Then, the offset of $j$ in the block
  $\SA(b \dd e]$ is $\rank{W}{\revstr{X}}{y}$. To efficiently
  determine $y$, we store the permutation that maps elements of $\S$
  sorted left-to-right to elements of $(\slex_i)_{i \in [1 \dd n']}$.
  This lets us determine $y$ in $\bigO(1)$ time. We then compute
  $\rank{W}{\revstr{X}}{y}$ in $\bigO(\log^{\epsilon} n)$ time using
  \cref{th:wavelet-tree}; see \cref{pr:sa-nonperiodic-isa}.
\end{description}

Let us now turn to the computation of $\SA[i]$ when $\SA[i] \not \in
\R$. Using a rank query on $\BVshort$ and an access to $\ARRshort$, we
first determine the length-$(3\tau{-}1)$ prefix of $\T[\SA[i] \dd n]$.
A lookup table lets us retrieve the prefix $X \in \D$ of $\T[\SA[i]
\dd n]$ and the boundaries of the corresponding block $\SA(b \dd e]$.
By Observation 2 above, it remains to determine the index $y$ of the
$(i - b)$th leftmost string in $W$ having $\revstr{X}$ as a prefix,
which we compute in $\bigO(\log^{\epsilon} n)$ time as $y =
\select{W}{\revstr{X}}{i-b}$ using \cref{th:wavelet-tree}. We then
have $\SA[i] = \slex_y - \deltatext$, where $\deltatext = |X| - 
2\tau$. See \cref{pr:sa-nonperiodic-sa} for details.

\bfparagraph{The Periodic Positions}

Let $j \in \R$. We again first focus on computing $\ISA[j]$. As
before, we aim to find the location of $j$ in the block $\SA(b \dd e]$
containing all suffixes of $\T$ prefixed with $X = \T[j \dd j + 3\tau
- 1)$. Note that all positions in $\SA(b \dd e]$ are in $\R$. The
challenge is thus to devise a way to compare suffixes starting in
$\R$. The problem with applying a similar approach as before is that
the size of $\R$ can reach $\Theta(n)$.  There exists, however, a subset of
$\R$ that, when combined with a bitvector representing remaining
positions in $\R$, can be applied here.  We derive it as follows:

\begin{description}[style=sameline,itemsep=1ex,font={\normalfont\itshape}]
\item[Structure of $\R$ in the left-to-right (text) order:] The gap between $|X| =
  3\tau-1$ and $\per(X) \leq \frac13 \tau$ ensures that every maximal
  block of positions in $\R$ corresponds to a \emph{$\tau$-run}, i.e.,
  a maximal substring of $\T$ of length $\geq 3\tau - 1$ whose
  shortest period is $\leq \frac13\tau$ (\cref{lm:run-end}).  Since
  any two $\tau$-runs overlap by at most $\frac23 \tau$ positions
  (see the proof of \cref{lm:gap}), their number is $\bigO(n / \tau)$.
  We can thus succinctly encode $\R$ by storing the set $\R'$ of
  $\tau$-run starting positions.
\item[Structure of $\R$ in the lexicographic order:] For $x \in \R$, let $\rend{x}$
  denote the position following the $\tau$-run containing $x$.
  Observe that, for every $x \in \R$, we can uniquely write $\T[x \dd
  \rend{x}) = H'H^kH''$, where $H$ is the lexicographically smallest
  rotation of $\T[x \dd x + p)$, $p = \per(\T[x \dd \rend{x}))$,
  and $H'$ (resp.\ $H''$) is a proper suffix (resp.\ prefix) of $H$
  (\cref{sec:sa-periodic-prelim}). Denote $\Lroot(x) = H$, $\Lhead(x)
  = |H'|$, $\Lexp(x) = k$, and $\Ltail(x) = |H''|$. Let also $\type(x)
  = -1$ if $\T[\rend{x}] \prec \T[\rend{x} - |H|]$ and $\type(x) = +1$
  otherwise. Then, in $\SA(b \dd e]$, all positions $x$ with $\type(x) = -1$
  precede all $x$ with $\type(x) = +1$. Moreover, the value of
  $\rend{x} - x$ is non-decreasing (resp.\ non-increasing) among the positions $x$
  with $\type(x) = -1$ (resp.\ $\type(x) = +1$); see \cref{lm:lce}.
\end{description}

Denote $\R^{-} = \{x \in \R: \type(x) = -1\}$,
$\R'^{-} = \R' \cap \R^{-}$, $\R_H = \{x \in \R : \Lroot(x) = H\}$,
$\R'^{-}_H = \R'^{-} \cap \R_H$, $\R_{s,H} = \{x \in \R_H : \Lhead(x)
= s\}$, and $\R^{-}_{s,H} = \R^{-} \cap \R_{s,H}$.
Assume $\type(j) = -1$ (the case of $\type(j) = +1$ is symmetric),
$\Lhead(j) = s$, and $\Lroot(j) = H$.  Given the above structural
insights, we can phrase locating $j$ in $\SA(b \dd e]$ as counting the positions $x
\in \R^{-}_{s,H}$ satisfying $\T[x \dd n] \preceq \T[j \dd n]$. By the
analysis above, all such positions satisfy $\rend{x} - x \leq \rend{j} - j$ and hence $\Lexp(x) \leq
\Lexp(j)$. We first compute the size of $\Posa(j) := \{x \in
\R^{-}_{s,H} : \Lexp(x) \leq \Lexp(j)\}$ and then subtract the size of
$\Poss(j) := \{x \in \R^{-}_{s,H} : \Lexp(x) = \Lexp(j)\text{ and
}\T[x \dd n] \succ \T[j \dd n]\}$ as follows:\footnote{Note that
here we first overestimate the number of smaller suffixes in
$\SA(b \dd e]$ and then subtract the larger suffixes. We explain
the reason for this counterintuitive approach in \cref{rm:sa-periodic}.}
\begin{itemize}
\item Since $|\Posa(j)|$ only depends on $\Lexp(j)$, it suffices to
  store a bitvector $\BVexp$ marking the boundaries in $\SA$ between
  blocks of positions with subsequent values of $\Lexp$. Computing
  $|\Posa(j)|$ then reduces to $\bigO(1)$-time rank and selection
  queries on $\BVexp$ (\cref{pr:isa-delta-a}).
\item As for $|\Poss(j)|$, we observe that $\Poss(j)$ contains at most
  one position in each $\tau$-run (\cref{lm:isa-delta-s}). We store
  all $x \in \R'^{-}_H$ sorted by the suffix starting right after the
  last occurrence of $H$, i.e., $\T[\rendfull{x} \dd n]$, where
  $\rendfull{x} := \rend{x} - \Ltail(x)$.  In a separate array, we also
  record $\rendfull{x} - x$ at the corresponding position. This lets
  us compute $|\Poss(j)|$ by first locating the block of positions $x
  \in \R'^{-}_H$ for which $\T[\rendfull{x} \dd n] \succ
  \T[\rendfull{j} \dd n]$, and then counting the ones with
  $\rendfull{x} - x \geq \rendfull{j} - j$ (\cref{pr:isa-delta-s}).
  Since the sum of $\rendfull{x} - x$ over all $x \in \R'$ is
  $\bigO(n)$ (\cref{sec:sa-periodic-ds}), we use a specialized
  structure for range counting (\cref{pr:range-queries}), bypassing
  the general $\Omega(\tfrac{\log n}{\log \log n})$-time lower
  bound~\cite{Patrascu07}.
\end{itemize}

Let us now turn to an $\SA[i]$ query with $\SA[i] \in \R^{-}$. First,
we determine $\T[\SA[i]\dd \SA[i]+3\tau-1)$ using $\BVshort$ and
$\ARRshort$, as well as $s = \Lhead(\SA[i])$ and $H = \Lroot(\SA[i])$
using a lookup table (\cref{pr:sa-root}).  Then, rank and selection 
queries on $\BVexp$ let us easily determine $\Lexp(\SA[i])$ and
$|\Poss(\SA[i])|$ (\cref{pr:sa-delta-a}).  As explained above, to
compute $\SA[i]$, it remains to first select the $k$th (where $k
= |\Poss(\SA[i])| + 1$) largest element $j \in \R'^{-}_H$ according to
the string $\T[\rendfull{j} \dd n]$, among positions $j'' \in
\R'^{-}_{H}$ satisfying $\rendfull{j''} - j'' \geq \Lhead(\SA[i]) +
\Lexp(\SA[i])\cdot |H|$.  The position $j' \in [j \dd \rend{j} - 3\tau + 2)$
with $\Lexp(j') = \Lexp(\SA[i])$ and $\Lhead(j') =
\Lhead(\SA[i])$ must then satisfy $\SA[i] = j'$. To compute $j$, we
use our specialized data structure for range queries
(\cref{pr:range-queries}).  Position $j'$ is then obtained by
subtracting $\Lhead(\SA[i]) + \Lexp(\SA[i])\cdot |H|$ from $\rendfull{j}$
(\cref{pr:sa-delta-s}).

In total, the query time for periodic positions is $\bigO(\log \log n)$
(\cref{pr:sa-periodic-isa,pr:sa-periodic-sa}).

\bfparagraph{Summary of New Techniques}

The key distinctive feature of our technique is the use of local
consistency without general orthogonal range queries, present in prior
approaches~\cite{ChristiansenEKN21, sss,resolution,dynsa}. This lets us
sidestep P\v{a}tra\c{s}cu's $\Omega(\tfrac{\log n}{\log \log n})$
lower bound~\cite{Patrascu07}, which is achieved in three steps:
\begin{itemize}
\item We replace range counting/selection in the nonperiodic case
  with \emph{prefix rank} and \emph{prefix selection}, for which we
  propose a new tradeoff by plugging in the technique of Belazzougui
  and Puglisi~\cite{Belazzougui2016}. This reveals the power of our
  reduction: we achieve a non-trivial tradeoff for complex queries by
  solving a simple bit-permuting problem (note that the tradeoff behind
  \cref{th:wavelet-tree}, e.g., occupies only 1.5 pages in 
  \cref{sec:prefix-queries}; the bulk of our paper is the reduction).
\item We replace range counting/selection in the periodic case
  by observing that the sum of coordinates is small
  (\cref{sec:sa-periodic-ds}). This lets us use a specialized
  solution (\cref{sec:range-queries}).
\item The above two cases occur at query time. The third case concerns
  the construction of the structure. More precisely, we completely
  eliminate range queries naturally occurring during the construction
  of components for periodic positions~\cite{sss,resolution} by a
  complex bit-optimal algorithm for the construction of the bitvector
  $\BVexp$ (see \cref{sec:sa-periodic-construction}).
\end{itemize}
As a result, we obtain a very general reduction stated in
\cref{th:sa-general}.

\subsection{Pattern Matching Queries}\label{sec:to-pm}

Let $\epsilon \in (0, 1)$ and $\T \in \Alphabet^n$ be as in
\cref{sec:to-sa}. We now give an overview of our data structure that,
given a packed representation of any pattern $\Pat \in \Alphabet^m$,
returns the pair of indexes $(b, e) = (\LB(\Pat, \T), \UB(\Pat, \T))$,
i.e., the boundaries of the $\SA$ block containing all suffixes having
$\Pat$ as a prefix, in $\bigO(m / \log_{\sigma} n + \log^{\epsilon}
n)$ time.  Note that having this range immediately gives us
$|\Occ(\Pat, \T)| = e - b$, i.e., it implements \emph{pattern
counting}. Moreover, combined with the result from \cref{sec:to-sa},
we obtain \emph{pattern reporting}, i.e., we can enumerate $\Occ(\Pat, \T)$
in $\bigO(m / \log_{\sigma} n + (|\Occ(\Pat, \T)| + 1) \log^{\epsilon}
n)$ time.

Let $\tau$ be as in \cref{sec:to-sa}. Our structure to compute
$(\LB(\Pat, \T), \UB(\Pat, \T))$ works differently depending on
whether $m \geq 3\tau - 1$ and whether $\per(\Pat[1 \dd 3\tau - 1]) \leq
\tfrac13\tau$ (such $\Pat$ is called \emph{periodic}) or not.
Checking if $\Pat$ is periodic is easily implemented via a lookup
table.

\bfparagraph{The Nonperiodic Patterns}

Let us assume $m \geq 3\tau - 1$ (shorter patterns are handled using a
precomputed array) and let $\S$ be as in \cref{sec:to-sa}. The basic
idea of the pattern matching query is to decompose $\Pat = XY$,
where $X \in \D$, and then utilize the following observation about
$\S$ (generalizing the second observation in \cref{sec:to-sa}):

\begin{description}[style=sameline,itemsep=1ex,font={\normalfont\itshape}]
\item[Observation: The order in the suffix array range corresponding to
  $\Occ(\Pat, \T)$ is consistent with $\S$.]  Let $(b, e)=
  (\LB(\Pat, \T), \UB(\Pat, \T))$. By the
  consistency condition, for every $i \in (b \dd e]$, we have
  $\Successor_{\S}(\SA[i]) = \SA[i] + \deltatext$, where $\deltatext =
  |X| - 2\tau$. Thus, letting $(\slex_i)_{i \in [1 \dd n']}$ be
  defined as in \cref{sec:to-sa}, the positions in $\SA(b \dd e]$
  increased by $\deltatext$ occur in $(\slex_i)_{i \in [1 \dd n']}$ in
  the same relative order. Therefore, to compute $|\Occ(\Pat, \T)|$,
  it suffices to first locate a range of $(\slex_i)_{i \in [1 \dd
  n']}$ consisting of positions followed by $\Pat(\deltatext \dd m]$
  in $\T$, and then count those which are additionally preceded with
  $X[1 \dd \deltatext]$. The first goal is implemented in
  $\bigO(m/\log_{\sigma} n + \log \log n)$ time via a compact trie
  over the set of strings $\{\T[\slex_i \dd n]\}_{i \in [1 \dd n']}$,
  reinterpreted as strings over alphabet of size $n^{\Theta(1)}$. The
  second step reduces to a prefix rank query over $W[1 \dd n']$
  (\cref{sec:to-sa}). This approach easily generalizes to return $(b, e)$
  instead of $|\Occ(\Pat, \T)|$; see \cref{lm:pm-nonperiodic}.
\end{description}

\bfparagraph{The Periodic Patterns}

Let us now assume that $m \geq 3\tau - 1$ and $\per(\Pat[1 \dd 3\tau -
1]) \leq \frac13\tau$. We first generalize the notion of $\Lroot(x)$,
$\rend{x}$, and all other functions from positions to strings
(\cref{sec:pm-periodic-prelim}). The main idea is to decompose $\Pat$
into the periodic prefix $\Pat[1 \dd \rend{\Pat})$ and the remaining
suffix $\Pat[\rend{\Pat} \dd |\Pat|]$. Let us consider the harder case
when $\rend{\Pat} = |\Pat| + 1$ (see \cref{lm:occ-a}).  We define
$\Occa(\Pat, \T) = \{j \in \R_{s,H} \cap \Occ(\Pat, \T) : \Lexp(j) >
\Lexp(\Pat)\}$ and $\Occs(\Pat, \T) = \{j \in \R_{s,H} \cap \Occ(\Pat,
\T) : \Lexp(j) = \Lexp(\Pat)\}$, where $H = \Lroot(\Pat)$ and $s =
\Lhead(\Pat)$.  The value $|\Occ(\Pat, \T)|$ is determined in two
steps:
\begin{itemize}
\item First, we compute the size of $\Occam(\Pat, \T) := \Occa(\Pat,
  \T) \cap \R^{-}$ (the size of $\Occap(\Pat, \T)$ is computed
  symmetrically) as in \cref{sec:to-sa} by utilizing rank and
  selection queries on $\BVexp$ (\cref{pr:pm-occ-a}). This only
  requires knowing $\Lexp(\Pat)$, which can be retrieved in $\bigO(1 +
  m / \log_{\sigma} n)$ time (\cref{pr:pm-root}).
\item Next, we compute the size of $\Occsm(\Pat, \T) := \Occs(\Pat,
  \T) \cap \R^{-}$.  We first show that $\Occsm(\Pat, \T)$ contains at
  most one position in every $\tau$-run (\cref{lm:occ-s}), i.e., an
  analogue of \cref{lm:isa-delta-s}. The computation is similar as in
  \cref{sec:to-sa}, except that we use a trie over meta-symbols to
  find the range of positions $x \in \R'^{-}_{H}$ with
  $\T[\rendfull{x}\dd n]$ prefixed by $\Pat[\rendfull{\Pat} \dd m]$.
  We then perform an $\bigO(\log \log n)$-time range query
  (\cref{pr:pm-occ-s}).  A small complication is to separate positions $x \in
  \R'^{-}$ with different $\Lroot(x)$ in the trie.  For this, we
  insert into the trie suffixes starting slightly earlier than
  $\rendfull{x}$; see \cref{sec:pm-periodic-ds}.
\end{itemize}

The above algorithm generalizes to the computation of $(b, e)$, rather
than $|\Occ(\Pat, \T)|$, except that handling ``fully periodic''
patterns (with $\rend{\Pat} = |\Pat| + 1$) requires some care (see \cref{rm:pm-periodic}).

\bfparagraph{Summary of New Techniques}

Our key technical contributions are as follows:
\begin{itemize}
\item We show how to directly apply string synchronizing
  sets~\cite{sss} to the problem of pattern matching (not by simply
  using $\SA$/$\ISA$ queries) and consequently obtain a very efficient
  reduction from pattern matching to prefix rank and prefix selection
  queries.
\item To achieve this, we prove several new combinatorial results for
  periodic patterns (\cref{sec:pm-periodic-prelim}), and then show for
  efficiently apply them
  (\cref{sec:pm-periodic-ds,sec:pm-periodic-nav,sec:pm-periodic-pm}).
\item As a result, we obtain the first optimal-size pattern matching
  index that is constructible in $o(n)$ time and supports:
  \begin{itemize}
  \item pattern occurrence counting in $\bigO(m / \log_{\sigma} n +
    \log^{\epsilon} n)$ time, and
  \item pattern occurrence reporting in $\bigO(m / \log_{\sigma} n +
    (|\Occ(\Pat, \T)| + 1) \log^{\epsilon} n)$ time.
  \end{itemize}
  This improves over~\cite{MunroNN20a,MunroNN20b} in either
  construction or query time and, perhaps more importantly, provides a
  very general reduction that enables achieving further time-space
  tradeoffs much easier (\cref{th:pm-general}).
\end{itemize}

\subsection{Suffix Tree Queries}\label{sec:to-st}

Let $\epsilon \in (0, 1)$ and $\T \in \Alphabet^n$ be as in
\cref{sec:to-sa}. In this section, we outline how to extend the
techniques presented in \cref{sec:to-sa,sec:to-pm} to obtain the full
suffix tree functionality in optimal $\bigO(n / \log_{\sigma} n)$
space. All operations (see \cref{tab:st-operations}) are supported in
$\bigO(\log^{\epsilon} n)$ time, and the data structure can be
constructed in $\bigO(n \min(1, \log \sigma / \sqrt{\log n}))$ time
and $\bigO(n / \log_{\sigma} n)$ working space.  Thus, e.g., for
$\sigma = 2$, our construction takes $\bigO(n / \sqrt{\log n}) = o(n)$
time.

Each node $v$ of the suffix tree of $\T$, denoted $\ST$, is encoded either as $(j,\ell)$ such that $\T[j
\dd j + \ell) = \str(v)$ or as $(b, e)=(\LB(\str(v),\T),\UB(\str(v),\T))$, where $\str(v)$ is the string
represented by $v$. Since the latter representation is more common
(e.g.~\cite{FischerMN09}), we adopt it as the default interface and
denote $\repr(v)$.

In this overview, we focus on the $\child$ operation, which
illustrates some of our key techniques. We remark, however, that other
operations require different combinatorial insights.

Let $\tau$ be as in \cref{sec:to-sa}. Our structure works differently
depending on whether the operation is performed on a node $v$ such
that $\str(v)$ is periodic or not (see \cref{sec:to-pm}).

\bfparagraph{The Nonperiodic Nodes}

Let $v$ be a node of $\ST$ such that $\str(v)$ is nonperiodic and let
$c \in \Sigma$.  Our goal is to compute $\repr(\child(v, c))$ given
$\repr(v)$. Let $\S$ again be as in \cref{sec:to-sa}. The basic idea
is to reduce the computation concerning the $\SA$-interval for string
$\str(v)$ to the computation involving only positions in $\S$. For
this purpose, we store the compact trie $\TSSS$ for $\{\T[\slex_i \dd
n]\}_{i \in [1 \dd n']}$.  Typically, each operation on $\ST$ then
involves the following steps:
\begin{enumerate}
\item Map the input node $v$ of $\ST$ (given as $\repr(v)$) to some
  node $u$ of $\TSSS$,
\item Perform some operation in $\TSSS$ resulting in a node $u'$ (in
  our case, $u' = \child(u, c)$),
\item Map $u'$ back to some node $v'$ of $\ST$ (producing $\repr(v')$
  as output).
\end{enumerate}

\begin{description}[style=sameline,itemsep=1ex,font={\normalfont\itshape}]
\item[Mapping from $\ST$ to $\TSSS$:] Let $(b, e) = \repr(v)$ and let
  $X \in \D$ be the distinguishing prefix of $\str(v)$. By the observation for
  nonperiodic patterns in \cref{sec:to-pm}, the left-to-right order of
  leaves of $\TSSS$ corresponding to suffixes in $\SA(b \dd e]$
  shifted by $\deltatext = |X| - 2\tau$ is consistent with their order
  in $\SA(b \dd e]$. Moreover, since we know how to compute the
  position in $(\slex_i)_{i \in [1 \dd n']}$ corresponding to suffix
  $\T[\SA[i] \dd n]$ for any $i \in [1 \dd n]$ such that $\SA[i] \in
  [1 \dd n] \setminus \R$ (see \cref{sec:to-sa}), we can compute the
  pointer to the leaf of $\TSSS$ corresponding to any suffix in $\SA(b
  \dd e]$. This implies that: (1) there exists a node $\mapToTSSS(v)
  := u$ in $\TSSS$ such that $\str(v) = X[1 \dd \deltatext] \cdot
  \str(u)$, and (2) a pointer to $u$ can be computed via a 
  lowest common ancestor (LCA) query from the leaves
  of $\TSSS$ corresponding to first and last suffix in $\SA(b \dd e]$
  (see \cref{sec:st-nonperiodic-nav}).
\item[Mapping from $\TSSS$ to $\ST$:] After computing $u' = \child(u,
  c)$, the next step is to go back to $\ST$. Our approach exploits
  a similar principle as when computing $\ISA[j]$: knowing $X \in \D$
  and a position in $(\slex_i)_{i \in [1 \dd n']}$ lets us determine
  the corresponding position in $\SA$. More generally, given $X \in
  \D$ and an interval of $(\slex_i)_{i \in [1 \dd n']}$, we can
  retrieve the corresponding interval in $\SA$.  The former is
  precomputed and stored with each node of $\TSSS$.  Note, however,
  the following complication: $\TSSS$ may have extra nodes between $u
  = \mapToTSSS(v)$ and $\widehat{u} = \mapToTSSS(\child(v, c))$, and
  then $\mapToTSSS(\child(v, c)) \ne \child(\mapToTSSS(v), c)$ (see
  also \cref{rm:st-nonperiodic-child}).  We thus need to first prove
  that applying the inverse mapping to \emph{any} node between $u$ and
  $\widehat{u}$ (in particular, to~$u'$) yields $\repr(\child(v, c))$
  (\cref{lm:st-nonperiodic-child}). This exploits properties of $\S$
  specific to the child operation; we omit the details here but
  remark that this step differs among core operations (see, e.g.,
  \cref{lm:st-nonperiodic-lca,lm:st-nonperiodic-wa}).
\end{description}

\bfparagraph{The Periodic Nodes}

Let us now assume that $\str(v)$ is periodic.  The basic idea is
similar as above: we keep a compact trie (denoted $\TZ$) letting us
search the suffixes in the set $\{\T[\rendfull{j} \dd n]\}_{j \in
\R'^{-}}$. Although the implementation of mapping and combinatorial
proofs are more technical, establishing these higher-level navigation
primitives results in simpler and more concise implementation of
queries (see, e.g., \cref{lm:st-periodic-wa}).

\bfparagraph{Summary of New Techniques}

Our key technical contributions are as follows:
\begin{itemize}
\item We show how to reduce all operations of a suffix tree to prefix
  rank and selection queries, resulting in the first $o(n)$-time
  construction of optimal-size compressed suffix tree, with all
  operations simultaneously matching the
  state-of-the-art~\cite{BoucherCGHMNR21,CaceresN22,
  FischerMN09,Gagie2020,PrezzaR21,RussoNO11,Sadakane07}.
\item To achieve this, we first define and efficiently
  implement mappings between the nodes of $\ST$ and of two
  auxiliary tries $\TSSS$ and $\TZ$. We then prove new combinatorial
  results (see, e.g.,
  \cref{lm:st-nonperiodic-lca,lm:st-nonperiodic-child,%
  lm:st-nonperiodic-pred,lm:st-nonperiodic-wa,lm:st-periodic-lca,%
  lm:st-periodic-child,lm:st-periodic-pred,lm:st-periodic-wa,lm:st-parent})
  showing that these high-level navigation primitives correctly
  handle all suffix tree operations.
\end{itemize}

\section{Auxiliary Tools}\label{sec:tools}

\subsection{Weighted Ancestors}\label{sec:wa}

Consider a rooted tree $\mathcal{T}$. Let $\Root(\mathcal{T})$ denote
the node at depth 0 and let $\parent(v)$ denote the immediate ancestor
of each node $v \neq \Root(\mathcal{T})$. We let
$\parent(\Root(\mathcal{T})) = \nil$. Assume that each node $v$ has an
associated weight $w(v)$ such that $w(\Root(\mathcal{T})) = 0$ and,
for every $v \neq \Root(\mathcal{T})$, it holds $w(\parent(v)) <
w(v)$.  We then say that the weight function $w$ is \emph{monotone}.
Given any node $v$ of $\mathcal{T}$ and an integer $0 \leq d \leq
w(v)$, we define $v' = \WA(v, d)$ (the \emph{weighted
ancestor}~\cite{FarachM96}) as the (unique) ancestor of $v$ in
$\mathcal{T}$ that satisfies $w(v') \geq d$ and for which $w(v')$ is
minimized.

\begin{theorem}[{\cite[Section~6.2.1]{AmirLLS07}}]\label{th:wa}
  Let $\mathcal{T}$ be a rooted tree with $n \le N$ nodes and a
  monotone weight function $w$ mapping nodes to $[0 \dd N)$.  There
  exists a data structure of size $\bigO(n)$ that answers weighted
  ancestor queries in $\mathcal{T}$ in $\bigO(\log \log N)$ time after
  $\bigO(n\log_{n}N)$-time preprocessing.
\end{theorem}
\begin{proof}
  A solution for $N=n$ was presented
  in~\cite[Section~6.2.1]{AmirLLS07}.  To generalize it to the case of
  $N\ge n$, it suffices to map all node weights to their ranks. For
  this, we sort all the node weights in $\bigO(n\log_n N)$ time (radix
  sort) and build a deterministic predecessor
  structure~\cite[Proposition~2]{wexp} in $\bigO(n)$ time so that the
  rank of a query threshold can be computed in $\bigO(\log \log N)$
  time.
\end{proof}

\subsection{Tries and Compact Tries}

A set of strings $\mathcal{S} \sub \Sigma^{+}$ is \emph{prefix-free}
if there are no $S, S' \in \mathcal{S}$ such that $S$ is a proper
prefix of $S'$.  For any prefix-free set of string $\mathcal{S} \sub
\Sigma^{+}$, its \emph{trie} is a minimal rooted tree $\mathcal{T}$,
with each edge labelled by some $c \in \Sigma$, such that: (1) no two
edges outgoing from the same node have the same label, (2) for
each $S \in \mathcal{S}$ there exists a path from
$\Root(\mathcal{T})$ to some node such that the concatenation of
edge-labels on that path is equal to $S$, and (3) children of every
node are ordered according to the lexicographical rank of the
connecting edge. A \emph{compact trie} of $\mathcal{S}$ is a trie of
$\mathcal{S}$ in which all maximal unary paths have been replaced with
edges labelled by substrings of elements of $\mathcal{S}$. The nodes
of the trie omitted in the compact trie are referred to as
\emph{implicit}. All other nodes are \emph{explicit}. Unless explicitly
stated otherwise, by \emph{node} we always mean an explicit node.

For any node $v$ of a (compact) trie $\mathcal{T}$, by $\str(v)$ we
denote the \emph{label} of $v$, i.e., the string obtained by
concatenating the labels of all edges on the path from
$\Root(\mathcal{T})$ to $v$. We denote $\sdepth(v) = |\str(v)|$. The
parent of $v$ in denoted $\parent(v)$.  For any $c \in \Sigma$, we
define $\child(v, c)$ as a child $v'$ of $v$ such that
$\str(v')[|\str(v)| + 1] = c$, or $\nil$ if no such node exists.  For
any $c \in \Sigma$, we also define $\pred(v, c)$ as follows:
\begin{itemize}
\item If there exists $c' < c$ such that $\child(v, c') \neq \nil$,
  then we let $\pred(v, c) = \child(v, c_{\max})$, where $c_{\max} =
  \max\{c' \in [0 \dd c) : \child(v, c') \neq \nil\}$.
\item Otherwise, we let $\pred(v, c) = \nil$.
\end{itemize}
We define $(\lrank(v), \rrank(v))$ as a pair of integers satisfying
$\lrank(v) = |\{S \in \mathcal{S} : S \prec \str(v)\}|$ and $\rrank(v)
- \lrank(v) = |\{S \in \mathcal{S} : \str(v)\text{ is a prefix of
}S\}|$.  Observe that then collecting every $i$th leftmost leaf
of $\mathcal{T}$, where $i \in (\lrank(v) \dd \rrank(v)]$ results in
precisely the set of leaves in the subtree rooted in $v$. Given any
node $v$ of $\mathcal{T}$ and an integer $0 \leq d \leq |\str(v)|$, we
let $v' = \WA(v, d)$ to be the weighted ancestor of $v$ assuming the
weight of each node is defined as $w(v) = \sdepth(v)$.  Thus, $v'$ is
the (unique) ancestor of $v$ in $\mathcal{T}$ that satisfies
$\sdepth(v') \geq d$ and for which $\sdepth(v')$ is minimized. For any
two nodes $v_1$ and $v_2$, the node $v = \LCA(v_1, v_2)$ (the
\emph{lowest common ancestor}) is defined as the (unique) ancestor of
both $v_1$ and $v_2$ with the maximal depth.

\begin{observation}\label{ob:lca}
  If $v_1$ and $v_2$ are nodes of a (compact) trie, then letting $v =
  \LCA(v_1, v_2)$ and $\ell = \lcp(\str(v_1), \str(v_2))$, it holds
  $\sdepth(v) = \ell$ and $\str(v) = \str(v_1)[1 \dd \ell] =
  \str(v_2)[1 \dd \ell]$.
\end{observation}

\subsubsection{Small Alphabet}

\begin{proposition}\label{pr:compact-trie}
  Given a packed representation of $\T \in \Alphabet^{n}$ with $2\le
  \sigma\le n$ and an array $A[1 \dd q]$ of $q$ positions in $\T$ such
  that, for any $1 \leq i < j \leq q$, it holds $\T[A[i] \dd n] \prec
  \T[A[j] \dd n]$, we can in $\bigO(q+n/\log_{\sigma}n)$ time
  construct a representation of the compact trie $\mathcal{T}$ of the
  set $\{\T[A[i] \dd n] : i \in [1 \dd q]\}$, augmented with auxiliary
  data structures to support the following operations on $\mathcal{T}$
  in $\bigO(1)$ time:
  \begin{itemize}
  \item Given $i \in [1 \dd q]$ return the $i$th leftmost leaf of
    $\mathcal{T}$,
  \item Given pointers to nodes $v_1$ and $v_2$ return a pointer to
    $\LCA(v_1, v_2)$,
  \item Given a pointer to node $v$, return $(\lrank(v), \rrank(v))$
    and $\sdepth(v)$.
  \end{itemize}
  It also supports the following operations in $\bigO(\log \log n)$
  time:
  \begin{itemize}
  \item Given a pointer to node $v$ and $d$ such that $0 \leq d \leq
    |\str(v)|$, return the pointer to $\WA(v, d)$,
  \item Given a pointer to node $v$ and $c \in \Sigma$, check if
    $\pred(v, c) \neq \nil$ (resp.\ $\child(v, c) \neq \nil$), and if
    so, return the pointer to $\pred(v, c)$ (resp.\ $\child(v, c)$).
  \end{itemize}
\end{proposition}
\begin{proof}
  The data structure consists of five components:
  \begin{enumerate}
  \item The packed representation of $\T$ using $\bigO(n /
    \log_{\sigma} n)$ space.
  \item The compact trie $\mathcal{T}$. Since we assumed that $\T[n]$
    is unique in $\T$ (see \cref{sec:prelim}), the set $\{\T[A[i] \dd
    n]\}_{i \in [1 \dd q]}$ is prefix-free, and hence $\mathcal{T}$
    has exactly $q$ leaves. Each node $v$ of $\mathcal{T}$ stores the
    precomputed values $\sdepth(v)$, $(\lrank(v), \rrank(v))$, and a
    predecessor data structure that, given any $c \in [0 \dd \sigma)$,
    returns a pointer to $\pred(v, c)$. Using the structure
    from~\cite[Proposition~2]{wexp}, we achieve linear space and
    $\bigO(\log \log n)$ query time.
  \item The array of pointers $L[1 \dd q]$ such that $L[i]$ is the
    pointer to the $i$th leftmost leaf of $\mathcal{T}$.
  \item The data structure of Bender and
    Farach-Colton~\cite{BenderF00} that augments $\mathcal{T}$ with
    support for $\LCA$ queries. The structure needs $\bigO(q)$ space
    and answers queries in $\bigO(1)$ time.
  \item The data structure from \cref{th:wa} for $\mathcal{T}$ with
    the weight function $\sdepth(v)\in [0\dd n]$. The structure needs
    $\bigO(q)$ space and answers queries in $\bigO(\log \log n)$ time.
  \end{enumerate}

  In total, the data structure takes $\bigO(q + n / \log_{\sigma} n)$
  space.

  Using the above structures, the implementation of all queries in the
  claim follows immediately (note that $\child(v, c)$ can be
  determined using $\pred(v, c + 1)$ and the packed representation of
  $\T$).

  \paragraph{Construction algorithm}

  The data structure is constructed as follows. We start by building a
  data structure that supports LCE queries for suffixes of
  $\T$. Using~\cite[Theorem~5.4]{sss}, the construction takes $\bigO(n
  / \log_{\sigma} n)$ time, and the resulting data structure answers
  queries in $\bigO(1)$ time. We then construct $\mathcal{T}$ by
  inserting elements of $\{\T[A[i] \dd n]\}_{i \in [1 \dd q]}$ in the
  order given by $A$. We maintain a stack containing the internal
  nodes on the rightmost path, with the deepest node on top. When
  inserting each string, we first determine the depth at which that
  string branches from the rightmost path using $\LCE$ queries on
  $\T$. We then update the rightmost path of the trie. Adding each
  string first removes some elements from the stack, and then adds at
  most two new elements. Since the total number of elements pushed on
  stack is $\bigO(q)$, the construction of $\mathcal{T}$ takes
  $\bigO(q)$ time. During the construction, we record the value
  $\sdepth(v)$ in each node and collect pointers to consecutive leaves
  in the $L[1 \dd q]$ array. With the single traversal of
  $\mathcal{T}$, we then precompute in $\bigO(q)$ time the values
  $\lrank(v)$ and $\rrank(v)$ for every node $v$ (note that at this
  point, pointers to all children of each node are stored simply using
  a list, since this traversal does not require fast lookups or
  predecessor queries). Next, we construct the predecessor data
  structure for every node.  Since the keys in every node are sorted,
  using~\cite[Proposition~2]{wexp}, over all nodes of $\mathcal{T}$,
  the construction takes $\bigO(q)$ time.  Finally, we construct the
  data structures supporting $\LCA$ and $\WA$ queries on
  $\mathcal{T}$. Using~\cite{BenderF00} and \cref{th:wa}, this takes
  $\bigO(q)$ and $\bigO(q \log_{q} n)= \bigO((q+\sqrt{n})
  \log_{q+\sqrt{n}}n)=\bigO(q+\sqrt{n})=\bigO(q+n/\log_{\sigma} n)$
  time, respectively.
\end{proof}

\subsubsection{Large Alphabet}

\begin{proposition}\label{pr:meta-trie}
  Given a packed representation of $\T \in \Alphabet^n$ with $2 \leq
  \sigma < n^{1/7}$ and an array $A[1 \dd q]$ of $q$ positions in $\T$
  such that, for any $1 \leq i < j \leq q$, it holds $\T[A[i] \dd n]
  \prec \T[A[j] \dd n]$, we can in $\bigO(q+n / \log_{\sigma} n)$ time
  construct a data structure that, given the packed representation of
  any $\Pat \in \Alphabet^m$, returns in $\bigO(m / \log_{\sigma} n +
  \log \log n)$ time a pair of integers $(\bpre, \epre)$ satisfying:
  \begin{itemize}
  \item $\bpre = |\{i \in [1 \dd q] : \T[A[i] \dd n] \prec \Pat\}|$,
    and
  \item $(\bpre \dd \epre] = \{i \in [1 \dd q] :
    \Pat\text{ is a prefix of }\T[A[i] \dd n]\}$.
  \end{itemize}
\end{proposition}
\begin{proof}

  The basic idea is to construct the compact trie of strings in
  $\{\T[A[i] \dd n]\}_{i \in [1 \dd q]}$ converted into strings over
  the alphabet of metasymbols (of $\Theta(\log_{\sigma} n)$ original
  symbols each). Our mapping of symbols to metasymbols does not,
  however, simply group symbols into blocks. We instead introduce a
  special mapping that will allow us to deduce the output range
  $(\bpre, \epre)$ using two predecessor queries in the image of the
  set $\{\T[A[i] \dd n]\}_{i \in [1 \dd q]}$ and some carefully
  crafted patterns. This will allow us to use the augmentation of
  tries of Fischer and Gawrychowski~\cite[Theorem~1]{wexp} in a
  black-box manner.

  \paragraph{Definitions}

  First, we introduce the mapping of strings over $\Alphabet$ into
  strings over metasymbols.  Let $\tau = \lfloor{\frac17 \log_\sigma
  n}\rfloor$ and $\kappa = 3\tau - 1$.  For any $X \in \Alphabet^{\le
  3\tau - 1}$, let $\Int(X)$ denote an integer constructed by appending
  $6\tau - 2|X|$ zeros and $|X|$ $c$s (where $c = \sigma - 1$) to $X$,
  and then interpreting the resulting string as a base-$\sigma$
  representation of a number in $[0 \dd \sigma^{6\tau})$. Note
  that $X \neq X'$ implies $\Int(X) \neq \Int(X')$. Let
  also $\Int'(X)$ denote an integer
  constructed by appending $6\tau - 2|X|$ $c$s (where $c = \sigma -
  1$) and $|X|$ zeros to $X$, and then interpreting the resulting
  string as a base-$\sigma$ representation of a number in $[0 \dd
  \sigma^{6\tau})$.  For any string $S \in \Alphabet^{*}$ of length
  $\ell \geq 0$, we define $\Intpad(S)$ as a string of length $\ell' =
  \lceil \tfrac{\ell + 1}{\kappa} \rceil > 0$ over alphabet $[0 \dd
  \sigma^{6\tau})$ such that, for any $i \in [1 \dd \ell']$, it holds
  $\Intpad(S)[i] = \Int(S((i{-}1) \cdot \kappa \dd \min(\ell, i \cdot
  \kappa)])$.  For $S \in \Alphabet^{*}$ of length $\ell \geq 0$, we
  define $\Intpadc(S)$ as a string of length $\ell' = \lceil
  \tfrac{\ell + 1}{\kappa} \rceil > 0$ over alphabet $[0 \dd
  \sigma^{6\tau})$ such that $\Intpadc(S)[1 \dd \ell') = \Intpad(S)[1
  \dd \ell')$ and $\Intpadc(S)[\ell'] = \Int'(S((\ell' - 1) \cdot
  \kappa \dd \min(\ell, \ell' \cdot \kappa)])$.  Note that if $\ell$
  is a multiple of $\kappa$, then the last symbol of $\Intpad(S)$ is
  $\Int(\emptystring) = 0$, whereas the last symbol of $\Intpadc(S)$
  is $\Int'(\emptystring)=\sigma^{6\tau}-1$.  Observe that
  \begin{itemize}
  \item For any set of strings $\mathcal{S} \subseteq \Alphabet^{*}$,
    the set $\{\Intpad(X) : X \in \mathcal{S}\}$ is prefix-free.
  \item For any strings $X, Y \in \Alphabet^{*}$, $X \prec Y$ holds if
    and only if $\Intpad(X) \prec \Intpad(Y)$.
  \item A string $\Pat \in \Alphabet^{*}$ is a prefix of $X \in
    \Alphabet^{*}$ if and only if $\Intpad(\Pat) \preceq \Intpad(X)
    \prec \Intpadc(\Pat)$.
  \end{itemize}
  For any set of strings $\mathcal{S}$ and any string $Y$, denote
  ${\rm rank}_{\mathcal{S}}(Y) := |\{X \in \mathcal{S} : X \prec
  Y\}|$.  Observe that by the above properties, letting $P_1 =
  \Intpad(\Pat)$, $P_2 = \Intpadc(\Pat)$, and $\mathcal{A} =
  \{\Intpad(\T[A[i] \dd n])\}_{i \in [1 \dd q]}$, we have $(\bpre,
  \epre) = ({\rm rank}_{\mathcal{A}}(P_1), {\rm
  rank}_{\mathcal{A}}(P_2))$.  Let $\mathcal{T}$ denote the compact
  trie of the set $\mathcal{A}$.

  \paragraph{Components}

  The data structure consists of two components:
  \begin{enumerate}
  \item The packed representation of $\T$ using $\bigO(n /
    \log_{\sigma} n)$ space.
  \item The trie $\mathcal{T}$ augmented
    using~\cite[Theorem~1]{wexp}. Note that this result requires that
    the alphabet of strings in $\mathcal{A}$ is of size
    $|\mathcal{A}|^{\bigO(1)}$, which may be violated for $q =
    n^{o(1)}$.  Thus, we actually define the alphabet to be $[0\dd
    \sigma')$, where $\sigma' =
    \sigma^{6\tau}+\lceil{\sqrt{n}}\rceil$, and insert to
    $\mathcal{A}$ additional $\lceil{\sqrt{n}}\rceil$ dummy length-1
    strings corresponding to the $\lceil{\sqrt{n}}\rceil$ largest
    characters in $[0\dd \sigma')$.  As a result, we must have
    $\sigma' = \bigO(n) = \bigO(|\mathcal{A}|^2)$.  At the same time, the
    dummy strings do not change ${\rm rank}_{\mathcal{A}}(Q)$ for any
    $Q\in [0\dd \sigma^{6\tau})^*$.  By~\cite[Theorem~1]{wexp}, such
    augmented $\mathcal{T}$ needs $\bigO(q + \sqrt{n})$ space.
  \end{enumerate}

  In total, the data structure takes $\bigO(q + n / \log_{\sigma} n)$
  space.

  \paragraph{Implementation of queries}

  Using $\T$ and $\mathcal{T}$, given the packed representation of
  $\Pat \in \Alphabet^{m}$, we compute the output pair $(\bpre,
  \epre)$ as follows.  First, note that, given the packed
  representation of $\Pat$ and $\T$, we can in $\bigO(1)$ time access
  any symbol of $\Intpad(\Pat)$, $\Intpadc(\Pat)$, and $\Intpad(\T[i
  \dd n])$ for any $i \in [1 \dd n]$. We start by computing $P_1$ and
  $P_2$ from $\Pat$ in $\bigO(m / \log_{\sigma} n + 1)$ time. Then,
  using~\cite[Theorem~1]{wexp}, we compute ${\rm
  rank}_{\mathcal{A}}(P_1)$ and ${\rm rank}_{\mathcal{A}}(P_2)$ in
  $\bigO(|P_1| + \log \log \sigma') = \bigO(m / \log_{\sigma} n + \log
  \log n)$ and $\bigO(|P_2| + \log \log \sigma') = \bigO(m /
  \log_{\sigma} n + \log \log n)$ time, respectively. By the above
  discussion, this gives us the output pair $(\bpre, \epre)$. During
  the query, the algorithm may want to access symbols of
  strings from $\mathcal{A}$. We do not store them explicitly (note
  that storing $\Intpad(\T)$ would not be enough), but instead perform
  the mapping on-the-fly.

  \paragraph{Construction algorithm}

  We start by building a
  data structure that supports LCE queries for suffixes of
  $\T$. Using~\cite[Theorem~5.4]{sss}, the construction takes $\bigO(n
  / \log_{\sigma} n)$ time, and the resulting data structure answers
  queries in $\bigO(1)$ time. Denote the length of the longest common
  prefix between suffixes $\T[i \dd n]$ and $\T[j \dd n]$ as $\LCE(i,
  j)$.  Observe that for any $i, j \in [1 \dd n]$ such that $i \neq
  j$, the longest common prefix of $\Intpad(\T[i \dd n])$ and
  $\Intpad(\T[j \dd n])$ has length $\lfloor \tfrac{\LCE(i,j)}{\kappa}
  \rfloor$, which can be computed in $\bigO(1)$ time.  We construct
  $\mathcal{T}$ by inserting elements of $\{\Intpad(\T[A[i] \dd
  n])\}_{i \in [1 \dd q]}$ in the order given by $A$. The
  construction proceeds as in the proof of \cref{pr:compact-trie} and
  takes $\bigO(q)$ time. Once the trie is constructed, we 
  add the $\lceil{\sqrt{n}}\rceil$ dummy length-1 strings and 
  augment
  $\mathcal{T}$ using~\cite[Theorem~1]{wexp} in $\bigO(q+\sqrt{n})$ time.
  In total, the construction takes $\bigO(q + n / \log_{\sigma} n)$ time.
\end{proof}

\subsection{(Prefix) Rank and Selection Queries}\label{sec:prefix-queries}

We start with an implementation of rank and selection queries for
larger alphabets.

\begin{lemma}[Belazzougui and Puglisi~\cite{Belazzougui2016}]\label{lm:rs}
  For all integers $N\ge n \ge \sigma\ge 2$ and every string $S\in
  [0\dd \sigma)^{\le n}$, there exists a data structure of $\bigO(n\log
  \sigma)$ bits that answers rank queries in $\bigO(\log \log N)$ time
  and selection queries in $\bigO(1)$ time.  Moreover, given a table
  precomputed in $\bigO(N)$ time (shareable across all instances with
  common parameter $N$) and the packed representation of $S$, the data
  structure can be constructed in $\bigO(\min(n,\,\sigma + n\log \sigma
  /\sqrt{\log N}))$ time using $\bigO(n\log \sigma)$ bits of space.
\end{lemma}
\begin{proof}
  If $\log^2 \sigma \ge \log N$, we use the data structure
  of~\cite[Lemma~C.2]{Belazzougui2016}, which occupies $\bigO(n\log
  \sigma)$ bits, answers rank queries in $\bigO(\log \log n)$ time and
  selection queries in $\bigO(1)$ time, and can be constructed in
  $\bigO(n)$ time using $\bigO(n\log \sigma)$ bits of space.\footnote{The
  statement of~\cite[Lemma~C.2]{Belazzougui2016} does not bound the
  space consumption of the construction algorithm.  Nevertheless, it
  is straightforward to implement the underlying construction
  procedure in $\bigO(n\log \sigma)$ bits of working space.  The
  original algorithm scans the input sequence $S$ from left to right
  and, for each $a\in \Sigma$, builds an array $P_a[1\dd n_a]$ such
  that $n_a = \rank{S}{a}{|S|}$ and $P_a[r]=\select{S}{a}{r}$ for
  $r\in [1\dd n_a]$. The array $P_a[1\dd n_a]$ is then converted to
  the Elias--Fano representation: an array $A_a[1\dd n_a]$ with
  $A_a[r] = P_a[r] \bmod \sigma$ for $r\in [1\dd n_a]$ and a bit
  vector $V_a = \unary((\lfloor P_a[r]/\sigma\rfloor-\lfloor
  P_a[r-1]/\sigma\rfloor)_{r\in [1\dd n_a]})$, where we assume
  $P_a[0]=0$ to streamline the formula. To achieve $\bigO(n\log \sigma)$
  bits of working space, instead of storing $P_a$ explicitly, we
  convert $P_a$ to the Elias--Fano representation on the fly as
  subsequent positions are appended to $P_a$.}  Otherwise, we use the
  data structure of~\cite[Lemma~C.3]{Belazzougui2016}, which occupies
  $\bigO(n\log \sigma)$ bits, answers rank queries in $\bigO(\log \log N)$
  time and selection queries in $\bigO(1)$ time, and can be constructed
  in $\bigO(\sigma + n\log^2\sigma / \log N)$ time using $\bigO(n
  \log\sigma)$ bits of space.
\end{proof}

The following proposition, instantiated with $h = \lceil
\ell^{\epsilon/2} \rceil$, immediately yields
\cref{th:wavelet-tree}.

\begin{proposition}\label{pr:wavelet_tree}
  For all integers $h,m,\ell,\sigma\in \Z_{\ge 1}$ satisfying $h\ge 2$
  and $m\ge \sigma^\ell\ge 2$, and every string $W\in ([0\dd
  \sigma)^\ell)^{\le m}$, there exists a data structure of size $\bigO(m
  \log_h (h\ell))$ that answers prefix rank queries in $\bigO(h \log
  \log m \log_h (h\ell))$ time and prefix selection queries in $\bigO(h
  \log_h (h\ell))$ time.  Moreover, it can be constructed in
  $\bigO(m\min(\ell,\,\sqrt{\log m})\log_h (h\ell))$ time using $\bigO(m
  \log_h (h\ell))$ space given the packed representation of $W$ and
  the parameter $h$.
\end{proposition}

\begin{proof}
  The data structure consists in the wavelet tree of $W$ and, when
  $h\le \ell$, an instance constructed recursively for an auxiliary
  string $\tilde{W}$ defined below.

  \paragraph{Wavelet tree}

  Let $\Sigma=[0\dd \sigma)$ so that the alphabet of $W$ is
  $\Sigma^\ell$.  The wavelet tree of $W$~\cite{wt} is the trie of
  $\Sigma^\ell$ with each internal node $v_X$ (representing a string
  $X\in \Sigma^{\le \ell-1}$) associated to a string $B_X[1\dd
  \rank{W}{X}{|W|}]\in \Sigma^*$ such that $B_X[r] =
  W[\select{W}{X}{r}][|X|+1]$ for $r\in [1\dd \rank{W}{X}{|W|}]$.  The
  strings $B_X$ are augmented with the component of \cref{lm:rs} (for
  rank and selection queries) with parameter~$N := m$.

  \paragraph{Recursive instance}

  We shall define $\tilde{W}$ as a string of length $|W|$ over the
  alphabet $\tilde{\Sigma}^{\tilde{\ell}}$, where
  $\tilde{\ell}:=\lfloor \ell/h\rfloor$, $\tilde{\sigma} = \sigma^h$,
  and $\tilde{\Sigma} := [0\dd \tilde{\sigma})$.  We identify
  $\tilde{\Sigma}$ with $\Sigma^h$, treating each string in $\Sigma^h$
  as the $h$-digit base-$\sigma$ representation of an integer in
  $\tilde{\Sigma}$.  For every string $X\in \Sigma^*$, define
  $\widetilde{X}\in \tilde{\Sigma}^*$ so that
  $|\widetilde{X}|=\lfloor{|X|/h}\rfloor$ and $\widetilde{X}[i]=
  X(h(i-1)\dd hi]$ for $i\in [1\dd |\widetilde{X}|]$.  Moreover, we
  set $\tilde{W}[1\dd |W|]$ so that $\tilde{W}[j]=\widetilde{W[j]}$
  for $j\in [1\dd |W|]$.  Note that the recursive application of
  \cref{pr:wavelet_tree} to $\tilde{W}$ is possible because $2\le
  \tilde{\sigma}^{\tilde{\ell}} \le \sigma^\ell \le m$ and
  $\tilde{\ell} \ge 1$ hold when $h\le \ell$.

  \paragraph{Data structure size}

  It is easy to see that, for a fixed length $d\in [0\dd \ell)$, the
  strings $B_X$ for $X\in \Sigma^d$ are of total length~$m$. Across
  all $X\in \Sigma^{\le \ell-1}$, this sums up to $m\ell$, so the raw
  strings $B_X$ occupy $\bigO(m\ell\log \sigma)=\bigO(m\log m)$ bits.
  The augmentation of $B_X$ using \cref{lm:rs} adds $\bigO((\sigma +
  |B_X|)\log \sigma)$ extra bits (we set $n := \max(\sigma, |B_X|) =
  \Theta(\sigma + |B_X|)$ to ensure $\sigma \leq n$), which sums up to
  $\bigO((\sigma^\ell + m\ell)\log \sigma)=\bigO(m\log m)$ bits, i.e.,
  $\bigO(m)$ machine words.  The recursion depth is
  $\bigO(\log_h(h\ell))$, so the overall size is
  $\bigO(m\log_h(h\ell))$.

  \paragraph{Answering queries}

  To handle any query concerning $X\in \Sigma^{\le \ell}$, we compute
  auxiliary strings $\widetilde{X}$ (as defined above) and $X':=X[1\dd
  |X|-(|X|\bmod h)]$ (obtained by expanding the letters in
  $\widetilde{X}$ into length-$h$~strings).

  Answering a prefix rank query $\rank{W}{X}{j}$, we traverse the path
  from $v_{X'}$ to $v_{X}$, maintaining a value $r$ such that
  $r=\rank{W}{Y}{j}$ holds while the algorithm visits $v_{Y}$.  We
  initialize $r:=j=\rank{W}{\emptystring}{j}$ if $X'=\emptystring$ and
  $r:=\rank{\tilde{W}}{\widetilde{X}}{j}$ (computed recursively)
  otherwise; this is valid due to $\rank{\tilde{W}}{\widetilde{X}}{j} =
  \rank{W}{X'}{j}$.  Upon entering a node $v_{Ya}$ from its parent
  $v_Y$, we set $r:= \rank{B_Y}{a}{r}$ since
  $\rank{W}{Ya}{j}=\rank{B_Y}{a}{\rank{W}{Y}{j}}$; see~\cite{wt}.
  When reaching $v_X$, we return $r=\rank{W}{X}{j}$.  The running time
  is $\bigO(h\log \log m)$ per recursive level, for a total of $\bigO(h
  \log \log m \cdot \log_h(h\ell))$.

  Answering a prefix selection query $\select{W}{X}{r}$, we traverse
  the path from $v_X$ to $v_{X'}$, maintaining a value $q$ such that
  $\select{W}{Y}{q}=\select{W}{X}{r}$ holds while the algorithm visits
  $v_{Y}$.  We initialize $q:=r$ and, upon entering a node $v_{Y}$
  from its child $v_{Ya}$, we set $q:=\select{B_Y}{a}{q}$ since
  $\select{W}{Ya}{r} = \select{W}{Y}{\select{B_Y}{a}{r}}$;
  see~\cite{wt}.  When reaching $v_{X'}$, we return $q =
  \select{W}{\emptystring}{q}$ if $X'=\emptystring$ and
  $\select{\tilde{W}}{\widetilde{X}}{q}$ otherwise; this is valid due to
  $\select{\tilde{W}}{\widetilde{X}}{q}=\select{W}{X'}{q}$. The running
  time is $\bigO(h)$ per recursive level, for a total of $\bigO(h \cdot
  \log_h(h\ell))$.

  \paragraph{Construction algorithm}

  If $\ell \le \sqrt{\log m}$, we use the original wavelet tree
  construction algorithm~\cite{wt}, which takes $\bigO(m\ell)$ time
  and $\bigO(m)$ space.  Building the data structure of \cref{lm:rs}
  for $B_X$ takes $\bigO(\sigma + |B_X|)$ time and $\bigO((\sigma +
  |B_X|)\log\sigma / \log m)$ space, which sums up to
  $\bigO(\sigma^{\ell} + m\ell) = \bigO(m\ell)$ time and $\bigO(m\ell
  \log \sigma / \log m)=\bigO(m)$ space across $X\in \Sigma^{\le
  \ell-1}$ (due to $\ell \log \sigma \le \log m$).  Precomputing the
  table shared by all instances of \cref{lm:rs} takes $\bigO(m)$ time
  and space.  Considering all levels of recursion, we get
  $\bigO(m\ell)$ time (due to $\tilde{\ell}\le \frac12\ell)$ and
  $\bigO(m \log_h(h\ell))$ space.
  
  If $\ell > \sqrt{\log m}$, on the other hand, we apply the
  bit-parallel wavelet tree construction algorithm
  of~\cite{MunroNV16,WaveletSuffixTree}, which has been adapted to
  large alphabets in~\cite[Lemma 6.4]{sss}.  Due to $\ell\log \sigma
  \le \log m$, this procedure takes $\bigO(m\ell\log \sigma
  /\sqrt{\log m}+m\ell\log^2 \sigma / \log m)=\bigO(m\sqrt{\log m})$
  time and $\bigO(m)$ space.  Building the data structure of
  \cref{lm:rs} for $B_X$ takes $\bigO(\sigma + (\sigma + |B_X|)\log
  \sigma / \sqrt{\log m}) = \bigO(\sigma + |B_X|\log \sigma /
  \sqrt{\log m})$ time and $\bigO((\sigma + |B_X|)\log\sigma / \log
  m)$ space, which sums up to $\bigO(m\sqrt{\log m})$ time and
  $\bigO(m)$ space across $X\in \Sigma^{\le \ell-1}$. Precomputing the
  table shared by all instances of \cref{lm:rs} takes $\bigO(m)$ time
  and space.  Considering all levels of recursion, we get a
  multiplicative overhead of $\bigO(\log_h(h\ell))$, for a total of
  $\bigO(m\sqrt{\log m} \log_h(h\ell))$ time and $\bigO(m
  \log_h(h\ell))$ space.
\end{proof}

\subsection{Range Counting and Selection}\label{sec:range-queries}

\prrangequeries*
\begin{proof}

  We use the following definitions.
  Denote $h = \lfloor \log m \rfloor$. For any $k \geq 0$, by $P_k[1
  \dd m_k]$, where $m_k = \rcount{A}{kh}{m}$, we denote the array
  defined by $P_k[i] = \rselect{A}{kh}{i}$. Let $v \geq 0$.  We define
  a bitvector $M_v[1 \dd m_k]$, where $k = \lfloor \tfrac{v}{h}
  \rfloor$ as follows. For any $i \in [1 \dd m_k]$, $M_v[i] = 1$ holds
  if and only if $A[P_k[i]] \geq v$. For any $k \geq 0$, we define the
  concatenation $M'_k = M_{kh} M_{kh+1} \cdots M_{(k+1)h-1}$.  Let
  $k_{\max} = \max\{k \geq 0 : m_k > 0\}$. Since all elements of $A$
  are nonnegative, and $\sum_{i=1}^{m'}A[i] \in \bigO(m \log m)$, we
  obtain $\max_{i \in [1 \dd m']} A[i] \in \bigO(m \log m)$, and
  consequently, $k_{\max} = \lfloor \tfrac{1}{h} \max_{i \in [1 \dd
  m']} A[i] \rfloor \in \bigO(m)$.

  \paragraph{Components}

  The data structure consists of two components:
  \begin{enumerate}
  \item First, for $k \in [0
    \dd k_{\max}]$, we store a plain representation of the sequence
    $P_k[1 \dd m_k]$ using $\bigO(m_k)$ space. Each array is augmented
    with a static predecessor data structure. We
    use~\cite[Proposition~2]{wexp}, and hence achieve linear space and
    $\bigO(\log \log m)$ query time. Each $i \in [1 \dd m']$ occurs in
    $\lceil \tfrac{A[i] + 1}{h} \rceil$ arrays. Thus, $\sum_{k \geq
    0}m_k = \sum_{i=1}^{m'} \lceil \tfrac{A[i] + 1}{h} \rceil \leq 2m' +
    \sum_{i=1}^{m'} \lfloor \tfrac{A[i]}{h} \rfloor \leq 2m' +
    \tfrac{1}{h}\sum_{i=1}^{m'}A[i] \in \bigO(m)$ and hence we can store
    the arrays $P_k$ (including the associated predecessor data
    structures) using $\bigO((k_{\max} + 1) + \sum_{k \geq 0}m_k)
    \subseteq \bigO(m)$ space, so that we can access each array in
    $\bigO(1)$ time.
  \item Second, for every $k \in [0 \dd k_{\max}]$, we
    store the plain representation of bitvector $M'_k$, augmented using
    \cref{th:binrksel}.  By $|M'_k| = h\cdot m_k$, the total length of
    bitvectors $M'_k$ is $\sum_{k \geq 0}|M'_k| = h\sum_{k \geq 0}m_k
    \in \bigO(m \log m)$. All bitvectors $M'_k$ can thus be stored in
    $\bigO((k_{\max} + 1) + \tfrac{1}{\log m} \sum_{k \geq 0}|M'_k|)
    \subseteq \bigO(m)$ words of space, so that we can access each in
    $\bigO(1)$ time. For a bitvector of length $t$, the augmentation of
    \cref{th:binrksel} adds only $\bigO(\log m +t)$ bits of space, and
    hence does not increase the space usage.
  \end{enumerate}

  In total, the data structure takes $\bigO(m)$ space.

  \paragraph{Implementation of queries}

  Using the above two components, we answer range counting/selection
  queries on $A$ as follows. To compute $\rcount{A}{v}{j}$, we first
  let $k = \lfloor \tfrac{v}{h} \rfloor$. If $k > k_{\max}$, then we
  return $\rcount{A}{v}{j} = 0$. Otherwise, we observe that if $j' =
  |\{i \in [1 \dd m_k] : P_k[i] \leq j\}|$, then $\rcount{A}{v}{j} =
  \rank{M_v}{1}{j'}$.  Computing $j'$ using the predecessor data
  structure takes $\bigO(\log \log m)$ time, and then
  $\rank{M_v}{1}{j'}$ is computed using the rank support data
  structure of the bitvector $M'_k$ as $\rank{M'_k}{1}{j' + (v -
  kh)m_k} - \rank{M'_k}{1}{(v - kh)m_k}$ in $\bigO(1)$ time. To
  compute $\rselect{A}{v}{r}$, we observe that letting again $k =
  \lfloor \tfrac{v}{h} \rfloor$, it holds $\rselect{A}{v}{r} =
  P_k[\select{M_v}{1}{r}]$.  The value $\select{M_v}{1}{r}$ is
  computed using the select support data structure of the bitvector
  $M'_k$ as $\select{M'_k}{1}{\rank{M'_k}{1}{(v - kh)m_k}+r}- (v -
  kh)m_k$ in $\bigO(1)$ time.

  \paragraph{Construction algorithm}

  We start by
  initializing $P_0[i] = i$ for $i \in [1 \dd m']$. For $k \in [1 \dd
  k_{\max}]$, the array $P_k$ is computed by iterating over
  $P_{k-1}$ and including only elements $P_{k-1}[i]$ satisfying
  $A[P_{k-1}[i]] \geq kh$. By $\sum_{k \geq 0}m_k \in \bigO(m)$, this
  takes $\bigO(m)$ time in total.  We then augment all arrays $P_k$
  with the predecessor data structures.  Since the arrays are sorted,
  using~\cite[Proposition~2]{wexp}, the construction altogether again
  takes $\bigO(m)$ time.  We then construct bitvectors $M'_k$ in the
  order of increasing $k \in [0 \dd k_{\max}]$. To build $M'_k$ we
  first scan $P_k$ and check if there exists $i \in [1 \dd m_k]$ such
  that $A[P_k[i]] < (k+1)h$.
  \begin{enumerate}
  \item If there is no such $i$, we set $M'_k := {\tt 1}^{hm_k}$ in
    $\bigO(1 + \tfrac{1}{\log m}hm_k) = \bigO(m_k)$ time.
  \item Otherwise, we scan again $P_k[1 \dd m_k]$ and prepare $h$
    lists $L_0, L_1, \ldots, L_{h-1}$ such that $L_{y}$ contains all
    $i \in [1 \dd m_k]$ satisfying $A[P_k[i]] = kh + y$.  Construction
    of all lists takes $\bigO(m_k + h)$ time. The bitvector $M'_k$ is
    then obtained as the concatenation of bitvectors $M_{kh},
    M_{kh+1}, \ldots, M_{(k+1)h-1}$ computed in this order. We first
    initialize $M_{kh} := {\tt 1}^{m_k}$ in $\bigO(1 +
    \tfrac{m_k}{\log m})$ time. The bitvector $M_{kh+y}$ for $y>0$ is
    obtained by first copying the bitvector $M_{kh+y-1}$ in $\bigO(1 +
    \tfrac{m_k}{\log m})$ time, and then setting $M_{kh+y}[i] = 0$ for
    every position $i$ stored in $L_{y-1}$. The total length of all
    lists $L_y$ is bounded by $m_k$. Thus, the construction of $M'_k$
    takes $\bigO(h + m_k + \tfrac{1}{\log m}hm_k) \subseteq \bigO(h +
    m_k)$ time.
  \end{enumerate}
  To bound the total time spent constructing bitvectors $M'_k$, we
  consider two cases:
  \begin{itemize}
  \item $k \leq \tfrac{m}{h}$: The total time spent in the
    construction of bitvectors $M'_k$ for such $k$ is bounded by the
    sum $\sum_{k=0}^{\lfloor m/h \rfloor} \bigO(h + m_k) \subseteq
    \bigO(m + \sum_{k \geq 0}m_k) \subseteq \bigO(m)$.
  \item $k > \tfrac{m}{h}$: Let $k' = \lfloor \tfrac{m}{h} \rfloor +
    1$.  Note that for any $t$, it holds $m_{t+1} \leq m_t$. Moreover,
    whenever Case~2 above happens for some $t$, it holds $m_{t+1} <
    m_t$.  Thus, Case~2 above can happen for $k > \tfrac{m}{h}$ only
    $m_{k'}$ times. Since for every $i \in [1 \dd m_{k'}]$ we have
    $A[P_{k'}[i]] \geq m$, by $\sum_{i \in [1 \dd m']}A[i] \in \bigO(m
    \log m)$ it holds $m_{k'} \in \bigO(\log m)$. The total time spend
    computing $M'_k$ for $k > \tfrac{m}{h}$ is thus bounded by
    $\bigO(m_{k'}(h + m_{k'}) + \sum_{k \geq k'}m_k) \subseteq
    \bigO(\log^2 m + \sum_{k \geq 0}m_k) \subseteq \bigO(m)$.
  \end{itemize}
  The total length of bitvectors $M'_k$ for $k \in [0 \dd k_{\max}]$,
  is $\sum_{k \in [0 \dd k_{\max}]}hm_k \in \bigO(hm)$. Thus,
  augmenting them all using \cref{th:binrksel} takes $\bigO((k_{\max}
  + 1) + \tfrac{1}{\log m}hm) \subseteq \bigO(m)$ time.
\end{proof}

\section{SA and ISA Queries}\label{sec:sa}

Let $\epsilon \in (0, 1)$ be any fixed constant and let $\T \in
\Alphabet^n$, where $2 \leq \sigma < n^{1/7}$. In this section, we
show how, given the packed representation of $\T$, in $\bigO(n \min(1,
\log \sigma / \sqrt{\log n}))$ time and using $\bigO(n / \log_{\sigma}
n)$ working space, to construct a data structure of size $\bigO(n /
\log_{\sigma} n)$ that answers SA and ISA queries in
$\bigO(\log^\epsilon n)$ time. We also derive a general reduction
depending on prefix rank and selection queries.

Let $\tau = \lfloor\mu \log_{\sigma} n\rfloor$, where $\mu$ is any
positive constant smaller than $\frac16$ such that $\tau \geq 1$ (such
$\mu$ exists by $\sigma < n^{1/7}$), be fixed for the duration of this
section. Throughout, we also use $\R$ as a shorthand for $\R(\tau,
\T)$.

\begin{definition}\label{def:position-periodicity}
  Let $j \in [1 \dd n]$. We call position $j$ \emph{periodic}
  if $j \in \R$. Otherwise, $j$ is \emph{nonperiodic}.
\end{definition}

\bfparagraph{Organization}

The structure and the query algorithm to compute $\SA[i]$
(resp.\ $\ISA[j]$), given any $i \in [1 \dd n]$ (resp.\ $j \in [1 \dd
n]$), are different depending on whether $\SA[i]$ (resp.\ $j$) is
periodic (\cref{def:position-periodicity}).  Our description is thus
split as follows. First (\cref{sec:sa-core}), we describe the set of
data structures called collectively the index ``core'' that enables
efficiently checking if $\SA[i] \in \R$ (resp.\ $j \in \R$); the core
also contain some common components utilized by the remaining parts.
In the following two parts
(\cref{sec:sa-nonperiodic,sec:sa-periodic}), we describe structures
handling each of the two cases. All ingredients are then put together
in \cref{sec:sa-final}. Finally, we present our result in the general
form (\cref{sec:sa-summary}).

\subsection{The Index Core}\label{sec:sa-core}

In this section, we present a data structure that, given any $j \in [1
\dd n]$ (resp.\ $i \in [1 \dd n]$), lets us in $\bigO(1)$ time determine
if $j \in \R$ (resp.\ $\SA[i] \in \R$).

The section is organized as follows. First, we introduce the
components of the data structure (\cref{sec:sa-core-ds}). Next, we
describe the query algorithms (\cref{sec:sa-core-nav}). Finally, we
show the construction algorithm (\cref{sec:sa-core-construction}).

\subsubsection{The Data Structure}\label{sec:sa-core-ds}

\bfparagraph{Definitions}

Let $\LTrange$ be a mapping from $X \in \Alphabet^{\leq 3\tau-1} :=
\{\emptystring\} \cup \Alphabet \cup \ldots \cup \Alphabet^{3\tau-1}$
to the pair of integers $(b, e) := (\LB(X, \T), \UB(X, \T))$. Let also
$\LTper$ denote the mapping from $\Alphabet^{3\tau - 1}$ to $\Zp$ such
that every $X$ is mapped to $\per(X)$.

Let $\BVshort[1 \dd n]$ be a bitvector defined such that $\BVshort[i]
= 1$ holds if and only if $i = n$, or $i < n$ and $X_{\SA[i]} \neq
X_{\SA[i+1]}$, where $X_j = \T[j \dd \min(n + 1, j + 3\tau - 1))$ for
every $j \in [1 \dd n]$.

Let $\ARRshort[1 \dd t]$ ($t \,{=}\, \rank{\BVshort}{1}{n}$) be
defined by $\ARRshort[i] \,{=}\, X_{\SA[j]}$, where $j =
\select{\BVshort}{1}{i}$.

\bfparagraph{Components}

The index core, denoted $\SACore(\T)$, consists of five components:
\begin{enumerate}
\item The packed representation of $\T$ using $\bigO(n / \log_{\sigma}
  n)$ space.
\item The lookup table $\LTrange$. When accessing $\LTrange$, strings
  $X \in \Alphabet^{\leq 3\tau - 1}$ are converted to  small integers
  using the mapping $\Int(X)$ defined in the proof of \cref{pr:meta-trie}.
  By $\Int(X) \in [0 \dd \sigma^{6\tau})$, $\LTrange$ needs
  $\bigO(\sigma^{6\tau}) = \bigO(n^{6\mu}) = \bigO(n / \log_{\sigma}
  n)$ space.
\item The lookup table $\LTper$. Similarly as above, we utilize
  the mapping $\Int(X)$. $\LTper$ thus also needs
  $\bigO(\sigma^{6\tau}) = \bigO(n / \log_{\sigma} n)$ space.
\item The bitvector $\BVshort$ augmented using \cref{th:binrksel} for
  rank and selection queries. The augmented bitvector takes $\bigO(n /
  \log n)$ space.
\item The array $\ARRshort$. Every string $X \in \{\ARRshort[i]\}_{i
  \in [1 \dd t]}$ is encoded as $\Int(X)$ using $6\tau\log \sigma =
  \bigO(\log n)$ bits.  This implicitly encodes the length of the
  string and ensures that all strings are encoded using the same
  number of bits.  By $\{\ARRshort[i]\}_{i \in [1 \dd t]} \sub
  \Alphabet^{\leq 3\tau-1}$, we have $t = \bigO(n^{1/2})$, and hence
  the array $\ARRshort$ needs $\bigO(n^{1/2}) = \bigO(n / \log n)$
  space.
\end{enumerate}

In total, $\SACore(\T)$ takes $\bigO(n / \log_{\sigma} n)$ space.

\subsubsection{Navigation Primitives}\label{sec:sa-core-nav}

\begin{proposition}\label{pr:sa-core-isa}
  Given $\SACore(\T)$, for any $j \in [1 \dd n]$ we can in $\bigO(1)$
  time determine if $j \in \R$.
\end{proposition}
\begin{proof}
  If $j > n - 3\tau + 2$, we return that $j \not\in \R$
  (\cref{def:sss}).  Otherwise, we use the packed encoding of $\T$ to
  extract $X = \T[j \dd j + 3\tau - 1)$ in $\bigO(1)$ time and convert
  it to $x = \Int(X)$. We then use the lookup table $\LTper$, to
  determine $p = \per(X)$, and return that $j \in \R$ if $p \leq
  \tfrac13\tau$.
\end{proof}

\begin{proposition}\label{pr:sa-core-sa}
  Given $\SACore(\T)$, for any $i \in [1 \dd n]$ we can in $\bigO(1)$
  time determine if $\SA[i] \in \R$.
\end{proposition}
\begin{proof}
  Given the position $i \in [1 \dd n]$, we first compute $y =
  \rank{\BVshort}{1}{i-1}$. The string $X = \ARRshort[y + 1]$ is then a
  prefix of $\T[\SA[i] \dd n]$. If $|X| < 3\tau - 1$, we must have
  $\SA[i] > n - 3\tau + 2$, and thus we return that $\SA[i] \in [1 \dd
  n] \setminus \R$ (see \cref{def:sss}). Otherwise (i.e., $|X| = 3\tau
  - 1$), using $\LTper$ we determine $p = \per(X)$ and return that
  $\SA[i] \in \R$ if $p \leq \frac13\tau$.
\end{proof}

\subsubsection{Construction Algorithm}\label{sec:sa-core-construction}

\begin{proposition}\label{pr:sa-core-construction}
  Given the packed representation of $\T \in \Alphabet^n$, we can
  construct $\SACore(\T)$ in $\bigO(n / \log_{\sigma} n)$ time.
\end{proposition}
\begin{proof}

  To compute $\LTrange$, we first compute for every $X \in
  \Alphabet^{\ell}$ (where $\ell = 3\tau - 1$), its frequency $f_X :=
  |\Occ(X, \T)|$. Using the simple generalization of the algorithm
  described in~\cite[Section 6.1.2]{sss}, this takes $\bigO(n /
  \log_{\sigma} n)$ time (note that the algorithm requires $\ell
  \sigma^{2\ell-1} = O(n / \log_{\sigma} n)$, which is satisfied here,
  since $2\ell-1 < 6\mu\log_{\sigma} n$ and $\mu < \frac16$). From the
  frequencies of $X \in \Alphabet^{3\tau-1}$ we then compute the
  values of $f_X$ for all $X \in \Alphabet^{<3\tau-1}$ by observing
  that unless $X$ is a nonempty suffix of $\T$, it holds $f_{X} =
  \sum_{c \in \Alphabet}f_{Xc}$, i.e., the frequency of each string
  shorter than $3\tau-1$ is obtained in $\bigO(\sigma)$ time. If $X$
  is a nonempty suffix of $\T$ (which we can check in $\bigO(1)$
  time), we additionally add one to the count.  Since each string
  contributes exactly once to the frequency of another string, over
  all $X \in \Alphabet^{<3\tau-1}$, this takes
  $\bigO(\sigma^{3\tau-1}) = \bigO(n/\log_{\sigma} n)$ time. Once
  $f_X$ is computed for all $X \in \Alphabet^{\leq 3\tau-1}$, we
  compute $\LTrange$ as follows. Denote $\Sigma = \Alphabet$.  Assume
  that $\LTrange[\Int(X)] = (b, e)$ holds for some $X \in
  \Alphabet^{<3\tau-1}$. Then, for any $c \in \Sigma$, it holds
  $\LTrange[\Int(Xc)] = (e - x - f_{Xc}, e - x)$, where $x =
  \sum_{c'\in\Sigma,c'> c}f_{Xc'}$, e.g., for $\sigma = 2$,
  $\LTrange[\Int(X0)] = (e - f_{X1} - f_{X0}, e - f_{X1})$.  We thus
  compute $\LTrange[\Int(X)]$ by initializing
  $\LTrange[\Int(\emptystring)] = (0, n)$, and then enumerating all $X
  \in \Alphabet^{\leq 3\tau-1}$ in the order of non-decreasing length
  (and, in case of ties, in the reverse lexicographical order). During
  the enumeration of strings of the form $Xc$, where $c \in \Sigma$,
  we maintain the sum $x = \sum_{c'\in\Sigma,c'>c}f_{Xc'}$.  Then,
  using the above formula, the value of $\LTrange[\Int(Xc)]$ can be
  obtained in $\bigO(1)$ time. Over all $X$, the computation of
  $\LTrange[\Int(X)]$ thus takes $\bigO(\sigma^{3\tau-1}) = \bigO(n /
  \log_{\sigma} n)$ time.

  To construct $\LTper$, we enumerate all $X \in \Alphabet^{3\tau -
  1}$, and for each $X$ in $\bigO(\tau^2)$ time we compute $\per(X)$
  by trying all $\ell \in [1 \dd 3\tau - 1)$.  Initializing $\LTper$
  takes $\bigO(\sigma^{6\tau}) = \bigO(n / \log_{\sigma} n)$. Over all
  $X \in \Alphabet^{3\tau - 1}$, we spend $\bigO(\sigma^{3\tau - 1}
  \tau^2) = \bigO(n^{1/2} \log^2 n) = \bigO(n / \log_{\sigma} n)$
  time.

  We finish with the construction of $\BVshort$ and $\ARRshort$.
  First, in $\bigO(n / \log n)$ time we initialize $\BVshort$ to
  zeros. Next, we initialize temporary counters $k$ and $f$ to zero,
  and simulate a preorder traversal of the trie of $\Alphabet^{3\tau -
  1}$. For each visited node with label $X$, we consider two cases:
  \begin{itemize}
  \item If $|X| < 3\tau - 1$, we check if $X$ is a suffix of $\T$. If
    so, increment $k$ and $f$ by one, and~report~$X$.
  \item Otherwise (i.e., if $|X| = 3\tau - 1$), if $f_X > 0$, we
    increment $k$ by one, $f$ by $f_X$, and report $X$.
  \end{itemize}
  Each time some string $X$ is reported, we set $\ARRshort[k] = X$ and
  $\BVshort[f] = 1$. The correctness of this procedure follows by
  noting that labels of nodes visited during the preorder traversal
  are lexicographically sorted.  The traversal takes
  $\bigO(\sigma^{3\tau}) = \bigO(n^{3\mu}) = \bigO(n / \log n)$
  time. Finally, using \cref{th:binrksel}, in $\bigO(n/\log n)$ time
  we augment $\BVshort$ with $\bigO(1)$-time rank and select queries.
\end{proof}

\subsection{The Nonperiodic Positions}\label{sec:sa-nonperiodic}

In this section, we describe a data structure that, given any $j \in [1
\dd n]$ (resp. $i \in [1 \dd n]$) satisfying $j \in [1 \dd n] \sm
\R$ (resp.\ $\SA[i] \in [1 \dd n] \sm \R$) computes
$\ISA[j]$ (resp.\ $\SA[i]$) in $\bigO(\log^{\epsilon} n)$ time.

The section is organized as follows. First, we introduce the
components of the data structure (\cref{sec:sa-nonperiodic-ds}). Next,
we describe the query algorithms
(\cref{sec:sa-nonperiodic-isa,sec:sa-nonperiodic-sa}). Finally, we
show the construction algorithm
(\cref{sec:sa-nonperiodic-construction}).

\subsubsection{The Data Structure}\label{sec:sa-nonperiodic-ds}

\bfparagraph{Definitions}

We fix some $\tau$-synchronizing set $\S$ of $\T$ obtained using
\cref{th:sss-packed-construction} (recall, that $\tau = \lfloor
\mu\log_{\sigma} n \rfloor$ is fixed for \cref{sec:sa}). We denote $n'
= |\S| = \bigO(n/\tau)$. Let $(\stext_t)_{t \in [1 \dd n']}$ be the
sequence containing the elements of $\S$ in sorted order, i.e., if
$i < j$ then $\stext_i < \stext_j$.  Let also
$(\slex_t)_{t \in [1 \dd n']}$ be the sequence containing elements
of $\S$ sorted according to the lexicographical order of the
corresponding suffixes, i.e., if $i < j$ then
$\T[\slex_i \dd n] \prec \T[\slex_j \dd n]$. Let $W[1 \dd n']$ be
a sequence of length-$3\tau$ strings such that $W[i] = \revstr{X_i}$,
where $X_i = \T^{\infty}[\slex_i - \tau \dd \slex_i + 2\tau)$.

For any $i \in [1 \dd n {-} 2\tau {+} 1]$, we define
$\Successor_{\S}(i) = \min\{j \in \S \cup \{n {-} 2\tau {+} 2\} : j
\geq i\}$ and denote $\D := \{\T[i \dd \Successor_{\S}(i) + 2\tau) : i
\in [1 \dd n {-} 3\tau {+} 2] \sm \R\}$.  Let $\LTD$ be a mapping from
$\Alphabet^{3\tau - 1}$ to $\Alphabet^{\leq 3\tau - 1}$ such that for
any $Y \in \Alphabet^{3\tau - 1}$, if there exists $X \in \D$ that is
a prefix of $Y$ (by the consistency of $\S$, there can be at most one
such $X$), then $\LTD$ maps $Y$ to $X$. Otherwise (i.e., there is no
such $X$), $\LTD$ maps $Y$ to $\emptystring$. Let $\LTrev$ be a
mapping that for every string $X \in \Alphabet^{\leq 3\tau - 1}$,
returns the packed representation of $\revstr{X}$.

Let $\BVS[1 \dd n]$ be a bitvector defined so that $\BVS[i] = 1$ holds
if and only if $i \in \S$.

Let $\ARRsmap[1 \dd n']$ be an array storing a permutation of $[1 \dd
n']$ such that $\ARRsmap[i] = j$ holds if and only if $\stext_j =
\slex_i$. Let $\ARRsinvmap[1 \dd n']$ be an array storing a permutation
of $[1 \dd n']$ such that $\ARRsinvmap[j] = i$ holds if and only if
$\stext_j = \slex_i$.

\bfparagraph{Components}

The data structure to handle nonperiodic positions consists of seven
components:
\begin{enumerate}
\item The index core $\SACore(\T)$ (\cref{sec:sa-core-ds}). It takes
  $\bigO(n / \log_{\sigma} n)$ space.
\item The lookup table $\LTrev$. When accessing $\LTrev$, strings
  $X \in \Alphabet^{\leq 3\tau - 1}$ are converted to $\Int(X)$. Thus,
  the mapping $\LTrev$ needs $\bigO(\sigma^{6\tau}) = \bigO(n^{6\mu})
  = \bigO(n / \log_{\sigma} n)$ space.
\item The lookup table $\LTD$. As above, $\LTD$ needs
  $\bigO(\sigma^{6\tau}) = \bigO(n / \log_{\sigma} n)$ space.
\item The bitvector $\BVS$ augmented using \cref{th:binrksel}.
  It needs $\bigO(n/\log n)$ space.
\item The array $\ARRsmap[1 \dd n']$ in plain form, using $n' =
  \bigO(n / \log_{\sigma} n)$ words of space.
\item The array $\ARRsinvmap[1 \dd n']$ in plain form, using $n' =
  \bigO(n / \log_{\sigma} n)$ words of space.
\item The data structure of \cref{th:wavelet-tree} for the sequence
  $W[1 \dd n']$. By $n' = \bigO(n / \log_{\sigma} n)$ and
  $\sigma^{3\tau} = \bigO(\sqrt{n}) = o(n / \log n)$, it needs
  $\bigO(n / \log_{\sigma} n)$ space.
\end{enumerate}

In total, the data structure takes $\bigO(n / \log_{\sigma} n)$ space.

\subsubsection{Implementation of
  \texorpdfstring{$\ISA$}{ISA} Queries}\label{sec:sa-nonperiodic-isa}

For any $j \in [1 \dd n - 3\tau + 2] \sm \R$, letting $X \in \D$ be a
prefix of $\T[j \dd n]$ (by~\cite[Lemma~6.1]{sss}, $\D$ is
prefix-free, and hence there is exactly one such $X$), we define
\[
  \Pos(j) = \{j' \in [1 \dd n] : \LCE(j, j') \geq |X|\text{ and }\T[j'
  \dd n] \preceq \T[j \dd n]\},
\]
and denote $\delta(j) := |\Pos(j)|$.

\begin{lemma}\label{lm:sa-nonperiodic-isa}
  Let $j \in [1 \dd n - 3\tau + 2] \sm \R$ and $X \in \D$ be a prefix
  of $\T[j \dd n]$. Denote $\deltatext = |X| - 2\tau$ and $b_X =
  \LB(X, \T)$. Then:
  \begin{enumerate}
  \item\label{lm:sa-nonperiodic-isa-it-1} It holds $\ISA[j] =
    b_X + \delta(j)$.
  \item\label{lm:sa-nonperiodic-isa-it-2} If $y \in [1 \dd n']$ is
    such that $\slex_y = j + \deltatext$, then $\delta(j) =
    \rank{W}{\revstr{X}}{y}$.
  \end{enumerate}
\end{lemma}
\begin{proof}
  1. Observe that $j' \in \Occ(X, \T)$ holds if and only if $\LCE(j,
  j') \geq |X|$.  Thus, by definition of $\ISA$, we have $\ISA[j] =
  \LB(X, \T) + |\{j' \in \Occ(X, \T) : \T[j' \dd n] \preceq \T[j \dd
  n]\}| = b_X + |\{j' \in [1 \dd n] : \LCE(j, j') \geq |X|\text{ and
  }\T[j' \dd n] \preceq \T[j \dd n]\}| = b_X + \delta(j)$.

  2. Denote $s = j + \deltatext$.  By definition of $\D$, we have $s
  \in \S$. By the consistency of $\S$, there exists a bijection (given
  by the mapping $j' \mapsto j' + \deltatext$) between positions $j'
  \in [1 \dd n] \sm \R$ satisfying $\T[j' \dd \Successor_{\S}(j') +
  2\tau) = X$ and $\T[j' \dd n] \preceq \T[j \dd n]$, and positions
  $s' \in \S$ such that $\T^{\infty}[s' - \deltatext \dd s' + 2\tau) =
  X$ and $\T[s' \dd n] \preceq \T[s \dd n]$.  Thus, letting $y \in [1
  \dd n']$ be such that $\slex_y = s$, we obtain that $\delta(j) =
  |\{i \in [1 \dd y] : \T^{\infty}[\slex_i - \deltatext \dd \slex_i +
  2\tau) = X\}|$. Since we defined $W[i] = \revstr{X_i}$, where
  $X_i = \T^{\infty}[\slex_i - \tau \dd \slex_i + 2\tau)$, we
  conclude that $\delta(j) = \rank{W}{\revstr{X}}{y}$.
\end{proof}

\begin{proposition}\label{pr:sa-nonperiodic-isa}
  Let $j \in [1 \dd n]$ be such that $j \in [1 \dd n] \setminus \R$.
  Given the data structure from \cref{sec:sa-nonperiodic-ds} and the
  position $j$, we can compute $\ISA[j]$ in $\bigO(\log^\epsilon n)$
  time.
\end{proposition}
\begin{proof}
  Given $j \in [1 \dd n] \sm \R$, we compute $\ISA[j]$ as follows.  If
  $j > n - 3\tau + 2$, then letting $X = \T[j \dd n]$, in $\bigO(1)$
  time we compute $(b_X, e_X) = (\LB(X, \T), \UB(X, \T))$ using the
  lookup table $\LTrange$. By definition of the lexicographical order,
  we then have $\SA[b + 1] = j$, and hence we return $\ISA[j] = b +
  1$. Let us thus assume $j \leq n - 3\tau - 2$.  By $j \not\in \R$
  and the density condition of $\S$ (see \cref{def:sss}), this implies
  that $\S \cap [j \dd j + \tau) \neq \emptyset$. In $\bigO(1)$ time
  we compute $x = \rank{\BVS}{1}{j-1}$. Then, in $\bigO(1)$ we compute $s
  = \select{\BVS}{1}{x+1} = \Successor_{\S}(j) \in \S$. We then have $X =
  \T[j \dd s + 2\tau) \in \D$, and in particular, $|X| = s + 2\tau -
  j$.  In $\bigO(1)$ time we lookup $(b_X, e_X) = \LTrange[\Int(X)]$,
  i.e., $b_X = \LB(X, \T)$. Letting $y = \ARRsinvmap[x + 1]$, we then
  have $\slex_y = s = j + |X| - 2\tau$. By
  \cref{lm:sa-nonperiodic-isa}, it thus remains to determine
  $\rank{W}{\revstr{X}}{y}$.  In $\bigO(1)$ time we compute
  $\revstr{X}$ using the lookup table $\LTrev$. In
  $\bigO(\log^{\epsilon} n)$ time, we then compute $\delta(j) =
  \rank{W}{\revstr{X}}{y}$ using \cref{th:wavelet-tree}, and finally
  return $\ISA[j] = b_X + \delta(j)$.
\end{proof}

\subsubsection{Implementation of
  \texorpdfstring{$\SA$}{SA} Queries}\label{sec:sa-nonperiodic-sa}

\begin{lemma}\label{lm:sa-nonperiodic-sa}
  Let $i \in [1 \dd n]$ be such that $\SA[i] \in [1 \dd n - 3\tau + 2]
  \setminus \R$ and $X \in \D$ be a prefix of $\T[\SA[i] \dd
  n]$. Denote $\deltatext = |X| - 2\tau$ and $b_X = \LB(X,
  \T)$. Then:
  \begin{enumerate}
  \item\label{lm:sa-nonperiodic-sa-it-1} It holds $i = b_X +
    \delta(\SA[i])$.
  \item\label{lm:sa-nonperiodic-sa-it-2} If $y = \select{W}{\revstr{X}}{i
    - b_X}$, then $\slex_{y} = \SA[i] + \deltatext$.
  \end{enumerate}
\end{lemma}
\begin{proof}
  1. Denote $j = \SA[i]$. By
  \cref{lm:sa-nonperiodic-isa}\eqref{lm:sa-nonperiodic-isa-it-1}, $i =
  \ISA[j] = b_X + \delta(j) = b_X + \delta(\SA[i])$.

  2. By the consistency of $\S$, we have $\SA[i] + \deltatext \in
  \S$. Thus, there exists $y \in [1 \dd n']$ such that $\slex_y =
  \SA[i] + \deltatext$.  By
  \cref{lm:sa-nonperiodic-isa}\eqref{lm:sa-nonperiodic-isa-it-2}
  applied for $j = \SA[i]$, for any such $y$ it holds $\delta(\SA[i])
  = \rank{W}{\revstr{X}}{y}$. By~\eqref{lm:sa-nonperiodic-sa-it-1}, we
  thus have $i - b_X = \rank{W}{\revstr{X}}{y}$.  Since $X$ is a
  prefix of $\T[\SA[i] \dd n]$, such $y$ must also satisfy $\T[\slex_y
  - \deltatext \dd \slex_y + 2\tau) = X$, or equivalently,
  $\revstr{X}$ must be a prefix of $W[y]$.  The only $y \in [1 \dd
  n']$ for which $\revstr{X}$ is a prefix of $W[y]$ and that
  satisfies $\rank{W}{\revstr{X}}{y} = i - b_X$, is $y =
  \select{W}{\revstr{X}}{i - b_X}$.
\end{proof}

\begin{proposition}\label{pr:sa-nonperiodic-sa}
  Let $i \in [1 \dd n]$ be such that $\SA[i] \in [1 \dd n] \setminus
  \R$.  Given the data structure from \cref{sec:sa-nonperiodic-ds} and
  the index $i$, we can compute $\SA[i]$ in $\bigO(\log^\epsilon n)$
  time.
\end{proposition}
\begin{proof}
  Given $i \in [1 \dd n]$ such that $\SA[i] \in [1 \dd n] \setminus
  \R$, we compute $\SA[i]$ as follows.  First, we compute $y =
  \rank{\BVshort}{1}{i-1}$.  The string $Y = \ARRshort[y + 1]$ is then a
  prefix of $\T[\SA[i] \dd n]$ of length $\min(3\tau - 1, n + 1 -
  i)$. If $|Y| < 3\tau - 1$, we therefore have $\SA[i] > n - 3\tau +
  2$ and moreover, $\SA[i] + |Y| = n + 1$. Thus, we return $\SA[i] = n
  + 1 - |Y|$. Otherwise (i.e., $|Y| = 3\tau - 1$), using $\LTD$ on $Y$
  we determine $x = \Int(X)$, where $X \in \D$ is a prefix of
  $\T[\SA[i] \dd n]$.  In $\bigO(1)$ time we lookup $(b_X, e_X) =
  \LTrange[x]$, i.e., $b_X = \LB(X, \T)$.  In $\bigO(\log^{\epsilon}
  n)$ time, we then compute $y = \select{W}{\revstr{X}}{i - b_X}$
  using \cref{th:wavelet-tree} (the packed representation of
  $\revstr{X}$ is obtained using the lookup table $\LTrev$ in
  $\bigO(1)$ time).  By
  \cref{lm:sa-nonperiodic-sa}\eqref{lm:sa-nonperiodic-sa-it-2}, we
  then have $\SA[i] = \slex_y - \deltatext$, where $\deltatext = |X| -
  2\tau$.  Using $\BVS$, in $\bigO(1)$ time we compute $j' =
  \select{\BVS}{1}{\ARRsmap[y]}$. We then have $j' = \slex_y$, and
  hence we return $\SA[i] = j' - \deltatext$. Altogether, the query
  takes $\bigO(\log^{\epsilon} n)$ time.
\end{proof}

\subsubsection{Construction Algorithm}\label{sec:sa-nonperiodic-construction}

\begin{proposition}\label{pr:sa-nonperiodic-construction}
  Given $\SACore(\T)$, we can augment it into a data structure from
  \cref{sec:sa-nonperiodic-ds} in $\bigO(n \min(1, \log \sigma /
  \sqrt{\log n}))$ time and using $\bigO(n / \log_{\sigma} n)$ working
  space.
\end{proposition}
\begin{proof}

  First, using \cref{th:sss-packed-construction}, we construct a
  $\tau$-synchronizing set $\S$ of size $\bigO(n / \tau)$ in
  $\bigO(n / \tau) = \bigO(n / \log_{\sigma} n)$ time from a packed
  representation of $\T$. The set $\S$ is returned as an array taking
  $\bigO(n / \log_{\sigma} n)$ space. Using this array, we initialize
  the bitvector $\BVS$ in $\bigO(n / \log_{\sigma} n)$ time.
  Augmenting $\BVS$ with \cref{th:binrksel} takes $\bigO(n / \log n)$
  time.

  Next, we construct the arrays $\ARRsinvmap$ and $\ARRsmap$.  We start
  by creating the sequence $(\stext_t)_{t \in [1 \dd n']}$ using
  select queries on $\BVS$. This takes $\bigO(n / \log_{\sigma} n)$ time.
  Then, given $(\stext_t)_{t \in [1 \dd n']}$, and the packed
  representation of $\T$, by~\cite[Theorem 4.3]{sss}, we compute the
  sequence $(\slex_t)_{t \in [1 \dd n']}$ in $\bigO(n / \log_{\sigma}
  n)$ time.  Given $(\slex_t)_{t \in [1 \dd n']}$, we then easily
  obtain the arrays $\ARRsinvmap$ and $\ARRsmap$: simply scan the
  sequence $(\slex_t)_{t \in [1 \dd n']}$ and for each $i \in [1 \dd
  n']$, let $j = \rank{\BVS}{1}{\slex_i}$ and note that then $\stext_j =
  \slex_i$ and hence we can set $\ARRsinvmap[j] = i$ and $\ARRsmap[i] =
  j$.

  Next, we initialize $\LTrev$.  In the RAM model, such array is
  easily initialized in $\bigO(\sigma^{6\tau}) =
  \bigO(n / \log_{\sigma} n)$ time. The sequence $W[1 \dd n']$
  is then obtained from $(\slex_t)_{t \in [1 \dd n']}$
  using $\LTrev$ in $\bigO(n / \log_{\sigma} n)$ time.  We then
  process $W$ using \cref{th:wavelet-tree}, which takes $\bigO(n
  \min(1, \log \sigma/\sqrt{\log n}))$ time and
  $\bigO(n/\log_{\sigma} n)$ working space.

  Finally, to construct $\LTD$, we first compute a lookup table that
  for every $Z \in \Alphabet^{2\tau}$ tells whether $\T[j \dd j +
  2\tau) = Z$ implies $j \in \S$ (by consistency of $\S$ this does not
  depend on $j$). Given the array containing the positions in $\S$ and
  the packed representation of $\T$, this takes $\bigO(\sigma^{2\tau}
  + |\S|) = \bigO(n/\log_{\sigma} n)$ time. Given such lookup table,
  we iterate through every $Y \in \Alphabet^{3\tau - 1}$ and in
  $\bigO(\tau)$ time we compute the shortest prefix $X$ of $Y$ whose
  length-$2\tau$ suffix is marked true in the above lookup table. If
  such $X$ exists, we have $X \in \D$. Accounting for the
  initialization of $\LTD$, over all $Y \in \Alphabet^{3\tau - 1}$,
  this takes $\bigO(\sigma^{6\tau} + \sigma^{3\tau - 1} \log_{\sigma}
  n) = \bigO(n / \log_{\sigma} n)$ time.
\end{proof}

\subsection{The Periodic Positions}\label{sec:sa-periodic}

In this section, we describe a data structure that, given any $j \in [1
\dd n]$ (resp. $i \in [1 \dd n]$) satisfying $j \in \R$
(resp.\ $\SA[i] \in \R$) computes $\ISA[j]$ (resp.\ $\SA[i]$)
in $\bigO(\log \log n)$ time.

The section is organized as follows. First, we present the toolbox of
combinatorial properties for periodic positions
(\cref{sec:sa-periodic-prelim}). Next, we introduce the components of
the data structure (\cref{sec:sa-periodic-ds}).  We then show how
using this structure to implement some basic navigational routines
(\cref{sec:sa-periodic-nav}).  Next, we describe the query algorithms
(\cref{sec:sa-periodic-isa,sec:sa-periodic-sa}). Finally, we show the
construction algorithm (\cref{sec:sa-periodic-construction}).

\subsubsection{Preliminaries}\label{sec:sa-periodic-prelim}

We start by introducing the definitions to express the properties
utilized in our data structures. For any $j \in \R$, we define
$\Lroot(j) = \min\{\T[j + t \dd j + t + p) : t \in [0 \dd p)\}$, where
$p = \per(\T[j \dd j + 3\tau - 1))$. We denote $\Lroots = \{\Lroot(j)
: j \in \R\}$.  For any $j \in \R$, let $\rend{j} = \min\{j' \geq j :
j' \not\in \R\} + 3\tau - 2$.

\begin{lemma}\label{lm:run-end}
  Let $j \in \R$ and $p = \per(\T[j \dd j + 3\tau - 1))$. Then:
  \begin{enumerate}
  \item\label{lm:run-end-it-1} If $j + 1 \in \R$ then $\per(\T[j + 1
    \dd j + 3\tau)) = p$,
  \item\label{lm:run-end-it-2} It holds $\rend{j} = j + p + \LCE(j, j
    + p)$.
  \end{enumerate}
\end{lemma}
\begin{proof}
  1. Denote $\Pat = \T[j \dd j + 3\tau - 1)$, $\Pat' = \T[j + 1 \dd j
  + 3\tau)$, and $p' = \per(\Pat')$.  Our goal is to show that $p' =
  p$. For a proof by contradiction, assume $p' \neq p$. By the
  assumption, $\per(\Pat) = p$. Denote $Y = \Pat'[1 \dd \tau]$, and
  note that since $\Pat$ and $\Pat'$ overlap by $3\tau - 2 \geq \tau$
  symbols, $Y$ is a substring of $\Pat$, and hence has periods $p$ and
  $p'$.  Observe that we cannot have $p \mid p'$ since this would
  imply that $Y[1 \dd p']$ is not primitive which would contradict $p'
  = \per(\Pat')$. Observe now that we have $p, p' \leq
  \tfrac{1}{3}\tau$. By the Weak Periodicity
  Lemma~\cite{fine1965uniqueness}, we thus have that $Y$ has period
  $p'' = \gcd(p, p')$. By our assumptions, this implies $p'' < p'$ and
  $p'' \mid p'$. Thus, again we obtain that $Y[1 \dd p']$ is not
  primitive. Therefore, we must have $p' = p$.

  2. Denote $j' = \rend{j} - 3\tau + 2$. By definition, we then have
  $[j \dd j') \sub \R$ and $j' \not\in \R$. By the above, for every $t
  \in [0 \dd j'-j)$, it holds $\per(\T[j + t \dd j + t + 3\tau - 1)) =
  p$. Thus, for every $j'' \in [j \dd j' + 3\tau - 2 - p)$, we have
  $\T[j''] = \T[j'' + p]$, i.e., the substring $\T[j \dd j' + 3\tau -
  2)$ has period $p$, and thus $\LCE(j, j + p) \geq (j' + 3\tau - 2) -
  j - p$, or equivalently, $j + p + \LCE(j, j + p) \geq j' + 3\tau - 2
  = \rend{j}$. To show that this lower bound on $j + p + \LCE(j, j +
  p)$ is tight, let us assume that $\rend{j} \leq n$ (otherwise, the
  claim follows immediately). Equivalently, we then have $j' + 3\tau -
  2 = \rend{j} \leq n$ and to finish the proof, it remains to show
  $\T[\rend{j}] \neq \T[\rend{j} - p]$.  Recall that $\per(\T[j' - 1
  \dd j' + 3\tau - 2)) = p$. Thus, $\T[j' + 3\tau - 2] = \T[\rend{j}]
  = \T[\rend{j} - p] = \T[j' + 3\tau - 2 - p]$ would imply that
  $\per(\T[j' \dd j' + 3\tau - 1)) = p$, or equivalently, that $j' \in
  \R$, a contradiction.
\end{proof}

Observe that by definition of $\Lroot$, letting $p = |\Lroot(j)|$,
there exists $s \in [0 \dd p)$ such that $\T[j + s \dd j + s
+ p) = \Lroot(j)$.  Combining this with \cref{lm:run-end} implies that
for every $j \in \R$, we can write $\T[j \dd \rend{j}) = H'H^{k}H''$,
where $H = \Lroot(j)$, and $H'$ (resp.\ $H''$) is a proper suffix
(resp.\ prefix) of $H$. We call such factorization the
\emph{L-decomposition} of $\T[j \dd \rend{j})$.  Note that the
L-decomposition is unique, since otherwise would contradict the
synchronization property of primitive
strings~\cite[Lemma~1.11]{AlgorithmsOnStrings}.  We denote
$\Lhead(j)=|H'|$, $\Lexp(j)=k$, and $\Ltail(j)=|H''|$.  For $j \in
\R$, we let $\type(j) = +1$ if $\rend{j} \leq n$ and $\T[\rend{j}]
\succ \T[\rend{j} - p]$ (where $p = |\Lroot(j)|$), and $\type(j) = -
1$ otherwise.  For any $j \in \R$, we denote $\rendfull{j} = \rend{j}
- \Ltail(j)$. Observe that $\rendfull{j} = j + \Lhead(j) + \Lexp(j)
\cdot |\Lroot(j)|$.

We repeatedly refer to the following subsets of $\R$.  First, denote
$\R^{-} = \{j \in \R : \type(j) = -1\}$ and $\R^{+} = \R \setminus
\R^{-}$.  For any $H \in \Sigma^{+}$ and any $s \in \Zz$ we then let
$\R_{H} = \{j \in \R : \Lroot(j) = H\}$, $\R^{-}_H = \R^{-} \cap
\R_H$, $\R^{+}_H = \R^{+} \cap \R_H$, $\R_{s,H} = \{j \in \R_H :
\Lhead(j)=s\}$, $\R^{-}_{s,H} = \R^{-} \cap \R_{s,H}$, and
$\R^{+}_{s,H} = \R^{+} \cap \R_{s,H}$.

The following lemmas establish the key properties of periodic
positions. First, we prove that the set of positions $\R_{s,H}$
occupies a contiguous block in $\SA$ and describe the structure of
such block.

\begin{lemma}\label{lm:lce}
  Let $j \in \R_{s,H}$. For any $j' \in [1 \dd n]$, $\LCE(j, j') \geq
  3\tau - 1$ holds if and only if $j' \in \R_{s,H}$.  Moreover, if
  $j'\in \R_{s,H}$ then, letting $t = \rend{j} - j$ and $t' =
  \rend{j'} - j'$, it holds $\LCE(j, j') \geq \min(t, t')$ and:
  \begin{enumerate}
  \item\label{lm:lce-it-1} If $\type(j) \neq \type(j')$, then $\T[j
    \dd n] \prec \T[j' \dd n]$ if and only if $\type(j) < \type(j')$,
  \item\label{lm:lce-it-2} If $\type(j) = \type(j') = -1$ and $t \neq
    t'$, then $\T[j \dd n] \prec \T[j' \dd n]$ if and only if $t <
    t'$,
  \item\label{lm:lce-it-3} If $\type(j) = \type(j') = +1$ and $t \neq
    t'$, then $\T[j \dd n] \prec \T[j' \dd n]$ if and only if $t >
    t'$,
  \item If $\type(j) \neq \type(j')$ or $t \neq
    t'$, then $\LCE(j, j') = \min(t, t')$.
  \end{enumerate}
\end{lemma}
\begin{proof}

  Let $j' \in [1 \dd n]$ be such that $\LCE(j, j') \geq 3\tau -
  1$. Denoting $p = \per(\T[j \dd j + 3\tau - 1))$ and $p' =
  \per(\T[j' \dd j' + 3\tau - 1))$ we then have $p' = p \leq
  \tfrac{1}{3}$. Thus, $j' \in \R$ and $\Lroot(j') = \min\{\T[j' + t
  \dd j' + t + p') : t \in [0 \dd p')\} = \min\{\T[j' + t \dd j' + t +
  p) : t \in [0 \dd p)\} = \min\{\T[j + t \dd j + t + p) : t \in [0
  \dd p)\} = H$.  To show that $\Lhead(j') = s$, note that by $|H|
  \leq \tau$, the string $H'H^2$ (where $H'$ is a length-$s$ suffix of
  $H$) is a prefix of $\T[j \dd j + 3\tau - 1) = \T[j' \dd j' + 3\tau
  - 1)$.  On the other hand, $\Lhead(j') = s'$ implies that
  $\widehat{H}'H^2$ (where $\widehat{H}'$ is a length-$s'$ suffix of
  $H$) is a prefix of $\T[j' \dd j' + 3\tau - 1)$.  Thus, by the
  synchronization property of primitive
  strings~\cite[Lemma~1.11]{AlgorithmsOnStrings} applied to the two
  copies of $H$, we have $s' = s$, and consequently, $j' \in
  \R_{s,H}$.  For the converse implication, assume $j' \in
  \R_{s,H}$. This implies that both $\T[j \dd \rend{j})$ and $\T[j'
  \dd \rend{j'})$ are prefixes of $H'\cdot H^{\infty}[1 \dd)$
  (where $H'$ is as above). Thus, by $\rend{j} - j,\, \rend{j'} -
  j' \geq 3\tau - 1$, we obtain $\LCE(j, j') \geq 3\tau - 1$.

  Let us now assume $j' \in \R_{s,H}$. Since, as noted above, both
  $\T[j \dd \rend{j}) = \T[j \dd j + t)$ and $\T[j' \dd \rend{j'}) =
  \T[j' \dd j' + t')$ are prefixes of $H'\cdot H^{\infty}[1 \dd)$,
  we have $\LCE(j, j') \geq \min(t, t')$.

  1. Assume $\type(j) < \type(j')$. Let $Q = H'\cdot H^{\infty}[1 \dd)$,
  where $H'$ is a length-$s$ suffix of $H$. We will prove $\T[j \dd n]
  \prec Q \prec \T[j' \dd n]$, which implies the claim. First, we note
  that $\type(j) = -1$ implies that either $\rend{j} = n + 1$, or
  $\rend{j} \leq n$ and $\T[\rend{j}] \prec \T[\rend{j} - |H|]$. In the
  first case, $\T[j \dd \rend{j}) = \T[j \dd n]$ is a proper prefix of
  $Q$ and hence $\T[j \dd n] \prec Q$. In the second case, we have $\T[j
  \dd \rend{j}) = \T[j \dd j + t) = Q[1 \dd t]$ and $\T[j + t] \prec
  \T[j + t - |H|] = Q[1 + t - |H|] = Q[1 + t]$. Consequently, $\T[j
  \dd n] \prec Q$. To show $Q \prec \T[j' \dd n]$ we observe that
  $\type(j') = +1$ implies $\rend{j'} \leq n$. Thus, we have $Q[1 \dd
  t'] = \T[j' \dd \rend{j'}) = \T[j' \dd j' + t')$ and $Q[1 + t'] =
  Q[1 + t' - |H|] = \T[j' + t' - |H|] \prec \T[j' + t']$. Hence, we
  obtain $Q \prec \T[j' \dd n]$.  We have thus obtained $\T[j \dd n]
  \prec Q \prec \T[j' \dd n]$ which implies $\T[j \dd n] \prec \T[j'
  \dd n]$.  The opposite implication follows easily by symmetry. More
  precisely, in a proof by contraposition, assuming $\type(j) \geq
  \type(j')$ we immediately obtain $\type(j) > \type(j')$ from the
  assumption. By the analogous argument as above we then have $\T[j
  \dd n] \succ \T[j' \dd n]$.

  2. Assume $t < t'$. Similarly as above, we consider two cases for
  $\rend{j}$. If $\rend{j} = n + 1$, then by $t < t'$, the string
  $\T[j \dd \rend{j}) = \T[j \dd n]$ is a proper prefix of $\T[j' \dd
  \rend{j'}) = \T[j' \dd j' + t')$ and hence $\T[j \dd n] \prec \T[j'
  \dd j' + t') \preceq \T[j' \dd n]$. On the other hand, if $\rend{j}
  \leq n$, then we have $\T[j \dd j + t) = \T[j' \dd j' + t)$ and by
  $t < t'$, $\T[j + t] \prec \T[j + t - |H|] = \T[j' + t - |H|] =
  \T[j' + t]$.  Hence, $\T[j \dd n] \prec \T[j' \dd n]$. The opposite
  implication follows by symmetry similarly as in \cref{lm:lce-it-1}.

  3. Assume $t > t'$. By $\type(j') = +1$ we have $\rend{j'} \leq
  n$. Thus, by $t > t'$, we have $\T[j \dd j + t') = \T[j' \dd j' +
  t')$ and $\T[j + t'] = \T[j + t' - |H|] = \T[j' + t' - |H|] \prec
  \T[j' + t']$. Hence, $\T[j \dd n] \prec \T[j' \dd n]$.  The opposite
  implication follows by symmetry similarly as in \cref{lm:lce-it-1}.

  4. By the earlier implication, $\LCE(j, j') \geq \min(t, t')$. Thus,
  it remains to show $\LCE(j, j') \leq \min(t, t')$.  First, let
  $\type(j) \neq \type(j')$ and without the loss of generality let us
  assume $\type(j) < \type(j')$ (i.e., $\type(j) = -1$ and $\type(j')
  = +1$).  Consider two cases:
  \begin{itemize}
  \item First, assume $t \leq t'$. Our goal is to prove $\LCE(j, j')
    \leq t$.  If $j + t = n + 1$, then we immediately obtain the
    claim.  Let us thus assume $j + t \leq n$. In the proof of
    \cref{lm:lce-it-1} we showed that in this case $\type(j) = -1$
    implies $\T[j + t] \prec Q[1 + t]$.  On the other hand, there we
    also proved that $\type(j') = +1$ implies $Q[1 \dd t'] = \T[j' \dd
    j' + t')$ and $Q[1 + t'] \prec \T[j' + t']$.  By $t \leq t'$, we
    thus obtain $Q[1 + t] \preceq \T[j' + t]$.  Consequently, $\T[j +
    t] \neq \T[j' + t]$ and hence $\LCE(j, j') \leq t$.
  \item Let us now assume $t > t'$. Our goal is to prove $\LCE(j, j')
    \leq t'$. In the proof of \cref{lm:lce-it-1} we showed that
    $\type(j) = -1$ implies that $\T[j \dd j + t) = Q[1 \dd t]$. Thus,
    by $t > t'$ we have $\T[j + t'] = Q[1 + t']$. On the other hand,
    in the proof of \cref{lm:lce-it-1} we also proved that $\type(j')
    = +1$ implies $Q[1 + t'] \prec \T[j' + t']$.  Thus, we obtain
    $\T[j + t'] \neq \T[j' + t']$ and hence $\LCE(j, j') \leq t'$.
  \end{itemize}
  This concludes the proof of the claim in the case $\type(j) \neq
  \type(j')$.  Let us thus assume $\type(j) = \type(j')$ and $t \neq
  t'$. First, consider the case $\type(j) = \type(j') = -1$ and assume
  without the loss of generality that $t < t'$ (to match the
  assumption in \cref{lm:lce-it-2}). Our goal is thus to show $\LCE(j,
  j') \leq t$. In the proof of \cref{lm:lce-it-2}, we showed that we
  then either have $\T[j \dd j + t) = \T[j \dd n]$ (in which case
  $\LCE(j, j') \leq n - j + 1 = t$), or $\T[j \dd j + t) = \T[j' \dd
  j' + t)$ and $\T[j + t] \prec \T[j' + t]$ (which also immediately
  implies $\LCE(j, j') \leq t$).  Let us now consider the case
  $\type(j) = \type(j') = +1$ and assume without the loss of
  generality that $t > t'$ (to match the assumption in
  \cref{lm:lce-it-3}).  Our goal is thus to show $\LCE(j, j') \leq
  t'$.  In the proof of \cref{lm:lce-it-3}, we showed that we then
  have $\T[j \dd j + t') = \T[j' \dd j' + t')$ and $\T[j + t'] \prec
  \T[j' + t']$. This implies $\LCE(j, j') \leq t'$.
\end{proof}

The key to the efficient computation of $\SA$ and $\ISA$ values for
periodic positions is processing of the elements of $\R$ in blocks
(note that unlike in \cref{lm:lce}, which describes the structure of
blocks in $\SA$, here we mean blocks of positions in the text). The
starting positions of these blocks are defined as $\R' := \{j\in \R :
j-1 \notin \R\}$.  We also let $\R'^{-} = \R' \cap \R^{-}$, $\R'^{+} =
\R' \cap \R^{+}$, $\R'^{-}_H = \R' \cap \R^{-}_H$, and $\R'^{+}_H =
\R' \cap \R^{+}_H$ for any $H \in \Sigma^{+}$. The following lemma
justifies this strategy.

\begin{lemma}\label{lm:R-block}
  For every $j \in \R \setminus \R'$ it holds:
  \begin{itemize}
  \item $\Lroot(j - 1) = \Lroot(j)$,
  \item $\rend{j - 1} = \rend{j}$,
  \item $\Ltail(j - 1) = \Ltail(j)$,
  \item $\rendfull{j-1} = \rendfull{j}$,
  \item $\type(j - 1) = \type(j)$.
  \end{itemize}
\end{lemma}
\begin{proof}

  Denote $p = \per(\T[j{-}1 \dd j{-}1{+}3\tau{-}1))$. By
  \cref{lm:run-end}\eqref{lm:run-end-it-1}, it holds $\per(\T[j \dd j
  + 3\tau - 1)) = p$. By $p \leq \tfrac{\tau}{3}$, we thus have $\T[j
  {-} 1 \dd j {-} 1 {+} p) = \T[j {-} 1 {+} p \dd j {-} 1 {+} 2p)$.
  Consequently, $\{\T[j {-} 1 {+} t \dd j {-} 1 {+} t {+} p) : t \in
  [0 \dd p)\} = \{\T[j {+} t \dd j {+} t {+} p) : t \in [0 \dd p)\}$,
  and hence $\Lroot(j - 1) = \Lroot(j)$.

  Denote $p' = \per(\T[j \dd j {+} 3\tau {-} 1))$.  By
  \cref{lm:run-end}\eqref{lm:run-end-it-2}, $\rend{j - 1} = j - 1 + p
  + \LCE(j - 1, j - 1 + p)$ and $\rend{j} = j + p' + \LCE(j, j + p')$.
  Thus, by $p = p'$ (following by the above) and $\T[j - 1] = \T[j - 1
  + p]$, we have $\rend{j - 1} = j - 1 + p + \LCE(j - 1, j - 1 + p) =
  j + p + \LCE(j, j + p) = j + p' + \LCE(j, j + p') = \rend{j}$.

  Assume $\T[j-1 \dd \rend{j-1}) = H'H^{k}H''$, where $H =
  \Lroot(j{-}1)$, $|H'| = \Lhead(j{-}1)$, and $|H''| =
  \Ltail(j{-}1)$. By $\rend{j{-}1} = \rend{j}$ and the uniqueness of
  L-decomposition, this implies that either $\T[j \dd \rend{j}) = H'[2
  \dd |H'|]H^kH''$ (if $|H'|>0$) or $\T[j \dd \rend{j}) = H[2 \dd
  |H|]H^{k-1}H''$ (otherwise) is the L-decomposition of $\T[j \dd
  \rend{j})$. In both cases, $\Ltail(j{-}1) = \Ltail(j) = |H''|$.

  By the above two properties, $\rendfull{j{-}1} = \rend{j{-}1} -
  \Ltail(j{-}1) = \rend{j} - \Ltail(j) = \rendfull{j}$.

  The last claim follows from the definition of type and equalities
  $\rend{j {-} 1} = \rend{j}$ and $p = p'$.
\end{proof}

The above is complemented by the following results establishing the
lower bound on the gap between blocks of positions in $\R$, and that a
mapping from $j$ to $\rendfull{j}$ establishes an injective mapping of
blocks of positions in $\R$ to positions in $\T$.

\begin{lemma}\label{lm:gap}
  Let $j, j', j'' \in [1 \dd n]$ be such that $j, j'' \in \R$, $j'
  \not\in \R$, and $j < j' < j''$. Then, it holds $\rend{j} \leq j'' +
  \tau - 1$ and $j'' - j \geq 2\tau$.
\end{lemma}
\begin{proof}
  Let $r = \min\{i \in (j' \dd j'']: i \in \R\}$. Then, $r \in \R'$.
  Observe that by \cref{lm:run-end} (resp.\ by $r - 1 \not\in \R$), it
  holds $\per(\T[j \dd \rend{j})) \leq \lfloor \tfrac{1}{3}\tau
  \rfloor$ (resp.\ $\per(\T[r \dd \rend{r})) \leq \lfloor
  \tfrac{1}{3}\tau \rfloor$), $\rend{j} - j \geq 3\tau - 1$
  (resp.\ $\rend{r} - r \geq 3\tau - 1$), and the substring $\T[j \dd
  \rend{j})$ (resp.\ $\T[r \dd \rend{r}$) cannot be extended in $\T$
  to the right (resp.\ left) without changing its shortest
  period. By~\cite[Fact~2.2.4]{phdtomek}, the fragments $\T[j \dd
  \rend{j})$ and $\T[r \dd \rend{r})$ must therefore overlap by less
  than $2\lfloor \frac13 \tau \rfloor$ symbols. In other words,
  $\rend{j} - r < 2\lfloor \tfrac{1}{3}\tau \rfloor$. By $r \leq j''$
  we thus obtain $\rend{j} \leq r + 2\lfloor \tfrac{1}{3}\tau \rfloor
  \leq j'' + \tau - 1$, i.e., the first claim.  Equivalently, we can
  state that $j'' \geq \rend{j} - \tau + 1$.  By combining this with
  $\rend{j} - j \geq 3\tau - 1$, we then obtain $j'' \geq \rend{j} -
  \tau + 1 \geq j + 3\tau - 1 - \tau + 1 = j + 2\tau$, i.e., the
  second claim.
\end{proof}

\begin{lemma}\label{lm:efull}
  For any $j, j' \in \R'$, $j \neq j'$ implies $\rendfull{j} \neq
  \rendfull{j'}$.
\end{lemma}
\begin{proof}
  Assume without the loss of generality that $j < j'$. Then, $j' - 1
  \not\in \R$. By \cref{lm:gap} applied for $j$, $j' - 1$, and $j'$ we
  obtain $\rend{j} \leq j' + \tau - 1$. Consequently, $\rendfull{j}
  \leq \rend{j} \leq j' + \tau - 1$.  Let now $r' = \min\{t \in (j'
  \dd n] : t \not\in \R\}$.  We then have $\rend{j'} = r' + 3\tau -
  2$. Since for every $t \in \R$, it holds $\rend{t} - \rendfull{t} =
  \Ltail(t) = |\Lroot(t)| \leq \lfloor \tfrac{1}{3}\tau \rfloor$, we
  thus have $\rendfull{j'} \geq \rend{j'} - \lfloor \tfrac{1}{3}\tau
  \rfloor = r' + 3\tau - 2 - \lfloor \tfrac{1}{3}\tau \rfloor \geq j'
  + 2\tau - 1$.  Combining this with the earlier upper bound on
  $\rendfull{j}$, we thus obtain $\rendfull{j} \leq j' + \tau - 1 < j'
  + 2\tau - 1 \leq \rendfull{j'}$. In particular, $\rendfull{j} \neq
  \rendfull{j'}$.
\end{proof}

\subsubsection{The Data Structure}\label{sec:sa-periodic-ds}

\bfparagraph{Definitions}

Let $q = |\R'^{-}|$ and let $(\rtextm_i)_{i \in [1 \dd q]}$ be a
sequence containing all elements of $\R'^{-}$ in sorted order, i.e,
for any $i, i' \in [1 \dd q]$, $i < i'$ implies $\rtextm_i <
\rtextm_{i'}$. Let $(\rlexm_i)_{i \in [1 \dd q]}$ also be a sequence
containing all elements $k \in \R'^{-}$, but sorted first according to
$\Lroot(k)$ and in case of ties, by $\T[\rendfull{k} \dd
n]$. Formally, for any $i, i' \in [1 \dd q]$, $i < i'$ implies that
$\Lroot(\rlexm_i) \prec \Lroot(\rlexm_{i'})$, or $\Lroot(\rlexm_{i}) =
\Lroot(\rlexm_{i'})$ and $\T[\rendfull{\rlexm_i} \dd n] \prec
\T[\rendfull{\rlexm_{i'}} \dd n]$. Note that by \cref{lm:efull}, the
sequence $(\rlexm_i)_{i \in [1 \dd q]}$ is well-defined.  Based on
$(\rlexm_i)_{i \in [1 \dd q]}$ we define the sequence of integers
$(\ell_i)_{i \in [1 \dd q]}$ as $\ell_i = \rendfull{\rlexm_i} -
\rlexm_i$.

Let $\LTroot$ denote the mapping from $\Alphabet^{3\tau - 1}$ to
$\N^2$ such that for any $X \in \Alphabet^{3\tau - 1}$ satisfying
$\per(X) \leq \tfrac{1}{3}\tau$, $\LTroot$ maps $X$ to a pair $(s,
p)$, where $p = \per(X)$ and $s \in [0 \dd p)$ is such that $X[1 {+} s
\dd 1 {+} s {+} p) = \min\{X[1 + t \dd 1 + t + p) : t \in [0 \dd
p)\}$.  We also define $\LTminexp: \Alphabet^{3\tau - 1} \to [1 \dd
n]$ as the mapping such that for every $X \in \Alphabet^{3\tau - 1}$
satisfying $\per(X) \leq \tfrac{1}{3}\tau$, if we let $p = \per(X)$,
$H = \min\{X[1 + t \dd 1 + t + p) : t \in [0 \dd p)\}$ and $s \in [0
\dd p)$ be such that $X[1{+}s \dd 1 {+} s {+} p) = H$, then assuming
$\R^{-}_{s,H} \neq \emptyset$, $\LTminexp$ maps $X$ to $\min
\{\Lexp(j) : j \in \R^{-}_{s,H}\}$.  Let $\LTruns$ be a
mapping, such that for every $H \in \Alphabet^{\leq \tau}$ and every
$H' \in \Alphabet^{\leq \tau}$, $\LTruns$ maps the pair $(H,H')$ to
$(b, e)$ defined by $b = |\{k \in \R'^{-} : \Lroot(k) \prec H\text{,
or }\Lroot(k) = H\text{ and }\T[\rendfull{k} \dd n] \prec H'\}|$ and
$e = b + |\{k \in \R'^{-}_H : H'\text{ is a prefix of }\T[\rendfull{k}
\dd n]\}|$. Note that then the set $\{\rlexm_{i} : i \in (b \dd e]\}$
consists of all positions $k \in \R'^{-}_H$ such that $H'$ is a prefix
of $\T[\rendfull{k} \dd n]$. In particular, every $(H, \emptystring)$
maps to a pair $(b, e)$ such that $e = \sum_{H' \preceq
H}|\R'^{-}_{H'}|$. For any $\ell > 0$, $H \in \Alphabet^{+}$, and $s
\in [0 \dd |H|)$, we define $\Pref_{\ell}(s, H)$ as the length-$\ell$
prefix of $H'\cdot H^{\infty}[1 \dd)$, where $H'$ is a length-$s$
suffix of $H$. Let $\LTpref$ denote the mapping that, given the
pair $(H, s)$, where $H \in \Alphabet^{\leq \tau}$, and $s
\in [0 \dd |H|)$, returns the packed encoding of
$\Pref_{3\tau - 1}(s, H)$.

Let $\BVexp[1 \dd n]$ be a bitvector such that for every $i \in [1 \dd
n]$, it holds $\BVexp[i] = 0$ if and only if $\SA[i] \in [1 \dd n]
\setminus \R^{-}$, or $i < n$ and the positions $j=\SA[i]$ and
$j'=\SA[i+1]$ satisfy $j,j' \in \R^{-}_{s,H}$ and $\Lexp(j) =
\Lexp(j')$ for some $H \in \Lroots$ and $s \in [0 \dd |H|)$.  Let
$\BVRprim[1 \dd n]$ be a bitvector defined such that $\BVRprim[i] = 1$
holds if and only if $i \in \R'$.

Let $\ARRnontail[1 \dd q]$ by an array defined by $\ARRnontail[i] =
\ell_i$.  Let $\ARRrmap[1 \dd q]$ be an array containing a permutation
of $[1 \dd q]$ such that $\ARRrmap[i] = i'$ holds if and only if
$\rtextm_i = \rlexm_{i'}$. By $\ARRrinvmap[1 \dd q]$ we denote an
array containing a permutation of $[1 \dd q]$ such that
$\ARRrinvmap[i'] = i$ holds if and only if $\rtextm_i = \rlexm_{i'}$.

\bfparagraph{Components}

The data structure consists of two parts. The first part, designed to
compute $\SA[i]$ (resp.\ $\ISA[j]$) for $i \in [1 \dd n]$ (resp.\ $j
\in [1 \dd n]$) satisfying $\SA[i] \in \R^{-}$ (resp.\ $j \in \R^{-}$),
consists of the following eleven components:
\begin{enumerate}
\item $\SACore(\T)$ (\cref{sec:sa-core-ds}). It takes $\bigO(n /
  \log_{\sigma} n)$ space.
\item The lookup table $\LTroot$.  When accessing $\LTroot$, strings
  $X \in \Alphabet^{3\tau - 1}$ are converted to $\Int(X)$. Thus,
  $\LTroot$ needs $\bigO(\sigma^{6\tau}) = \bigO(n^{6\mu}) =
  \bigO(n / \log_{\sigma} n)$ space.
\item The lookup table $\LTminexp$. As above, $\LTminexp$ also needs
  $\bigO(\sigma^{6\tau}) = \bigO(n / \log_{\sigma} n)$ space.
\item The lookup table $\LTruns$. When storing the mapping from the
  key $(H,H')$, we first concatenate $H$ and $H'$ and convert it to an
  integer $x = \Int(HH')$ in the range $[0 \dd \sigma^{6\tau})$. We
  then create a triple $(x, |H|)$ (this contains enough information to
  decode $H$ and $H'$) and injectively map it to a positive integer
  not exceeding $\sigma^{6\tau} \tau < \sigma^{6\tau}\log_{\sigma} n$.
  Thus, $\LTruns$ can be stored using
  $\bigO(n^{6\mu}\log n) = \bigO(n/\log_{\sigma} n)$ space.
\item The lookup table $\LTpref$. When storing the mapping, we convert
  the string $H$ to $\Int(H)$. By $\Int(H) \in [0 \dd \sigma^{6\tau})$
  and $|H| \leq \tau$, each pair $(\Int(H), s)$ can then be injectively
  encoded as an integer in the range of size $\sigma^{6\tau} \tau <
  \sigma^{6\tau}\log_{\sigma} n$ and hence $\LTpref$ needs
  $\bigO(n^{6\mu} \log n) = \bigO(n/\log_{\sigma} n)$ space.
\item The bitvector $\BVexp$ augmented using \cref{th:binrksel}. It
  needs $\bigO(n / \log n)$ space.
\item The bitvector $\BVRprim$ augmented using \cref{th:binrksel}. It
  needs $\bigO(n / \log n)$ space.
\item The array $\ARRnontail[1 \dd q]$ augmented with a structure
  from~\cref{pr:range-queries}. To analyze its space usage, consider
  any $j_1, j_2, j_3 \in \R'$ such that $j_1 < j_2 < j_3$.  Then, $j_2
  - 1 \not\in \R$ and $j_3 - 1 \not\in \R$. By \cref{lm:gap} applied
  first for $j_1$, $j_2 - 1$, and $j_2$ we have $\rend{j_1} \leq j_2 +
  \tau - 1$. Applying it again for $j_2$, $j_3 - 1$, and $j_3$, we
  obtain $j_3 - j_2 \geq 2\tau$, or equivalently, $j_2 \leq j_3 -
  2\tau$.  Combining the two inequalities, we thus obtain that
  $\rend{j_1} \leq j_3 - \tau - 1 < j_3$.  This implies that each
  position of $\T$ belongs to at most two intervals in the collection
  $\{[j \dd \rend{j}) : j \in \R'\}$, and consequently,
  $\sum_{i=1}^{q}\ell_i \leq 2n$. On the other hand, by \cref{lm:gap},
  for every $j, j' \in \R'$, $j \neq j'$ implies $|j' - j| \geq
  2\tau$.  Thus, $q = \bigO(n / \tau) = \bigO(n / \log_{\sigma}
  n)$. The array $A$ augmented using \cref{pr:range-queries} thus
  needs $\bigO(n / \log_{\sigma} n)$ space.
\item The array $\ARRrmap$ in plain form using $\bigO(1 + q) =
  \bigO(n / \log_{\sigma} n)$ space.
\item The array $\ARRrinvmap$ in plain form using $\bigO(1 + q) =
  \bigO(n / \log_{\sigma} n)$ space
\item The $\bigO(n / \log_{\sigma} n)$-space data structure
  from~\cite[Theorem~5.4]{sss} that, given any $i,i' \in [1 \dd n]$,
  returns $\LCE(i,i')$ in $\bigO(1)$ time.
\end{enumerate}

The second part of the structure, designed to compute $\SA[i]$
(resp.\ $\ISA[j]$) for $i \in [1 \dd n]$ (resp.\ $j \in [1 \dd n]$)
satisfying $\SA[i] \in \R^{+}$ (resp.\ $j \in \R^{+}$), consists of
the symmetric counterparts adapted according to \cref{lm:lce}.

In total, the data structure takes $\bigO(n / \log_{\sigma} n)$ space.

\subsubsection{Navigation Primitives}\label{sec:sa-periodic-nav}

\begin{proposition}\label{pr:isa-root}
  Given the data structure from \cref{sec:sa-periodic-ds} and any
  position $j \in \R$, we can in $\bigO(1)$ time compute the values
  $\Lroot(j)$, $\Lhead(j)$, $\Lexp(j)$, $\Ltail(j)$, and $\type(j)$.
\end{proposition}
\begin{proof}
  We first compute $x \in [0 \dd \sigma^{6\tau})$ such that $x =
  \Int(\T[j \dd j + 3\tau - 1))$. Given the packed encoding of text
  $\T$, such $x$ is obtained in $\bigO(1)$ time. We then look up $(s,
  p) = \LTroot[x]$, and in $\bigO(1)$ time obtain $\Lroot(j) = \T[j
  {+} s \dd j {+} s {+} p)$ and $\Lhead(j) = s$.  Next, we compute
  $\Lexp(j)$ and $\Ltail(j)$. For this we recall that by
  \cref{lm:run-end}\eqref{lm:run-end-it-2}, it holds $\rend{j} = j + p
  + \LCE(j,j+p)$. Thus, given $j$ and $p$, we can compute $\rend{j}$
  in $\bigO(1)$ time. We then obtain $\Lexp(j) = \lfloor
  \frac{\rend{j} - j - s}{p} \rfloor$ and $\Ltail(j) = (\rend{j} - j -
  s)\bmod p$.  Finally, to test if $\type(j) = +1$, we check whether
  $\rend{j} \leq n$, and if so, whether $\T[\rend{j}] \succ
  \T[\rend{j} - p]$.
\end{proof}

\begin{proposition}\label{pr:sa-root}
  Let $i \in [1 \dd n]$ be such that $\SA[i] \in \R$. Given the data
  structure from \cref{sec:sa-periodic-ds} and the index $i$, in
  $\bigO(1)$ time we can compute $\Lroot(\SA[i])$ and
  $\Lhead(\SA[i])$.
\end{proposition}
\begin{proof}
  We first compute $y = \rank{\BVshort}{1}{i-1}$. The string $X =
  \ARRshort[y + 1]$ is then a prefix of $\T[\SA[i] \dd n]$ of length
  $3\tau - 1$.  Let $x = \Int(X)$.  We then look up $(s, p) =
  \LTroot[x]$, and in $\bigO(1)$ time obtain $\Lroot(\SA[i]) = X[1 {+}
  s \dd 1 {+} s {+} p)$ and $\Lhead(\SA[i]) = s$.
\end{proof}

\subsubsection{Implementation of
  \texorpdfstring{$\ISA$}{ISA} Queries}\label{sec:sa-periodic-isa}

For any $j \in \R$, we define
\[
  \Pos(j) = \{j' \in [1 \dd n] : \LCE(j, j') \geq 3\tau - 1
  \text{ and }\T[j' \dd n] \preceq \T[j \dd n]\},
\]
and denote $\delta(j) = |\Pos(j)|$.

\begin{lemma}\label{lm:isa-delta-1}
  Let $j \in \R$ and $X = \T[j \dd j + 3\tau - 1)$. Then, $\ISA[j] =
  \LB(X, \T) + \delta(j)$.
\end{lemma}
\begin{proof}
  It suffices to observe that $j' \in \Occ(X, \T)$ holds if and only
  if $\LCE(j, j') \geq 3\tau - 1$.  Thus, it holds by definition of
  $\ISA[j]$ that $\ISA[j] = \LB(X, \T) + |\{j' \in \Occ(X, \T) : \T[j'
  \dd n] \preceq \T[j \dd n]\}| = \LB(X,\T) + |\{j' \in [1 \dd n] :
  \LCE(j, j') \geq 3\tau - 1\text{ and }\T[j' \dd n] \preceq \T[j \dd
  n]\}| = \LB(X, \T) + \delta(j)$.
\end{proof}

We focus on computing $\delta(j)$ for $j \in \R^{-}$. The elements of
$\R^{+}$ are processed symmetrically (see the proof of
\cref{pr:sa-periodic-isa}).  For any $H\in \Lroots$, $s\in [0\dd
|H|)$, and $j \in \R^{-}_{s,H}$, we define $\Posa(j) = \{j' \in
\R^{-}_{s,H} : \Lexp(j') \leq \Lexp(j)\}$ and $\Poss(j) = \{j' \in
\R^{-}_{s,H} : \Lexp(j') = \Lexp(j)\text{ and }\T[j' \dd n] \succ \T[j
\dd n]\}$. For any $j \in \R^{-}$, we denote $\deltaa(j) =
|\Posa(j)|$ and $\deltas(j) = |\Poss(j)|$.

\begin{lemma}\label{lm:isa-delta-2}
  For any $j \in \R^{-}$, it holds $\delta(j) = \deltaa(j) -
  \deltas(j)$.
\end{lemma}
\begin{proof}

  We will prove that $\Posa(j)$ is a disjoint union of $\Pos(j)$ and
  $\Poss(j)$. This implies $\delta(j) + \deltas(j) = \deltaa(j)$, and
  consequently, the equality in the claim.

  By \cref{lm:lce}, letting $j \in \R^{-}_{s,H}$, we have $\Pos(j) =
  \{j' \in \R^{-}_{s,H} : \T[j' \dd n] \preceq \T[j \dd n]\}$, and
  moreover, if $j' \in \Pos(j)$, then $\rend{j'} - j' \leq \rend{j} -
  j$. In particular, $\Lexp(j') = \lfloor \tfrac{\rend{j'} - j' -
  s}{|H|} \rfloor \leq \lfloor \tfrac{\rend{j} - j - s}{|H|} \rfloor
  = \Lexp(j)$. Hence, $\Pos(j) \subseteq \Posa(j)$. On the other hand,
  clearly $\Poss(j) \subseteq \Posa(j)$ and $\Poss(j) \cap \Pos(j) =
  \emptyset$. Thus, to obtain the claim, it suffices to show that
  $\Posa(j) \setminus \Poss(j) \subseteq \Pos(j)$.

  Let $j' \in \Posa(j) \setminus \Poss(j)$. Consider two cases.  If
  $\Lexp(j') = \Lexp(j)$, then by definition of $\Poss(j)$, it must
  hold $\T[j' \dd n] \preceq \T[j \dd n]$. Thus, we have $j' \in
  \Pos(j)$. Let us therefore assume $\Lexp(j') < \Lexp(j)$. Then,
  $\rend{j'} - j' = s + \Lexp(j') \cdot |H| + \Ltail(j') < s +
  \Lexp(j')\cdot |H| + |H| \le s + \Lexp(j) \cdot |H| \le s +
  \Lexp(j)\cdot |H| + \Ltail(j) = \rend{j} - j$.  By
  \cref{lm:lce}\eqref{lm:lce-it-2}, this implies $\T[j' \dd n]
  \prec \T[j \dd n]$, and consequently, $j' \in \Pos(j)$.
\end{proof}

\bfparagraph{Computing $\deltaa(j)$}

We now describe the algorithm to compute $\deltaa(j)$ for
$j \in \R^{-}$.

\begin{proposition}\label{pr:isa-delta-a}
  Given the data structure from \cref{sec:sa-periodic-ds} and any $j
  \in \R^{-}$, in $\bigO(1)$ time we can compute $\deltaa(j)$.
\end{proposition}
\begin{proof}
  Let $X = \T[j \dd j + 3\tau - 1)$. First, using the lookup table
  $\LTrange$, we compute $(b_X, e_X) = (\LB(X, \T), \UB(X,
  \T))$. Then, by \cref{lm:lce}, $\SA(b_X \dd e_X]$ contains all
  positions from $\R_{s,H}$, where $H = \Lroot(j)$ and $s=\Lhead(j)$.
  Next, using \cref{pr:isa-root}, we compute in $\bigO(1)$ time the
  value $k = \Lexp(j)$. Finally, we retrieve $k_{\min} =
  \LTminexp[\Int(X)]$. Observe now that by \cref{lm:lce}, all
  positions in $\R^{-}_{s,H}$ occur in $\SA(b_X \dd e_X]$ before
  $\R^{+}_{s,H}$. Furthermore, by \cref{lm:lce}\eqref{lm:lce-it-2},
  $[k_{\min} \dd k]\sub \{\Lexp(j') : j'\in \R^{-}_{s,H}\}$ (for
  $k'\in (k_{\min} \dd k]$, we can take $j' = j+(k-k')|H|$).  Thus, by
  the definition of $\BVexp$, we can finally return $\deltaa(j) =
  \select{\BVexp}{1}{\rank{\BVexp}{1}{b_X}+ (k - k_{\min}) + 1} - b_X
  $ in $\bigO(1)$ time.
\end{proof}

\bfparagraph{Computing $\deltas(j)$}

Next, we describe the algorithm to compute $\deltas(j)$ for any
position $j \in \R^{-}$.

\begin{lemma}\label{lm:isa-delta-s}
  Assume $i,j \in \R^{-}_H$ and let $\ell = \rend{i} - i - 3\tau + 2$.
  Then $|\Poss(j) \cap [i \dd i + \ell)| \leq 1$.  Moreover,
  $|\Poss(j) \cap [i \dd i + \ell)| = 1$ if and only if
  $\T[\rendfull{i} \dd n] \succ \T[\rendfull{j} \dd n]$ and
  $\rendfull{i} - i \geq \rendfull{j} - j$.
\end{lemma}
\begin{proof}

  By \cref{lm:R-block}, we have $[i \dd i + \ell) \subseteq \R^{-}_H$
  with $\rend{i + \delta} = \rend{i}$ for every $\delta \in [0 \dd
  \ell)$. Moreover, by the uniqueness of L-decomposition, $\Ltail(i +
  \delta) = \Ltail(i)$.  Together, these imply that $\rendfull{i +
  \delta} = \rendfull{i}$, and consequently $\rendfull{i + \delta} -
  (i + \delta) = \rendfull{i} - i - \delta$. It remains to observe
  that, letting $j \in \R^{-}_{s,H}$, for $j' \in \Poss(j)$ it holds
  $\rendfull{j'} - j' = s + \Lexp(j') \cdot |H| = s + \Lexp(j) \cdot
  |H| = \rendfull{j} - j$. Thus, $i + \delta \in \Poss(j)$ implies
  $\rendfull{i + \delta} - (i + \delta) = \rendfull{i} - (i + \delta)
  = \rendfull{j} - j$, or equivalently, $\delta = (\rendfull{i} - i) -
  (\rendfull{j} - j)$, and therefore, $|\Poss(j) \cap [i \dd i +
  \ell)| \leq 1$.

  For the second part, assume first that $i + \delta \in \Poss(j)$
  holds for some $\delta \in [0 \dd \ell)$. Then, as noted above, we
  have $\rendfull{j} - j = \rendfull{i} - (i + \delta) \leq
  \rendfull{i} - i$. Moreover, letting $j \in \R^{-}_{s,H}$, by
  definition of $\Poss(j)$, we have $i + \delta \in \R^{-}_{s,H}$,
  $\Lexp(j) = \Lexp(i + \delta)$, and $\T[i {+} \delta \dd n] \succ
  \T[j \dd n]$.  Therefore, we obtain that $\T[i {+} \delta \dd
  \rendfull{i + \delta}) = \T[i {+} \delta \dd \rendfull{i}) = \T[j
  \dd \rendfull{j}) = H'H^k$ (where $k = \Lexp(j)$ and $H'$ is the
  length-$s$ suffix of $H$), and consequently, $\T[\rendfull{i} \dd n]
  \succ \T[\rendfull{j} \dd n]$.  To show the converse implication,
  assume $\T[\rendfull{i} \dd n] \succ \T[\rendfull{j} \dd n]$ and
  $\rendfull{i} - i \geq \rendfull{j} - j$. Let $\delta =
  (\rendfull{i} - i) - (\rendfull{j} - j)$. We will prove that $\delta
  \in [0 \dd \ell)$ and $i + \delta \in \Poss(j)$. Clearly $\delta
  \geq 0$. To show $\delta < \ell$, we first prove $\rend{i} -
  \rendfull{i} \geq \rend{j} - \rendfull{j}$.  Suppose that $q =
  \rend{i} - \rendfull{i} < \rend{j} - \rendfull{j}$.  By $i \in
  \R^{-}_{H}$, we then either have $\rendfull{i} + q = n + 1$, or
  $\rendfull{i} + q \leq n$ and $\T[\rendfull{i} + q] \prec
  \T[\rendfull{i} + q - |H|] = \T[\rendfull{j} + q - |H|] =
  \T[\rendfull{j} + q]$, both of which contradict $\T[\rendfull{i} \dd
  n] \succ \T[\rendfull{j} \dd n]$. Thus, $\rend{i} - \rendfull{i}
  \geq \rend{j} - \rendfull{j}$.  This implies, $\rend{i} - (i +
  \delta) = (\rendfull{i} - (i + \delta)) + (\rend{i} - \rendfull{i})
  = (\rendfull{j} - j) + (\rend{i} - \rendfull{i}) \geq (\rendfull{j}
  - j) + (\rend{j} - \rendfull{j}) = \rend{j} - j \geq 3\tau - 1$, or
  equivalently $\delta \leq \rend{i} - i - 3\tau + 1 < \ell$.  To show
  $i + \delta \in \Poss(j)$, it remains to observe that $\rendfull{i +
  \delta} - (i + \delta) = \rendfull{i} - (i + \delta) = \rendfull{j}
  - j$ and $i + \delta, j \in \R^{-}_H$ imply $\T[i + \delta \dd
  \rendfull{i}) = \T[j \dd \rendfull{j})$.  This in particular gives,
  letting $j \in \R_{s,H}$, that $i + \delta \in \R_{s,H}$ and
  $\Lexp(i + \delta) = \Lexp(j)$.  Moreover, combining it with
  $\T[\rendfull{i} \dd n] \succ \T[\rendfull{j} \dd n]$ yields $\T[i +
  \delta \dd n] \succ \T[j \dd n]$. Finally, by \cref{lm:R-block},
  $\type(i + \delta) = \type(i) = -1$. Therefore, $i + \delta \in
  \Poss(j)$.
\end{proof}

\begin{proposition}\label{pr:isa-delta-s}
  Given the data structure from \cref{sec:sa-periodic-ds} and any $j
  \in \R^{-}$, in $\bigO(\log \log n)$ time we can compute
  $\deltas(j)$.
\end{proposition}
\begin{proof}
  Given $j \in \R^{-}$, we first compute $H = \Lroot(j)$, $s =
  \Lhead(j)$, and $k = \Lexp(j)$.  By \cref{pr:isa-root}, this takes
  $\bigO(1)$ time. This lets us deduce $\rendfull{j} = j + s +
  k|H|$. Then, we compute $i \in [1 \dd q]$ satisfying $j \in
  [\rtextm_i \dd \rend{\rtextm_i} - 3\tau + 2)$, i.e., $j$ is in the
  maximal block of positions from $\R^{-}$ starting at position
  $\rtextm_i$. Using $\BVRprim$ we obtain $i = \rank{\BVRprim}{1}{j}$
  in $\bigO(1)$ time. Observe now that, letting $j' = \rtextm_i$, by
  $\rendfull{j'} = \rendfull{j}$, we have $\T[\rendfull{j'} \dd n] =
  \T[\rendfull{j} \dd n]$.  Therefore, letting $x = \ARRrmap[i]$ and
  $x' = \sum_{H' \preceq H}|\R'^{-}_{H'}|$ (obtained in $\bigO(1)$
  time using $\LTruns$), by \cref{lm:isa-delta-s} we have $\deltas(j)
  = |\{i' \in (x \dd x'] : \ell_{i'} \geq \rendfull{j} - j\}| =
  \rcount{\ARRnontail}{\rendfull{j} - j}{x'} -
  \rcount{\ARRnontail}{\rendfull{j} - j}{x}$, which we compute in
  $\bigO(\log \log n)$ time using the data structure from
  \cref{pr:range-queries}.
\end{proof}

\begin{remark}\label{rm:sa-periodic}
  In \cref{lm:isa-delta-2}, we presented an equation relating the
  sizes of $\Posa(j)$ and $\Poss(j)$, and the size of $\Pos(j)$, where
  $j \in \R^{-}$. In this formula, some positions are first counted as
  part of $\Posa(j)$, and then canceled when subtracting the size of
  $\Poss(j)$. To see the reason for this counterintuitive formula, let
  $J := \{j' \in \R^{-}_{s,H} : \Lexp(j') = \Lexp(j)\}$, where $s =
  \Lhead(j)$ and $H = \Lroot(j)$, and consider the problem of
  computing the size of $J' = \{j' \in J : \T[j' \dd n] \succeq \T[j
  \dd n]\}$.  As shown in \cref{lm:isa-delta-s}, to count such
  positions, it suffices to first align all $j'' \in \R'^{-}$ by the
  position $\rendfull{j''}$, and then count those $j''$ that satisfy
  (1) $\T[\rendfull{j''} \dd n] \succeq \T[\rendfull{j} \dd n]$, and
  (2) $\rendfull{j''} - j'' \geq \rendfull{j} - j$. For every $j'' \in
  \R'^{-}$ satisfying these conditions, there exists exactly one $j'
  \in \R^{-}_{s,H}$ such that $[j'' \dd j'] \sub \R$, $\Lexp(j') =
  \Lexp(j)$ and $\T[j' \dd n] \succeq \T[j \dd n]$, because for $j,
  j'' \in \R^{-}$, $\T[\rendfull{j''} \dd n] \succeq \T[\rendfull{j}
  \dd n]$ implies $\rend{j''} - \rendfull{j''} \geq \rend{j} -
  \rendfull{j}$ (\cref{lm:lce}).  Thus, letting $\ell = \rendfull{j} -
  j$, such $j'$ is given by $j' = \rendfull{j''} - \ell$.  In
  particular, such $j'$ satisfies $j' \in \R$ because $(\rend{j''} -
  \rendfull{j''}) + \ell \geq \rend{j} - \rendfull{j} + \ell =
  \rend{j} - j \geq 3\tau - 1$.

  Consider now the problem of computing the size of $J'' = \{j' \in J
  : \T[j' \dd n] \prec \T[j \dd n]\}$ (defining $\Poss(j)$ as $J''$
  may seem like a simpler alternative to the current definition).
  Observe that the above method does not work for this problem.  The
  reason for this is that position $j'' \in \R'^{-}$ satisfying
  $\T[\rendfull{j'} \dd n] \prec \T[\rendfull{j} \dd n]$ does not
  necessarily imply that $\rendfull{j''} - \ell \in \R$. This is
  because we may have $\rend{j''} - \rendfull{j''} < \rend{j} -
  \rendfull{j}$, which implies that it is possible that $(\rend{j''} -
  \rendfull{j''}) + \ell < 3\tau - 1$.  This motivates the current
  definition of $\Poss(j)$.
\end{remark}

\bfparagraph{Summary}

By combining all above results, we obtain the following algorithm
to compute $\ISA[j]$ for periodic positions.

\begin{proposition}\label{pr:sa-periodic-isa}
  Given the data structure from \cref{sec:sa-periodic-ds} and any $j
  \in \R$, in $\bigO(\log \log n)$ time we can compute $\ISA[j]$.
\end{proposition}
\begin{proof}
  First, in $\bigO(1)$ time we compute $x = \Int(X)$, where $X = \T[j
  \dd j + 3\tau - 1)$. In $\bigO(1)$ we then look up $(b_X, e_X) =
  \LTrange[x]$. In particular, we have $b_X = \LB(X, \T)$. Then, using
  \cref{pr:isa-root} we determine $\type(j)$. Depending on whether $j
  \in \R^{-}$ or $j \in \R^{+}$ we use either a combination of
  \cref{pr:isa-delta-a,pr:isa-delta-s}, or their symmetric
  counterparts (more precisely, if $j \in \R^{+}$, letting $s =
  \Lhead(j)$ and $H = \Lroot(j)$, we have $\deltaa(j) = |\Posa(j)|$
  and $\deltas(j) = |\Poss(j)|$, where $\Posa(j) = \{j' \in
  \R^{+}_{s,H}: \Lexp(j') \leq \Lexp(j)\}$ and $\Poss(j) = \{j' \in
  \R^{+}_{s,H} : \Lexp(j) = \Lexp(j)\text{ and }\T[j' \dd n] \prec
  \T[j \dd n]\}$), to compute $\deltaa(j)$ and $\deltas(j)$ in
  $\bigO(1)$ and $\bigO(\log \log n)$ time, respectively.  If $j \in
  \R^{-}$, then by \cref{lm:isa-delta-2} we have $\delta(j) =
  \deltaa(j) - \deltas(j)$.  Otherwise, by the counterpart of
  \cref{lm:isa-delta-2}, $\delta(j) = (e_X - b_X) - (\deltaa(j) -
  \deltas(j))$.  Finally, we return $\ISA[j] = b_X + \delta(j)$ as the
  answer. In total, the query takes $\bigO(\log \log n)$ time.
\end{proof}

\subsubsection{Implementation of
  \texorpdfstring{$\SA$}{SA} Queries}\label{sec:sa-periodic-sa}

We focus on positions $i \in [1 \dd n]$ satisfying $\SA[i] \in
\R^{-}$.  Positions satisfying $\SA[i] \in \R^{+}$ are processed
symmetrically (see the proof of \cref{pr:sa-periodic-sa}).  The
algorithm to query $\SA[i]$ for $i \in [1 \dd n]$ satisfying $\SA[i]
\in \R^{-}$ proceeds in two steps. First, we compute $\Lexp(\SA[i])$
and $\deltas(\SA[i])$. In the second steps, these values are used to
compute $\SA[i]$.

\bfparagraph{Computing $\Lexp(\SA[i])$ and $\deltas(\SA[i])$}

We now describe the first step during the computation of $\SA[i]$ for
$i \in [1 \dd n]$ satisfying $\SA[i] \in \R$.

\begin{proposition}\label{pr:sa-delta-a}
  Let $i \in [1 \dd n]$ be such that $\SA[i] \in \R$.  Given the data
  structure from \cref{sec:sa-periodic-ds} and the index $i$, in
  $\bigO(1)$ time we can check if $\type(\SA[i]) = -1$, and if so,
  return $\Lexp(\SA[i])$ and $\deltas(\SA[i])$.
\end{proposition}
\begin{proof}
  To check if $\type(\SA[i]) = -1$, we first compute $y =
  \rank{\BVshort}{1}{i-1}$.  The string $X = \ARRshort[y + 1]$ is then
  a prefix of $\T[\SA[i] \dd n]$ of length $3\tau - 1$.  Let $x =
  \Int(X)$. In $\bigO(1)$ time we then look up $(b_X, e_X) =
  \LTrange[x]$.  By \cref{lm:lce} we then have $\type(\SA[i]) = -1$ if
  and only if $\BVexp[i \dd e_X]$ contains a bit with value 1. This
  can be checked in $\bigO(1)$ time by checking if
  $\rank{\BVexp}{1}{e_X} > \rank{\BVexp}{1}{i-1}$. Let us assume
  $\type(\SA[i]) = -1$. To compute $\Lexp(\SA[i])$, we first in
  $\bigO(1)$ retrieve $k_{\min} = \LTminexp[x]$, and then compute
  $\Lexp(\SA[i]) = k_{\min} + (\rank{\BVexp}{1}{i-1} -
  \rank{\BVexp}{1}{b_X})$.  Then, $\deltaa(\SA[i])$ can be computed in
  $\bigO(1)$ time as $\deltaa(\SA[i]) =
  \select{\BVexp}{1}{\rank{\BVexp}{1}{i-1}+1}-b_X$.  Finally, by
  applying \cref{lm:isa-delta-1} and \cref{lm:isa-delta-2} for $j =
  \SA[i]$, it holds $i - b_X = \deltaa(\SA[i]) -
  \deltas(\SA[i])$. Thus, we obtain $\deltas(\SA[i]) = b_X +
  \deltaa(\SA[i]) - i$.
\end{proof}

\bfparagraph{Computing $\SA[i]$}

We now describe the algorithm to complete the computation of
$\SA[i]$ for any $i \in [1 \dd n]$ such that $\SA[i] \in \R^{-}$.

\begin{proposition}\label{pr:sa-delta-s}
  In $\bigO(n / \log_{\sigma} n)$ time, we can augment the structure
  of \cref{pr:sa-root} so that, given any $i \in [1 \dd n]$ such that
  $\SA[i] \in \R^{-}$, along with $\Lexp(\SA[i])$ and
  $\deltas(\SA[i])$, we can compute $\SA[i]$ in $\bigO(\log \log n)$
  time.
\end{proposition}
\begin{proof}
  First, we compute $H = \Lroot(\SA[i])$ and $\Lhead(\SA[i])$ in
  $\bigO(1)$ time using \cref{pr:sa-root}. This lets us deduce that
  $\rendfull{\SA[i]} - \SA[i] = \ell$, where $\ell = \Lhead(\SA[i]) +
  \Lexp(\SA[i])|H|$.  Let $x = \sum_{H' \preceq H} |\R'^{-}_{H'}|$
  (obtained using $\LTruns$ in $\bigO(1)$ time). Next, we compute
  $\delta = \rcount{\ARRnontail}{\ell}{x}$.  Using the structure from
  \cref{pr:range-queries}, this takes $\bigO(\log \log n)$ time. Let
  $k = \delta - \deltas(\SA[i])$.  We then compute the position $p \in
  [1 \dd q]$ of the $k$th leftmost element in $\ARRnontail$ that is
  greater or equal than $\ell$. Using \cref{pr:range-queries}, we
  compute $p = \rselect{\ARRnontail}{\ell}{k}$ in $\bigO(1)$ time.
  By \cref{lm:efull} and \cref{lm:isa-delta-s}, we then have
  $\rendfull{\rlexm_p} = \rendfull{\SA[i]}$. By combining
  \cref{lm:R-block} and \cref{lm:efull}, for any $j', j'' \in \R$ such
  that $j' < j''$ and $\rendfull{j'} = \rendfull{j''}$, it holds $[j'
  \dd j''] \sub \R$, i.e., $j'$ and $j''$ must belong to the same
  contiguous block of positions from $\R$. Since $\rlexm_p \in \R'$,
  we thus have $\SA[i] \in [\rlexm_p \dd \rend{\rlexm_p} - 3\tau + 2)
  \sub \R^{-}_H$.  In $\bigO(1)$ time we obtain $p' = \ARRrinvmap[p]$
  and $j := \select{B_{R'}}{1}{p'} = \rlexm_p$.  Observe now that in
  the block $[j \dd \rend{j} - 3\tau + 2)$ there is at most one
  element with given values of $\Lexp$ and $\Lhead$, and we already
  have values $\Lexp(\SA[i])$ and $\Lhead(\SA[i])$.  We thus proceed
  as follows. First, we compute $\rend{j}$.  For this, we recall that
  by \cref{lm:run-end}\eqref{lm:run-end-it-2}, it holds $\rend{j} = j
  + |H| + \LCE(j, j + |H|)$. Thus, given $j$ and $|H|$, we can compute
  $\rend{j}$ in $\bigO(1)$ time.  We then in $\bigO(1)$ time compute
  $s = \Lhead(j)$ using the lookup table $\LTroot$.  This lets us
  determine $\rendfull{j} = \rend{j} - ((\rend{j} - j - s) \bmod
  |H|)$.  In $\bigO(1)$ time we then obtain $\SA[i] = \rendfull{j} -
  \Lhead(\SA[i]) - \Lexp(\SA[i])|H|$. In total, the query takes
  $\bigO(\log \log n)$ time.
\end{proof}

\bfparagraph{Summary}

By combining all above results, we obtain the following algorithm
to compute $\SA[i]$ for periodic positions.

\begin{proposition}\label{pr:sa-periodic-sa}
  Let $i \in [1 \dd n]$ be such that $\SA[i] \in \R$. Given the
  data structure from \cref{sec:sa-periodic-ds} and the index $i$, in
  $\bigO(\log \log n)$ time we can compute $\SA[i]$.
\end{proposition}
\begin{proof}
  First, using \cref{pr:sa-delta-a}, in $\bigO(1)$ time we compute
  $\type(\SA[i])$. Depending on whether $\SA[i] \in \R^{-}$ or $\SA[i]
  \in \R^{+}$, we use either a combination of
  \cref{pr:sa-delta-a,pr:sa-delta-s} or their symmetric counterparts
  (see the proof of \cref{pr:sa-periodic-isa}), to first compute
  $\Lexp(\SA[i])$ and $\deltas(\SA[i])$ in $\bigO(1)$ time, and then
  $\SA[i]$ in $\bigO(\log \log n)$ time.
\end{proof}

\subsubsection{Construction Algorithm}\label{sec:sa-periodic-construction}

\begin{proposition}\label{pr:sa-periodic-construction}
  Given $\SACore(\T)$, we can in $\bigO(n / \log_{\sigma} n)$ time
  augment it into a data structure from \cref{sec:sa-periodic-ds}.
\end{proposition}
\begin{proof}
  Due to a large number of components, as well as dependency of some
  components on others, we present the description in separate
  paragraphs, in the order in which it occurs.

  \paragraph{Construction of $\LTroot$}

  To compute $\LTroot$, we observe that, given $X \in \Alphabet^{3\tau
  - 1}$, we can check in $\bigO(\tau^2)$ time if $\per(X) \leq
  \tfrac{1}{3}\tau$, and if so, determine the value $\LTroot[\Int(X)]
  = (s, p)$. To compute $\per(X)$, we try all $\ell \in [1 \dd
  \lfloor\frac{\tau}{3}\rfloor]$ until we find that $\ell$ is a period
  of $X$, or that there is no such $\ell$. Assuming $p := \per(X) \leq
  \tfrac{1}{3}\tau$, finding $s \in [0 \dd p)$ satisfying $X[1 + s \dd
  1 + s + p) = \min\{X[t \dd t + p) : t \in [1 \dd p]\}$ also takes
  $\bigO(\tau^2)$ time. Initializing $\LTroot$ takes
  $\bigO(\sigma^{6\tau}) = \bigO(n / \log_{\sigma} n)$. Over all $X
  \in \Alphabet^{3\tau - 1}$, we spend $\bigO(\sigma^{3\tau - 1}
  \tau^2) = \bigO(n^{1/2} \log^2 n) = \bigO(n / \log_{\sigma} n)$
  time.

  \paragraph{Construction of the structure for $\LCE$ queries}

  By~\cite[Theorem~5.4]{sss}, the data structure for $\LCE$ queries on
  $\T$ can be constructed from the packed representation of $\T$ in
  $\bigO(n/\log_{\sigma} n)$ time.

  \paragraph{Construction of $\BVexp$}

  To simplify the notation, for the duration of this proof, we denote
  $E := \BVexp$.
  We use the following definitions. For any $H \in \Lroots$ and $s \in
  [0 \dd |H|)$, let $E^{-}_{s,H}$ denote the block of $E$
  corresponding to suffixes starting in $\R^{-}_{s,H}$, i.e.,
  $E^{-}_{s,H} = E(b \dd e]$, where $(b \dd e] \sub [1 \dd n]$ is such
  that $\R^{-}_{s,H} = \{\SA[i] : i \in (b \dd e]\}$ (such $(b \dd e]$
  exists by \cref{lm:lce}\eqref{lm:lce-it-1}). Finally, let $\unary(x)
  := {\tt 0}^{x}{\tt 1}$ denote the unary encoding of an integer $x
  \geq 0$, and let $\unary^{+}(x)$ be $\unary(x)$ with the first
  symbol removed (in particular, $\unary^{+}(0)$ is the empty
  string). If $(a_i)_{i\in [1\dd k]}$ is a sequence of non-negative
  integers, we define $\unary((a_i)_{i\in [1\dd k]}) :=
  \bigodot_{i=1}^{k}\unary(a_i)$, where $\bigodot$ denotes
  concatenation.  Analogously, $\unary^{+}((a_i)_{i\in [1\dd k]}) :=
  \bigodot_{i=1}^{k} \unary^{+}(a_i)$.  The definitions of $\unary$
  and $\unary^{+}$ are naturally extended to infinite sequences
  $(a_i)_{i\in [1\dd \infty)}$.

  Let $\alpha < 1$ be a positive constant. We first show an algorithm
  that, given the set
  of positions $\R'^{-}_{H}$ (where $H \in \Lroots$) as input,
  computes all bitvectors $E^{-}_{0, H}, \ldots, E^{-}_{|H| - 1, H}$
  in $\bigO(|\R'^{-}_H| + |\R^{-}_H|/\log n + n^{\alpha})$ time.  For
  any $s \in [0 \dd |H|)$ and $k \geq 0$, denote $e_{s,k,H} = |\{j'
  \in \R^{-}_{s,H} : \Lexp(j') = k\}|$.  We start by observing that by
  \cref{lm:lce}\eqref{lm:lce-it-2}, $E^{-}_{s,H} =
  \unary^{+}((e_{s,k,H})_{k\in [0\dd \infty)})$.  The values
  $e_{s,k,H}$ can be efficiently determined based on the following
  observation.  First, note that if $j \in \R'^{-}_H$, then $[j \dd
  \rend{j} - 3\tau + 2) \subseteq \R^-_H$, and $j - 1, \rend{j} -
  3\tau + 2 \not\in \R$, i.e., the block of positions in $\R^-_H$ is
  maximal.  By \cref{lm:R-block}, for any $j' \in [j \dd \rend{j} -
  3\tau + 2)$, it holds $\rend{j'} = \rend{j}$.  Thus, for any $j' \in
  [j \dd \rend{j} - 3\tau + 2)$, we have $\Lexp(j') = \lfloor \tfrac{e
  - j'}{|H|} \rfloor$ and $\Lhead(j') = (e - j') \bmod |H|$, where $e
  = \rend{j} - \Ltail(j)$. With this in mind, for any $j \in
  \R'^{-}_{H}$, we let $\mathcal{I}_j = (3\tau \,{-}\, 2 \,{-}\, t \dd
  s + k|H|]$, where $s = \Lhead(j)$, $k = \Lexp(j)$, and $t =
  \Ltail(j)$. By the above discussion, for any $s \in [0 \dd |H|)$ and
  $k \geq 0$, we have $e_{s,k,H} = |\{j\in \R'^-_H : s + k|H| \in
  \mathcal{I}_j\}|$.  The algorithm consists of three steps:

  \begin{enumerate}
  \item First, we compute the string $\unary((e_{0,k,H})_{k = 0}
    ^{k_{\max}})$, where $k_{\max} = \max \{\Lexp(j') : j' \in
    \R^{-}_H\}$. We start by computing $k_{\max}$. For this we observe
    that $k_{\max} = \max\{\Lexp(j') : j' \in \R'^{-}_H\}$. Thus,
    using \cref{pr:isa-root}, we can compute $k_{\max}$ in
    $\bigO(|\R'^{-}_H|)$ time. To compute $\unary((e_{0, k,
      H})_{k\in[0\dd k_{\max}]})$, we generate the sequence of
    ``events'' from $\R'^{-}_H$, sort them, and then output
    $\unary((e_{0,k,H})_{k \in [0\dd k_{\max}]})$ left-to-right. More
    precisely, let $m = |\R'^{-}_H|$, and let $(p_i,v_i)_{i\in [0\dd
    2m]}$ be a sequence containing the multiset
    $\{(0,0),(k_{\max}+1,0)\} \cup \{(\lceil \min \mathcal{I}_j/|H|
    \rceil, +1) : j \in \R'^{-}_H\} \cup \{(\lfloor \max
    \mathcal{I}_j/|H| \rfloor + 1, -1) : j \in \R'^{-}_H\}$ such that
    for any $i \in [1 \dd 2m]$, it holds $p_{i - 1} \leq p_i$. To
    compute the sequence $(p_i, v_i)_{i\in [0\dd 2m]}$, we observe
    that, given $j \in \R'^{-}_H$, we can compute $\mathcal{I}_j$ in
    $\bigO(1)$ time using \cref{pr:isa-root}. Thus, in $\bigO(m)$ time
    we can generate all pairs in the above multiset. We then sort the
    pairs by the first element. Using $\lceil 1 / \alpha \rceil$-round
    radix sort, this takes $\bigO(m + n^{\alpha})$ time. Consequently,
    we can compute $(p_i, v_i)_{i\in [0\dd 2m]}$ in $\bigO(|\R'^{-}_H|
    + n^{\alpha})$ time.  Given the sequence $(p_i, v_i)_{i\in [0\dd
    2m]}$, we compute $\unary((e_{0, k, H})_{k\in[0\dd k_{\max}]})$ as
    follows.  First, we initialize the output bitvector to the empty
    string and set $v = 0$.  We then iterate through $i = 1, \ldots,
    2m$. For every $i$, we first append $p_i - p_{i-1}$ copies of the
    string $\unary(v)$ to the output string. We then add $v_i$ to
    $v$. To efficiently append multiple copies of $\unary(v)$ to the
    output, we first precompute (in $\bigO(\log^2 n) =
    \bigO(n^{\alpha})$ time) the prefix of length $\log n$ of the
    string $\unary(x)^{\infty}[1 \dd)$ for every $x \in [0 \dd \log n)$.
    This way, we can append $\unary(v)^{\ell}$ to the output in
    $\bigO(1 + (v+1)\ell/\log n)$ time. Consequently, the construction
    of $\unary((e_{0, k, H})_{k\in[0\dd k_{\max}]})$ takes
    $\bigO(|\R'^{-}_H| + |\unary((e_{0, k, H})_{k\in[0\dd k_{\max}]})|
    / \log n + n^{\alpha}) = \bigO(|\R'^{-}_H| + |E^{-}_{0, H}| / \log
    n + n^{\alpha}) = \bigO(|\R'^{-}_H| + |\R^{-}_H| / \log n +
    n^{\alpha})$ time. To show the first upper bound, observe that
    $k_{\max} \leq |E^{-}_{0, H}| + \bigO(\tau/|H|)$. Thus,
    $|\unary((e_{0, k, H})_{k\in[0\dd k_{\max}]})| =
    |\unary^{+}((e_{0, k, H})_{k\in[0\dd \infty)})| + k_{\max} + 1 =
    |E^{-}_{0, H}| + k_{\max} + 1 \leq 2|E^{-}_{0, H}| + \bigO(\log
    n)$.  The second upper bound follows by observing that
    $|E^{-}_{0,H}| + \cdots + |E^{-}_{|H| - 1, H}| = |\R^{-}_H|$.

  \item The second step of the algorithm is to compute the strings
    $\unary((e_{s, k, H})_{k\in[0\dd k_{\max}]})$ for $s \in [1 \dd
    |H|)$.  For any $s \in [1 \dd |H|)$, let $(q^{(s)}_{i},
    p^{(s)}_{i}, v^{(s)}_{i})_{i\in [0\dd m_s]}$ denote the sequence
    containing all the elements $(q, p, v)$ of the multiset $\{(q, 0,
    0) : q \in [1 \dd |H|)\} \cup \{(q, k_{\max} + 1, 0) : q \in [1
    \dd |H|)\} \cup \{(\min \mathcal{I}_j \bmod |H|, \lfloor \min
    \mathcal{I}_j / |H| \rfloor, +1) : j \in \R'^{-}_H\} \cup \{((\max
    \mathcal{I}_j + 1) \bmod |H|, \lfloor (\max \mathcal{I}_j + 1) /
    |H| \rfloor, -1) : j \in \R'^{-}_H\}$ that satisfy $q = s$, and
    for any $i \in [1 \dd m_s]$, it holds $p^{(s)}_{i-1} \leq
    p^{(s)}_{i}$ (note that the elements of this multiset satisfying
    $q=0$ are not included in any sequence).  To compute the sequences
    $(q^{(s)}_i, p^{(s)}_i, v^{(s)}_i)_{i\in [0\dd m_s]}$ for all $s
    \in [1 \dd |H|)$, we first enumerate all triples in the above
    multiset. Using \cref{pr:isa-root}, this takes $\bigO(m)$ time.
    We then sort the triples lexicographically. Using $\lceil 1 /
    \alpha \rceil$-round radix sort, this takes $\bigO(m +
    n^{\alpha})$ time.  This yields all sequences concatenated
    together. It is easy to discard unused elements, and to detect
    boundaries between lists with a single scan.  Consequently, we can
    construct all sequences in $\bigO(|\R'^{-}_H| + n^{\alpha})$
    time. Given the above sequences, we can compute the strings
    $\unary((e_{s, k, H})_{k\in[0\dd k_{\max}]})$ for $s \in [1 \dd
    |H|)$ as follows. The algorithm computes the strings in the order
    of increasing $s$. More precisely, given the string $U :=
    \unary((e_{s - 1, k, H})_{k\in[0\dd k_{\max}]})$ and the sequence
    $(q^{(s)}_i, p^{(s)}_i, v^{(s)}_i)_{i\in [0\dd m_s]}$ (where $s
    \in [1 \dd |H|)$), we compute the string $\unary((e_{s, k,
    H})_{k\in[0\dd k_{\max}]})$ in $\bigO(m_s + |U|/\log n)$ time as
    follows. First, we initialize the output bitvector to the empty
    string, and set $v = 0$ and $y = 0$.  We then iterate through $i =
    1, \ldots, m_s$. For every $i$, we first check if $p^{(s)}_i >
    p^{(s)}_{i - 1}$. If yes, we perform the following three
    steps. First, find the position $y'$ of the $p^{(s)}_i$th 1-bit in
    $U$. Second, append the substring $U(y \dd y']$ to the output,
    except we first prepend it with $v$ zeros (if $v \geq 0$) or
    discard its first $-v$ bits (if $v < 0$). Finally, we set $y = y'$
    and $v = 0$. Then (regardless of whether $p^{(s)}_i > p^{(s)}_{i -
    1}$), we add $v^{(s)}_i$ to $v$. To efficiently compute $y'$ we
    observe that the arguments of the consecutive select queries are
    increasing. We can thus precompute in $\bigO(n^{\alpha})$ time a
    lookup table such that the computation of $y'$ takes $\bigO(1 +
    (y' - y)/\log n)$ time (these lookup tables can be shared among
    algorithms for different $s$).  Note that for any $s \in [0 \dd
    |H|)$, we have $k_{\max} \leq |E^{-}_{s,H}| +
    \bigO(\tau/|H|)$. Thus, $|U| \leq 2|E^{-}_{s - 1, H}| + \bigO(\log
    n)$, and hence the algorithm runs in $\bigO(m_s + |E^{-}_{s - 1,
    H}|/\log n)$ time.  Consequently, by $m_0 + \cdots + m_{|H| - 1}
    \leq 2|\R'^{-}_H| + 2|H|$ and $|E^{-}_{0, H}| + \cdots +
    |E^{-}_{|H| - 1, H}| = |\R^{-}_H|$, over all $s \in [1 \dd |H|)$,
    we spend $\bigO(|\R'^{-}_{H}| + |\R^{-}_H|/\log n + n^{\alpha})$
    time.

  \item The third and final step of the algorithm is to convert the
    string $\unary((e_{s, k, H})_{k\in[0\dd k_{\max}]})$ into
    $\unary^{+}((e_{s, k, H})_{k\in[0\dd k_{\max}]}) = E^{-}_{s,H}$
    for every $s \in [0 \dd |H|)$. Let us fix some $s \in [0 \dd
    |H|)$.  Observe that to implement the conversion, it suffices
    to remove the first bit, as well as every bit following a 1-bit in
    $\unary((e_{s, k, H})_{k\in[0\dd k_{\max}]})$. In the RAM model,
    such local operation is easy implemented in
    $\bigO(1+|\unary((e_{s, k, H})_{k\in[0\dd k_{\max}]})|/\log n)$
    time after a $\bigO(n^{\alpha})$-time preprocessing (we do the
    preprocessing once for all $s \in [0 \dd |H|)$). As observed
    above, $|\unary((e_{s, k, H})_{k\in[0\dd k_{\max}]})| \leq
    2|E^{-}_{s,H}| + \bigO(\log n)$.  Thus, the total time to perform
    the conversion for all $s$ is $\bigO(|R^{-}_H|/\log n +
    n^{\alpha})$.
  \end{enumerate}

  Using the above algorithm, we construct $E$ as follows. We start by
  computing the set $\{(\Int(\Lroot(j)), j)\}_{j \in \R'^{-}}$. For
  this, observe that for every $\tau$-synchronizing set $\mathsf{P}$
  of $\T$, by the density condition (see
  also~\cite[Section~6.1.2]{sss}), $i \in \R'$ implies that either $i
  = 1$ or $i > 1$ and $i - 1 \in \mathsf{P}$.  In particular,
  $|\R'^{-}| \leq |\R'| \leq 1 + |\mathsf{P}|$.  We thus proceed as
  follows. First, using~\cite[Theorem~8.11]{sss} in
  $\bigO(n/\log_{\sigma} n)$ time we construct any
  $\tau$-synchronizing set $\mathsf{P}$ of $\T$ of size
  $\bigO(n/\tau)$.  Then, using the above observation together with
  \cref{pr:isa-root}, we enumerate the set $\{(\Int(\Lroot(j)),
  j)\}_{j \in \R'^{-}}$ in $\bigO(1 + |\mathsf{P}|) =
  \bigO(n/\log_{\sigma} n)$ time. We then discard $\mathsf{P}$.  Using
  $\lceil 1/\alpha \rceil$-round radix sort we then sort in
  $\bigO(|\R'^{-}| + n^{\alpha}) = \bigO(n / \log_{\sigma} n +
  n^{\alpha})$ time the set of pairs by the first coordinate. This
  yields the representation of sets $\R'^{-}_H$ for all $H \in
  \Lroots$. For each $H \in \Lroots$, we then use the above algorithm
  to compute bitvectors $E^{-}_{0, H}, \ldots, E^{-}_{|H| - 1, H}$ in
  $\bigO(|\R'^{-}_H| + |\R^{-}_H| / \log n + n^{\alpha})$ time. By
  $\Lroots \subseteq \Alphabet^{\leq \tau}$, over all $H$, this takes
  $\bigO(|\R'^{-}| + |R^{-}| / \log n + n^{\alpha + \mu})$ time
  (recall that $\tau = \lfloor \mu \log_{\sigma} n \rfloor$ and
  $\mu < \frac16$). Choosing $\alpha < 1 - \mu$
  results in $\bigO(n / \log_{\sigma} n)$ total time.  When bitvectors
  $E^{-}_{s,H}$ are computed for all $H \in \Lroots$ and $s \in [0 \dd
  |H|)$, we initialize $E$ to the string ${\tt 0}^{n}$ in $\bigO(n /
  \log n)$ time, and then ``paste'' all the non-empty bitvectors
  $E^{-}_{s,H}$ into their correct positions. Given $H \in \Lroots$
  and $s \in [0 \dd |H|)$, we first compute in $\bigO(\log n)$ time
  the corresponding string $X \in \Alphabet^{3\tau - 1}$, and then
  compute the position to paste $E^{-}_{s,H}$ using the lookup table
  $\LTrange$.  Over all $H \in \Lroots$ and $s \in [0 \dd |H|)$, this
  takes $\bigO(n^{\mu} \log^2 n + n / \log n) = \bigO(n /
  \log_{\sigma} n)$ time. Thus, altogether, constructing $E$ and
  augmenting it using \cref{th:binrksel} takes $\bigO(n /
  \log_{\sigma} n)$ time.

  \paragraph{Construction of $\LTminexp$}

  Observe that in the above algorithm, if $i$ is the position of the
  leftmost 0-bit in $\unary((e_{s, k, H})_{k\in[0\dd k_{\max}]})$,
  then $\min \{\Lexp(j) : j \in \R^{-}_{s,H}\} = i - 1$. Given the
  packed representation of $\unary((e_{s, k, H})_{k\in[0\dd
  k_{\max}]})$, the position $i$ can be easily found in
  $\bigO(1+|\unary((e_{s, k, H})_{k\in[0\dd k_{\max}]})| / \log n)$
  time. Thus, accounting for the computation of $X \in
  \Alphabet^{3\tau - 1}$ corresponding to the choice of $H \in
  \Lroots$ and $s \in [0 \dd |H|)$, we can initialize $\LTminexp$ in
  $\bigO(n / \log n + n^{\mu} \log^2 n) = \bigO(n / \log_{\sigma} n)$
  time.

  \paragraph{Construction of $\BVRprim$}

  As seen above, we can enumerate $\R'$, and thereby compute
  $\BVRprim$, in $\bigO(n / \log_{\sigma} n)$ time.  Augmenting
  $\BVRprim$ with \cref{th:binrksel} takes $\bigO(n / \log n)$ time.

  \paragraph{Construction of $\ARRrmap$}

  Since for any $j \in \R$, we can in $\bigO(1)$ compute $\Lroot(j)$,
  $\rend{j}$, $s = \Lhead(j)$, and $k=\Lexp(j)$, in $\bigO(n /
  \log_{\sigma} n)$ time we can also enumerate all $j \in
  \R'^{-}$. The key challenge is computing the sequence $(\rlexm_i)_{i
  \in [1 \dd q]}$. By the density condition, for every
  $\tau$-synchronizing set $\mathsf{P}$ of $\T$, it holds that if $j
  \in \R$, then $\rend{j} - 2\tau + 1 \in \mathsf{P}$ (for a proof,
  simply compare the claims of \cref{lm:run-end}
  and~\cite[Fact~3.2]{sss}). This lets us compute $(\rlexm_i)_{i \in
  [1 \dd q]}$ as follows.  First, using~\cite[Theorem~8.11]{sss}, in
  $\bigO(n / \log_{\sigma} n)$ time we construct any
  $\tau$-synchronizing set $\mathsf{P}$ of $\T$ of size $\bigO(n /
  \tau)$. The set $\mathsf{P}$ is returned as an array of size
  $|\mathsf{P}|$. In $\bigO(n/\log_{\sigma} n)$ time we then create a
  bitvector $B_{\mathsf{P}}[1 \dd n]$ such that $B_{\mathsf{P}}[i] =
  1$ holds if and only if $i \in \mathsf{P}$. In $\bigO(n/\log n)$
  time we augment $B_{\mathsf{P}}$ using \cref{th:binrksel}. Let
  $(p^{\rm text}_t)_{t \in [1 \dd |\mathsf{P}|]}$ denote a sequence
  containing elements of $\mathsf{P}$ in increasing order and let
  $(p^{\rm lex}_t)_{t \in [1 \dd |\mathsf{P}|]}$ denote a sequence
  containing $\mathsf{P}$ sorted according to the lexicographical
  order of the corresponding suffixes in $\T$, i.e., such that for any
  $i, i' \in [1 \dd |\mathsf{P}|]$, $i < i'$ implies $\T[p^{\rm
  lex}_{i} \dd n] \prec \T[p^{\rm lex}_{i'} \dd n]$.  Given the
  array containing $\mathsf{P}$, we compute the sequence $(p^{\rm
  lex}_t)_{t \in [1 \dd |\mathsf{P}|]}$ in $\bigO(n/\log_{\sigma}
  n)$ time using~\cite[Theorem~4.3]{sss}. Let $\ISA_{\mathsf{P}}[1 \dd
  |\mathsf{P}|]$ be an array storing a permutation of $[1 \dd
  |\mathsf{P}|]$ such that $\ISA_{\mathsf{P}}[j] = i$ holds if and
  only if $p^{\rm text}_j = p^{\rm lex}_i$. Using the sequence
  $(p^{\rm lex}_t)_{t \in [1 \dd |\mathsf{P}|]}$ and the bitvector
  $B_{\mathsf{P}}$, we compute $\ISA_{\mathsf{P}}$ in
  $\bigO(|\mathsf{P}|) = \bigO(n/\log_{\sigma} n)$ time: For every $i
  \in [1 \dd |\mathsf{P}|]$, we first compute $j =
  \rank{B_{\mathsf{P}}}{1}{p^{\rm lex}_i}$ and then set
  $\ISA_{\mathsf{P}}[j] = i$.  Next, for each $j \in \R'^{-}$, letting
  $H = \Lroot(j)$ and $j_{\mathsf{P}} =
  \rank{B_{\mathsf{P}}}{1}{\rend{j} - 2\tau + 1}$, we form a tuple
  $(\Int(H), \rend{j} - \rendfull{j},
  \ISA_{\mathsf{P}}[j_{\mathsf{P}}], j)$. Observe, that
  $X \prec X'$ holds if and only if $\Int(X) <
  \Int(X')$.  Let $j,j' \in \R'^{-}_H$. Note that since both
  $\T[\rendfull{j} \dd \rend{j})$ and $\T[\rendfull{j'} \dd
  \rend{j'})$ are prefixes of $H$, by definition of $\R^{-}$,
  $\rend{j} - \rendfull{j} < \rend{j'} - \rendfull{j'}$ implies
  $\T[\rendfull{j} \dd n] \prec \T[\rendfull{j'} \dd n]$. If $\rend{j}
  - \rendfull{j} = \rend{j'} - \rendfull{j'}$, then $\T[\rend{j} -
  2\tau + 1 \dd \rendfull{j}) = \T[\rend{j'} - 2\tau + 1 \dd
  \rendfull{j'})$, and consequently, $\T[\rendfull{j} \dd n] \prec
  \T[\rendfull{j'} \dd n]$ holds if and only if
  $\ISA_{\mathsf{P}}[j_{\mathsf{P}}] <
  \ISA_{\mathsf{P}}[j'_{\mathsf{P}}]$. We have thus shown that sorting
  the tuples lexicographically yields a sequence $(\rlexm_i)_{i \in [1
  \dd q]}$ on the fourth coordinate. Given $j \in \R'^{-}$, we can
  compute the corresponding tuple in $\bigO(1)$ time. Thus, since all
  its elements are integers in the range $[1 \dd n]$, using LSD
  radix-sort, we can compute $(\rlexm_i)_{i \in [1 \dd q]}$ in
  $\bigO(n / \log_{\sigma} n)$ time.  With a single scan of
  $(\rlexm_i)_{i \in [1 \dd q]}$ and the help of rank queries on
  $\BVRprim$ we can then compute table $\ARRrmap$ in $\bigO(n /
  \log_{\sigma} n)$ time.

  \paragraph{Construction of $\ARRrinvmap$}

  Given $\ARRrmap$, we can compute $\ARRrinvmap$ in $\bigO(q) =
  \bigO(n / \log_{\sigma} n)$ time, since these two arrays are
  inverses of each other.

  \paragraph{Construction of $\LTruns$}

  In $\bigO(\sigma^{\tau} + |\R'^{-}|)$ time we perform a synchronized
  enumeration of all $H \in \Alphabet^{\leq \tau}$ in lexicographical
  order and the $\Lroot$ values (obtained using \cref{pr:isa-root})
  for positions in the sequence $(\rlexm_i)_{i \in [1 \dd q]}$.  This
  lets us obtain the pair $(b_H,e_H)$ satisfying $\{\rlexm_i : i \in
  (b_H \dd e_H]\} = \R'^{-}_H$ for every $H \in \Alphabet^{\leq \tau}$
  satisfying $\R'^{-}_{H} \neq \emptyset$. For each such $H$, we then
  enumerate all $H' \in \Alphabet^{\leq \tau}$ and for each we find
  corresponding subrange of $(b_H \dd e_H]$ in $\bigO(\tau \log n)$
  time using binary search. Overall, the initialization of $\LTruns$
  takes $\bigO(\sigma^{6\tau} \tau + |\R'^{-}| + \sigma^{2\tau} \tau
  \log n) = \bigO(n / \log_{\sigma} n)$ time.

  \paragraph{Construction of $\LTpref$}

  To construct $\LTpref$, we enumerate all possible $H \in
  \Alphabet^{\leq \tau}$.  For each $H$, we try all $s
  \in [0 \dd |H|)$, and for each we construct the string $\Pref_{3\tau
  - 1}(s,H)$ in $\bigO(\tau)$ time.  Over all $H$, and including the
  initialization of $\LTpref$, this takes $\bigO(\sigma^{6\tau} \tau +
  \sigma^{\tau} \tau^2) = \bigO(n^{6\mu} \log n) = \bigO(n /
  \log_{\sigma} n)$ time.

  \paragraph{Construction of range counting/selection for $A$}

  From $(\rlexm_i)_{i \in [1 \dd q]}$ we construct in $\bigO(n /
  \log_{\sigma} n)$ time the sequence $(\ell_i)_{i \in [1 \dd q]}$,
  and then build the array $\ARRnontail[1 \dd q]$ and augment it with
  a range counting/selection data structure. Using
  \cref{pr:range-queries}, by $q = \bigO(n / \log_{\sigma} n)$ and
  $\sum_{i=1}^{q} \ARRnontail[i] = \bigO(n)$, this takes $\bigO(n /
  \log_{\sigma} n)$ time.

  \paragraph{Construction of the remaining components}

  After the above components are constructed, we then analogously
  construct their symmetric counterparts (adapted according to
  \cref{lm:lce}).
\end{proof}

\subsection{The Final Data Structure}\label{sec:sa-final}

In this section, we put together
\cref{sec:sa-core,sec:sa-nonperiodic,sec:sa-periodic} to obtain a data
structure that, given any $j \in [1 \dd n]$ (resp. $i \in [1 \dd n]$)
computes $\ISA[j]$ (resp.\ $\SA[i]$) in $\bigO(\log^{\epsilon} n)$ time.

The section is organized as follows. First, we introduce the
components of the data structure (\cref{sec:sa-ds}). Next, we describe
the query algorithms (\cref{sec:sa-isa,sec:sa-sa}). Finally, we show
the construction algorithm (\cref{sec:sa-construction}).

\subsubsection{The Data Structure}\label{sec:sa-ds}

The data structure consists of two components:
\begin{enumerate}
\item The structure from \cref{sec:sa-nonperiodic-ds} (used to handle
  nonperiodic positions).
\item The structure from \cref{sec:sa-periodic-ds} (used to handle
  periodic positions).
\end{enumerate}

In total, the data structure needs $\bigO(n / \log_{\sigma} n)$ space.

\subsubsection{Implementation of
  \texorpdfstring{$\ISA$}{ISA} Queries}\label{sec:sa-isa}

\begin{proposition}\label{pr:sa-isa}
  Given the data structure from \cref{sec:sa-ds} and any $j \in [1 \dd
    n]$, in $\bigO(\log^{\epsilon} n)$ time we can compute $\ISA[j]$.
\end{proposition}
\begin{proof}
  First, we use \cref{pr:sa-core-isa} to check in $\bigO(1)$ time if
  $j \in \R$. Depending on whether $j \in \R$ or not, we use
  \cref{pr:sa-nonperiodic-isa} or \cref{pr:sa-periodic-isa} to compute
  $\ISA[j]$ in $\bigO(\log^\epsilon n)$ or $\bigO(\log \log n)$ time
  (respectively).
\end{proof}

\subsubsection{Implementation of
  \texorpdfstring{$\SA$}{SA} Queries}\label{sec:sa-sa}

\begin{proposition}\label{pr:sa-sa}
  Given the data structure from \cref{sec:sa-ds} and any $i \in [1 \dd
  n]$, in $\bigO(\log^{\epsilon} n)$ time we can compute $\SA[i]$.
\end{proposition}
\begin{proof}
  First, we use \cref{pr:sa-core-sa} to check in $\bigO(1)$ time if
  $\SA[i] \in \R$. Depending on whether $\SA[i] \in \R$ or not, we use
  \cref{pr:sa-nonperiodic-sa} or \cref{pr:sa-periodic-sa} to compute
  $\SA[i]$ in $\bigO(\log^\epsilon n)$ or $\bigO(\log \log n)$ time
  (respectively).
\end{proof}

\subsubsection{Construction Algorithm}\label{sec:sa-construction}

\begin{proposition}\label{pr:sa-construction}
  Given the packed representation of $\T \in \Alphabet^n$,
  we can construct the data structure from \cref{sec:sa-ds} in
  $\bigO(n \min(1, \log \sigma / \sqrt{\log n}))$ time and
  $\bigO(n / \log_{\sigma} n)$ working space.
\end{proposition}
\begin{proof}
  First, from a packed representation of $\T$, we construct
  $\SACore(\T)$ in $\bigO(n / \log_{\sigma} n)$ time using
  \cref{pr:sa-core-construction}. Then, using
  \cref{pr:sa-nonperiodic-construction,pr:sa-periodic-construction},
  we augment $\SACore(\T)$ into the two components of the structure
  from \cref{sec:sa-ds} in $\bigO(n \min(1, \log \sigma /
  \sqrt{\log n}))$ and $\bigO(n / \log_{\sigma} n)$ time
  (respectively) and using $\bigO(n / \log_{\sigma} n)$ working
  space.
\end{proof}

\subsection{Summary}\label{sec:sa-summary}

By combining \cref{pr:sa-isa,pr:sa-sa,pr:sa-construction}
we obtain the following final result of this section.

\begin{theorem}\label{th:sa}
  Given any constant $\epsilon \in (0,1)$ and the packed
  representation of a text $\T \in \Alphabet^n$ with $2 \leq \sigma <
  n^{1/7}$, in $\bigO(n \min(1, \log \sigma / \sqrt{\log n}))$ time
  and $\bigO(n / \log_{\sigma} n)$ working space we can construct a
  data structure of size $\bigO(n / \log_{\sigma} n)$ that:
  \begin{itemize}
  \item Given any $i \in [1 \dd n]$ returns $\SA[i]$ in
    $\bigO(\log^{\epsilon} n)$ time,
  \item Given any $j \in [1 \dd n]$ returns $\ISA[j]$ in
    $\bigO(\log^{\epsilon} n)$ time.
  \end{itemize}
\end{theorem}

We also immediately obtain the following more general result.

\begin{theorem}\label{th:sa-general}
  Consider a data structure answering prefix rank and selection
  queries that, for any string of length $m$ over alphabet
  $\Alphabet^\ell$, achieves the following complexities:
  \begin{enumerate}
  \item Space usage $S(m, \ell, \sigma)$,
  \item Preprocessing time $P_t(m, \ell, \sigma)$,
  \item Preprocessing space $P_s(m, \ell, \sigma)$,
  \item Query time $Q(m, \ell, \sigma)$.
  \end{enumerate}
  For every $\T \in \Alphabet^n$ with $2 \leq \sigma < n^{1/7}$, there
  exist $m = \bigO(n/\log_{\sigma} n)$ and $\ell = \bigO(\log_{\sigma}
  n)$ such that, given the packed representation of $\T$, we can in
  $\bigO(n / \log_{\sigma} n + P_t(m, \ell, \sigma))$ time and
  $\bigO(n / \log_{\sigma} n + P_s(m, \ell,\sigma))$ working space
  build a structure of size $\bigO(n/\log_{\sigma} n + S(m, \ell,
  \sigma))$~that:
  \begin{itemize}
  \item Given any $i \in [1 \dd n]$ returns $\SA[i]$ in $\bigO(\log
    \log n + Q(m, \ell, \sigma))$ time,
  \item Given any $j \in [1 \dd n]$ returns $\ISA[j]$ in $\bigO(\log
    \log n + Q(m, \ell, \sigma))$ time.
  \end{itemize}
\end{theorem}

\section{Pattern Matching Queries}\label{sec:pm}

Let $\epsilon \in (0, 1)$ be any fixed constant and let $\T \in
\Alphabet^n$, where $2 \leq \sigma < n^{1/7}$.  In this section we
show how, given the packed representation of $\T$, to construct in
$\bigO(n \min(1, \log \sigma / \sqrt{\log n}))$ time and $\bigO(n /
\log_{\sigma} n)$ working space a structure of size
$\bigO(n/\log_{\sigma} n)$ that, given the packed representation of a
pattern $\Pat \in \Alphabet^m$, returns $\LB(\Pat, \T)$ and $\UB(\Pat,
\T)$ in $\bigO(m / \log_{\sigma} n + \log^{\epsilon} n)$ time. We also
derive a general reduction depending on prefix rank and selection
queries.

As in \cref{sec:sa}, we let $\tau = \lfloor \mu \log_{\sigma} n
\rfloor$, where $\mu$ is some positive constant smaller than $\frac16$
such that $\tau \geq 1$, be fixed for the duration of this section.
Throughout, we also use $\R$ as a shorthand for $\R(\tau, \T)$.

\begin{definition}\label{def:pattern-periodicity}
  Let $\Pat \in \Alphabet^m$. We call pattern $\Pat$ \emph{periodic}
  if it holds that $m \geq 3\tau - 1$ and $\per(\Pat[1 \dd 3\tau {-}
  1]) \leq \tfrac{1}{3}\tau$.  Otherwise, $\Pat$ is
  \emph{nonperiodic}.
\end{definition}

\bfparagraph{Organization}

The structure and the query algorithm for a pattern $\Pat$ are
different depending on whether $\Pat$ is periodic
(\cref{def:pattern-periodicity}).  Our description is thus split as
follows. First (\cref{sec:pm-core}), we describe the set of data
structures called collectively the index ``core'' that enables
efficiently checking if $\Pat$ is periodic (it is also used to handle
very short patterns and contains some common components utilized by
the remaining parts). In the following two parts
(\cref{sec:pm-nonperiodic,sec:pm-periodic}), we describe structures
handling each of the two cases. All ingredients are then put together
in \cref{sec:pm-final}. Finally, we present our result in the general
form (\cref{sec:pm-summary}).

\subsection{The Index Core}\label{sec:pm-core}

In this section, we present a data structure that, given a packed
representation of any pattern $\Pat \in \Alphabet^{m}$, lets us in
$\bigO(1)$ time check if $\Pat$ is periodic. It also let us compute
$(\LB(\Pat, \T), \UB(\Pat, \T))$ if $m < 3\tau - 1$.

The section is organized as follows. First, we introduce the
components of the data structure (\cref{sec:pm-core-ds}).  We then
show how using this structure to implement the periodicity check
(\cref{sec:pm-core-nav}). Next, we describe the query algorithm for
short patterns (\cref{sec:pm-core-pm}). Finally, we show the
construction algorithm (\cref{sec:pm-core-construction}).

\subsubsection{The Data Structure}\label{sec:pm-core-ds}

The index core, denoted $\CountCore(\T)$ consists of the following
subset of components of $\SACore(\T)$:
\begin{enumerate}
\item The packed representation of $\T$ using $\bigO(n / \log_{\sigma}
  n)$ space.
\item The lookup table $\LTrange$ using $\bigO(\sigma^{6\tau}) =
  \bigO(n / \log_{\sigma} n)$ space.
\item The lookup table $\LTper$ using $\bigO(\sigma^{6\tau}) = \bigO(n
  / \log_{\sigma} n)$ space.
\end{enumerate}

In total, $\CountCore(\T)$ needs $\bigO(n / \log_{\sigma} n)$ space.

\subsubsection{Navigation Primitives}\label{sec:pm-core-nav}

\begin{proposition}\label{pr:pm-core-periodicity}
  Given $\CountCore(\T)$ and a packed representation of $\Pat \in
  \Alphabet^{m}$, we can in $\bigO(1)$ time determine whether $\Pat$
  is periodic (\cref{def:pattern-periodicity}).
\end{proposition}
\begin{proof}
  If $m < 3\tau - 1$, we return false. Otherwise, in $\bigO(1)$ time
  we compute $x = \Int(X)$, where $X = \Pat[1 \dd 3\tau {-} 1]$.  We
  then look up $p = \LTper[x]$ and return true if and only if $p \leq
  \frac{1}{3}\tau$.
\end{proof}

\subsubsection{Implementation of Queries}\label{sec:pm-core-pm}

\begin{proposition}\label{pr:pm-core-pm}
  Let $\Pat \in \Alphabet^m$ be a pattern satisfying $m < 3\tau - 1$.
  Given $\CountCore(\T)$ and the packed representation of $\Pat$, in
  $\bigO(1)$ time we can compute $(\LB(\Pat, \T), \UB(\Pat, \T))$.
\end{proposition}
\begin{proof}
  Using $\LTrange$ on $\Pat$, we immediately obtain and return
  $(\LB(\Pat, \T), \UB(\Pat, \T))$ in $\bigO(1)$ time.
\end{proof}

\subsubsection{Construction Algorithm}\label{sec:pm-core-construction}

\begin{proposition}\label{pr:pm-core-construction}
  Given the packed representation of $\T \in \Alphabet^n$, we can
  construct $\CountCore(\T)$ in $\bigO(n / \log_{\sigma} n)$ time.
\end{proposition}
\begin{proof}
  Since $\CountCore(\T)$ contains a subset of components of
  $\SACore(\T)$, this follows by \cref{pr:sa-core-construction}.
\end{proof}

\subsection{The Nonperiodic Patterns}\label{sec:pm-nonperiodic}

In this section, we describe a data structure that, given a packed
representation of any nonperiodic pattern $\Pat \in \Alphabet^{m}$
(see \cref{def:pattern-periodicity}), computes $(\LB(\Pat, \T),
\UB(\Pat, \T))$ in $\bigO(m / \log_{\sigma} n + \log^{\epsilon} n)$
time.

The section is organized as follows. First, we introduce the
components of the data structure (\cref{sec:pm-nonperiodic-ds}). Next,
we describe the query algorithm
(\cref{sec:pm-nonperiodic-pm}). Finally, we show the construction
algorithm (\cref{sec:pm-nonperiodic-construction}).

\subsubsection{The Data Structure}\label{sec:pm-nonperiodic-ds}

\bfparagraph{Definitions}

Let $\S$ be a $\tau$-synchronizing set, as defined in
\cref{sec:sa-nonperiodic-ds}. Let $\ARRslex[1 \dd n']$ be an array
defined by $\ARRslex[i] = \slex_i$ (where $(\slex_t)_{t \in [1 \dd
n']}$ is a sequence as defined in \cref{sec:sa-nonperiodic-ds}).

\bfparagraph{Components}

The data structure to handle nonperiodic patterns consists of three
components:
\begin{enumerate}
\item The index core $\CountCore(\T)$ (\cref{sec:pm-core-ds}) using
  $\bigO(n / \log_{\sigma} n)$ space.
\item The data structure from \cref{sec:sa-nonperiodic-ds} using
  $\bigO(n / \log_{\sigma} n)$ space.
\item The data structure from \cref{pr:meta-trie} for the array
  $\ARRslex[1 \dd n']$. By $n' = \bigO(n / \log_{\sigma} n)$ and
  \cref{pr:meta-trie}, it needs $\bigO(n / \log_{\sigma} n)$ space.
\end{enumerate}

In total, the data structure takes $\bigO(n / \log_{\sigma} n)$ space.

\subsubsection{Implementation of Queries}\label{sec:pm-nonperiodic-pm}

\begin{lemma}\label{lm:pm-nonperiodic}
  Let $\Pat \in \Alphabet^m$ be a nonperiodic pattern satisfying $m
  \geq 3\tau - 1$, and let $X \in \D$ be a prefix of $\Pat$.  Denote
  $\deltatext = |X| - 2\tau$ and $\Pat' = \Pat(\deltatext \dd m]$.
  Let $(\bpre, \epre)$ be such that $\bpre = |\{i \in [1 \dd n'] :
  \T[\slex_i \dd n] \prec \Pat'\}|$ and $(\bpre \dd \epre] = \{i \in
  [1 \dd n'] : \Pat'\text{ is a prefix of }\T[\slex_i \dd n]\}$.
  Then, it holds
  \vspace{1.5ex}
  \[
    (\LB(\Pat, \T), \UB(\Pat, \T)) =
      (b_X + \delta_1, b_X + \delta_2),
    \vspace{1.5ex}
  \]
  where $b_X = \LB(X, \T)$, $\delta_1 = \rank{W}{\revstr{X}}{\bpre}$,
  and $\delta_2 = \rank{W}{\revstr{X}}{\epre}$.
\end{lemma}
\begin{proof}
  Observe that by the consistency of $\S$ and $X \in \D$, $j \in
  \Occ(X, \T)$ implies $j + \deltatext \in \S$. Thus, $\Occ(X, \T) =
  \{s - \deltatext : s \in \S\text{ and }s - \deltatext \in \Occ(X,
  \T)\}$.  Note also that if $S_1$ is a prefix of $S_2$ then $\LB(S_2,
  \T) = \LB(S_1, \T) + |\{j \in \Occ(S_1, \T) : T[j \dd n] \prec
  S_2\}|$.  Together with the definition of $\bpre$, this implies
  \begin{align*}
  \LB(\Pat, \T)
    &= \LB(X,\T) + |\{j \in \Occ(X, \T) : \T[j \dd n] \prec \Pat\}|\\
    &= b_X + |\{s - \deltatext : s \in \S\text,\,s - \deltatext \in
      \Occ(X, \T),\text{ and }\T[s - \deltatext \dd n] \prec \Pat\}|\\
    &= b_X + |\{s \in \S : s - \deltatext \in
      \Occ(X, \T)\text{ and }\T[s - \deltatext \dd n] \prec \Pat\}|\\
    &= b_X + |\{s \in \S : s - \deltatext \in
      \Occ(X, \T)\text{ and }\T[s \dd n] \prec \Pat'\}|\\
    &= b_X + |\{i \in [1 \dd n'] : \slex_i - \deltatext \in \Occ(X, \T)
       \text{ and }\T[\slex_i \dd n] \prec \Pat'\}|\\
    &= b_X + |\{i \in [1 \dd n'] : \slex_i - \deltatext \in \Occ(X, \T)
       \text{ and }i \leq \bpre\}|\\
   &= b_X + |\{i \in [1 \dd \bpre] : \slex_i - \deltatext \in
       \Occ(X, \T)\}|\\
    &= b_X + \rank{W}{\revstr{X}}{\bpre}\\
    &= b_X + \delta_1,
  \end{align*}
  where the second-to-last equality follows by $W[i] = \revstr{X_i}$,
  where $X_i = \T^{\infty}[\slex_i - \tau \dd \slex_i + 2\tau)$, since
  $\slex_i - \deltatext \in \Occ(X, \T)$ holds if and only if $X$ if a
  suffix of $X_i$ (i.e., if $\revstr{X}$ is a prefix of
  $\revstr{X_i}$).

  Next, we show that $|\Occ(\Pat, \T)| = \delta_2 - \delta_1$. We start
  by observing that (similarly as above, except applied to $\Pat$) by the
  consistency of $\S$ and $X \in \D$ being a prefix of $\Pat$, $j \in
  \Occ(\Pat, \T)$ implies $j + \deltatext \in \S$. Thus, $\Occ(\Pat,
  \T) = \{s - \deltatext : s \in \S\text{ and }s - \deltatext \in
  \Occ(\Pat, \T)\}$ and hence,
  \begin{align*}
  |\Occ(\Pat, \T)|
    &= |\{s - \deltatext : s \in \S,\,
        s - \deltatext \in \Occ(\Pat, \T)\}|\\
    &= |\{s \in \S : s - \deltatext \in \Occ(\Pat, \T)\}|\\
    &= |\{i \in [1 \dd n'] : \slex_i - \deltatext \in \Occ(\Pat, \T)\}|\\
    &= |\{i \in [1 \dd n'] : \slex_i  - \deltatext \in \Occ(X, \T)
       \text{ and }\slex_i \in \Occ(\Pat', \T)\}|\\
    &= |\{i \in [1 \dd n'] : \slex_i - \deltatext \in \Occ(X, \T)
       \text{ and }\bpre < i \leq \epre\}|\\
    &= |\{i \in (\bpre \dd \epre] :
       \slex_i - \deltatext \in \Occ(X, \T)\}|\\
    &= \rank{W}{\revstr{X}}{\bpre} - \rank{W}{\revstr{X}}{\epre}\\
    &= \delta_1 - \delta_2.
  \end{align*}
  Combining the above with the earlier equality, we obtain
  $\UB(\Pat, \T) = \LB(\Pat, \T) + |\Occ(\Pat, \T)| = b_X +
  \delta_2$, i.e., the second part of the claim.
\end{proof}

\begin{remark}
  Note that since the range $(\bpre \dd \epre]$ is well-defined
  even if $\epre - \bpre = 0$, the above lemma holds even if
  $|\Occ(\Pat, \T)| = 0$.
\end{remark}

\begin{proposition}\label{pr:pm-nonperiodic-range}
  Let $\Pat \in \Alphabet^m$ be a nonperiodic pattern satisfying $m
  \geq 3\tau - 1$. Given the data structure from
  \cref{sec:pm-nonperiodic-ds} and the packed representation of
  $\Pat$, in $\bigO(m / \log_{\sigma} n + \log^{\epsilon} n)$ time we
  can compute $(\LB(\Pat, \T), \UB(\Pat, \T))$.
\end{proposition}
\begin{proof}
  Let $Y = \Pat[1 \dd 3\tau {-}
  1]$. First, using the lookup table $\LTrange$ on $Y$, in $\bigO(1)$
  time we compute $(b_Y, e_Y) = (\LB(Y, \T), \UB(Y, \T))$.  If $e_Y -
  b_Y = 0$, then $\Occ(Y, \T) = \emptyset$, and it is easy to see that
  then we have $\LB(\Pat, \T) = \LB(Y, \T)$ and $\UB(\Pat, \T) =
  \UB(Y, \T)$. We thus return $(b_Y, e_Y)$. Let us thus assume $b_Y
  \neq e_Y$, i.e., $\Occ(Y, \T) \neq \emptyset$.  Together with
  $\per(Y) > \tfrac{1}{3}\tau$, this implies (see
  \cref{sec:sa-nonperiodic-ds}) that there exists a unique prefix $X
  \in \D$ of $\Pat$. Using $\LTD$ on $Y$ in $\bigO(1)$ time we compute
  the prefix $X \in \D$ of $\Pat$. Let $\delta = |X| - 2\tau$. Using
  again the lookup table $\LTrange$, in $\bigO(1)$ time we compute
  $(b_X, e_X) = (\LB(X, \T), \UB(X, \T))$.  Using \cref{pr:meta-trie},
  we then compute in $\bigO(m / \log_{\sigma} n + \log \log n)$ time
  the pair $(\bpre, \epre)$ for the pattern $\Pat' := \Pat(\delta \dd
  m]$. By \cref{lm:pm-nonperiodic}, we then return $(\LB(\Pat, \T),
  \UB(\Pat, \T)) = (b_X + \rank{W}{\revstr{X}}{\bpre}, b_X +
  \rank{W}{\revstr{X}}{\epre})$, with the two prefix rank queries
  implemented using \cref{th:wavelet-tree}, in $\bigO(\log^\epsilon n)$
  time each (the string $\revstr{X}$ is obtained using the lookup
  table $\LTrev$).  Altogether, the query time is $\bigO(m /
  \log_{\sigma} n + \log^\epsilon n)$.
\end{proof}

\subsubsection{Construction Algorithm}\label{sec:pm-nonperiodic-construction}

\begin{proposition}\label{pr:pm-nonperiodic-construction}
  Given $\CountCore(\T)$, we can in $\bigO(n \min(1, \log \sigma /
  \sqrt{\log n}))$ time and in $\bigO(n / \log_{\sigma} n)$ working space
  augment it into a data structure from \cref{sec:pm-nonperiodic-ds}.
\end{proposition}
\begin{proof}
  First, we combine
  \cref{pr:sa-core-construction,pr:sa-nonperiodic-construction}
  (recall that the packed representation of $\T$ is a component of
  $\CountCore(\T)$) to construct the structure from
  \cref{sec:sa-nonperiodic-ds} in $\bigO(n \min(1, \log \sigma /
  \sqrt{\log n}))$ time and using $\bigO(n / \log_{\sigma} n)$ working
  space. In particular, this constructs $(\slex_i)_{i \in [1 \dd n']}$.
  We thus initialize $\ARRslex[i] = \slex_i$ for $i \in [1 \dd n']$ and
  in $\bigO(n/\log_{\sigma} n)$ time and $\bigO(n/\log_{\sigma} n)$
  working space construct the data structure from \cref{pr:meta-trie}.
  The overall runtime is $\bigO(n \min(1, \log \sigma / \sqrt{\log n}))$.
  The working space never exceed $\bigO(n / \log_{\sigma} n)$ words.
\end{proof}

\subsection{The Periodic Patterns}\label{sec:pm-periodic}

In this section, we describe a data structure that, given a packed
representation of any periodic pattern $\Pat \in \Alphabet^{m}$ (see
\cref{def:pattern-periodicity}), computes $(\LB(\Pat, \T), \UB(\Pat,
\T))$ in $\bigO(m / \log_{\sigma} n + \log \log n)$ time.

The section is organized as follows. First, we present the toolbox of
combinatorial properties for periodic patterns
(\cref{sec:pm-periodic-prelim}). Next, we introduce the components of
the data structure (\cref{sec:pm-periodic-ds}).  We then show how
using this structure to implement some basic navigational routines
(\cref{sec:pm-periodic-nav}). Next, we describe the query algorithm
(\cref{sec:pm-periodic-pm}). Finally, we show the construction
algorithm (\cref{sec:pm-periodic-construction}).

\subsubsection{Preliminaries}\label{sec:pm-periodic-prelim}

Let $\Pat \in \Alphabet^{m}$ be a periodic pattern (see
\cref{def:pattern-periodicity}).  We define $\Lroot(\Pat) =
\min\{\Pat[1 + t \dd 1 + t + p) : t \in [0 \dd p)\}$, where $p =
\per(\Pat[1 \dd 3\tau - 1])$.  Let $H = \Lroot(\Pat)$. We define
$\pend{\Pat} = 1 + p + \lcp(\Pat[1 \dd m], \Pat[1 + p \dd m])$, where
$p = |H|$. By definition, there exists $s \in [0 \dd p)$ such that
$\Pat[1 + s \dd 1 + s + p) = H$.  Thus, we can write $\Pat[1 \dd
\pend{\Pat}) = H'H^{k}H''$, where $H'$ (resp.\ $H''$) is a proper
suffix (resp.\ prefix) of $H$. By $\pend{\Pat} \geq 3\tau$ and $|H|
\leq \tau$, such decomposition is unique (see also
\cref{sec:sa-periodic-prelim}). We denote $\Lhead(\Pat) = |H'|$,
$\Lexp(\Pat) = k$, and $\Ltail(\Pat) = |H''|$.  We also let
$\pendfull{\Pat} = \pend{\Pat} - \Ltail(\Pat)$. We define $\type(\Pat)
= +1$ if $\pend{\Pat} \leq m$ and $\Pat[\pend{\Pat}] \succ
\Pat[\pend{\Pat} - p]$ (where $p = |\Lroot(\Pat)|$), and $\type(\Pat)
= - 1$ otherwise.

\begin{lemma}\label{lm:pm-lce}
  Let $\Pat \in \Alphabet^{m}$ be a periodic pattern and let $s =
  \Lhead(\Pat)$ and $H = \Lroot(\Pat)$.  For any $j \in [1 \dd n]$,
  $\lcp(\Pat, \T[j \dd n]) \geq 3\tau - 1$ holds if and only if $j \in
  \R_{s,H}$. Moreover, if $j \in \R_{s,H}$ then, letting $t = \pend{\Pat}
  - 1$ and $t' = \rend{j} - j$, it holds $\lcp(\Pat, \T[j \dd n])
  \geq \min(t, t')$ and:
  \begin{enumerate}
  \item\label{lm:pm-lce-it-1} If $\type(\Pat) \neq \type(j)$, then
    $\Pat \prec \T[j \dd n]$ if and only if $\type(\Pat) < \type(j)$,
  \item\label{lm:pm-lce-it-2} If $\type(\Pat) = \type(j) = -1$ and $t
    \neq t'$, then $\Pat \prec \T[j \dd n]$ if and only if $t < t'$,
  \item\label{lm:pm-lce-it-3} If $\type(\Pat) = \type(j) = +1$ and $t
    \neq t'$, then $\Pat \prec \T[j \dd n]$ if and only if $t > t'$,
  \item\label{lm:pm-lce-it-4} If $\type(\Pat) \neq \type(j)$ or $t
    \neq t'$, then $\Pat \neq \T[j \dd n]$ and $\lcp(\Pat, \T[j \dd n])
    = \min(t, t')$.
  \end{enumerate}
\end{lemma}
\begin{proof}

  Let $j \in [1 \dd n]$ be such that $\lcp(\Pat, \T[j \dd n]) \geq
  3\tau - 1$. Denoting $p = \per(\Pat[1 \dd 3\tau - 1])$ and $p' =
  \per(\T[j \dd j + 3\tau - 1))$ we then have $p' = p \leq
  \tfrac{1}{3}\tau$.  Thus, $j \in \R$. Moreover, this implies
  $\Lroot(j) = \min\{\T[j + \delta \dd j + \delta + p') : \delta \in
  [0 \dd p')\} = \min\{\T[j + \delta \dd j + \delta + p) : \delta \in
  [0 \dd p)\} = \min\{\Pat[1 + \delta \dd 1 + \delta + p) : \delta \in
  [0 \dd p)\} = H$.  To show that $\Lhead(j) = s$, note that by $|H|
  \leq \tau$, the string $H'H^2$ (where $H'$ is a length-$s$ suffix of
  $H$) is a prefix of $\Pat[1 \dd 3\tau - 1] = \T[j \dd j + 3\tau -
  1)$.  On the other hand, $\Lhead(j) = s'$ implies that
  $\widehat{H}'H^2$ (where $\widehat{H}'$ is a length-$s'$ suffix of
  $H$) is a prefix of $\T[j \dd j + 3\tau - 1)$.  Thus, by the
  synchronization property of primitive
  strings~\cite[Lemma~1.11]{AlgorithmsOnStrings} applied to the two
  copies of $H$, we have $s' = s$, and hence, $j \in \R_{s,H}$.  For
  the converse implication, assume $j \in \R_{s,H}$. This implies that
  both $\Pat[1 \dd \pend{\Pat})$ and $\T[j \dd \rend{j})$ are prefixes
  of $H'\cdot H^{\infty}[1 \dd)$ (where $H'$ is as above). Thus, by
  $\pend{\Pat} - 1,\, \rend{j} - j \geq 3\tau - 1$, we obtain
  $\lcp(\Pat, \T[j \dd n]) \geq 3\tau - 1$.

  Let us now assume $j \in \R_{s,H}$. Since, as noted above, both
  $\Pat[1 \dd \pend{\Pat}) = \Pat[1 \dd t]$ and $\T[j \dd \rend{j})
  = \T[j \dd j + t')$ are prefixes of $H'\cdot H^{\infty}[1 \dd)$,
  we have $\lcp(\Pat, \T[j \dd n]) \geq \min(t, t')$.

  1. Let $Q = H'\cdot H^{\infty}[1 \dd)$, where $H'$ is a length-$s$
  suffix of $H$. In the proof of \cref{lm:lce},
  it is shown that $\type(j) = -1$
  implies $\T[j \dd n] \prec Q$, and $\type(j) = +1$ implies $Q \prec
  \T[j \dd n]$.  We now prove an analogous fact for $\Pat$. We first
  note that $\type(\Pat) = -1$ implies that either $\pend{\Pat} = m +
  1$, or $\pend{\Pat} \leq m$ and $\Pat[\pend{\Pat}] \prec
  \Pat[\pend{\Pat} - |H|]$. In the first case, $\Pat[1 \dd
  \pend{\Pat}) = \Pat$ is a proper prefix of $Q$ and hence $\Pat \prec
  Q$. In the second case, we have $\Pat[1 \dd t] = Q[1 \dd t]$ and
  $\Pat[1 + t] \prec \Pat[1 + t - |H|] = Q[1 + t - |H|] = Q[1 +
  t]$. Consequently, $\Pat \prec Q$. If $\type(\Pat) = +1$ holds, then
  $\pend{\Pat} \leq m$.  Thus, we have $Q[1 \dd t] = \Pat[1 \dd t]$
  and $Q[1 + t] = Q[1 + t - |H|] = \Pat[1 + t - |H|] \prec \Pat[1 +
  t]$. Hence, we obtain $Q \prec \Pat$. We are now ready to prove the
  claim. Assume first that $\type(\Pat) < \type(j)$. By the above we
  then have $\Pat \prec Q \prec \T[j \dd n]$. The opposite implication
  is proved by contraposition. Assume $\type(\Pat) > \type(j)$. By the
  above we then have $\T[j \dd n] \prec Q \prec \Pat$.

  2. Assume $t < t'$. If $\pend{\Pat} = m + 1$, then $\Pat[1 \dd t] =
  \Pat[1 \dd \pend{\Pat}) = \Pat$ is proper prefix of $\T[j \dd j +
  t') = \T[j \dd \rend{j})$, and hence $\Pat \prec \T[j \dd \rend{j})
  \preceq \T[j \dd n]$. If $\pend{\Pat} \leq m$, then we have $\Pat[1
  \dd t] = \T[j \dd j + t)$ and by $t < t'$, $\Pat[1 + t] \prec \Pat[1
  + t - |H|] = \T[j + t - |H|] = \T[j + t]$. Thus, we also obtain
  $\Pat \prec \T[j \dd n]$.  The opposite implication is proved by
  contraposition. Assume $t > t'$.  If $\rend{j} = n + 1$, then by $t
  > t'$, the string $\T[j \dd j + t') = \T[j \dd \rend{j}) = \T[j \dd
  n]$ is a proper prefix of $\Pat[1 \dd t] = \Pat[1 \dd \pend{\Pat})$,
  and hence $\T[j \dd n] \prec \Pat[1 \dd \pend{\Pat}) \preceq
  \Pat$. If $\rend{j} \leq n$, then we have $\T[j \dd j + t') = \Pat[1
  \dd t']$ and by $t > t'$, $\T[j + t'] \prec \T[j + t' - |H|] =
  \Pat[1 + t' - |H|] = \Pat[1 + t']$.  Consequently, we also obtain
  $\T[j \dd n] \prec \Pat$.

  3. Assume $t > t'$. By $\type(j) = +1$, we have $\rend{j} \leq
  n$. Thus, by $t > t'$ we have $\Pat[1 \dd t'] = \T[j \dd j + t')$
  and $\Pat[1 + t'] = \Pat[1 + t' - |H|] = \T[j + t' - |H|] \prec \T[j
  + t']$.  Consequently, $\Pat \prec \T[j \dd n]$. The opposite
  implication is proved by contraposition.  Assume $t < t'$.  By
  $\type(\Pat) = +1$, we have $\pend{\Pat} \leq m$. Thus, by $t < t'$
  we have $\T[j \dd j + t) = \Pat[1 \dd t]$ and $\T[j + t] = \T[j + t
  - |H|] = \Pat[1 + t - |H|] \prec \Pat[1 + t]$. Consequently, we
  obtain $\T[j \dd n] \prec \Pat$.

  4. By the earlier implication, $\lcp(\Pat, \T[j \dd n]) \geq \min(t,
  t')$.  Thus, it remains to show that $\Pat \neq \T[j \dd n]$ and
  $\lcp(\Pat, \T[j \dd n]) \leq \min(t, t')$.  First, let us assume
  $\type(\Pat) < \type(j)$ (i.e., $\type(\Pat) = -1$ and $\type(j) =
  +1$). Consider two cases:
  \begin{itemize}
  \item First, assume $t \leq t'$. Our goal is to prove $\Pat \neq
    \T[j \dd n]$ and $\lcp(\Pat, \T[j \dd n]) \leq t$.  First, recall
    from the proof of \cref{lm:lce}\eqref{lm:lce-it-1} that $\type(j)
    = +1$ implies $j + t' \leq n$, $Q[1 \dd t'] = \T[j \dd j + t')$,
    and $Q[1 + t'] \prec \T[j + t']$. Consider now two subcases.  If
    $\pend{\Pat} = m + 1$, then $t = m$, and hence $\lcp(\Pat, \T[j
    \dd n]) \leq m = t$. By $t' \leq n - j$ we then also have $|\Pat|
    = t < t' + 1 \leq n - j + 1 = |\T[j \dd n]|$.  Thus, $\Pat \neq
    \T[j \dd n]$. Let us thus assume $\pend{\Pat} \leq m$.  In the
    proof of \cref{lm:pm-lce-it-1} we showed that in this case
    $\type(\Pat) = -1$ implies $\Pat[1 + t] \prec Q[1 + t]$.  On the
    other hand, as noted above, $\type(j) = +1$ implies $Q[1 \dd t'] =
    \T[j \dd j + t')$, and $Q[1 + t'] \prec \T[j + t']$.  By $t \leq
    t'$ we thus have $Q[1 + t] \preceq \T[j + t]$. Consequently,
    $\Pat[1 + t] \neq \T[j + t]$. This immediately implies $\Pat \neq
    \T[j \dd n]$ and $\lcp(\Pat, \T[j \dd n]) \leq t$.
  \item Let us now assume $t > t'$. Our goal is to prove $\Pat \neq
    \T[j \dd n]$ and $\lcp(\Pat, \T[j \dd n]) \leq t'$.  In the proof
    of \cref{lm:pm-lce-it-1} we showed that $\type(\Pat) = -1$ implies
    $\Pat[1 \dd t] = Q[1 \dd t]$. Thus, by $t > t'$ we have $\Pat[1 +
    t'] = Q[1 + t']$.  On the other hand, in the proof of
    \cref{lm:lce}\eqref{lm:lce-it-1} we showed that $\type(j) = +1$
    implies $Q[1 + t'] \prec \T[j + t']$. Thus, we obtain $\Pat[1 +
    t'] \neq \T[j + t']$. This immediately implies $\Pat \neq \T[j \dd
    n]$ and $\lcp(\Pat, \T[j \dd n]) \leq t'$.
  \end{itemize}
  Assume now $\type(\Pat) > \type(j)$ (i.e., $\type(\Pat) = +1$ and
  $\type(j) = -1$). Consider two cases:
  \begin{itemize}
  \item First, assume $t < t'$. Our goal is to prove $\Pat \neq \T[j
    \dd n]$ and $\lcp(\Pat, \T[j \dd n]) \leq t$. In the proof of
    \cref{lm:lce}\eqref{lm:lce-it-1} we showed that $\type(j) = -1$
    implies $\T[j \dd j + t') = Q[1 \dd t']$. Thus, by $t < t'$ we
    have $\T[j + t] = Q[1 + t]$. On the other hand, in the proof of
    \cref{lm:pm-lce-it-1} we showed that $\type(\Pat) = +1$ implies
    $Q[1 + t] \prec \Pat[1 + t]$. Thus, we obtain $\T[j + t] \neq
    \Pat[1 + t]$.  This immediately implies $\Pat \neq \T[j \dd n]$
    and $\lcp(\Pat, \T[j \dd n]) \leq t$.
  \item Let us now assume $t \geq t'$. Our goal is to prove $\Pat \neq
    \T[j \dd n]$ and $\lcp(\Pat, \T[j \dd n]) \leq t'$. First, recall
    from the proof of \cref{lm:pm-lce-it-1} that $\type(\Pat) = +1$
    implies $t + 1 = \rend{\Pat} \leq m$, $Q[1 \dd t] = \Pat[1 \dd
    t]$, and $Q[1 + t] \prec \Pat[1 + t]$. Consider now two subcases.
    If $\rend{j} = n + 1$, then $j + t' = n + 1$ (or equivalently, $t'
    = n - j + 1$) and hence $\lcp(\Pat, \T[j \dd n]) \leq |\T[j \dd
    n]| = t'$. By $t + 1 \leq m$ we then also have $|\T[j \dd n]| = t'
    < t + 1 \leq m = |\Pat|$.  Thus, $\Pat \neq \T[j \dd n]$. Let
    us thus assume $\rend{j} \leq n$.  In the proof of
    \cref{lm:lce}\eqref{lm:lce-it-1} we showed that in this case
    $\type(j) = -1$ implies $\T[j + t'] \prec Q[1 + t']$. On the
    other hand, as noted above, $\type(\Pat) = +1$ implies $Q[1 \dd t]
    = \Pat[1 \dd t]$ and $Q[1 + t] \prec \Pat[1 + t]$.  By $t \geq t'$
    we thus have $Q[1 + t'] \preceq \Pat[1 + t']$.  Consequently,
    $\T[j + t'] \neq \Pat[1 + t']$. This immediately implies $\Pat
    \neq \T[j \dd n]$ and $\lcp(\Pat, \T[j \dd n]) \leq t'$.
  \end{itemize}
  This concludes the proof of the claim if $\type(\Pat) \neq
  \type(j)$. Let us now assume $\type(\Pat) = \type(j) = -1$ and
  $t \neq t'$. Consider two cases:
  \begin{itemize}
  \item First, assume $t < t'$. Our goal is to prove $\Pat \neq \T[j
    \dd n]$ and $\lcp(\Pat, \T[j \dd n]) \leq t$. In the proof of
    \cref{lm:pm-lce-it-2} we showed that either it holds $\Pat[1 \dd
    t] = \Pat$ (in which case $\lcp(\Pat, \T[j \dd n]) \leq |\Pat|
    = t$ and $|\Pat| = t < t' = \rend{j} - j \leq n + 1 - j = |\T[j
    \dd n]|$ which in turn implies $\Pat \neq \T[j \dd n]$), or
    $\Pat[1 + t] \prec \T[j + t]$ (which also implies $\Pat \neq \T[j
    \dd n]$ and $\lcp(\Pat, \T[j \dd n]) \leq t$).
  \item Let us now assume $t > t'$. Our goal is to prove $\Pat \neq
    \T[j \dd n]$ and $\lcp(\Pat, \T[j \dd n]) \leq t'$. In the proof
    of \cref{lm:pm-lce-it-2}, we showed that either it holds $\T[j \dd
    j + t') = \T[j \dd n]$ (in which case $\lcp(\Pat, \T[j \dd n])
    \leq n - j + 1 = t'$ and $|\T[j \dd n]| = n - j + 1 = t' < t =
    \pend{\Pat} - 1 \leq |\Pat|$ which in turn implies $\Pat \neq \T[j
    \dd n]$), or $\T[j + t'] \prec \Pat[1 + t']$ (which also implies
    $\Pat \neq \T[j \dd n]$ and $\lcp(\Pat, \T[j \dd n]) \leq t'$).
  \end{itemize}
  Let us now assume $\type(\Pat) = \type(j) = +1$ and $t \neq t'$.
  Consider two cases:
  \begin{itemize}
  \item First, assume $t < t'$. Our goal is to prove $\Pat \neq \T[j
    \dd n]$ and $\lcp(\Pat, \T[j \dd n]) \leq t$. In the proof of
    \cref{lm:pm-lce-it-3} we showed that $\T[j + t] \prec \Pat[1 +
    t]$. This immediately implies the claims.
  \item Let us now assume $t > t'$. Our goal is to prove $\Pat \neq
    \T[j \dd n]$ and $\lcp(\Pat, \T[j \dd n]) \leq t'$. In the proof
    of \cref{lm:pm-lce-it-3} we showed that $\Pat[1 + t'] \prec \T[j +
    t']$. This immediately implies the claims. \qedhere
  \end{itemize}
\end{proof}

\begin{lemma}\label{lm:pm-lce-2}
  Let $\Pat \in \Alphabet^{m}$ be a periodic pattern. For every $S \in
  \Alphabet^{+}$, $\lcp(P, S) \geq 3\tau - 1$ implies that $S$ is
  periodic, and that it holds $\Lroot(S) = \Lroot(\Pat)$ and
  $\Lhead(S) = \Lhead(\Pat)$.
\end{lemma}
\begin{proof}
  Denote $X = \Pat[1 \dd 3\tau - 1]$. Letting $p := \per(X)$ we then
  have $p \leq \tfrac{1}{3}\tau$.  By $\lcp(P, S) \geq 3\tau - 1$, $X$
  is thus a prefix of $S$ and hence $\per(S[1 \dd 3\tau - 1]) = p \leq
  \tfrac{1}{3}\tau$, i.e., $S$ is periodic.  Moreover, we then have
  $\Lroot(S) = \min\{S[1 + t \dd 1 + t + p) : t \in [0 \dd p)\} =
  \min\{X[1 + t \dd 1 + t + p) : t \in [0 \dd p)\} = \min\{\Pat[1 + t
  \dd 1 + t + p) : t \in [0 \dd p)\} = \Lroot(\Pat)$. To show the last
  claim, denote $s = \Lhead(\Pat)$ and $s' = \Lhead(S)$.  Then,
  letting $H = \Lroot(\Pat) = \Lroot(S)$, the string $H' H^2$
  (resp.\ $\widehat{H}' H^2$) is a prefix of $\Pat$ (resp.\ $S$),
  where $H'$ (resp.\ $\widehat{H}'$) is a length-$s$
  (resp.\ length-$s'$) suffix of $H$.  Note, however, that $s, s' <
  |H| = p \leq \tfrac{1}{3}\tau$ and $|X| \geq \tau \geq 3|H|$.  This
  implies that $H' H^2$ and $\widehat{H}' H^2$ are both prefixes of
  $X$. By the synchronization property of primitive
  strings~\cite[Lemma~1.11]{AlgorithmsOnStrings}, this implies $|H'| =
  |\widehat{H}'|$. Thus, we obtain $\Lhead(\Pat) = s = |H'| =
  |\widehat{H}'| = s' = \Lhead(S)$.
\end{proof}

\begin{lemma}\label{lm:pm-lce-3}
  Let $\Pat \in \Alphabet^{+}$ be a periodic pattern satisfying
  $\rend{\Pat} \leq |\Pat|$. Then:
  \begin{enumerate}
  \item\label{lm:pm-lce-3-it-1} For every $S \in \Alphabet^{+}$,
    $\lcp(\Pat, S) \geq \rend{\Pat}$ (in particular, $\Pat$ being
    a prefix of $S$) implies that $S$ is periodic and it holds:
    \begin{itemize}
    \item $\rend{S} = \rend{\Pat}$,
    \item $\Ltail(S) = \Ltail(\Pat)$,
    \item $\rendfull{S} = \rendfull{\Pat}$,
    \item $\Lexp(S) = \Lexp(\Pat)$,
    \item $\type(S) = \type(\Pat)$.
    \end{itemize}
  \item\label{lm:pm-lce-3-it-2} If $j \in \Occ(\Pat, \T)$, then
    $j \in \R$ and it holds:
    \begin{itemize}
    \item $\rend{j} - j = \rend{\Pat} - 1$,
    \item $\Ltail(j) = \Ltail(\Pat)$,
    \item $\rendfull{j} - j = \rendfull{\Pat} - 1$,
    \item $\Lexp(j) = \Lexp(\Pat)$,
    \item $\type(j) = \type(\Pat)$.
    \end{itemize}
  \end{enumerate}
\end{lemma}
\begin{proof}
  Denote $X = \Pat[1 \dd 3\tau - 1]$, $H = \Lroot(\Pat)$, $s =
  \Lhead(\Pat)$, and $p = \per(X) = |H| \leq \tfrac{1}{3}\tau$.

  1. First, observe that by definition, $\rend{\Pat} = 1 + p +
  \lcp(\Pat, \Pat[1 + p \dd |\Pat|]) > |X|$.  Thus, $\lcp(\Pat, S)
  \geq \rend{\Pat}$ implies that $X$ is a prefix of $S$, and hence $S$
  is periodic.  By \cref{lm:pm-lce-2}, we then also have $\Lroot(S) =
  H$ and $\Lhead(S) = s$.  To show $\rend{S} = \rend{\Pat}$, observe
  that by $\rend{\Pat} \leq |\Pat|$ and the definition of
  $\rend{\Pat}$, we have $\Pat[\rend{\Pat}] \neq \Pat[\rend{\Pat} -
  p]$. Consequently, $\lcp(\Pat, S) \geq \rend{\Pat}$ yields
  $\lcp(\Pat, \Pat[1 + p \dd |\Pat|]) = \lcp(S, S[1 + p \dd |S|])$.
  Combining this with $|\Lroot(S)| = p$ we thus obtain $\rend{S} = 1 +
  p + \lcp(S, S[1 + p \dd |S|]) = 1 + p + \lcp(\Pat, \Pat[1 + p \dd
  |\Pat|]) = \rend{\Pat}$.  By $\Lhead(S) = s$ we then also obtain
  $\Ltail(S) = (\rend{S} - 1 - \Lhead(S)) \bmod |\Lroot(S)| =
  (\rend{\Pat} - 1 - s) \bmod |H| = \Ltail(\Pat)$, and consequently
  $\rendfull{S} = \rend{S} - \Ltail(S) = \rend{\Pat} - \Ltail(\Pat) =
  \rendfull{\Pat}$. We then also have $\Lexp(S) = \lfloor
  \tfrac{\rend{S} - 1 - \Lhead(S)}{|\Lroot(S)|} \rfloor = \lfloor
  \tfrac{\rend{\Pat} - 1 - s}{|H|} \rfloor = \Lexp(\Pat)$.  Finally,
  by $\rend{\Pat} \leq |\Pat|$ and $\lcp(\Pat, S) \geq \rend{\Pat}$,
  we then also have $S[\rend{S}] = \Pat[\rend{\Pat}]$. Consequently,
  $S[\rend{S}] \prec S[\rend{S} - p]$ holds if and only of
  $\Pat[\rend{\Pat}] \prec \Pat[\rend{\Pat} - p]$. Therefore,
  $\type(S) = \type(\Pat)$.

  2. We start by noting that $j \in \Occ(\Pat, \T)$ implies $j \in
  \Occ(X, \T)$. Thus, $\per(\T[j \dd j + 3\tau - 1)) = \per(X) = p
  \leq \tfrac{1}{3}\tau$, and hence $j \in \R$.  By \cref{lm:pm-lce},
  we then also have $\Lroot(j) = H$ and $\Lhead(j) = s$. To show
  $\rend{j} - j = \rend{\Pat} - 1$, denote $S = \T[j \dd n]$. Since
  $\Pat$ is a prefix of $S$, by \cref{lm:pm-lce-3-it-1}, it follows
  that $\rend{S} = \rend{\Pat}$.  By $\Lroot(S) = \Lroot(\Pat)$
  (\cref{lm:pm-lce-2}) and the definition of $\rend{S}$, we thus have
  $1 + p + \lcp(S, S[1 + p \dd |S|]) = \rend{S} = \rend{\Pat}$, or
  equivalently, $\lcp(S, S[1 + p \dd |S|]) = \rend{\Pat} - p - 1$.
  Since $\lcp(S, S[1 + p \dd |S|]) = \LCE(j, j + p)$, we thus obtain
  $p + \LCE(j, j + p) = \rend{\Pat} - 1$. It remains to note that for
  $j \in \R$, by \cref{lm:run-end}\eqref{lm:run-end-it-2}, $\rend{j} -
  j = p + \LCE(j, j + p)$.  Therefore, we have $\rend{j} - j =
  \rend{\Pat} - 1$. Combining this with $\Lroot(j) = H$ and $\Lhead(j)
  = s$ yields $\Ltail(j) = (\rend{j} - j - \Lhead(j)) \bmod
  |\Lroot(j)| = (\rend{\Pat} - 1 - s) \bmod |H| = \Ltail(\Pat)$,
  $\rendfull{j} - j = \rend{j} - j - \Ltail(j) = \rend{\Pat} - 1 -
  \Ltail(\Pat) = \rendfull{\Pat} - 1$, and $\Lexp(j) = \lfloor
  \tfrac{\rend{j} - j - s}{|\Lroot(j)|} \rfloor = \lfloor
  \tfrac{\rend{\Pat} - 1 - s}{|H|} \rfloor = \Lexp(\Pat)$.  Finally,
  by $\rend{\Pat} \leq |\Pat|$ we have $\rend{j} \leq n$ and
  $\T[\rend{j}] = \Pat[\rend{\Pat}]$. Consequently, $\T[\rend{j}]
  \prec \T[\rend{j} - p]$ holds if and only if $\Pat[\rend{\Pat}]
  \prec \Pat[\rend{\Pat} - p]$. Therefore, $\type(j) = \type(\Pat)$.
\end{proof}

\subsubsection{The Data Structure}\label{sec:pm-periodic-ds}

\bfparagraph{Definitions}

Let $q = |\R'^{-}|$. Recall (\cref{sec:sa-periodic-ds}), that
$(\rlexm_i)_{i \in [1 \dd q]}$ denotes the sequence containing all
positions $j \in \R'^{-}$ sorted first by $\Lroot(j)$, and in case of
ties, by $\T[\rendfull{j} \dd n]$.  Recall also that $\Lroots =
\{\Lroot(j) : j \in \R\}$.  For any string $H\in \Lroots$, let
$\Pow(H) = H^\infty[1\dd |H|\lceil\tfrac{\tau}{|H|}\rceil]$.  This
function satisfies the following properties:
\begin{itemize}
\item The set $\{\Pow(H) : H \in \Lroots\}$ is prefix-free.
\item For any $X, Y \in \Lroots$, $X \prec Y$ implies $\Pow(X) \prec
  \Pow(Y)$.
\end{itemize}
For a proof, consider $X, Y \in \Lroots$ such that $X \prec
Y$. By~\cite[Fact~9.1.6]{phdtomek}, it holds $X \preceq \Pow(X) \prec
X^\infty[1 \dd) \prec Y \preceq \Pow(Y)$. Since $|Y| < \tau \leq
|\Pow(X)|$, the set $\{\Pow(X), \Pow(Y)\}$ is prefix-free.

We define $\Zset = \{\rendfull{j} - |\Pow(\Lroot(j))| : j \in
\R'^{-}\}$.  We also define an array $\ARRzlex[1 \dd q]$ so that, for
any $i \in [1 \dd q]$, $\ARRzlex[i] = \rendfull{j} - |\Pow(H_i)|$,
where $j = \rlexm_i$ and $H_i = \Lroot(\rlexm_i)$. Note that
$\{\ARRzlex[i] : i \in [1 \dd q]\} = \Zset$. Observe also that
$\T[\ARRzlex[i] \dd n] = \Pow(H_i)\cdot \T[\rendfull{j} \dd n]$.
Together with the properties of the $\Pow$ function and with the
definition of $(\rlexm_i)_{i \in [1 \dd q]}$, this implies that the
positions in $\ARRzlex$ are sorted according to the lexicographic
order of the corresponding suffixes of $\T$, i.e., $i < i'$ implies
$\T[\ARRzlex[i] \dd n] \prec \T[\ARRzlex[i'] \dd n]$.

\bfparagraph{Components}

The data structure to handle periodic patterns consists of two
parts. The first part (designed to handle periodic patterns $\Pat$
satisfying $\type(\Pat) = -1$) consists of three components:
\begin{enumerate}
\item The index core $\CountCore(\T)$ (\cref{sec:pm-core-ds}) using
  $\bigO(n / \log_{\sigma} n)$ space.
\item The first part of the structure from \cref{sec:sa-periodic-ds}
  using $\bigO(n / \log_{\sigma} n)$ space.
\item The data structure from \cref{pr:meta-trie} for the array
  $\ARRzlex[1 \dd q]$. By $q = \bigO(n / \log_{\sigma} n)$ and
  \cref{pr:meta-trie}, it needs $\bigO(n / \log_{\sigma} n)$ space.
\end{enumerate}

The second part of the structure (to handle $\Pat$ satisfying
$\type(\Pat) = +1$) consists of the symmetric counterparts of the
above components adapted according to \cref{lm:pm-lce}.

In total, the data structure takes $\bigO(n / \log_{\sigma} n)$ space.

\subsubsection{Navigation Primitives}\label{sec:pm-periodic-nav}

\begin{proposition}\label{pr:pm-root}
  Let $\Pat \in \Alphabet^m$ be a periodic pattern. Given the data
  structure from \cref{sec:pm-periodic-ds} and the packed
  representation of $\Pat$, we can in $\bigO(1 + m / \log_{\sigma} n)$
  time compute $\Lroot(\Pat)$, $\Lhead(\Pat)$, $\Lexp(\Pat)$,
  $\Ltail(\Pat)$, and $\type(\Pat)$.
\end{proposition}
\begin{proof}
  We first compute $x \in [0 \dd \sigma^{6\tau})$ such that $x =
  \Int(\Pat[1 \dd 3\tau {-} 1])$. Given the packed encoding of $\Pat$,
  such $x$ is obtained in $\bigO(1)$ time. We then look up $(s, p) =
  \LTroot[x]$, and in $\bigO(1)$ time obtain $\Lroot(\Pat) = \Pat[1
  {+} s \dd 1 {+} s {+} p)$ and $\Lhead(\Pat) = s$. Next, we compute
  $\Lexp(\Pat)$ and $\Ltail(\Pat)$. For this, we first determine the
  length $\ell$ of the longest common prefix of $\Pat$ and $\Pat(p \dd
  m]$. Using the packed representation of $\Pat$, we can do this in
  $\bigO(1 + m / \log_{\sigma} n)$ time (see,
  e.g.,~\cite[Proposition~2.3]{sss}). Consequently, we obtain
  $\pend{\Pat} = 1 + p + \ell$, $\Lexp(\Pat) = \lfloor
  \tfrac{\pend{\Pat} - 1 - s}{p} \rfloor$, and $\Ltail(\Pat) =
  (\pend{\Pat} - 1 - s) \bmod p$.  Finally, to test if $\type(\Pat) =
  +1$, we check whether $\pend{\Pat} \leq m$, and if so, whether
  $\Pat[\pend{\Pat}] \succ \Pat[\pend{\Pat} - p]$.
\end{proof}

\subsubsection{Implementation of Queries}\label{sec:pm-periodic-pm}

\bfparagraph{Overview}

The query algorithm is derived in two steps.  First, we establish how,
given the structure from \cref{sec:pm-periodic-ds} and a packed
representation of any periodic pattern $\Pat \in \Alphabet^{m}$ to
compute $|\Occ(\Pat, \T)|$ in $\bigO(m / \log_{\sigma} n + \log \log
n)$ time. This culminates in \cref{pr:pm-occ}.  We then show how to
extend this algorithm to instead return $(\LB(\Pat, \T), \UB(\Pat,
\T))$ in the same time complexity, culminating in
\cref{pr:pm-periodic-range}. The reason for this two-step approach is
explained in \cref{rm:pm-periodic}.

\bfparagraph{Computing $|\Occ(\Pat, \T)|$}

Let $\Pat \in \Alphabet^{m}$ be a periodic pattern.  Denote $s =
\Lhead(\Pat)$ and $H = \Lroot(\Pat)$.  We define $\Occa(\Pat, \T) =
\{j \in \R_{s,H} \cap \Occ(\Pat, \T) : \Lexp(j) > \Lexp(\Pat)\}$ and
$\Occs(\Pat, \T) = \{j \in \R_{s,H} \cap \Occ(\Pat, \T) : \Lexp(j) =
\Lexp(\Pat)\}$.

\begin{lemma}\label{lm:occ}
  For any periodic pattern $\Pat \in \Alphabet^{m}$, the set
  $\Occ(\Pat, \T)$ is a disjoint union of $\Occa(\Pat, \T)$ and
  $\Occs(\Pat, \T)$.
\end{lemma}
\begin{proof}
  By definition, $\Occa(\Pat, \T) \cap \Occs(\Pat, \T) = \emptyset$
  and $\Occa(\Pat, \T) \cup \Occs(\Pat, \T) \subseteq \Occ(\Pat, \T)$.
  Thus, it suffices to show $\Occ(\Pat, \T) \subseteq \Occa(\Pat, \T)
  \cup \Occs(\Pat, \T)$.  Assume $j \in \Occ(\Pat, \T)$. By $m \geq
  3\tau - 1$, this implies $\lcp(\T[j \dd n], \Pat) \geq 3\tau -
  1$. Thus, by \cref{lm:pm-lce}, it holds $j \in \R_{s,H}$, where $s =
  \Lhead(\Pat)$ and $H = \Lroot(\Pat)$. To obtain $j \in \Occa(\Pat,
  \T) \cup \Occs(\Pat, \T)$ it remains to show $\Lexp(j) \geq
  \Lexp(\Pat)$.  First, note that for any $t \in [1 \dd m)$, $j \in
  \Occ(\Pat, \T)$ implies $\LCE(j, j + t) \geq \lcp(\Pat[1 \dd m],
  \Pat[1 + t \dd m])$.  In particular, letting $p = |H|$, by
  definition of $\rend{\Pat}$ and
  \cref{lm:run-end}\eqref{lm:run-end-it-2}, we have $\rend{j} - j = p
  + \LCE(j, j + p) \geq p + \lcp(\Pat[1 \dd m], \Pat[1 + p \dd m]) =
  \rend{\Pat} - 1$.  Consequently, $\Lexp(j) = \lfloor \tfrac{\rend{j}
  - j - s}{p} \rfloor \geq \lfloor \tfrac{\rend{\Pat} - 1 - s}{p}
  \rfloor = \Lexp(\Pat)$.
\end{proof}

By the above lemma, if $\Pat \in \Alphabet^{m}$ is periodic,
then $\Occ(\Pat, \T) \sub \R$.  We focus on computing sizes of sets
$\Occam(\Pat, \T) := \Occa(\Pat, \T) \cap \R^{-}$ and $\Occsm(\Pat,
\T) := \Occs(\Pat, \T) \cap \R^{-}$. The sizes of the sets
$\Occap(\Pat, \T) := \Occa(\Pat, \T) \cap \R^{+}$ and $\Occsp(\Pat,
\T) := \Occs(\Pat, \T) \cap \R^{+}$ are computed analogously
(see \cref{pr:pm-occ}).

We now describe the algorithm to compute $|\Occam(\Pat, \T)|$
for any periodic pattern $\Pat \in \Alphabet^{m}$.

\begin{lemma}\label{lm:occ-a}
  Assume that $\Pat \in \Alphabet^{m}$ is periodic. If
  $\pend{\Pat} \leq m$, then it holds $\Occam(\Pat, \T) =
  \emptyset$. Otherwise, it holds $\Occam(\Pat, \T) = \{j \in
  \R^{-}_{s,H} : \Lexp(j) > \Lexp(\Pat)\}$, where $s = \Lhead(\Pat)$
  and $H = \Lroot(\Pat)$.
\end{lemma}
\begin{proof}

  Let $\pend{\Pat} \leq m$. Denote $k = \Lexp(\Pat)$.  Suppose
  $\Occam(\Pat, \T) \neq \emptyset$, and let $j \in \Occam(\Pat, \T)$.
  By definition, $s + k|H| \leq \pend{\Pat} - 1 < s + (k+1)|H|$ and
  $\Pat[\pend{\Pat}] \neq \Pat[\pend{\Pat} - |H|]$.  On the other
  hand, by $j \in \R_{s,H}$ and $\Lexp(j) > k$, the string $H'H^{k+1}$
  (where $H'$ is a length-$s$ suffix of $H$) is a prefix of $\T[j \dd
  n]$.  Thus, we have $\T[j + \pend{\Pat} - 1] = \T[j + \pend{\Pat} -
  1 - |H|] = \Pat[\pend{\Pat} - |H|] \neq \Pat[\pend{\Pat}]$.  This
  implies $j \not\in \Occ(\Pat, \T)$, contradicting $j \in
  \Occam(\Pat, \T)$. Thus, $\Occam(\Pat, \T) = \emptyset$.

  Let $\pend{\Pat} > m$. The inclusion $\Occam(\Pat, \T) \sub \{j \in
  \R^{-}_{s,H} : \Lexp(j) > \Lexp(\Pat)\}$ follows by definition. To
  show the opposite inclusion, let $j \in \R^{-}_{s,H}$ be such that
  $\Lexp(j) > \Lexp(\Pat)$. Denote $k = \Lexp(\Pat)$.  Then, $\Pat =
  H'H^kH''$, where $|H'| = s$, and $H'$ (resp.\ $H''$) is a suffix
  (resp.\ prefix) of $H$. Thus, $\Pat$ is a prefix of $H'H^{k+1}$. The
  latter string, on the other hand, is by $\Lexp(j) \geq k + 1$ and $j
  \in \R_{s,H}$, a prefix of $\T[j \dd n]$.  Thus, $j \in \Occ(\Pat,
  \T)$. By $j \in \R^{-}_{s,H}$ and $\Lexp(j) > \Lexp(\Pat)$, we
  therefore also have $j \in \Occam(\Pat, \T)$.
\end{proof}

\begin{proposition}\label{pr:pm-occ-a}
  Let $\Pat \in \Alphabet^m$ be a periodic pattern. Given the data
  structure from \cref{sec:pm-periodic-ds} and the packed
  representation of $\Pat$, we can compute $|\Occam(\Pat, \T)|$ in
  $\bigO(1 + m / \log_{\sigma} n)$ time.
\end{proposition}
\begin{proof}
  First, using \cref{pr:pm-root}, we compute $s = \Lhead(\Pat)$, $H =
  \Lroot(\Pat)$, $k = \Lexp(\Pat)$, and $t = \Ltail(\Pat)$ in $\bigO(1
  + m / \log_{\sigma} n)$ time. This lets us determine $\pend{\Pat} =
  1 {+} s {+} k|H| {+} t$. If $\pend{\Pat} \leq m$, then by
  \cref{lm:occ-a}, we return $|\Occam(\Pat, \T)| = 0$.  Otherwise,
  using the array $\LTrange$, we compute in $\bigO(1)$ time a pair of
  integers $b, e$ such that $\SA(b \dd e]$ contains the starting
  positions of all suffixes of $\T$ prefixed with $X = \Pat[1 \dd
  3\tau {-} 1]$. Equivalently, by \cref{lm:lce} (see also the
  implementation of queries in \cref{pr:isa-delta-a}), $\SA(b \dd e]$
  contains all positions from $\R_{s,H}$. If $b = e$, then it holds
  $\R_{s,H} = \emptyset$, and thus we return $|\Occam(\Pat, \T)| = 0$.
  Let us thus assume $b < e$. Our goal now is to determine the
  subrange of $\SA(b \dd e]$ containing all positions in $\{j \in
  \R^{-}_{s,H} : \Lexp(j) > \Lexp(\Pat)\}$ (these positions form a
  subrange by \cref{lm:lce}). For that, we first compute $d =
  \rank{\BVexp}{1}{e} - \rank{\BVexp}{1}{b}$ in $\bigO(1)$ time.  If
  $d = 0$, then $\R^{-}_{s,H} = \emptyset$, and hence we return
  $|\Occam(\Pat, \T)| = 0$. Otherwise, we retrieve $k_{\min} =
  \LTminexp[\Int(X)]$ in $\bigO(1)$ time. Then, letting $k_{\max} =
  k_{\min} + d - 1$, we have $k_{\min} \leq k_{\max}$ and $[k_{\min}
  \dd k_{\max}] = \{\Lexp(j) : j \in \R^{-}_{s,H}\}$ (see the proof
  of \cref{pr:isa-delta-a}). If $k \geq k_{\max}$, by \cref{lm:occ-a},
  we return $|\Occam(\Pat, \T)| = 0$. Otherwise, we have two cases.
  Let $p = \rank{\BVexp}{1}{b}$. If $k < k_{\min}$, then we return
  $|\Occam(\Pat, \T)| = |\R^{-}_{s,H}| = \select{\BVexp}{1}{p + d} -
  b$.  Otherwise (i.e., $k \geq k_{\min}$), we return $|\Occam(\Pat,
  \T)| = \select{\BVexp}{1}{p + d} - \select{\BVexp}{1}{p + k -
    k_{\min} + 1}$. In total, the query takes $\bigO(1 + m /
  \log_{\sigma} n)$ time.
\end{proof}

Next, we now describe the algorithm compute $|\Occsm(\Pat, \T)|$ for
any periodic pattern $\Pat \in \Alphabet^{m}$.

\begin{lemma}\label{lm:occ-s}
  Let $\Pat \in \Alphabet^m$ be a periodic pattern.  Denote $H =
  \Lroot(\Pat)$.  Assume $i \in \R^{-}_{H}$ and let $\ell = \rend{i} -
  i - 3\tau + 2$. Then, $|\Occsm(\Pat, \T) \cap [i \dd i + \ell)| \leq
  1$. Moreover, $|\Occsm(\Pat, \T) \cap [i \dd i + \ell)| = 1$ holds
  if and only if $\Pat[\pendfull{\Pat} \dd m]$ is a prefix of
  $\T[\rendfull{i} \dd n]$ and $\rendfull{i} - i \geq \pendfull{\Pat}
  - 1$.
\end{lemma}
\begin{proof}

  As observed in the proof of \cref{lm:isa-delta-s}, $[i \dd i + \ell)
  \sub \R^{-}_{H}$, and for any $\delta \in [0 \dd \ell)$, it holds
  $\rend{i + \delta} = \rend{i}$, $\Ltail(i + \delta) = \Ltail(i)$,
  and consequently, $\rendfull{i + \delta} = \rendfull{i}$ and
  $\rendfull{i + \delta} - (i + \delta) = \rendfull{i} - i - \delta$.
  Moreover, by definition of $\Occsm(\Pat, \T)$, letting $\Lhead(\Pat)
  = s$, for any $j \in \Occsm(\Pat, \T)$ it holds $\rendfull{j} - j =
  s + \Lexp(j) \cdot |H| = s + \Lexp(\Pat) \cdot |H| = \pendfull{\Pat}
  - 1$.  Thus, $i + \delta \in \Occsm(\Pat, \T)$ implies $\rendfull{i
  + \delta} - (i + \delta) = \rendfull{i} - (i + \delta) =
  \pendfull{\Pat} - 1$, or equivalently, $\delta = (\rendfull{i} - i)
  - (\pendfull{\Pat} - 1)$, and therefore, $|\Occsm(\Pat, \T) \cap [i
  \dd i + \ell)| \leq 1$.

  For the second part, assume first that $i + \delta \in \Occsm(\Pat,
  \T)$ holds for some $\delta \in [0 \dd \ell)$.  Then, as noted
  above, we have $\pendfull{\Pat} - 1 = \rendfull{i} - (i + \delta)
  \leq \rendfull{i} - i$. Moreover, letting $\Lhead(\Pat) = s$, by
  definition of $\Occsm(\Pat, \T)$, we have $i + \delta \in
  \R^{-}_{s,H}$, $\Lexp(\Pat) = \Lexp(i + \delta)$, and $\T[i {+}
  \delta \dd i {+} \delta + m) = \Pat$.  Therefore, we obtain that
  $\T[i {+} \delta \dd \rendfull{i + \delta}) = \T[i {+} \delta \dd
  \rendfull{i}) = \Pat[1 \dd \pendfull{\Pat}) = H'H^k$ (where $k =
  \Lexp(\Pat)$ and $H'$ is the length-$s$ suffix of $H$), and
  consequently, $\Pat[\pendfull{\Pat} \dd m]$ is a prefix of
  $\T[\rendfull{i} \dd n]$.  To show the converse implication, assume
  that $\Pat[\pendfull{\Pat} \dd m]$ is a prefix of $\T[\rendfull{i}
  \dd n]$ and $\rendfull{i} - i \geq \pendfull{\Pat} - 1$. Let $\delta
  = (\rendfull{i} - i) - (\pendfull{\Pat} - 1)$.  We will prove that
  $\delta \in [0 \dd \ell)$ and $i + \delta \in \Occsm(\Pat,
  \T)$. Clearly $\delta \geq 0$.  To show $\delta < \ell$, we first
  prove $\rend{i} - \rendfull{i} \geq \pend{\Pat} - \pendfull{\Pat}$.
  Suppose that $q = \rend{i} - \rendfull{i} < \pend{\Pat} -
  \pendfull{\Pat}$.  By $i \in \R^{-}_{H}$, we then either have
  $\rendfull{i} + q = n + 1$, or $\rendfull{i} + q \leq n$ and
  $\T[\rendfull{i} + q] \neq \T[\rendfull{i} + q - |H|] =
  \Pat[\pendfull{\Pat} + q - |H|] = \Pat[\pendfull{\Pat} + q]$, both
  of which contradict that $\Pat[\pendfull{\Pat} \dd m]$ is a prefix
  of $\T[\rendfull{i} \dd n]$. Thus, $\rend{i} - \rendfull{i} \geq
  \pend{\Pat} - \pendfull{\Pat}$.  This implies, $\rend{i} - (i +
  \delta) = (\rendfull{i} - (i + \delta)) + (\rend{i} - \rendfull{i})
  = (\pendfull{\Pat} - 1) + (\rend{i} - \rendfull{i}) \geq
  (\pendfull{\Pat} - 1) + (\pend{\Pat} - \pendfull{\Pat}) =
  \pend{\Pat} - 1 \geq 3\tau - 1$, or equivalently $\delta \leq
  \rend{i} - i - 3\tau + 1 < \ell$.  To show $i + \delta \in
  \Occsm(\Pat, \T)$, it remains to observe that $\rendfull{i + \delta}
  - (i + \delta) = \rendfull{i} - (i + \delta) = \pendfull{\Pat} - 1$
  and $\Lroot(i + \delta) = \Lroot(i) = H = \Lroot(\Pat)$ (following
  from \cref{lm:R-block}) imply $\T[i + \delta \dd \rendfull{i}) =
  \Pat[1 \dd \pendfull{\Pat})$.  This in particular gives, letting
  $\Lhead(\Pat) = s$, that $i + \delta \in \R_{s,H}$ and $\Lexp(i +
  \delta) = \Lexp(\Pat)$.  Moreover, combining it with
  $\Pat[\pendfull{\Pat} \dd m]$ being a prefix of $\T[\rendfull{i} \dd
  n]$ yields $\T[i + \delta \dd i + \delta + m) = \Pat$.  Finally, by
  \cref{lm:R-block}, $\type(i + \delta) = \type(i) = -1$. Therefore,
  $i + \delta \in \Occsm(\Pat, \T)$.
\end{proof}

\begin{proposition}\label{pr:pm-occ-s}
  Let $\Pat \in \Alphabet^m$ be a periodic pattern. Given the data
  structure from \cref{sec:pm-periodic-ds} and the packed
  representation of $\Pat$, we can compute $|\Occsm(\Pat, \T)|$ in
  $\bigO(m / \log_{\sigma} n + \log \log n)$ time.
\end{proposition}
\begin{proof}
  First, using \cref{pr:pm-root}, we compute $s = \Lhead(\Pat)$, $H =
  \Lroot(\Pat)$, and $k = \Lexp(\Pat)$ in $\bigO(1 + m / \log_{\sigma}
  n)$ time. This lets us determine $\pendfull{\Pat} = 1 {+} s {+}
  k|H|$ and $\Pat' := \Pat[\pendfull{\Pat} - |\Pow(H)| \dd m]$. Then,
  using \cref{pr:meta-trie}, we compute in $\bigO(m / \log_{\sigma} n
  + \log \log n)$ time a range $(\bpre\dd \epre] = \{i\in [1\dd q]:
  \Pat'\text{ is a prefix of } T[\ARRzlex[i]\dd n]\}$.  Observe that
  the set $\{\rlexm_i: i\in (\bpre\dd \epre]\}$ consists of all
  positions $j \in \R'^{-}_{H}$ such that $\Pat[\pendfull{\Pat} \dd
  m]$ is a prefix of $\T[\rendfull{j} \dd n]$.  Thus, by
  \cref{lm:occ-s}, we have $|\Occsm(\Pat, \T)| = |\{i \in (\bpre \dd
  \epre] : \rendfull{\rlexm_i} - \rlexm_i \geq \pendfull{\Pat} -
  1\}|$, which we compute in $\bigO(\log \log n)$ time using the range
  counting structure as $\rcount{\ARRnontail}{\pendfull{\Pat} -
  1}{\epre} - \rcount{\ARRnontail}{\pendfull{\Pat} - 1}{\bpre}$
  (recall, that $\ARRnontail[i] = \rendfull{\rlexm_i} - \rlexm_i$; see
  \cref{sec:sa-periodic-ds}).
\end{proof}

By combining all above results, we obtain the following algorithm to
compute $|\Occ(\Pat, \T)|$ for any periodic pattern $\Pat$.

\begin{proposition}\label{pr:pm-occ}
  Let $\Pat \in \Alphabet^m$ be a periodic pattern. Given the data
  structure from \cref{sec:pm-periodic-ds} and the packed
  representation of $\Pat$, we can compute $|\Occ(\Pat, \T)|$ in
  $\bigO(m / \log_{\sigma} n + \log \log n)$ time.
\end{proposition}
\begin{proof}
  Given a packed representation of a periodic pattern $\Pat$, we
  compute $|\Occ(\Pat, \T)| = |\Occam(\Pat, \T)| + |\Occsm(\Pat, \T)|
  + |\Occap(\Pat, \T)| + |\Occsp(\Pat, \T)|$ using
  \cref{pr:pm-occ-a,pr:pm-occ-s} and their symmetric counterparts
  (adapted according to \cref{lm:pm-lce}). The total time is $\bigO(m
  / \log_{\sigma} n + \log \log n)$.
\end{proof}

\bfparagraph{Generalizing the Query Algorithm}

We now show how to generalize the above algorithms to compute
$|\Occ(\Pat, \T)|$, to instead return $(\LB(\Pat, \T), \UB(\Pat,
\T))$.

For any periodic pattern $\Pat \in \Alphabet^m$ we define
\[
  \Pos(\Pat, \T) = \{j \in [1 \dd n] : \lcp(\T[j \dd n], \Pat) \geq
  3\tau - 1\text{ and }\T[j \dd n] \prec \Pat\},
\]
and denote $\delta(\Pat, \T) = |\Pos(\Pat, \T)|$.

\begin{lemma}\label{lm:pm-delta-1}
  Let $\Pat \in \Alphabet^m$ be a periodic pattern and let $X =
  \Pat[1 \dd 3\tau {-} 1]$. Then, it holds $\LB(\Pat, \T) = \LB(X, \T)
  + \delta(\Pat, \T)$.
\end{lemma}
\begin{proof}
  It suffices to observe that $j \in \Occ(X, \T)$ holds if and only if
  $\lcp(\T[j \dd n], \Pat) \geq 3\tau - 1$. Thus, it holds by
  definition of $\LB(\Pat, \T)$ that $\LB(\Pat, \T) = \LB(X, \T) +
  |\{j \in \Occ(X, \T) : \T[j \dd n] \prec \Pat\}| = \LB(X,\T) + |\{j
  \in [1 \dd n] : \lcp(\T[j \dd n], \Pat) \geq 3\tau - 1\text{ and
  }\T[j \dd n] \prec \Pat\}| = \LB(X, \T) + \delta(\Pat, \T)$.
\end{proof}

We focus on computing $\delta(\Pat, \T)$ for $\Pat$ satisfying
$\type(\Pat) = -1$ (the structure for $\Pat$ satisfying $\type(\Pat) =
+1$ is symmetric; see the proof of \cref{pr:pm-periodic-range}).  We
define $\Posa(\Pat, \T) = \{j \in \R^{-}_{s,H} : \Lexp(j) \leq
\Lexp(\Pat)\}$ and $\Poss(\Pat, \T) = \{j \in \R^{-}_{s,H} : \Lexp(j)
= \Lexp(\Pat)\text{ and }\T[j \dd n] \succeq \Pat\}$, where $s =
\Lhead(\Pat)$ and $H = \Lroot(\Pat)$. We denote $\deltaa(\Pat, \T) =
|\Posa(\Pat, \T)|$ and $\deltas(\Pat, \T) = |\Poss(\Pat, \T)|$.

\begin{lemma}\label{lm:pm-delta-2}
  For any periodic pattern $\Pat \in \Alphabet^m$ that satisfies
  $\type(\Pat) = -1$, it holds $\delta(\Pat, \T) = \deltaa(\Pat, \T) -
  \deltas(\Pat, \T)$.
\end{lemma}
\begin{proof}

  We will prove that $\Posa(\Pat, \T)$ is a disjoint union of
  $\Pos(\Pat, \T)$ and $\Poss(\Pat, \T)$. This implies $\delta(\Pat,
  \T) + \deltas(\Pat, \T) = \deltaa(\Pat, \T)$, and consequently, the
  equality in the claim.

  Denote $s = \Lhead(\Pat)$ and $H = \Lroot(\Pat)$.
  By \cref{lm:pm-lce}, letting $j \in \R^{-}_{s,H}$, we
  have $\Pos(\Pat, \T) = \{j \in \R^{-}_{s,H} : \T[j \dd n] \prec
  \Pat\}$, and moreover, if $j \in \Pos(\Pat, \T)$, then $\rend{j} - j
  \leq \pend{\Pat} - 1$. In particular, $\Lexp(j) = \lfloor
  \tfrac{\rend{j} - j - s}{|H|} \rfloor \leq \lfloor
  \tfrac{\pend{\Pat} - 1 - s}{|H|} \rfloor = \Lexp(\Pat)$.  Hence,
  $\Pos(\Pat, \T) \subseteq \Posa(\Pat, \T)$. On the other hand,
  clearly $\Poss(\Pat, \T) \subseteq \Posa(\Pat, \T)$ and $\Poss(\Pat,
  \T) \cap \Pos(\Pat, \T) = \emptyset$. Thus, to obtain the claim, it
  suffices to show that $\Posa(\Pat, \T) \setminus \Poss(\Pat, \T)
  \subseteq \Pos(\Pat, \T)$.

  Let $j \in \Posa(\Pat, \T) \setminus \Poss(\Pat, \T)$. Consider two
  cases.  If $\Lexp(j) = \Lexp(\Pat)$, then by definition of
  $\Poss(\Pat, \T)$, it must hold $\T[j \dd n] \prec \Pat$. Thus, we
  have $j \in \Pos(\Pat, \T)$. Let us therefore assume $\Lexp(j) <
  \Lexp(\Pat)$. Then, $\rend{j} - j = s + \Lexp(j) \cdot |H| +
  \Ltail(j) < s + \Lexp(j)\cdot |H| + |H| \le s + \Lexp(\Pat) \cdot
  |H| \le s + \Lexp(\Pat)\cdot |H| + \Ltail(\Pat) = \pend{\Pat} - 1$.
  By \cref{lm:pm-lce}\eqref{lm:pm-lce-it-2} and
  \cref{lm:pm-lce}\eqref{lm:pm-lce-it-4}, this implies $\T[j
  \dd n] \prec \Pat$, and consequently, $j \in \Pos(\Pat, \T)$.
\end{proof}

We now describe how, given any periodic pattern $\Pat \in \Alphabet^m$
that satisfies $\type(\Pat) = -1$, to compute $\deltaa(\Pat, \T)$.

\begin{proposition}\label{pr:pm-delta-a}
  Let $\Pat \in \Alphabet^m$ be a periodic pattern satisfying
  $\type(\Pat) = -1$. Given the data structure from
  \cref{sec:pm-periodic-ds} and the packed representation of $\Pat$,
  we can in $\bigO(1 + m / \log_{\sigma} n)$ time compute
  $\deltaa(\Pat, \T)$.
\end{proposition}
\begin{proof}
  First, using \cref{pr:pm-root}, we compute $H = \Lroot(\Pat)$ and $k
  = \Lexp(\Pat)$ in $\bigO(1 + m / \log_{\sigma} n)$ time.  Then,
  using $\LTrange$, we compute in $\bigO(1)$ time a pair of integers
  $b, e$ such that $\SA(b \dd e]$ contains the starting positions of
  all suffixes of $\T$ prefixed with $X = \Pat[1 \dd 3\tau {-}
  1]$. Equivalently, by \cref{lm:pm-lce}, $\SA(b \dd e]$ contains all
  positions from $\R_{s,H}$, where $s = \Lhead(\Pat)$. If $b = e$,
  then it holds $\R_{s,H} = \emptyset$, and thus we return
  $\deltaa(\Pat, \T) = 0$.  Let us thus assume $b < e$. Our goal now
  is to determine the subrange of $\SA(b \dd e]$ containing all
  positions in $\{j \in \R^{-}_{s,H} : \Lexp(j) \leq \Lexp(\Pat)\}$
  (these positions form a subrange by \cref{lm:lce}). For that, we
  first compute $d = \rank{\BVexp}{1}{e} - \rank{\BVexp}{1}{b}$ in
  $\bigO(1)$ time.  If $d = 0$, then $\R^{-}_{s,H} = \emptyset$, and
  hence we return $\deltaa(\Pat, \T) = 0$. Otherwise, we retrieve
  $k_{\min} = \LTminexp[\Int(X)]$ in $\bigO(1)$ time. Then, letting
  $k_{\max} = k_{\min} + d - 1$, we have $k_{\min} \leq k_{\max}$ and
  $[k_{\min} \dd k_{\max}] = \{\Lexp(j) : j \in \R^{-}_{s,H}\}$ (see
  the proof of \cref{pr:isa-delta-a}). If $k < k_{\min}$, we return
  $\deltaa(\Pat, \T) = 0$. Otherwise, we have two cases.  Let $p =
  \rank{\BVexp}{1}{b}$. If $k \geq k_{\max}$, then we return
  $\deltaa(\Pat, \T) = |\R^{-}_{s,H}| = \select{\BVexp}{1}{p + d} -
  b$.  Otherwise (i.e., $k < k_{\max}$), we return $\deltaa(\Pat, \T)
  = \select{\BVexp}{1}{p + k - k_{\min} + 1} - b$.  In total, the
  query takes $\bigO(1 + m / \log_{\sigma} n)$ time.
\end{proof}

We now describe how, given any periodic pattern $\Pat \in \Alphabet^m$
that satisfies $\type(\Pat) = -1$, to compute $\deltas(\Pat, \T)$.

\begin{lemma}\label{lm:pm-delta-s}
  Let $\Pat \in \Alphabet^m$ be a periodic pattern that satisfies
  $\type(\Pat) = -1$. Denote $H = \Lroot(\Pat)$.  Assume $i \in
  \R^{-}_{H}$ and let $\ell = \rend{i} - i - 3\tau + 2$. Then, we have
  $|\Poss(\Pat, \T) \cap [i \dd i + \ell)| \leq 1$. Moreover,
  $|\Poss(\Pat, \T) \cap [i \dd i + \ell)| = 1$ holds if and only if
  $\T[\rendfull{i} \dd n] \succeq \Pat[\pendfull{\Pat} \dd m]$ and
  $\rendfull{i} - i \geq \pendfull{\Pat} - 1$.
\end{lemma}
\begin{proof}

  In the proof of \cref{lm:occ-s}, it is shown that $[i \dd i + \ell)
  \subseteq \R^{-}_H$, and for any $\delta \in [0 \dd \ell)$, it holds
  $\rendfull{i+\delta}-(i + \delta) = \rendfull{i} - i - \delta$.  By
  definition of $\Poss(\Pat, \T)$, letting $s = \Lhead(\Pat)$, for any
  $j \in \Poss(\Pat, \T)$ it holds $\rendfull{j} - j = s + \Lexp(j)
  \cdot |H| = s + \Lexp(\Pat) \cdot |H| = \pendfull{\Pat} - 1$. Thus,
  $i + \delta \in \Poss(\Pat, \T)$ implies $\rendfull{i + \delta} - (i
  + \delta) = \rendfull{i} - (i + \delta) = \pendfull{\Pat} - 1$, or
  equivalently, $\delta = (\rendfull{i} - i) - (\pendfull{\Pat} - 1)$,
  and therefore, $|\Poss(\Pat, \T) \cap [i \dd i + \ell)| \leq 1$.

  For the second part, assume first that $i + \delta \in \Poss(\Pat,
  \T)$ holds for some $\delta \in [0 \dd \ell)$.  Then, as noted
  above, we have $\pendfull{\Pat} - 1 = \rendfull{i} - (i + \delta)
  \leq \rendfull{i} - i$. Moreover, letting $\Lhead(\Pat) = s$, by
  definition of $\Poss(\Pat, \T)$, we have $i + \delta \in
  \R^{-}_{s,H}$, $\Lexp(\Pat) = \Lexp(i + \delta)$, and $\T[i {+}
  \delta \dd n] \succeq \Pat$.  Therefore, we obtain that $\T[i {+}
  \delta \dd \rendfull{i + \delta}) = \T[i {+} \delta \dd
  \rendfull{i}) = \Pat[1 \dd \pendfull{\Pat}) = H'H^k$ (where $k =
  \Lexp(\Pat)$ and $H'$ is the length-$s$ suffix of $H$), and
  consequently, $\T[\rendfull{i} \dd n] \succeq \Pat[\pendfull{\Pat}
  \dd m]$.  To show the converse implication, assume that
  $\T[\rendfull{i} \dd n] \succeq \Pat[\pendfull{\Pat} \dd m]$ and
  $\rendfull{i} - i \geq \pendfull{\Pat} - 1$. Let $\delta =
  (\rendfull{i} - i) - (\pendfull{\Pat} - 1)$.  We will prove that
  $\delta \in [0 \dd \ell)$ and $i + \delta \in \Poss(\Pat,
  \T)$. Clearly $\delta \geq 0$.  To show $\delta < \ell$, we first
  prove $\rend{i} - \rendfull{i} \geq \pend{\Pat} - \pendfull{\Pat}$.
  Suppose that $q = \rend{i} - \rendfull{i} < \pend{\Pat} -
  \pendfull{\Pat}$.  By $i \in \R^{-}_{H}$, we then either have
  $\rendfull{i} + q = n + 1$, or $\rendfull{i} + q \leq n$ and
  $\T[\rendfull{i} + q] \prec \T[\rendfull{i} + q - |H|] =
  \Pat[\pendfull{\Pat} + q - |H|] = \Pat[\pendfull{\Pat} + q]$, both
  of which contradict $\T[\rendfull{i} \dd n] \succeq
  \Pat[\pendfull{\Pat} \dd m]$. Thus, $\rend{i} - \rendfull{i} \geq
  \pend{\Pat} - \pendfull{\Pat}$.  This implies, $\rend{i} - (i +
  \delta) = (\rendfull{i} - (i + \delta)) + (\rend{i} - \rendfull{i})
  = (\pendfull{\Pat} - 1) + (\rend{i} - \rendfull{i}) \geq
  (\pendfull{\Pat} - 1) + (\pend{\Pat} - \pendfull{\Pat}) =
  \pend{\Pat} - 1 \geq 3\tau - 1$, or equivalently $\delta \leq
  \rend{i} - i - 3\tau + 1 < \ell$.  To show $i + \delta \in
  \Poss(\Pat, \T)$, it remains to observe that $\rendfull{i + \delta}
  - (i + \delta) = \rendfull{i} - (i + \delta) = \pendfull{\Pat} - 1$
  and $\Lroot(i + \delta) = \Lroot(i) = H = \Lroot(\Pat)$ (following
  from \cref{lm:R-block}) imply $\T[i + \delta \dd \rendfull{i}) =
  \Pat[1 \dd \pendfull{\Pat})$.  This in particular gives, letting
  $\Lhead(\Pat) = s$, that $i + \delta \in \R_{s,H}$ and $\Lexp(i +
  \delta) = \Lexp(\Pat)$.  Moreover, combining it with
  $\T[\rendfull{i} \dd n] \succeq \Pat[\pendfull{\Pat} \dd m]$ yields
  $\T[i + \delta \dd n] \succeq \Pat$.  Finally, by \cref{lm:R-block},
  $\type(i + \delta) = \type(i) = -1$. Therefore, $i + \delta \in
  \Poss(\Pat, \T)$.
\end{proof}

\begin{proposition}\label{pr:pm-delta-s}
  Let $\Pat \in \Alphabet^m$ be a periodic pattern satisfying
  $\type(\Pat) = -1$. Given the data structure from
  \cref{sec:pm-periodic-ds} and the packed representation of $\Pat$,
  we can in $\bigO(m / \log_{\sigma} n + \log \log n)$ time compute
  $\deltas(\Pat, \T)$.
\end{proposition}
\begin{proof}
  First, using \cref{pr:pm-root}, we compute $s = \Lhead(\Pat)$, $H =
  \Lroot(\Pat)$, and $k = \Lexp(\Pat)$ in $\bigO(1 + m / \log_{\sigma}
  n)$ time. This lets us determine $\pendfull{\Pat} = 1 {+} s {+}
  k|H|$ and $\Pat' := \Pat[\pendfull{\Pat} - |\Pow(H)| \dd m]$. Then,
  using \cref{pr:meta-trie}, we compute in $\bigO(m / \log_{\sigma} n
  + \log \log n)$ time a value $x = |\{i \in [1 \dd q] :
  \T[\ARRzlex[i] \dd n] \prec \Pat'\}|$. Then, letting $x' = \sum_{H'
  \preceq H} |\R'^{-}_{H'}|$ (obtained from $\LTruns$ in $\bigO(1)$
  time as explained in the proof of \cref{pr:isa-delta-s}), by
  definition of $\ARRzlex$ and properties of function $\Pow$ (see the
  proof of \cref{pr:pm-occ-s}), the set $\{\rlexm_i : i \in (x \dd
  x']\}$ (where $\rlexm_i$ is defined as in the proof of
  \cref{pr:pm-occ-s}) consists of all positions $j \in \R'^{-}_{H}$
  satisfying $\T[\rendfull{j} \dd n] \succeq \Pat[\pendfull{\Pat} \dd
  m]$.  Thus, by \cref{lm:pm-delta-s}, it holds $\deltas(\Pat, \T) =
  |\Poss(\Pat, \T)| = |\{i \in (x \dd x'] : \ell_{i} \geq
  \pendfull{\Pat} - 1\}|$ (where $\ell_i$ is defined as in
  \cref{pr:pm-occ-s}), which we compute in $\bigO(\log \log n)$ time
  using the range counting structure as
  $\rcount{\ARRnontail}{\pendfull{\Pat} - 1}{x'} -
  \rcount{\ARRnontail}{\pendfull{\Pat} - 1}{x}$.
\end{proof}

By combining the above results, we obtain the algorithm to efficiently
compute the pair $(\LB(\Pat, \T), \UB(\Pat, \T))$ for periodic patterns.

\begin{proposition}\label{pr:pm-periodic-range}
  Let $\Pat \in \Alphabet^m$ be a periodic pattern. Given the data
  structure from \cref{sec:pm-periodic-ds} and the packed
  representation of $\Pat$, we can in $\bigO(m / \log_{\sigma} n +
  \log \log n)$ time compute the pair $(\LB(\Pat, \T),\allowbreak
  \UB(\Pat, \T))$.
\end{proposition}
\begin{proof}
  First, using \cref{pr:pm-occ} in $\bigO(m / \log_{\sigma} n + \log
  \log n)$ time we compute $|\Occ(\Pat, \T)|$. Next, using the lookup
  table $\LTrange$, in $\bigO(1)$ time we compute $(b_X, e_X) =
  (\LB(X, \T), \allowbreak \UB(X, \T))$, where $X = \Pat[1 \dd 3\tau
  {-} 1]$. Then, in $\bigO(1 + m / \log_{\sigma} n)$ time using
  \cref{pr:pm-root} we determine $\type(\Pat)$. Depending on whether
  $\type(\Pat) = -1$ or $\type(\Pat) = +1$, we use either a
  combination of \cref{pr:pm-delta-a,pr:pm-delta-s}, or their
  symmetric counterparts (more precisely, if $\type(\Pat) = +1$, we
  have $\deltaa(\Pat, \T) = |\Posa(\Pat, \T)|$ and $\deltas(\Pat, \T)
  = |\Poss(\Pat, \T)|$, where $\Posa(\Pat, \T) = \{j \in \R^{+}_{s,H}:
  \Lexp(j) \leq \Lexp(\Pat)\}$ and $\Poss(\Pat, \T) = \{j \in
  \R^{+}_{s,H} : \Lexp(j) = \Lexp(j)\text{ and }\T[j \dd n] \prec
  \Pat\}$), to compute $\deltaa(\Pat, \T)$ and $\deltas(\Pat, \T)$ in
  $\bigO(1 + m / \log_{\sigma} n)$ and $\bigO(m / \log_{\sigma} n +
  \log \log n)$ time, respectively. If $\type(\Pat) = -1$, then by
  \cref{lm:pm-delta-2} we have $\delta(\Pat, \T) = \deltaa(\Pat, \T) -
  \deltas(\Pat, \T)$.  Otherwise, by the counterpart of
  \cref{lm:pm-delta-2}, $\delta(\Pat, \T) = (e_X - b_X) -
  (\deltaa(\Pat, \T) - \deltas(\Pat, \T))$.  Finally, we return
  $(\LB(\Pat, \T), \UB(\Pat, \T)) = (b_X + \delta(\Pat, \T), b_X +
  \delta(\Pat, \T) + |\Occ(\Pat, \T)|)$ (see \cref{lm:pm-delta-1}) as
  the answer. In total, the query takes $\bigO(m / \log_{\sigma} n +
  \log \log n)$ time.
\end{proof}

\begin{remark}\label{rm:pm-periodic}
  Note the subtle difference in the type of symmetry used during the
  computation of $|\Pos(\Pat, \T)|$ and $|\Occ(\Pat, \T)|$.  When
  computing $\delta(\Pat, \T) = |\Pos(\Pat, \T)|$, by \cref{lm:pm-lce}
  we have $\Pos(\Pat, \T) \sub \R^{-}$ for any $\Pat$ satisfying
  $\type(\Pat) = -1$ (and $\Pos(\Pat, \T) \sub \R^{+}$ for $\Pat$
  satisfying $\type(\Pat) = +1$).  However, when computing
  $|\Occ(\Pat, \T)|$ for $\Pat$ satisfying $\type(\Pat) = -1$, it is
  possible that $\Occ(\Pat, \T) \cap \R^{-} \neq \emptyset$ and
  $\Occ(\Pat, \T) \cap \R^{+} \neq \emptyset$. Consequently, during
  the computation of $|\Occ(\Pat, \T)|$, we partition the output set
  $\Occa(\Pat, \T)$ (resp.\ $\Occs(\Pat, \T)$) into two subsets
  $\Occam(\Pat, \T)$ and $\Occap(\Pat, \T)$ (resp.\ $\Occsm(\Pat, \T)$
  and $\Occsp(\Pat, \T)$), but the computation is always performed
  regardless of $\type(\Pat)$, leading to two queries for each
  periodic pattern $\Pat$. During the computation of $\delta(\Pat,
  \T)$, on the other hand, the computation is performed separately for
  $\Pat$ satisfying $\type(\Pat) = -1$ and $\Pat$ satisfying
  $\type(\Pat) = +1$, without the need to partition $\Pos(\Pat, \T)$
  within each case, leading to a single query but only on the
  appropriate structure depending on $\type(\Pat)$. This is the reason
  for why the seemingly related computation of $|\Pos(\Pat, \T)|$ and
  $|\Occ(\Pat, \T)|$ is (unlike for nonperiodic patterns; see
  \cref{sec:pm-nonperiodic-pm}) described separately.
\end{remark}

\subsubsection{Construction Algorithm}\label{sec:pm-periodic-construction}

\begin{proposition}\label{pr:pm-periodic-construction}
  Given $\CountCore(\T)$, we can in $\bigO(n / \log_{\sigma} n)$ time
  augment it into a data structure from \cref{sec:pm-periodic-ds}.
\end{proposition}
\begin{proof}
  First, we combine
  \cref{pr:sa-core-construction,pr:sa-periodic-construction} (recall
  that the packed representation of $\T$ is a component of
  $\CountCore(\T)$) to construct the data structure from
  \cref{sec:sa-periodic-ds} in $\bigO(n / \log_{\sigma} n)$ time.  In
  particular, this constructs $(\rlexm_i)_{i \in [1 \dd q]}$.  Using
  \cref{pr:isa-root}, we can now compute $\ARRzlex[i]$ for any $i \in
  [1 \dd q]$ in $\bigO(1)$ time.  Then, in $\bigO(n / \log_{\sigma}
  n)$ time, we construct the data structure from \cref{pr:meta-trie}.

  After the above components are constructed, we then analogously
  construct their symmetric counterparts (adapted according to
  \cref{lm:pm-lce}).
\end{proof}

\subsection{The Final Data Structure}\label{sec:pm-final}

In this section, we put together
\cref{sec:pm-core,sec:pm-nonperiodic,sec:pm-periodic} to obtain a data
structure that, given a packed representation of any pattern $\Pat \in
\Alphabet^{m}$, computes $(\LB(\Pat, \T), \UB(\Pat, \T))$ in $\bigO(m
/ \log_{\sigma} n + \log^{\epsilon} n)$ time.

The section is organized as follows. First, we introduce the
components of the data structure (\cref{sec:pm-ds}). Next, we describe
the query algorithms (\cref{sec:pm-pm}). Finally, we show the
construction algorithm (\cref{sec:pm-construction}).

\subsubsection{The Data Structure}\label{sec:pm-ds}

The data structure consists of two components:

\begin{enumerate}
\item The structure from \cref{sec:pm-nonperiodic-ds} (used to handle
  nonperiodic patterns).
\item The structure from \cref{sec:pm-periodic-ds} (used to handle
  periodic patterns).
\end{enumerate}

In total, the data structure takes $\bigO(n / \log_{\sigma} n)$ space.

\subsubsection{Implementation of Queries}\label{sec:pm-pm}

\begin{proposition}\label{pr:pm-pm}
  Given the data structure from \cref{sec:pm-ds} and the packed
  representation of any $\Pat \in \Alphabet^m$, we can in $\bigO(m /
  \log_{\sigma} n + \log^{\epsilon} n)$ time compute the pair
  $(\LB(\Pat, \T),\allowbreak \UB(\Pat, \T))$.
\end{proposition}
\begin{proof}
  First, using \cref{pr:pm-core-periodicity}, in $\bigO(1)$ time we
  check if $\Pat$ is periodic. If so, we obtain $(\LB(\Pat,
  \T),\allowbreak \UB(\Pat, \T))$ in $\bigO(m / \log_{\sigma} n + \log
  \log n)$ time using \cref{pr:pm-periodic-range}. Otherwise (i.e., if
  $\Pat$ is not periodic), we consider two cases, depending on whether
  it holds $m < 3\tau - 1$.  If so, then we obtain $(\LB(\Pat,
  \T),\allowbreak \UB(\Pat, \T))$ in $\bigO(1)$ time using
  \cref{pr:pm-core-pm}. Otherwise, we obtain $(\LB(\Pat,
  \T),\allowbreak \UB(\Pat, \T))$ in $\bigO(m / \log_{\sigma} n +
  \log^{\epsilon} n)$ time using \cref{pr:pm-nonperiodic-range}.
\end{proof}

\subsubsection{Construction Algorithm}\label{sec:pm-construction}

\begin{proposition}\label{pr:pm-construction}
  Given the packed representation of $\T \in \Alphabet^n$, we can
  construct the data structure from \cref{sec:pm-ds} in $\bigO(n
  \min(1, \log \sigma / \sqrt{\log n}))$ time and
  $\bigO(n / \log_{\sigma} n)$ working space.
\end{proposition}
\begin{proof}
  First, from a packed representation of $\T$, we construct
  $\CountCore(\T)$ in $\bigO(n / \log_{\sigma} n)$ time using
  \cref{pr:pm-core-construction}. Then, using
  \cref{pr:pm-nonperiodic-construction,pr:pm-periodic-construction},
  we augment $\CountCore(\T)$ into the two components of the structure
  from \cref{sec:pm-ds} in $\bigO(n \min(1, \log \sigma / \sqrt{\log
  n}))$ and $\bigO(n / \log_{\sigma} n)$ time (respectively) and using
  $\bigO(n / \log_{\sigma} n)$ working space.  The overall runtime is
  thus $\bigO(n \min(1, \log \sigma / \sqrt{\log n}))$.
\end{proof}

\subsection{Summary}\label{sec:pm-summary}

By combining \cref{pr:pm-pm} and \cref{pr:pm-construction}
we obtain the following final result of this section.

\begin{theorem}\label{th:pm}
  Given any constant $\epsilon \in (0,1)$ and the packed
  representation of a text $\T \in \Alphabet^n$ with $2 \leq \sigma <
  n^{1/7}$, we can in $\bigO(n \min(1, \log \sigma / \sqrt{\log n}))$
  time and $\bigO(n / \log_{\sigma} n)$ working space construct a data
  structure of size $\bigO(n/\log_{\sigma} n)$ that, given the packed
  representation of any $\Pat \in \Alphabet^m$, returns the pair
  $(\LB(\Pat, \T), \UB(\Pat, \T))$ (and hence, in particular, the
  value $|\Occ(\Pat, \T)|$) in $\bigO(m / \log_{\sigma} n +
  \log^\epsilon n)$ time.
\end{theorem}

By combining the above result with \cref{th:sa}, we moreover obtain
the following result.

\begin{theorem}\label{th:pm-reporting}
  Given any constant $\epsilon \in (0,1)$ and the packed
  representation of a text $\T \in \Alphabet^n$ with $2 \leq \sigma <
  n^{1/7}$, we can in $\bigO(n \min(1, \log \sigma / \sqrt{\log n}))$
  time and $\bigO(n / \log_{\sigma} n)$ working space construct a data
  structure of size $\bigO(n/\log_{\sigma} n)$ that, given the packed
  representation of any $\Pat \in \Alphabet^m$, returns the set
  $\Occ(\Pat, \T)$ in $\bigO(m / \log_{\sigma} n + (|\Occ(\Pat, \T)| +
  1)\log^\epsilon n)$ time.
\end{theorem}

By observing that the dominating operations in the above index are
prefix rank and selection queries, we obtain the following more
general result.

\begin{theorem}\label{th:pm-general}
  Consider a data structure answering prefix rank and selection
  queries that, for any string of length $m$ over alphabet
  $\Alphabet^\ell$, achieves the following complexities:
  \begin{enumerate}
  \item Space usage $S(m, \ell, \sigma)$,
  \item Preprocessing time $P_t(m, \ell, \sigma)$,
  \item Preprocessing space $P_s(m, \ell, \sigma)$,
  \item Query time $Q(m, \ell, \sigma)$.
  \end{enumerate}
  For every $\T \in \Alphabet^n$ with $2 \leq \sigma < n^{1/7}$, there
  exist $m = \bigO(n/\log_{\sigma} n)$ and $\ell = \bigO(\log_{\sigma}
  n)$ such that, given the packed representation of $\T$, we can in
  $\bigO(n / \log_{\sigma} n + P_t(m, \ell, \sigma))$ time and
  $\bigO(n / \log_{\sigma} n + P_s(m, \ell,\sigma))$ working space
  build a structure of size $\bigO(n/\log_{\sigma} n + S(m, \ell,
  \sigma))$~that, given the packed representation of any $\Pat \in
  \Alphabet^p$, performs the following queries:
  \begin{itemize}
  \item Return $(\LB(\Pat, \T), \UB(\Pat, \T))$ in $\bigO(p /
    \log_{\sigma} n + \log \log n + Q(m, \ell, \sigma))$ time,
  \item Return $\Occ(\Pat, \T)$ in $\bigO(p / \log_{\sigma} n +
    (|\Occ(\Pat, \T)| + 1)(\log \log n + Q(m, \ell, \sigma)))$ time.
  \end{itemize}
\end{theorem}

\section{Suffix Tree Queries}\label{sec:st}

\begin{table}[t!]
  \centering
  \setlength{\tabcolsep}{4pt}
  \begin{tabular}{ll}
    \toprule Operation & Description\\
    \midrule

    \hyperref[sec:st-isleaf]{$\isleaf(v)$} &
      Return true if and only if $v$ is a leaf
      \\

    \hyperref[sec:st-index]{$\ind(v)$} &
      Any position $j \in \Occ(\str(v), \T)$
      \\

    \hyperref[sec:st-findleaf]{$\findleaf(j)$} &
      The leaf $v$ satisfying $\str(v) = \T[j \dd n]$
      \\

    \hyperref[sec:st-count]{$\cnt(v)$} &
      The number of leaves in the subtree rooted in $v$
      \\

    \hyperref[sec:st-sdepth]{$\sdepth(v)$} &
      The string-depth of node $v$, i.e., $|\str(v)|$
      \\

    \hyperref[sec:st-parent]{$\parent(v)$} &
      The parent of $v \neq \Root(\ST)$
      \\

    \hyperref[sec:st-firstchild]{$\firstchild(v)$} &
      The leftmost child of $v$, or $\nil$ if $v$ is a leaf
      \\

    \hyperref[sec:st-lastchild]{$\lastchild(v)$} &
      The rightmost child of $v$, or $\nil$ if $v$ is a leaf
      \\

    \hyperref[sec:st-rightsibling]{$\rightsibling(v)$} &
      The right sibling of $v$, or $\nil$ if there is no such node
      \\

    \hyperref[sec:st-leftsibling]{$\leftsibling(v)$} &
      The left sibling of $v$, or $\nil$ if there is no such node
      \\

    \hyperref[sec:st-slink]{$\slink(v)$} &
      A node $v'$ satisfying $\str(v') = \str(v)[2 \dd |\str(v)|]$
      \\

    \hyperref[sec:st-slink-iter]{$\slink(v, i)$} &
      A node $v'$ satisfying $\str(v') = \str(v)[i{+}1 \dd |\str(v)|]$,
      i.e., iterated $\slink$
      \\

    \hyperref[sec:st-wlink]{$\wlink(v, c)$} &
      A node $v'$ satisfying $\str(v') = c\cdot \str(v)$, or $\nil$ if
      there is no such node~\tablefootnote{Our data structure supports
      also a slightly stronger operation $\wlinkprim(v, c)$ (see
      \cref{pr:st-wlinkprim}), that returns a node $v'$ satisfying
      $\repr(v') = (\LB(c \cdot \str(v), \T), \UB(c \cdot \str(v),
      \T))$, if such node exists. This generalizes $\wlink(v, c)$,
      since $\wlink(v, c) \neq \nil$ holds if and only if
      $\wlinkprim(v, c) \neq \nil$ and $\sdepth(\wlinkprim(v, c)) =
      \sdepth(v) + 1$. Therefore, we can use $\wlinkprim(v, c)$ to
      compute $\wlink(v, c)$. Note, however, that it is possible that
      $\wlink(v, c) = \nil$ and yet $\wlinkprim(v, c) \neq \nil$. In
      that case, there exists a node corresponding to $\wlink(v, c)$
      in the suffix \emph{trie} of $\T$, but in $\ST$ this node is not
      explicit.}
      \\

    \hyperref[sec:st-child]{$\child(v, c)$} &
      A child $v'$ of $v$ satisfying $\str(v')[|\str(v)| {+} 1] \,{=}\,
      c$, or $\nil$ if there is no such node
      \\

    \hyperref[sec:st-pred]{$\pred(v, c)$} &
      A node $\child(v, c')$, where $c' = \max \{c'' \in [0 \dd c) :
      \child(v, c'') \neq \nil\}$ (or $\nil$)
      \\

    \hyperref[sec:st-letter]{$\letter(v, i)$} &
      The $i$th leftmost character of $\str(v)$
      \\

    \hyperref[sec:st-wa]{$\WA(v, d)$} &
      The most shallow ancestor of $v$ satisfying $\sdepth(v) \geq
      d$
      \\

    \hyperref[sec:st-lca]{$\LCA(u, v)$} &
      The lowest common ancestor of nodes $u$ and $v$
      \\

    \hyperref[sec:st-isancestor]{$\isancestor(u, v)$} &
      Return true if and only if $u$ is an ancestor of $v$
      \\

    \bottomrule
  \end{tabular}
  \caption{Operations on suffix tree $\ST$ supported by our data
    structure.}\label{tab:st-operations}
\end{table}

Let $\epsilon \in (0, 1)$ be any fixed constant and let $\T \in
\Alphabet^n$, where $2 \leq \sigma < n^{1/7}$. Let $\ST$ denote the
suffix tree of $\T$, i.e., a compact trie of the set $\{\T[1 \dd n],
\T[2 \dd n], \dots, \T[n]\}$.  In this section, we show how given the
packed representation of $\T$, to construct in $\bigO(n \min(1, \log
\sigma / \sqrt{\log n}))$ time and $\bigO(n / \log_{\sigma} n)$
working space a representation of $\ST$ occupying $\bigO(n /
\log_{\sigma} n)$ space, and supporting each of the operations listed
in \cref{tab:st-operations} in $\bigO(\log^{\epsilon} n)$
time.~\footnote{Similarly as in prior CST
implementations~\cite{cst,FischerMN09,RussoNO11,Gagie2020,%
BoucherCGHMNR21,CaceresN22}, the time complexity of some
operations is actually $\bigO(1)$. We also note that some
prior CST implementations
(e.g.,~\cite{cst,RussoNO11,Gagie2020,CaceresN22}) support two
additional operations called \emph{tree depth} and \emph{tree level
ancestor} which are analogous to $\sdepth(v)$ and $\WA(v, d)$ but
with distance to the root defined by the number of ancestor nodes
rather than the total length of edge labels.} We also derive a general
reduction depending on prefix rank and selection queries.

As in \cref{sec:sa,sec:pm}, we let $\tau = \lfloor \mu \log_{\sigma} n
\rfloor$, where $\mu$ is any positive constant smaller than $\frac16$
such that $\tau \geq 1$, be fixed for the duration of this section.
Throughout, we also use $\R$ as a shorthand for $\R(\tau, \T)$.

\begin{definition}\label{def:node-periodicity}
  Let $v$ be an explicit node of $\ST$. The node $v$ is
  said to be \emph{periodic} if $\str(v)$ is periodic
  (\cref{def:pattern-periodicity}). Otherwise, $v$ is
  \emph{nonperiodic}.
\end{definition}

\bfparagraph{Representation of a Node}

For any explicit node $v$ of $\ST$ we denote $\Occ(v) := \Occ(\str(v),
\T)$. In our data structure we represent each explicit node $v$ of
$\ST$ in one of two ways:
\begin{itemize}
\item A pair $(j, \ell)$, where $j \in \Occ(v)$ (i.e., $j$ is the
  starting position of some occurrence of $\str(v)$ in $\T$) and $\ell
  = \sdepth(v)$).
\item A pair $(\lrank(v),\rrank(v))$. Note that since $v$ is a node of
  suffix tree, in this special case we have $(\lrank(v), \rrank(v)) =
  (\LB(\str(v), \T), \UB(\str(v), \T))$.  Thus, letting $(b, e) =
  (\lrank(v), \rrank(v))$, we then have $\{\SA[i]\}_{i \in (b \dd e]}
  = \Occ(v)$. Note also that $b < e$.
\end{itemize}
In most cases, the latter representation leads to a more convenient
implementation. Thus, we adopt it as a default and denote $\repr(v) :=
(\lrank(v), \rrank(v))$ (while using the first one mostly as a
temporary internal representation). We also define $\repr(\nil) = (0,
0)$.

\bfparagraph{Organization}

The structure and query algorithms for a node $v$ are different
depending on whether $v$ is periodic (\cref{def:node-periodicity}).
Our description is thus split as follows. First (\cref{sec:st-core}),
we describe the set of data structures called collectively the index
``core'' that enables efficiently checking if $v$ is periodic (it is
also used to perform operations on nodes with very small depth and
contains some common components utilized by the remaining parts). In
the following two parts (\cref{sec:st-nonperiodic,sec:st-periodic}),
we describe structures handling each of the two cases. All ingredients
are then put together in \cref{sec:st-final}. Finally, we present our
result in the general form (\cref{sec:st-summary}).

\subsection{The Index Core}\label{sec:st-core}

In this section, we describe a data structure used to check in
$\bigO(1)$ time if a given node is periodic. It also lets us perform
operations concerning nodes at depth smaller than $3\tau - 1$ in
$\bigO(1)$ time.

The section is organized as follows. First, we introduce the
components of the data structure (\cref{sec:st-core-ds}).  We then
show how using this structure to implement some basic navigational
routines (\cref{sec:st-core-nav}). Next, we describe the query
algorithms for the fundamental operations
(\cref{sec:st-core-lca,sec:st-core-child,%
sec:st-core-pred,sec:st-core-wa}).  Finally, we show the
construction algorithm (\cref{sec:st-core-construction}).

\subsubsection{The Data Structure}\label{sec:st-core-ds}

\bfparagraph{Definitions}

For any $k \geq 1$, let $\Su_{k} := \{S \in \Alphabet^k : S\text{
occurs in }\T\}$ denote the set of length-$k$ substrings of $\T$.
Let $\Tshort$ denote the compact trie of $\Su_{3\tau-1}$.

\bfparagraph{Components}

The index core, denoted $\STCore(\T)$, consists of two components:
\begin{enumerate}
\item The index core $\SACore(\T)$ (\cref{sec:sa-core-ds}). It takes
  $\bigO(n/\log_{\sigma} n)$ space.
\item The compact trie $\Tshort$.  All nodes of
  $\Tshort$ are stored in an array and pointers to nodes are
  implemented as indexes to this array. Each node $v$ of $\Tshort$
  stores the string $\str(v)$ encoded as an integer $\Int(\str(v))$,
  the pointer $\parent(v)$, the value $\sdepth(v)$, and the doubly
  linked list containing pointers to all children of $v$, in ascending
  order of the first letter on the connecting edge. Since each node
  $v$ of $\Tshort$ corresponds to a unique string $S \in
  \Alphabet^{\leq 3\tau-1}$, in total $\Tshort$ needs
  $\bigO(\sigma^{3\tau}) = \bigO(\sqrt{n})$ space.  The trie $\Tshort$
  is augmented with the following structures:
  \begin{enumerate}[label=(\alph*)]
  \item A linear-space data structure answering the $\LCA$ queries in
    $\Tshort$ in $\bigO(1)$ time~\cite{BenderF00}. By the above bound,
    the data structure uses $\bigO(\sqrt{n})$ space.
  \item A lookup table $\LTchild$ that for each edge of $\Tshort$
    connecting a node $v$ to its parent $p$ and labeled with a string
    starting with the character $c$, maps the pair $(i_p, c)$ to
    $i_v$, where $i_p$ and $i_v$ are pointers to $p$ and
    $v$. $\Tshort$ has less than $2\sigma^{3\tau-1}$ nodes and thus
    $i_v < 2\sigma^{3\tau-1}$. On the other hand, $c \in
    \Alphabet$. Thus, each pair $(i_v, c)$ can be (in $\bigO(1)$ time)
    injectively mapped to an integer not exceeding $2\sigma^{3\tau} =
    \bigO(\sqrt{n})$ and hence $\LTchild$ needs $\bigO(n /
    \log_{\sigma} n)$ space.
  \item A lookup table $\LTwa$ that for every node $v$ of $\Tshort$
    and every $d \in [0 \dd 3\tau-1)$, maps the pair $(i_v, d)$ to
    $i_{u}$, where $u = \WA(v, d)$ and $i_v$ (resp.\ $i_u$) is the
    pointer to $v$ (resp.\ $u$). Since $i_v < 2\sigma^{3\tau-1}$ and
    $\tau = \bigO(\log n)$, each pair $(i_v, d)$ can be injectively
    mapped to an integer not exceeding $\bigO(\sqrt{n} \log n)$ and
    hence $\LTwa$ needs $\bigO(n / \log_{\sigma} n)$ space.
  \item An array storing the pointers to leaves of $\Tshort$ in the
    left-to-right order. Since the number of leaves is
    $\bigO(\sigma^{3\tau - 1})$, the array needs $\bigO(n /
    \log_{\sigma} n)$ space.
  \end{enumerate}
\end{enumerate}

In total, $\STCore(\T)$ takes $\bigO(n/\log_{\sigma} n)$ space.

\begin{remark}
  Note that $\Tshort$ corresponds to $\ST$ truncated at depth
  $3\tau-1$. The key reason motivating this definition is that the
  pair $(b, e) = (\LB(\str(v), \T), \UB(\str(v), \T))$ for every node
  $v$ of $\Tshort$ at depth $3\tau - 1$ that corresponds to an
  implicit node of $\ST$ (in the middle of an edge connecting some
  node $v'$ of $\ST$ to one of its children $v''$) satisfies $(b, e) =
  (\LB(\str(v''), \T), \UB(\str(v''), \T))$.  In all our uses, this is
  sufficient, and the value $\sdepth(v'')$ is never needed.
\end{remark}

\subsubsection{Navigation Primitives}\label{sec:st-core-nav}

\bfparagraph{Mapping from $\ST$ to $\Tshort$}

For any explicit node $v$ of $\ST$, we define $\mapToShort(v)$ as the
deepest explicit node $u$ of $\Tshort$ such that $\str(u)$ is a prefix
of $\str(v)$.

\begin{lemma}\label{lm:st-core-map}
  Let $v$ be an explicit node of $\ST$ and $u = \mapToShort(v)$. Then,
  $\sdepth(v) \geq 3\tau - 1$ holds if and only if $\sdepth(u) = 3\tau
  - 1$. Moreover,
  \begin{enumerate}
  \item\label{lm:st-core-map-it-1} If $\sdepth(v) \geq 3\tau - 1$ then
    $\str(u) = \str(v)[1 \dd 3\tau {-} 1]$.
  \item\label{lm:st-core-map-it-2} Otherwise (i.e., if $\sdepth(v) <
    3\tau - 1$), $\str(u) = \str(v)$.
  \end{enumerate}
\end{lemma}
\begin{proof}
  1. If $\sdepth(v) \geq 3\tau - 1$ then $\str(v)[1 \dd 3\tau {-} 1]
  \in \Su_{3\tau - 1}$. Therefore, by definition of $\Tshort$,
  there exists a node $u'$ in $\Tshort$ satisfying $\str(u') =
  \str(v)[1 \dd 3\tau {-} 1]$.  Since $\str(u')$ is a prefix of
  $\str(v)$ and $u'$ is a leaf of $\Tshort$, we thus have $u' = u$,
  and hence $\str(u) = \str(v)[1 \dd 3\tau {-} 1]$.

  2. Let $\sdepth(v) < 3\tau - 1$ and $X = \str(v)$. Since $v$ is
  explicit, there exists distinct $c, c' \in \Sigma$ such that $Xc$
  and $Xc'$ occur in $\T$. By $|X| < 3\tau - 1$, $\Tshort$ therefore
  has an explicit node $u'$ satisfying $\str(u') = X$.  By
  $\sdepth(u') = \sdepth(v)$, we thus have $u' = u$ and hence $\str(u)
  = \str(v)$.

  The equivalence follows immediately from the two items.
\end{proof}

\begin{lemma}\label{lm:st-core-map-2}
  Let $v$ be an explicit node of $\ST$. Let $i_1 = \lrank(v) + 1$,
  $i_2 = \rrank(v)$, $y_1 = \rank{\BVshort}{1}{i_1 - 1} + 1$, $y_2 =
  \rank{\BVshort}{1}{i_2 - 1} + 1$, $u_1$ (resp.\ $u_2$) be the $y_1$th
  (resp.\ $y_2$th) leftmost leaf of $\Tshort$, and $u = \LCA(u_1,
  u_2)$. Then, $\mapToShort(v) = u$.
\end{lemma}
\begin{proof}
  By definition of $\Tshort$ and $\ARRshort$ (\cref{sec:sa-core}), if
  $\widehat{u}$ is the $k$th leftmost leaf of $\Tshort$, then
  $\str(\widehat{u}) = \ARRshort[k]$. Thus, $\str(u_1) = \ARRshort[y_1]$
  and $\str(u_2) = \ARRshort[y_2]$.  Denote $Q = \str(v)$ and consider
  two cases:
  \begin{itemize}
  \item Let $\sdepth(v) \geq 3\tau - 1$. Denote $X = Q[1 \dd 3\tau {-}
    1]$.  By $i_1 = \lrank(v) + 1$ and $i_2 = \rrank(v)$, we then have
    $\SA[i_1], \SA[i_2] \in \Occ(Q, \T) \sub \Occ(X, \T)$.  By
    definition of $\BVshort$ and $\ARRshort$, positions $y_1 =
    \rank{\BVshort}{1}{i_1 - 1} + 1$ and $y_2 = \rank{\BVshort}{1}{i_2 -
    1} + 1$ then satisfy $\ARRshort[y_1] = \ARRshort[y_2] = X$. Thus, by
    the above observation, $\str(u_1) = \str(u_2) = X$ and hence, by
    \cref{ob:lca}, $\str(u) = X = \str(v)[1 \dd 3\tau -
    1]$. Consequently, by
    \cref{lm:st-core-map}\eqref{lm:st-core-map-it-1} and since all
    nodes of $\Tshort$ have different value of $\str$, this yields
    $\mapToShort(v) = u$.
  \item Let us now assume $\sdepth(v) < 3\tau - 1$.  Let $v_1$
    (resp.\ $v_2$) be the $i_1$th (resp.\ $i_2$th) leftmost leaf of
    $\ST$.  Then, $\str(v_1) = \T[\SA[i_1] \dd n]$ and $\str(v_2) =
    \T[\SA[i_2] \dd n]$.  By $i_1 = \lrank(v) + 1$ and $i_2 =
    \rrank(v)$ we have $v = \LCA(v_1, v_2)$. Thus, by \cref{ob:lca},
    $\lcp(\T[\SA[i_1] \dd n], \T[\SA[i_2] \dd n]) = \sdepth(v) = |Q|$.
    Observe now that:
    \begin{itemize}
    \item By definition of $\ARRshort$ and $\BVshort$, the string
      $\ARRshort[y_1]$ (resp.\ $\ARRshort[y_2]$) is a prefix of
      $\T[\SA[i_1] \dd n]$ (resp.\ $\T[\SA[i_2] \dd n]$).  Thus, it
      holds $\lcp(\ARRshort[y_1], \ARRshort[y_2]) \leq \lcp(\T[\SA[i_1]
      \dd n], \T[\SA[i_2] \dd n]) = |Q|$.
    \item On the other hand, since $Q$ is a prefix of $\str(v_1) =
      \T[\SA[i_1] \dd n]$ and $\str(v_2) = \T[\SA[i_2] \dd n]$, and it
      holds $|\ARRshort[y_1]| = \min(3\tau - 1, n - \SA[i_1] + 1)$,
      $|\ARRshort[y_2]| = \min(3\tau - 1, n - \SA[i_2] + 1)$, and $|Q| <
      3\tau - 1$, we obtain that $Q$ is a prefix of $\ARRshort[y_1]$ and
      $\ARRshort[y_2]$.  Thus, $\lcp(\ARRshort[y_1], \ARRshort[y_2]) \geq
      |Q|$.
    \end{itemize}
    We thus proved that $Q$ is a prefix of $\ARRshort[y_1]$ and
    $\ARRshort[y_2]$, and $\lcp(\ARRshort[y_1], \ARRshort[y_2]) = |Q|$.
    Thus, since $\str(u_1) = \ARRshort[y_1]$ and $\str(u_2) =
    \ARRshort[y_2]$, we obtain from \cref{ob:lca}, that $u = \LCA(u_1,
    u_2)$ satisfies $\str(u) = Q$. By
    \cref{lm:st-core-map}\eqref{lm:st-core-map-it-2} and since all
    nodes of $\Tshort$ have different value of $\str$, this yields
    $\mapToShort(v) = u$. \qedhere
  \end{itemize}
\end{proof}

\begin{proposition}\label{pr:st-core-map}
  Let $v$ be an explicit node of $\ST$.  Given $\STCore(\T)$ and
  $\repr(v)$, in $\bigO(1)$ time we can compute the pointer to the
  node $\mapToShort(v)$.
\end{proposition}
\begin{proof}
  Denote $(b, e) = \repr(v)$, $i_1 = b + 1$, and $i_2 = e$.  First, in
  $\bigO(1)$ time we compute $y_1 = \rank{\BVshort}{1}{i_1 - 1} + 1$
  and $y_2 = \rank{\BVshort}{1}{i_2 - 1} + 1$. In $\bigO(1)$ time we
  then retrieve the $y_1$th and $y_2$th leftmost leaves $u_1$ and $u_2$
  of $\Tshort$ (respectively). Finally, using the $\LCA$ structure for
  $\Tshort$, we compute in $\bigO(1)$ time the pointer to node $u =
  \LCA(u_1, u_2)$ of $\Tshort$. By \cref{lm:st-core-map-2}, we then
  have $\mapToShort(v) = u$.
\end{proof}

\begin{proposition}\label{pr:st-core-periodicity}
  Let $v$ be an explicit node of $\ST$.  Given $\STCore(\T)$ and
  $\repr(v)$, we can in $\bigO(1)$ time check if $v$ is periodic.
  If $v$ is not periodic, then in $\bigO(1)$ we can additionally
  determine if it holds $\sdepth(v) < 3\tau - 1$.
\end{proposition}
\begin{proof}
  First, using \cref{pr:st-core-map}, in $\bigO(1)$ time we compute
  the pointer to $u = \mapToShort(v)$. If $\sdepth(u) = 3\tau - 1$
  then by \cref{lm:st-core-map} it holds $\sdepth(v) \geq 3\tau - 1$,
  and we can in $\bigO(1)$ time determine if $v$ is periodic by
  checking if $\per(X) \leq \tfrac{1}{3}\tau$ for $X = \str(v)[1 \dd
  3\tau {-} 1] = \str(u)$ (stored with $u$) using the lookup table
  $\LTper$. If $\sdepth(u) < 3\tau - 1$, then by
  \cref{lm:st-core-map}, we have $\sdepth(v) < 3\tau - 1$, and hence
  $v$ is nonperiodic. Note that in the above algorithm, whenever $v$
  is nonperiodic, we always know if $\sdepth(v) < 3\tau - 1$. Thus, we
  can additionally return this information at no extra cost. Each of
  the steps takes $\bigO(1)$ time.
\end{proof}

\subsubsection{Implementation of
  \texorpdfstring{$\LCA(u, v)$}{LCA(u, v)}}\label{sec:st-core-lca}

\begin{lemma}\label{lm:st-core-lca}
  Let $v_1$ and $v_2$ be explicit nodes of $\ST$. Then,
  \[
    \mapToShort(\LCA(v_1, v_2)) = \LCA(\mapToShort(v_1),
    \mapToShort(v_2)).
  \]
\end{lemma}
\begin{proof}
  Let $u_1 = \mapToShort(v_1)$, $u_2 = \mapToShort(v_2)$, $v =
  \LCA(v_1, v_2)$, and $u = \LCA(u_1, u_2)$.  Then, the claim is that
  $\mapToShort(v) = u$. Denote $\ell = \lcp(\str(v_1), \str(v_2))$ and
  recall that by \cref{ob:lca}, we have $\sdepth(v) = \ell$. We
  consider two cases:
  \begin{itemize}
  \item First, assume $\sdepth(v) \geq 3\tau - 1$.  By $\sdepth(v_1)
    \geq \sdepth(v) \geq 3\tau - 1$, we obtain from
    \cref{lm:st-core-map}\eqref{lm:st-core-map-it-1} and
    \cref{ob:lca} that $\str(u_1) = \str(v_1)[1 \dd 3\tau {-} 1] =
    \str(v)[1 \dd 3\tau {-} 1]$.  Analogously, $\str(u_2) = \str(v)[1
    \dd 3\tau {-} 1]$, and consequently, $\str(u) = \str(v)[1 \dd
    3\tau {-} 1]$.  Since all nodes in $\Tshort$ have different values
    of $\str$, by \cref{lm:st-core-map}\eqref{lm:st-core-map-it-1},
    this implies $u = \mapToShort(v)$.
  \item Let us now assume $\sdepth(v) < 3\tau - 1$. We will show that
    then $\str(u) = \str(v)$. Since all nodes in $\Tshort$ have
    different values of $\str$, by
    \cref{lm:st-core-map}\eqref{lm:st-core-map-it-2}, this immediately
    implies $u = \mapToShort(v)$.  We first show that $\sdepth(u_1)
    \geq \ell$.  Consider two cases.  If $\sdepth(v_1) \geq 3\tau -
    1$, then by \cref{lm:st-core-map}\eqref{lm:st-core-map-it-1},
    $\str(u_1) = \str(v_1)[1 \dd 3\tau - 1]$, i.e., $\sdepth(u_1) =
    3\tau - 1 > \sdepth(v) = \ell$. Otherwise, by
    \cref{lm:st-core-map}\eqref{lm:st-core-map-it-2}, it holds
    $\str(u_1) = \str(v_1)$, and thus also $\sdepth(u_1) =
    \sdepth(v_1) \geq \sdepth(v) = \ell$.  By the analogous argument,
    $\sdepth(u_2) \geq \ell$. Recall now that, by definition,
    $\str(u_1)$ (resp.\ $\str(u_2)$) is a prefix of $\str(v_1)$
    (resp.\ $\str(v_2)$). Thus, $\str(u_1)[1 \dd \ell] = \str(v_1)[1
    \dd \ell] = \str(v_2)[1 \dd \ell] = \str(u_2)[1 \dd \ell]$.
    Denoting $\ell' = \lcp(\str(u_1), \str(u_2))$, we therefore have
    $\ell' \geq \ell$.  On the other hand, $\str(u_1)$
    (resp.\ $\str(u_2)$) being a prefix of $\str(v_1)$
    (resp.\ $\str(v_2)$), implies $\ell' \leq \ell$. Consequently,
    $\ell' = \ell$.  By \cref{ob:lca}, we therefore obtain
    $\str(u) = \str(\LCA(u_1, u_2)) = \str(u_1)[1 \dd \ell'] =
    \str(v_1)[1 \dd \ell] = \str(\LCA(v_1, v_2)) = \str(v)$. \qedhere
  \end{itemize}
\end{proof}

\begin{proposition}\label{pr:st-core-lca}
  Let $v_1$ and $v_2$ be explicit nodes of $\ST$.  Given $\STCore(\T)$
  and the pairs $\repr(v_1)$ and $\repr(v_2)$, we can in $\bigO(1)$
  time check if $\sdepth(\LCA(v_1, v_2)) \geq 3\tau - 1$. If so, in
  $\bigO(1)$ time we can additionally determine if $\LCA(v_1, v_2)$ is
  periodic. Otherwise (i.e., if $\sdepth(\LCA(v_1, v_2)) < 3\tau - 1$)
  in $\bigO(1)$ time we can compute $\repr(\LCA(v_1, v_2))$.
\end{proposition}
\begin{proof}
  Denote $v = \LCA(v_1, v_2)$.  First, using \cref{pr:st-core-map}, in
  $\bigO(1)$ time we compute pointers to $u_1 = \mapToShort(v_1)$ and
  $u_2 = \mapToShort(v_2)$ of $\Tshort$. Then, using the $\LCA$
  structure for $\Tshort$, we compute in $\bigO(1)$ time the pointer
  to node $u = \LCA(u_1, u_2)$ of $\Tshort$. By \cref{lm:st-core-lca},
  we now have $\mapToShort(v) = u$.  If $\sdepth(u) = 3\tau - 1$, by
  \cref{lm:st-core-map} it holds $\sdepth(v) \geq 3\tau - 1$ and
  $\str(u) = \str(v)[1 \dd 3\tau {-} 1]$, and thus we can in
  $\bigO(1)$ determine if $v$ is periodic by checking if $\per(X) \leq
  \tfrac{1}{3}\tau$ for $X = \str(u)$ (stored with $u$) using the
  lookup table $\LTper$.  Otherwise (i.e., if $\sdepth(u) < 3\tau -
  1$), by \cref{lm:st-core-map} we have $\sdepth(v) < 3\tau - 1$ and
  $\str(u) = \str(v)$.  Thus, we return that $v$ is nonperiodic and in
  $\bigO(1)$ time we obtain the pair $\repr(v) = (\LB(\str(v), \T),
  \UB(\str(v), \T)) = (\LB(\str(u), \T), \UB(\str(u), \T))$ using the
  lookup table $\LTrange$ on $\str(u)$.  Each of the steps
  takes $\bigO(1)$ time.
\end{proof}

\subsubsection{Implementation of
  \texorpdfstring{$\child(v, c)$}{child(v, c)}}\label{sec:st-core-child}

\begin{lemma}\label{lm:st-core-child}
  Let $v$ be an explicit internal node of $\ST$ satisfying $\sdepth(v)
  < 3\tau - 1$. Let $u = \mapToShort(v)$.  For any $c \in \Alphabet$,
  $\child(v, c) = \nil$ holds if and only if $\child(u, c) = \nil$.
  Moreover, if $\child(v, c) \neq \nil$ then, letting $u' = \child(u,
  c)$, it holds
  \[
    \repr(\child(v, c)) = (\LB(\str(u'), \T), \UB(\str(u'), \T)).
  \]
\end{lemma}
\begin{proof}
  By \cref{lm:st-core-map}\eqref{lm:st-core-map-it-2}, $u$ satisfies
  $\str(u) = \str(v)$. Thus, since $\Tshort$ is a compact trie of
  substrings of $\T$ of length $3\tau-1$, we immediately obtain that
  for any $c \in \Alphabet$, $\child(v, c) \neq \nil$ if and only if
  $\child(u, c)$.  Let us assume for some $c \in \Alphabet$, it holds
  $\child(v, c) = v' \neq \nil$. If $\sdepth(v') \leq 3\tau-1$, then
  by definition of $\Tshort$, the node $u' = \child(u, c)$ must
  satisfy $\str(v') = \str(u')$. This implies the claim
  immediately. Otherwise ($\sdepth(v') > 3\tau-1$), $u'$ satisfies
  $\sdepth(u') = 3\tau - 1$, and corresponds to the implicit node of
  $\ST$ on the edge connecting $v$ to $v'$. By definition of suffix
  tree, however, letting $S$ be such that $\str(v)S = \str(v')[1 \dd
  3\tau{-}1]$, we have $(\LB(\str(v'), \T), \UB(\str(v'), \T)) {=}
  (\LB(\str(v)S, \T), \UB(\str(v)S, \T)) = (\LB(\str(u'), \T),
  \UB(\str(u'), \T))$, which by definition of $\repr$ implies the
  claim.
\end{proof}

\begin{proposition}\label{pr:st-core-child}
  Let $v$ be an explicit internal node of $\ST$ satisfying $\sdepth(v)
  < 3\tau - 1$.  Given $\STCore(\T)$, $\repr(v)$, and $c \in
  \Alphabet$, we can in $\bigO(1)$ time compute $\repr(\child(v, c))$.
\end{proposition}
\begin{proof}
  First, using \cref{pr:st-core-map}, in $\bigO(1)$ time we compute a
  pointer to $u = \mapToShort(v)$. Using the lookup table $\LTchild$,
  in $\bigO(1)$ time we check if $\child(u, c) = \nil$. If so, then by
  \cref{lm:st-core-child}, it holds $\child(v, c) = \nil$ and we
  return $\repr(\child(v, c)) = (0, 0)$.  Otherwise (i.e., $\child(u,
  c) \neq \nil$), we obtain a pointer to $u' = \child(u, c)$. By
  \cref{lm:st-core-child}, we then have $\repr(\child(v, c)) =
  (\LB(\str(u'), \T), \UB(\str(u'), \T))$, which we obtain using the
  lookup table $\LTrange$ on $\str(u')$. Each of the steps takes
  $\bigO(1)$ time.
\end{proof}

\subsubsection{Implementation of
  \texorpdfstring{$\pred(v, c)$}{pred(v, c)}}\label{sec:st-core-pred}

\begin{proposition}\label{pr:st-core-pred}
  Let $v$ be an explicit internal node of $\ST$ satisfying $\sdepth(v)
  < 3\tau - 1$.  Given $\STCore(\T)$, $\repr(v)$, and $c \in
  \Alphabet$, we can in $\bigO(1)$ time compute $\LB(\str(v)c, \T)$.
\end{proposition}
\begin{proof}
  First, using \cref{pr:st-core-map} we compute a pointer to $u =
  \mapToShort(v)$. By $\sdepth(v) < 3\tau - 1$ and
  \cref{lm:st-core-map}\eqref{lm:st-core-map-it-2}, node $u$ satisfies
  $\str(u) = \str(v)$. Thus, we have $\LB(\str(v)c, \T) =
  \LB(\str(u)c, \T)$. Next, we compute $Y = \str(u)c$ (recall, that
  $\str(u)$ is stored with $u$).  Using the lookup table $\LTrange$,
  we then compute and return $\LB(Y, \T)$.  Each of the steps takes
  $\bigO(1)$ time.
\end{proof}

\subsubsection{Implementation of
  \texorpdfstring{$\WA(v, d)$}{WA(v, d)}}\label{sec:st-core-wa}

\begin{lemma}\label{lm:st-core-wa}
  Let $v$ be an explicit node of $\ST$ and $d$ be such that $0 \leq d
  \leq |\str(v)|$ and $d < 3\tau \,{-}\, 1$. Then, letting $u =
  \mapToShort(v)$ and $u' = \WA(u, d)$, it holds
  \[
    \repr(\WA(v, d)) = (\LB(\str(u'), \T), \UB(\str(u'), \T)).
  \]
\end{lemma}
\begin{proof}
  Denote $v' = \WA(v, d)$. We consider two cases:
  \begin{itemize}
  \item First, assume $\sdepth(v) \geq 3\tau - 1$.  By
    \cref{lm:st-core-map}\eqref{lm:st-core-map-it-1}, we then have
    $\str(u) = \str(v)[1 \dd 3\tau {-} 1]$. Therefore, utilizing one
    of the assumptions about $d$, we have $d < 3\tau - 1 =
    \sdepth(u)$, i.e., $u'$ is well-defined (see \cref{sec:wa}).
    Moreover, this implies that for any ancestor $\bar{v}$ of $v$ at
    depth at most $3\tau - 1$, there exist a corresponding ancestor
    $\bar{u}$ of $u$ and there exists a one-to-one mapping between
    ancestors of $\bar{v}$ in $\ST$ and ancestors of $\bar{u}$ in
    $\Tshort$ (with corresponding nodes having equal root-to-node
    labels). Therefore, if $\sdepth(v') \leq 3\tau - 1$ then $\str(u')
    = \str(v')$ and the claim follows. Otherwise ($\sdepth(v') > 3\tau
    - 1$), by $d < 3\tau - 1$, we must have $u' = u$ and $u'$ then
    corresponds to the implicit node on the edge connecting $v'$ to
    $\parent(v')$.  This implies $\repr(v') = (\LB(\str(u'), \T),
    \UB(\str(u'), \T))$.
  \item Let us now assume $\sdepth(v) < 3\tau - 1$. By
    \cref{lm:st-core-map}\eqref{lm:st-core-map-it-2}, we then have
    $\str(u) = \str(v)$. In particular, utilizing one of the
    assumptions on $d$, we have $d \leq \sdepth(v) = \sdepth(u)$,
    i.e., $u'$ is well-defined (see \cref{sec:wa}).  Moreover, this
    implies that there is a one-to-one correspondence between
    ancestors of $v$ in $\ST$ and ancestors or $u$ in $\Tshort$. In
    particular, $\str(v') = \str(u')$, which implies the
    claim. \qedhere
  \end{itemize}
\end{proof}

\begin{proposition}\label{pr:st-core-wa}
  Let $v$ be an explicit node of $\ST$. Given $\STCore(\T)$,
  $\repr(v)$, and an integer $d$ satisfying $0 \leq d \leq |\str(v)|$
  and $d < 3\tau - 1$, in $\bigO(1)$ time we can compute $\repr(\WA(v,
  d))$.
\end{proposition}
\begin{proof}
  First, using \cref{pr:st-core-map}, we compute a pointer to node $u
  = \mapToShort(v)$. Then, using the lookup table $\LTwa$, in
  $\bigO(1)$ time we obtain the pointer to $u' = \WA(u, d)$. By
  \cref{lm:st-core-wa}, we then have $\repr(\WA(v, d)) =
  (\LB(\str(u'), \T), \UB(\str(u'), \T))$, which is obtained using the
  lookup table $\LTrange$ on $\str(u')$. Each of the steps takes
  $\bigO(1)$ time.
\end{proof}

\subsubsection{Construction Algorithm}\label{sec:st-core-construction}

\begin{proposition}\label{pr:st-core-construction}
  Given the packed representation of $\T \in \Alphabet^n$, we can
  construct $\STCore(\T)$ in $\bigO(n / \log_{\sigma} n)$ time.
\end{proposition}
\begin{proof}

  First, in $\bigO(n / \log_{\sigma} n)$ time we construct
  $\SACore(\T)$ using \cref{pr:sa-core-construction}.
  Note that during the construction,
  we compute the frequency $f_X = |\Occ(X, \T)|$ for every
  $X \in \Alphabet^{\leq 3\tau - 1}$ (note that by definition
  of $\Occ(X, \T)$ (\cref{sec:prelim}), we have
  $f_X = n$ for the empty string $X = \emptystring$).

  Next, we construct the trie $\Tshort$ and the associated data
  structures. Observe that for every $X \in \Alphabet^{\leq
  3\tau-1}$, the trie $\Tshort$ contains an explicit node $v$
  satisfying $\str(v) = X$ if and only if $f_X > 0$, and either $|X| =
  3\tau - 1$ or $|X| < 3\tau - 1$ and there exist distinct $c, c'
  \in \Alphabet$ such that $f_{Xc} > 0$ and $f_{Xc'} > 0$.~\footnote{Note,
  that this holds also for $X = \emptystring$ because we defined
  $f_{\emptystring} = n$ and assumed in \cref{sec:prelim}
  that $\T$ contains at least two distinct symbols.}
  Thus, given any $X$, we can in $\bigO(\sigma)$ time
  check if there exists a node of $\Tshort$ corresponding to $X$.
  Moreover, if such $v$ exists and $|X| > 0$ then to find $X'$
  satisfying $\str(\parent(v)) = X'$, it suffices to compute the
  longest prefix $X'$ of $X$ such that $\LTrange$ for $X'$ is
  different from $\LTrange$ for $X$. Thus, such $X'$ can be computed
  in $\bigO(\tau) = \bigO(\log n)$ time. Using the above observations,
  we construct $\Tshort$ as follows. We maintain a lookup table
  $\LTnode$ that for any $X \in \Alphabet^{\leq 3\tau - 1}$ maps the
  integer $\Int(X)$ to a pointer to the node $v$ of $\Tshort$
  satisfying $\str(v) = X$ if such $v$ exists. By $\Int(X) \in [0 \dd
  \sigma^{6\tau})$, the table needs $\bigO(\sigma^{6\tau}) =
  \bigO(n/\log_{\sigma} n)$ space and its initialization takes
  $\bigO(n/\log_{\sigma} n)$ time.  During the construction, nodes are
  stored in a dynamic array with amortized $\bigO(1)$-time insertion
  at the end, and pointers are implemented as indexes of this array.
  We enumerate all $X \in \Alphabet^{\leq 3\tau - 1}$ in the order of
  non-decreasing length, and in case of ties, in lexicographical
  order. For each $X$, using the above method in $\bigO(\sigma)$ time
  we check whether there should be a node in $\Tshort$ satisfying
  $\str(v) = X$.  If so, we create a new node $v$, add it to the array
  of nodes, and update the lookup table $\LTnode$. Associated with $v$
  we store the string $X$ encoded as $\Int(X)$ and the length $|X| =
  \sdepth(v)$. If $|X| > 0$, in $\bigO(\log n)$ time we then compute
  the longest prefix $X'$ of $X$ for which $(\LB(X', \T), \UB(X',\T))
  \neq (\LB(X, \T), \UB(X, \T))$ (utilizing the lookup table
  $\LTrange$), and then using $\LTnode$ obtain $v'$ satisfying
  $\str(v') = X'$.  We then set $\parent(v) = v'$ and add $v$ to the
  list of children of $v'$, updating also the links between children
  of $v'$. Over all $X \in \Alphabet^{\leq 3\tau - 1}$, the
  construction takes $\bigO(\sigma^{3\tau-1} (\sigma + \log n)) =
  \bigO(n/\log_{\sigma} n)$ time. After constructing $\Tshort$, we
  augment it with auxiliary structures as follows:
  \begin{enumerate}[label=(\alph*)]
  \item In $\bigO(\sigma^{3\tau-1}) = \bigO(n/\log_{\sigma} n)$ time
    we preprocess $\Tshort$ for $\bigO(1)$-time $\LCA$ queries using
    the structure from~\cite{BenderF00}.
  \item Next, we perform a traversal of $\Tshort$. For each node $v$
    different from the root, we obtain the pointer to $p = \parent(v)$
    and $c = \str(v)[|\str(p)| + 1]$.  We then injectively map $(i_p,
    c)$ (where $i_p$ is the pointer to $p$) to an integer $x$ not
    exceeding $2\sigma^{3\tau} = \bigO(\sqrt{n})$ and set $\LTchild[x]
    := i_v$, where $i_v$ is the pointer to $v$.  The construction
    takes $\bigO(\sqrt{n}) = \bigO(n / \log_{\sigma} n)$ time.
  \item Next, starting from each node $v$ of $\Tshort$ we compute the
    pointer to $\WA(v, d)$ for every $d \in [0 \dd 3\tau - 1)$.  It
    suffices to perform one traversal towards the roots and thus this
    takes $\bigO(\log n)$ time per node (in total for all $d$).  For
    each computed node $v'$, we map the pair $(i_v, d)$ (where $i_v$
    is the pointer to $v$) to an integer $x$ not exceeding
    $\bigO(\sqrt{n} \log n)$ and set $\LTwa[x] := i_{v'}$, where
    $i_{v'}$ is the pointer to $v'$.  Including the initialization of
    $\LTwa$, the construction takes $\bigO(\sqrt{n} \log n) =
    \bigO(n/\log_{\sigma} n)$ time.
  \item Finally, we perform the in-order traversal of the tree,
    collecting the leaves of $\Tshort$ in an array. By the bound on the
    number of nodes, this takes $\bigO(n / \log_{\sigma} n)$
    time. \qedhere
  \end{enumerate}
\end{proof}

\subsection{The Nonperiodic Nodes}\label{sec:st-nonperiodic}

In this section, we describe a data structure used to perform
operations on nonperiodic nodes (see \cref{def:node-periodicity}) in
$\bigO(\log^{\epsilon} n)$ time.

The section is organized as follows. First, we introduce the
components of the data structure (\cref{sec:st-nonperiodic-ds}).  We
then show how using this structure to implement some basic
navigational routines (\cref{sec:st-nonperiodic-nav}). Next, we
describe the query algorithms for the fundamental operations
(\cref{sec:st-nonperiodic-lca,sec:st-nonperiodic-child,%
sec:st-nonperiodic-pred,sec:st-nonperiodic-wa}).  Finally, we show
the construction algorithm (\cref{sec:st-nonperiodic-construction}).

\subsubsection{The Data Structure}\label{sec:st-nonperiodic-ds}

\bfparagraph{Definitions}

Let $\S$ be a $\tau$-synchronizing set, as defined in
\cref{sec:sa-nonperiodic-ds}.  Recall (\cref{sec:pm-nonperiodic-ds})
that $\ARRslex[1 \dd n']$ is an array defined by $\ARRslex[i] =
\slex_i$. Let $\TSSS$ denote the compact trie of the set $\{\T[i \dd
n] : i \in \S\}$.

\bfparagraph{Components}

The data structure to handle nonperiodic nodes consists of three
components:
\begin{enumerate}
\item The index core $\STCore(\T)$ (\cref{sec:st-core-ds}). It takes
  $\bigO(n/\log_{\sigma} n)$ space.
\item The data structure from \cref{sec:sa-nonperiodic-ds} using
  $\bigO(n / \log_{\sigma} n)$ space.
\item The compact trie $\TSSS$ represented as in
  \cref{pr:compact-trie} (i.e., for the array $\ARRslex[1 \dd n']$
  defined above). By $n' = \bigO(n / \log_{\sigma} n)$ and
  \cref{pr:compact-trie}, it needs $\bigO(n/\log_{\sigma} n)$ space.
\end{enumerate}

In total, the data structure takes $\bigO(n/\log_{\sigma} n)$ space.

\subsubsection{Navigation Primitives}\label{sec:st-nonperiodic-nav}

\bfparagraph{Mapping from $\ST$ to $\TSSS$}

For any explicit nonperiodic node $v$ of $\ST$ satisfying $\sdepth(v)
\geq 3\tau - 1$, we define $\mapToTSSS(v) = u$ as a node of $\TSSS$
satisfying $\str(v) = X[1 \dd \deltatext] \cdot \str(u)$, where $X \in
\D$ is a prefix of $\str(v)$ and $\deltatext = |X| - 2\tau$ (such $X$
exists and is unique, since for $Y = \str(v)[1 \dd 3\tau {-} 1]$ it
holds $\Occ(Y, \T) \neq \emptyset$ and $\per(Y) > \tfrac{1}{3}\tau$;
see \cref{sec:sa-nonperiodic}).

\begin{lemma}\label{lm:st-nonperiodic-map}
  Let $v$ be an explicit nonperiodic node of $\ST$ satisfying
  $\sdepth(v) \geq 3\tau - 1$.
  \begin{enumerate}
  \item The node $\mapToTSSS(v)$ is well-defined.
  \item\label{lm:st-nonperiodic-map-it-2} Let $X \in \D$ be a prefix
    of $\str(v)$, $b_X = \LB(X, \T)$, $i_1 = \lrank(v) + 1$, $i_2 =
    \rrank(v)$, $y_1 = \select{W}{\revstr{X}}{i_1 - b_X}$, $y_2 =
    \select{W}{\revstr{X}}{i_2 - b_X}$, $u_1$ (resp.\ $u_2$) be the
    $y_1$th (resp.\ $y_2$th) leftmost leaf of $\TSSS$, and
    $u = \LCA(u_1, u_2)$. Then, $\mapToTSSS(v) = u$.
  \end{enumerate}
\end{lemma}
\begin{proof}
  1. Let $X \in \D$ be a prefix of $\str(v)$ and let $\deltatext = |X|
  - 2\tau$. If $v$ is a leaf of $\ST$, then for
  $i \in \Occ(\str(v), \T)$, it holds
  that $X$ is a prefix of $\T[i \dd n]$. Thus, by the consistency of
  $\S$, $i + \deltatext \in \S$, and consequently, there exist $u$ in
  $\TSSS$ such that $\str(v) = X[1 \dd \deltatext] \cdot
  \str(u)$. Otherwise (i.e., $v$ is an internal node), consider any
  two different leaves $v_1$ and $v_2$ in the subtree rooted in $v$
  such that $v = \LCA(v_1, v_2)$.  Let $i_1 \in \Occ(\str(v_1), \T)$
  and $i_2 \in \Occ(\str(v_2), \T)$.
  Then, $\str(v)$ is a prefix of both $\T[i_1 \dd n]$ and
  $\T[i_2 \dd n]$.  Since $X$ is a prefix of $\str(v)$, $X$ is
  therefore also a prefix of $\T[i_1 \dd n]$ and $\T[i_2 \dd
  n]$. Thus, again by the consistency of $\S$, we have $i_1 +
  \deltatext, i_2 + \deltatext \in \S$. Consequently, there exist
  nodes $u_1$ and $u_2$ in $\TSSS$ satisfying $\str(v_1) = X[1 \dd
  \deltatext] \cdot \str(u_1)$ and $\str(v_2) = X[1 \dd \deltatext]
  \cdot \str(u_2)$.  By \cref{ob:lca} applied to $v_1$ and $v_2$,
  for $\ell = \lcp(\str(v_1), \str(v_2))$ it holds $\str(v) =
  \str(v_1)[1 \dd \ell]$.  On the other hand, applying
  \cref{ob:lca} to $u_1$ and $u_2$ implies that for $\ell' =
  \lcp(\str(u_1), \str(u_2))$ and $u = \LCA(u_1, u_2)$, it holds
  $\str(u) = \str(u_1)[1 \dd \ell']$. Finally, by $\deltatext < |X|$,
  we have $\ell = \lcp(\str(v_1), \str(v_2)) = \lcp(X[1 \dd
  \deltatext] \cdot \str(u_1), X[1 \dd \deltatext] \cdot \str(u_2)) =
  \deltatext + \lcp(\str(u_1), \str(u_2)) = \deltatext + \ell'$. Thus,
  \begin{align*}
   \str(v) & = \str(v_1)[1 \dd \ell]\\
     &= X[1 \dd \deltatext] \cdot
        \str(u_1)[1 \dd \ell - \deltatext] \\
     &= X[1 \dd \deltatext] \cdot \str(u_1)[1 \dd \ell'] \\
     &= X[1 \dd \deltatext] \cdot \str(u),
  \end{align*}
  i.e., there exists $u$ in $\TSSS$ satisfying $\str(v) = X[1 \dd
  \deltatext] \cdot \str(u)$, i.e., $\mapToTSSS(v)$ is well-defined.

  2. Let $\deltatext = |X| - 2\tau$. To see that $y_1$ and $y_2$ in
  the definition are well-defined (i.e., that $i_1-b_X, i_2-b_X \in [1
  \dd \rank{W}{\revstr{X}}{n'}]$), recall first that (similarly as in
  the proof of \cref{lm:pm-nonperiodic,lm:sa-nonperiodic-isa}) by
  consistency of $\S$, there exists a bijection (given by $j \mapsto j
  + \deltatext$) between $\Occ(X,\T)$ and positions $s \in \S$ such
  that $\T^{\infty}[s - \deltatext \dd s + 2\tau) = X$.  In
  particular, by definition of $W[1 \dd n']$, this implies
  $\rank{W}{\revstr{X}}{n'} = |\Occ(X,\T)| = e_X - b_X$, where $e_X =
  \UB(X, \T)$. Observe now that in $\ST$, for any node $v$, it holds
  $\lrank(v) = \LB(\str(v), \T)$ and $\rrank(v) = \UB(\str(v), \T)$
  (this property does not hold, e.g., in $\TSSS$). Thus, since $X$ is
  a prefix of $\str(v)$, we have $b_X < i_1 \leq i_2 \leq e_X$.
  Combining with the above, we thus obtain $1 \leq i_1 - b_X \leq i_2
  - b_X \leq \rank{W}{\revstr{X}}{n'}$.

  Let $v_1$ (resp.\ $v_2$) be the $i_1$th (resp.\ $i_2$th) leftmost
  leaf of $\ST$. Denote $\str(v) = Q$.  By $i_1, i_2 \in (b_X \dd
  e_X]$, the string $X$ is a prefix of $\str(v_1)$ and
  $\str(v_2)$. Since $v = \LCA(v_1, v_2)$, $X$ is therefore also a
  prefix of $\str(v)$. To show $\str(v) = X[1 \dd \deltatext] \cdot
  \str(u)$, it thus suffices to show $\str(u) = Q(\deltatext \dd
  |Q|]$.  To this end, we will prove that $Q(\deltatext \dd |Q|]$ is a
  prefix of $\str(u_1)$ and $\str(u_2)$, and that it holds
  $\lcp(\str(u_1), \str(u_1)) = |Q| - \deltatext$.  By $u = \LCA(u_1,
  u_2)$, this immediately implies the claim.  Let $i \in (b_X \dd
  e_X]$.  By \cref{lm:sa-nonperiodic-isa} for $j = \SA[i]$, if $\slex_y =
  \SA[i] + \deltatext$ then $\delta(\SA[i]) =
  \rank{W}{\revstr{X}}{y}$. Since by definition of $b_X$ it holds
  $\delta(\SA[i]) = i - b_X$, we obtain $i - b_X =
  \rank{W}{\revstr{X}}{y}$.  Since $y$ also satisfies $\T[\slex_y -
  \deltatext \dd \slex_y + 2\tau) = X$, it must hold $y =
  \select{W}{\revstr{X}}{i - b_X}$. For such $y$ we have $\slex_y =
  \SA[i] + \deltatext$. Applied for $i_1$ and $i_2$, we obtain
  $\slex_{y_1} = \SA[i_1] + \deltatext$ and $\slex_{y_2} = \SA[i_2] +
  \deltatext$. Recall now that the sequence $(\slex_i)_{i\in [1 \dd
  n']}$ contain the positions in $\S$ sorted according to the
  lexicographical order of the corresponding suffixes of $\T$.  This
  implies that the $y_1$th (resp.\ $y_2$th) leftmost leaf $u_1$
  (resp.\ $u_2$) of $\TSSS$ satisfies $\str(u_1) =
  \T[\slex_{y_1} \dd n] = \T[\SA[i_1] + \deltatext \dd n]$
  (resp.\ $\str(u_2) = \T[\slex_{y_2} \dd n] = \T[\SA[i_2] +
  \deltatext \dd n]$).  Since all suffixes of $\T$ with starting
  positions in $\SA(\lrank(v) \dd \rrank(v)]$ have $Q$ as a prefix,
  and we clearly have $i_1, i_2 \in (\lrank(v) \dd \rrank(v)]$, we
  immediately obtain that $Q(\deltatext \dd |Q|]$ is a prefix of both
  $\T[\SA[i_1] + \deltatext \dd n] = \str(u_1)$ and $\T[\SA[i_2] +
  \deltatext \dd n] = \str(u_2)$. To show the second claim, we first
  note that we have $\lcp(\T[\SA[i_1] \dd n], \T[\SA[i_2] \dd n]) =
  \lcp(\str(v_1), \str(v_2)) = |\str(\LCA(v_1, v_2))| = |\str(v)| =
  |Q|$.  Together with $\deltatext \leq |X| \leq |Q|$, this implies
  $\lcp(\str(u_1), \str(u_2)) = \lcp(\T[\SA[i_1] + \deltatext \dd n],
  \T[\SA[i_2] + \deltatext \dd n]) = \lcp(\T[\SA[i_1] \dd n],
  \T[\SA[i_2] \dd n]) - \deltatext = |Q| - \deltatext$.  As noticed
  earlier, these two facts yield $\str(u) = Q(\deltatext \dd |Q|]$,
  and consequently $\str(v) = X[1 \dd \deltatext] \cdot
  \str(u)$. Thus, $\mapToTSSS(v) = u$.
\end{proof}

\begin{proposition}\label{pr:st-nonperiodic-map}
  Let $v$ be an explicit nonperiodic node of $\ST$ satisfying
  $\sdepth(v) \geq 3\tau - 1$.  Given the data structure from
  \cref{sec:st-nonperiodic-ds} and the pair $\repr(v)$, we can in
  $\bigO(\log^{\epsilon} n)$ time compute the pointer to
  $\mapToTSSS(v)$.
\end{proposition}
\begin{proof}
  Denote $(b, e) = \repr(v)$. First, using \cref{pr:st-core-map}, in
  $\bigO(1)$ time we compute pointer to $u = \mapToShort(v)$.  By
  \cref{lm:st-core-map}\eqref{lm:st-core-map-it-1}, we have $\str(u) =
  \str(v)[1 \dd 3\tau {-} 1]$.  Letting $Y = \str(u)$, we then have
  $\per(Y) > \tfrac{1}{3}\tau$ and $\Occ(Y, \T) \neq \emptyset$. This
  implies (see \cref{sec:sa-nonperiodic-ds}) that there exists a unique
  prefix $X \in \D$ of $\str(v)$. Using $\LTD$ on $Y$, in $\bigO(1)$
  time we obtain $X$. Using the lookup table $\LTrange$ (stored as
  part of $\STCore(\T)$; see \cref{sec:st-core}), in $\bigO(1)$ time
  we compute $b_X = \LB(X, \T)$.  Using the lookup table $\LTrev$
  stored in the structure from \cref{sec:st-nonperiodic-ds}, we then
  obtain $\revstr{X}$. Next, letting $i_1 = b + 1$ and $i_2 = e$
  (recall that $\repr(v) = (\lrank(v), \rrank(v))$), using
  \cref{th:wavelet-tree} in $\bigO(\log^{\epsilon} n)$ time we
  compute $y_1 = \select{W}{\revstr{X}}{i_1 - b_X}$ and $y_2 =
  \select{W}{\revstr{X}}{i_2 - b_X}$.  Then, using
  \cref{pr:compact-trie} in $\bigO(1)$ time we compute the pointers to
  the $y_1$th and $y_2$th leftmost leaves $u_1$ and $u_2$
  (respectively) of $\TSSS$.  Then, again using
  \cref{pr:compact-trie}, in $\bigO(1)$ time we compute and return the
  pointer to $u = \LCA(u_1, u_2)$. By
  \cref{lm:st-nonperiodic-map}\eqref{lm:st-nonperiodic-map-it-2}, it
  holds $\mapToTSSS(v) = u$.
\end{proof}

\bfparagraph{Mapping from $\TSSS$ to $\ST$}

For any string $X \in \Alphabet^{\leq 3\tau - 1}$ and any node $u$ of
the trie $\TSSS$, we define $\pseudoInvTSSS{X}{u} = (b_X + \delta_1,
b_X + \delta_2)$, where $b_X = \LB(X, \T)$, $z_1 = \lrank(u)$, $z_2 =
\rrank(u)$, $\delta_1 = \rank{W}{\revstr{X}}{z_1}$, and $\delta_2 =
\rank{W}{\revstr{X}}{z_2}$.

\begin{remark}\label{rm:st-nonperiodic-invmap}
  Note that the mapping from $\ST$ to $\TSSS$ is not necessarily
  injective, and hence it may not have an inverse. To perform the
  mapping from $\TSSS$ to $\ST$, we will use the above function.
  Note, however, that although the pair $\pseudoInvTSSS{X}{u}$ is
  always defined, not for \emph{every} $X$ and $u$ it yields
  $\repr(v)$ for some node $v$ of $\ST$. Below we show a simple but
  useful condition where it does. In the following sections we show
  more subtle uses of $\pseudoInvTSSS{X}{u}$ (see, e.g.,
  \cref{rm:st-nonperiodic-child}).
\end{remark}

\begin{lemma}\label{lm:st-nonperiodic-invmap}
  Let $v$ be an explicit nonperiodic node of $\ST$ satisfying
  $\sdepth(v) \geq 3\tau - 1$ and let $u = \mapToTSSS(v)$. Then,
  letting $X \in \D$ be a prefix of $\str(v)$, it holds $\repr(v) =
  \pseudoInvTSSS{X}{u}$.
\end{lemma}
\begin{proof}
  First, recall that $\D \sub \Alphabet^{\leq 3\tau - 1}$
  (\cref{sec:sa-nonperiodic-ds}). Thus, $\pseudoInvTSSS{X}{u}$ is
  well-defined.  Denote $b_X = \LB(X, \T)$, $Q = \str(v)$, $\deltatext
  = |X| - 2\tau$, $\Qsuf = Q(\deltatext \dd |Q|]$. Note that $\str(u)
  = \Qsuf$. Since $\{\slex_i\}_{i \in [1 \dd n']} = \S$ and
  $(\T[\slex_i \dd n])_{i \in [1 \dd n']}$ is lexicographically
  sorted, it holds by definition of $\TSSS$ that $\lrank(u) = |\{i \in
  [1 \dd n'] : \T[\slex_i \dd n] \prec \Qsuf\}|$ and $(\lrank(u) \dd
  \rrank(u)] = \{i \in [1 \dd n'] : \Qsuf\text{ is a prefix of
  }\T[\slex_i \dd n]\}$ (in particular, we have $\{\slex_i\}_{i \in
    (\lrank(u) \dd \rrank(u)]} = \Occ(\Qsuf,\T)$).  Therefore, letting
  $\delta_1 = \rank{W}{\revstr{X}}{\lrank(u)}$ and $\delta_2 =
  \rank{W}{\revstr{X}}{\rrank(u)}$, by \cref{lm:pm-nonperiodic}, it
  holds $\repr(v) = (\LB(Q, \T), \allowbreak \UB(Q, \T)) = (b_X +
  \delta_1, b_X + \delta_2) = \pseudoInvTSSS{X}{u}$.
\end{proof}

\begin{proposition}\label{pr:st-nonperiodic-invmap}
  Let $u$ be a node of $\TSSS$. Given the data structure from
  \cref{sec:st-nonperiodic-ds}, a pointer to $u$, and the value
  $\Int(X)$ for some $X \in \Alphabet^{\leq 3\tau - 1}$, we can in
  $\bigO(\log^{\epsilon} n)$ time compute the pair
  $\pseudoInvTSSS{X}{u}$.
\end{proposition}
\begin{proof}
  First, using the $\LTrange$ lookup table (stored as part of
  $\STCore(\T)$; see \cref{sec:st-core}), we compute $b_X = \LB(X,
  \T)$. Using the lookup table $\LTrev$ stored in the structure from
  \cref{sec:st-nonperiodic-ds}, we compute $\revstr{X}$.  In
  $\bigO(1)$ we obtain $z_1 = \lrank(u)$ and $z_2 = \rrank(u)$
  (\cref{pr:compact-trie}).  Finally, using \cref{th:wavelet-tree},
  in $\bigO(\log^{\epsilon} n)$ time we compute $\delta_1 =
  \rank{W}{\revstr{X}}{z_1}$ and $\delta_2 =
  \rank{W}{\revstr{X}}{z_2}$, and return $\pseudoInvTSSS{X}{u} = (b_X
  + \delta_1, b_X + \delta_2)$.
\end{proof}

\subsubsection{Implementation of
  \texorpdfstring{$\LCA(u, v)$}{LCA(u, v)}}\label{sec:st-nonperiodic-lca}

\begin{lemma}\label{lm:st-nonperiodic-lca}
  Let $v_1$ and $v_2$ be explicit nodes of $\ST$ such that $\LCA(v_1,
  v_2)$ is nonperiodic and it holds $\sdepth(\LCA(v_1, v_2)) \geq 3\tau -
  1$. Then, $v_1$ and $v_2$ are nonperiodic and it holds $\sdepth(v_1)
  \geq 3\tau - 1$ and $\sdepth(v_2) \geq 3\tau - 1$.  Moreover,
  \[
    \mapToTSSS(\LCA(v_1, v_2)) = \LCA(\mapToTSSS(v_1),
    \mapToTSSS(v_2)).
  \]
\end{lemma}
\begin{proof}
  Denote $v = \LCA(v_1, v_2)$ and $Y = \str(v)[1 \dd 3\tau {-} 1]$.
  By the assumption, we have $\per(Y) > \tfrac{1}{3}\tau$. Since by
  definition of $v$, the string $Y$ is a prefix of $\str(v_1)$ and
  $\str(v_2)$, we thus obtain that $v_1$ and $v_2$ are nonperiodic and
  it holds $\sdepth(v_1) \geq 3\tau - 1$ and $\sdepth(v_2) \geq 3\tau
  - 1$. Thus, $u_1 = \mapToTSSS(v_1)$ and $u_2 = \mapToTSSS(v_2)$
  are well-defined (see \cref{sec:st-nonperiodic-nav}).

  Let $u = \LCA(u_1, u_2)$, $\ell' = \sdepth(u)$, and $\ell =
  \sdepth(v)$.  By \cref{ob:lca}, we have $\ell = \lcp(\str(v_1),
  \str(v_2)), \ell' = \lcp(\str(u_1), \str(u_2))$, $\str(v) =
  \str(v_1)[1 \dd \ell]$, and $\str(u) = \str(u_1)[1 \dd \ell']$.  Let
  now $X \in \D$ be a prefix of $\str(v)$ (such $X$ exists and is
  unique since $\per(Y) > \tfrac{1}{3}\tau$ and since $\str(v)$ being
  a substring of $\T$ implies $\Occ(Y, \T) \neq \emptyset$; see also
  \cref{sec:sa-nonperiodic-ds}). By definition of $\mapToTSSS(v_1)$ and
  $\mapToTSSS(v_2)$, we have $\str(v_1) = X[1 \dd \deltatext] \cdot
  \str(u_1)$ and $\str(v_2) = X[1 \dd \deltatext] \cdot \str(u_2)$,
  where $\deltatext = |X| - 2\tau$.  This implies $\ell = \deltatext +
  \ell'$, and consequently, $\str(v) = \str(v_1)[1 \dd \ell] = X[1 \dd
  \deltatext] \cdot \str(u_1)[1 \dd \ell - \deltatext] = X[1 \dd
  \deltatext] \cdot \str(u_1)[1 \dd \ell'] = X[1 \dd \deltatext] \cdot
  \str(u)$.  Since $v$ is an explicit node of $\ST$, and no two nodes
  of $\TSSS$ have the same value of $\str$, we therefore obtain
  $\mapToTSSS(v) = u$.
\end{proof}

\begin{proposition}\label{pr:st-nonperiodic-lca}
  Let $v_1$ and $v_2$ be explicit nodes of $\ST$ such that $\LCA(v_1,
  v_2)$ is nonperiodic and satisfies $\sdepth(\LCA(v_1, v_2)) \geq
  3\tau - 1$.  Given the data structure from
  \cref{sec:st-nonperiodic-ds} and the pairs $\repr(v_1)$ and
  $\repr(v_2)$, we can in $\bigO(\log^{\epsilon} n)$ time compute
  $\repr(\LCA(v_1, v_2))$.
\end{proposition}
\begin{proof}
  Denote $v = \LCA(v_1, v_2)$.  By \cref{lm:st-nonperiodic-lca}, $v_1$
  and $v_2$ are nonperiodic and it holds $\sdepth(v_1) \geq 3\tau - 1$
  and $\sdepth(v_2) \geq 3\tau - 1$. Thus, $u_1 = \mapToTSSS(v_1)$ and
  $u_2 = \mapToTSSS(v_2)$ are well-defined (see
  \cref{sec:st-nonperiodic-nav}).  Using \cref{pr:st-nonperiodic-map},
  in $\bigO(\log^{\epsilon} n)$ time we compute pointers to $u_1$ and
  $u_2$. Next, using the representation of $\TSSS$ stored as part of
  the structure in \cref{sec:st-nonperiodic-ds}, and
  \cref{pr:compact-trie}, in $\bigO(1)$ time we compute a pointer to
  $u = \LCA(u_1, u_2)$. By \cref{lm:st-nonperiodic-lca}, it holds
  $\mapToTSSS(v) = u$.  We exploit this connection to compute
  $\repr(v)$.  Since $v$ is nonperiodic, letting $Y = \str(v)[1 \dd
  3\tau {-} 1]$, it holds $\per(Y) > \tfrac{1}{3}\tau$.  Since
  $\str(v)$ is a substring of $\T$, we have $\Occ(Y, \T) \neq
  \emptyset$. Together with $\per(Y) > \tfrac{1}{3}\tau$, this implies
  (see \cref{sec:sa-nonperiodic-ds}) that there exists a unique prefix
  $X \in \D$ of $\str(v)$.  We compute $X$ as follows. First, using
  \cref{pr:st-core-map}, in $\bigO(1)$ time we compute pointers to
  $u'_1 = \mapToShort(v_1)$ and $u'_2 = \mapToShort(v_2)$ of
  $\Tshort$. Then, using the $\LCA$ structure for $\Tshort$, we
  compute in $\bigO(1)$ time the pointer to node $u' = \LCA(u'_1,
  u'_2)$ of $\Tshort$. By \cref{lm:st-core-lca}, we now have
  $\mapToShort(v) = u'$. Moreover, by $\sdepth(v) \geq 3\tau - 1$ and
  \cref{lm:st-core-map}\eqref{lm:st-core-map-it-1}, $\str(u') =
  Y$. Using $\LTD$ on $Y$, in $\bigO(1)$ time we thus obtain $X$.
  Using \cref{pr:st-nonperiodic-invmap}, in $\bigO(\log^{\epsilon} n)$
  time, we then compute the pair $(b, e) = \pseudoInvTSSS{X}{u}$. As
  noted above, $\mapToTSSS(v) = u$. By
  \cref{lm:st-nonperiodic-invmap}, we therefore have $\repr(v) = (b,
  e)$.
\end{proof}

\subsubsection{Implementation of
  \texorpdfstring{$\child(v, c)$}{child(v, c)}}\label{sec:st-nonperiodic-child}

\begin{lemma}\label{lm:st-nonperiodic-child}
  Let $c \in \Alphabet$ and $v$ be an explicit nonperiodic internal
  node of $\ST$ satisfying $\sdepth(v) \geq 3\tau - 1$. Let $u =
  \mapToTSSS(v)$. If $\child(u,c) = \nil$ then $\child(v,c) =
  \nil$. Otherwise, letting $u' = \child(u,c)$, it holds
  \vspace{1.5ex}
  \[
    \repr(\child(v,c)) =
    \begin{cases}
      (b, e) & \text{if $b \neq e$},\\
      (0, 0)   & \text{otherwise},
    \end{cases}
    \vspace{1.5ex}
  \]
  where $(b,e) = \pseudoInvTSSS{X}{u'}$ and $X \in \D$ is a prefix of
  $\str(v)$.
\end{lemma}
\begin{proof}
  Denote $\Pat = \str(v)c$, $\deltatext = |X| - 2\tau$, and $\Pat' =
  \Pat(\deltatext \dd |\Pat|]$.  Observe that since $\str(v)$ is
  nonperiodic and it holds $\sdepth(v) \geq 3\tau - 1$, $\Pat$ is also
  nonperiodic and satisfies $|\Pat| \geq 3\tau - 1$.  Let $(\bpre,
  \epre)$ be such that $\bpre = |\{i \in [1 \dd n'] : \T[\slex_i \dd
  n] \prec \Pat'\}|$ and $(\bpre \dd \epre] = \{i \in [1 \dd n'] :
  \Pat'\text{ is a prefix of }\T[\slex_i \dd n]\}$.  Recall now that
  $u$ satisfies $\str(u) = \str(v)(\deltatext \dd |\str(v)|]$, or
  equivalently, $\str(u)c = \Pat'$.  By definition of $\TSSS$ and
  $\child(u, c)$, we thus obtain that $\child(u, c) = \nil$ implies
  $\epre - \bpre = 0$.  Consequently, by \cref{lm:pm-nonperiodic}, it
  holds $|\Occ(\Pat, \T)| = \UB(\Pat, \T) - \LB(\Pat, \T) = (b_X +
  \rank{W}{\revstr{X}}{\epre}) - (b_X + \rank{W}{\revstr{X}}{\bpre}) =
  \rank{W}{\revstr{X}}{\bpre} - \rank{W}{\revstr{X}}{\bpre} = 0$, and
  thus $\child(v, c) = \nil$.

  Let us now assume $\child(u, c) = u' \neq \nil$. By definition of
  $\pseudoInvTSSS{X}{u'}$, we then have $(b, e) = (b_X +
  \rank{W}{\revstr{X}}{\lrank(u')}, b_X +
  \rank{W}{\revstr{X}}{\rrank(u')})$, where $b_X = \LB(X, \T)$.  By
  definition of $\TSSS$ and $\child(u, c)$, however, we also have
  $\bpre = \lrank(u')$ and $\epre = \rrank(u')$.  Thus, by
  \cref{lm:pm-nonperiodic},
  \begin{align*}
    (\LB(\Pat, \T), \UB(\Pat, \T))
      &= (b_X + \rank{W}{\revstr{X}}{\bpre},
          b_X + \rank{W}{\revstr{X}}{\epre})\\
      &= (b_X + \rank{W}{\revstr{X}}{\lrank(u')},
          b_X + \rank{W}{\revstr{X}}{\rrank(u')})\\
      &= (b, e).
  \end{align*}
  By the above, if $b \neq e$, then $\Occ(\Pat, \T) \neq
  \emptyset$. This implies $\child(v, c) \neq \nil$ and
  $\repr(\child(v, c)) = (\LB(\Pat, \T), \UB(\Pat, \T))$. We thus
  indeed have $\repr(\child(v, c)) = (b, e)$.  Otherwise (i.e.,
  if $b = e$), by the above we have $\Occ(\Pat, \T) = \emptyset$. This
  implies $\child(v, c) = \nil$ and hence indeed we also have
  $\repr(\child(v, c)) = (0, 0)$.
\end{proof}

\begin{remark}\label{rm:st-nonperiodic-child}
  Note that even though in the above result we have $\mapToTSSS(v) =
  u$ and $\child(u, c)$ contains information used to determine
  $\child(v, c)$, it does not necessarily hold that
  $\mapToTSSS(\child(v, c)) = \child(u, c)$. The procedure is
  nevertheless correct, because such one-to-one correspondence is not
  required.  The details of this mapping, however, become relevant for
  the $\WA(v, d)$ operation and are explained in detail in the proof
  of \cref{lm:st-nonperiodic-wa}.
\end{remark}

\begin{proposition}\label{pr:st-nonperiodic-child}
  Let $v$ be an explicit nonperiodic internal node of $\ST$ satisfying
  $\sdepth(v) \geq 3\tau - 1$.  Given the data structure from
  \cref{sec:st-nonperiodic-ds}, $\repr(v)$, and $c \in \Alphabet$, in
  $\bigO(\log^{\epsilon} n)$ time we can compute $\repr(\child(v,
  c))$.
\end{proposition}
\begin{proof}
  First, using \cref{pr:st-nonperiodic-map}, in $\bigO(\log^{\epsilon}
  n)$ time we compute a pointer to $u = \mapToTSSS(v)$. Then, using
  the representation of $\TSSS$ stored as part of the structure in
  \cref{sec:st-nonperiodic-ds}, and \cref{pr:compact-trie}, in
  $\bigO(\log \log n)$ time we check if $\child(u, c) = \nil$. If so, by
  \cref{lm:st-nonperiodic-child} we have $\child(v, c) = \nil$, and
  thus we return $\repr(\child(v, c)) = (0, 0)$. Otherwise ($\child(u,
  c) \neq \nil$), we obtain a pointer to $u' = \child(u, c)$.  Next,
  using \cref{pr:st-core-map} in $\bigO(1)$ time we compute a pointer
  to $u'' = \mapToShort(v)$. By
  \cref{lm:st-core-map}\eqref{lm:st-core-map-it-1}, $\sdepth(v) \geq
  3\tau - 1$ implies $\sdepth(u'') = 3\tau - 1$ and $\str(u'') =
  \str(v)[1 \dd 3\tau {-} 1]$.  We obtain $Y = \str(u'')$ (stored with
  $u''$) in $\bigO(1)$ time.  Using $\LTD$ on $Y$, in $\bigO(1)$ time
  we then compute a prefix $X \in \D$ of $Y$ (see
  \cref{sec:st-nonperiodic-nav}).  Finally, using
  \cref{pr:st-nonperiodic-invmap}, in $\bigO(\log^{\epsilon} n)$ time
  we compute the pair $(b, e) = \pseudoInvTSSS{X}{u'}$. If $b = e$
  then by \cref{lm:st-nonperiodic-child} we return $\repr(\child(v,
  c)) = (0, 0)$. Otherwise, we return $\repr(\child(v, c)) = (b, e)$.
\end{proof}

\subsubsection{Implementation of
  \texorpdfstring{$\pred(v, c)$}{pred(v, c)}}\label{sec:st-nonperiodic-pred}

\begin{lemma}\label{lm:st-nonperiodic-pred}
  Let $c \in \Alphabet$ and $v$ be an explicit nonperiodic internal
  node of $\ST$ satisfying $\sdepth(v) \geq 3\tau - 1$. Let $u =
  \mapToTSSS(v)$. If $\pred(u,c) = \nil$ then $\LB(\str(v)c, \T) =
  \LB(\str(v), \T)$.  Otherwise, letting $u' = \pred(u,c)$, it holds
  \vspace{0.5ex}
  \[
    \LB(\str(v)c, \T) = e,
    \vspace{0.5ex}
  \]
  where $(b,e) = \pseudoInvTSSS{X}{u'}$ and $X \in \D$ is a prefix of
  $\str(v)$.
\end{lemma}
\begin{proof}
  We start by characterizing $\LB(\str(v), \T)$ using
  \cref{lm:pm-nonperiodic} for pattern $\str(v)$. First, note that $v$
  is nonperiodic and it holds $\sdepth(v) \geq 3\tau - 1$. On the
  other hand, by definition, we have $\str(u) = \str(v)(\deltatext \dd
  |\str(v)|]$, where $\deltatext = |X| - 2\tau$. Finally, by
  definition of $\TSSS$, we have $|\{i \in [1 \dd n'] : \T[\slex_i \dd
  n] \prec \str(v)(\deltatext \dd |\str(v)|]\}| = \lrank(u)$.  Thus,
  by \cref{lm:pm-nonperiodic}, we have $\LB(\str(v), \T) = b_X +
  \rank{W}{\revstr{X}}{\lrank(u)}$, where $b_X = \LB(X, \T)$.

  Denote $\Pat = \str(v)c$ and $\Pat' = \Pat(\deltatext \dd
  |\Pat|]$. Since $\str(v)$ is nonperiodic and satisfies $|\str(v)|
  \geq 3\tau - 1$, $\Pat$ is also nonperiodic and it holds $|\Pat|
  \geq 3\tau - 1$. Note also that by $\str(u) = \str(v)(\deltatext \dd
  |\str(v)|]$, we have $\str(u)c = \Pat'$.

  Let us first assume $\pred(u, c) = \nil$. By definition, this
  implies that $|\{i \in [1 \dd n'] : \T[\slex_i \dd n] \prec
  \str(u)c\}| = \lrank(u)$.  Equivalently, by $\str(u)c = \Pat'$,
  $|\{i \in [1 \dd n'] : \T[\slex_i \dd n] \prec \Pat'\}| =
  \lrank(u)$. By \cref{lm:pm-nonperiodic} for pattern $\Pat =
  \str(v)c$ we thus obtain $\LB(\str(v)c, \T) = b_X +
  \rank{W}{\revstr{X}}{\lrank(u)}$. Since above we also established
  that $\LB(\str(v), \T) = b_X + \rank{W}{\revstr{X}}{\lrank(u)}$, we
  have thus proved that $\pred(u, c) = \nil$ implies $\LB(\str(v)c,
  \T) = \LB(\str(v), \T)$.

  Let us now assume $\pred(u, c) = u' \neq \nil$. Observe that by
  definition of $\pred(u, c)$, this implies $|\{i \in [1 \dd n'] :
  \T[\slex_i \dd n] \prec \str(u)c\}| = \rrank(u')$. By
  \cref{lm:pm-nonperiodic} applied for pattern $\Pat = \str(v)c$, we
  thus obtain $\LB(\str(v)c, \T) = b_X +
  \rank{W}{\revstr{X}}{\rrank(u')}$. On the other hand, observe that
  by definition of $(b, e) = \pseudoInvTSSS{X}{u'}$, we have $e = b_X
  + \rank{W}{\revstr{X}}{\rrank(u')}$. We thus obtain $\LB(\str(v)c,
  \T) = e$.
\end{proof}

\begin{proposition}\label{pr:st-nonperiodic-pred}
  Let $v$ be an explicit nonperiodic internal node of $\ST$ satisfying
  $\sdepth(v) \geq 3\tau - 1$.  Given the data structure from
  \cref{sec:st-nonperiodic-ds}, $\repr(v)$, and $c \in \Alphabet$, in
  $\bigO(\log^{\epsilon} n)$ time we can compute $\LB(\str(v)c, \T)$.
\end{proposition}
\begin{proof}
  First, using \cref{pr:st-nonperiodic-map}, in $\bigO(\log^{\epsilon}
  n)$ time we compute a pointer to $u = \mapToTSSS(v)$. Then, using
  the representation of $\TSSS$ stored as part of the structure in
  \cref{sec:st-nonperiodic-ds}, and \cref{pr:compact-trie}, in
  $\bigO(\log \log n)$ time we check if $\pred(u, c) = \nil$. If so,
  by \cref{lm:st-nonperiodic-pred} we have $\LB(\str(v)c, \T) =
  \LB(\str(v), \T)$ and hence we return $\LB(\str(v), \T)$ (given as
  input) as the result. Otherwise ($\pred(u, c) \neq \nil$), we obtain
  a pointer to $u' = \pred(u, c)$.  Next, using \cref{pr:st-core-map}
  in $\bigO(1)$ time we compute a pointer to $u'' =
  \mapToShort(v)$. By
  \cref{lm:st-core-map}\eqref{lm:st-core-map-it-1}, $\sdepth(v) \geq
  3\tau - 1$ implies $\sdepth(u'') = 3\tau - 1$ and $\str(u'') =
  \str(v)[1 \dd 3\tau {-} 1]$.  We obtain $Y = \str(u'')$ (stored with
  $u''$) in $\bigO(1)$ time.  Using $\LTD$ on $Y$, in $\bigO(1)$ time
  we then compute a prefix $X \in \D$ of $Y$ (see
  \cref{sec:st-nonperiodic-nav}).  Finally, using
  \cref{pr:st-nonperiodic-invmap}, in $\bigO(\log^{\epsilon} n)$ time
  we compute the pair $(b, e) = \pseudoInvTSSS{X}{u'}$. By
  \cref{lm:st-nonperiodic-pred}, it holds $\LB(\str(v)c, \T) = e$.  We
  thus return $e$ as the answer.
\end{proof}

\subsubsection{Implementation of
  \texorpdfstring{$\WA(v,d)$}{WA(v, d)}}\label{sec:st-nonperiodic-wa}

\begin{lemma}\label{lm:st-nonperiodic-wa}
  Let $v$ be an explicit nonperiodic node of $\ST$ and $d$ be an
  integer satisfying $3\tau - 1 \leq d \leq |\str(v)|$. Then, letting
  $u = \mapToTSSS(v)$, it holds
  \[
    \repr(\WA(v, d)) = \pseudoInvTSSS{X}{\widehat{u}},
  \]
  where $X \in \D$ is a prefix of $\str(v)$, $\deltatext = |X| -
  2\tau$, and $\widehat{u} = \WA(u, d - \deltatext)$.
\end{lemma}
\begin{proof}
  Denote $f^{(0)}(x) = x$ and $f^{(i)}(x) = f(f^{(i-1)}(x))$ for $i
  \in \Zp$. Let
  \vspace{-1ex}
  \begin{align*}
    \mathcal{V} &:=
      \{\parent^{(i)}(v) : i \in\Zz\text{ and }
      \sdepth(\parent^{(i)}(v)) \geq |X|\}\text{ and }\\
    \mathcal{U} &:=
      \{\parent^{(i)}(u) : i \in \Zz\text{ and }
      \sdepth(\parent^{(i)}(u)) \geq 2\tau\}
  \end{align*}
  For any $v' \in \mathcal{V}$, the node $u' = \mapToTSSS(v')$
  satisfies $\str(v') = X[1 \dd \deltatext] \cdot \str(u')$.  In
  particular, $\str(u) = \str(v)(\deltatext \dd |\str(v)|]$.  Since
  for any $v' \in \mathcal{V}$, $\str(v') = \str(v)[1 \dd
  |\str(v')|]$, we thus obtain that for $u' = \mapToTSSS(v')$ it
  holds $\str(u') = \str(v')(\deltatext \dd |\str(v')|] =
  \str(v)(\deltatext \dd |\str(v')|] = \str(u)[1 \dd |\str(u')|]$,
  i.e., $u'$ is an ancestor of $u$.  Moreover, $\sdepth(u') =
  |\str(v')| - \deltatext \geq |X| - \deltatext = 2\tau$.
  Consequently, $\mathcal{U}' := \{\mapToTSSS(v') : v' \in
  \mathcal{V}\}$ satisfies $\mathcal{U}' \sub \mathcal{U}$. Note also,
  that $v' \neq v''$ implies $\mapToTSSS(v') \neq \mapToTSSS(v'')$.

  For any $u' \in \mathcal{U}$, denote $(\fbeg(u'), \fend(u')) =
  \pseudoInvTSSS{X}{u'}$.  We prove the following property of
  $\mathcal{U}'$. Let $w, w' \in \mathcal{U}$ be such that $w' =
  \parent(w)$. We claim, that $(\fbeg(w), \fend(w)) \neq (\fbeg(w'),
  \fend(w'))$ implies $w' \in \mathcal{U}'$. Denote $Q' = \str(w')$
  and $Q = X[1 \dd \deltatext] \cdot Q'$. The proof consists of three
  steps:
  \begin{itemize}
  \item First, we show that it holds $\{\SA[i]\}_{i \in (s' \dd t']} =
    \Occ(Q, \T)$, where $s' = \fbeg(w')$ and $t' = \fend(w')$. By the
    above discussion, we have $\str(v) = X[1 \dd \deltatext] \cdot
    \str(u)$. Thus, $X(\deltatext \dd |X|]$ is a prefix of
    $\str(u)$. On the other hand, by $w' \in \mathcal{U}$, $\str(w')$
    is a prefix of $\str(u)$ and we have $|\str(w')| \geq
    2\tau$. Consequently, $X(\deltatext \dd |X|]$ is a prefix of
    $Q'$. As noted in the proof of \cref{lm:st-nonperiodic-invmap}, by
    the consistency of $\S$ we then have $\{\slex_i\}_{i \in
    (\lrank(w') \dd \rrank(w')]} = \Occ(Q', \T)$ and consequently
    $\{\SA[i]\}_{i \in (s' \dd t']} = \Occ(X[1 \dd \deltatext] \cdot
    Q', \T) = \Occ(Q, \T)$.
  \item Second, we prove that there exists a node $v'$ in $\ST$ such
    that $\str(v') = Q$. First, note that $\{\SA[i]\}_{i \in (s' \dd
    t']} = \Occ(Q, \T)$ already implies that there exists some node
    $v'$ of $\ST$ such that $\repr(v') = (s', t')$ and $Q$ is a prefix
    of $\str(v')$. It thus remains to show that $\str(v') = Q$. For
    this, it suffices to show that there exists $c, c' \in \Alphabet$
    such that $c \neq c'$, $\Occ(Qc, \T) \neq \emptyset$, and
    $\Occ(Qc', \T) \neq \emptyset$. Observe that by $(\fbeg(w'),
    \fend(w')) \neq (\fbeg(w), \fend(w))$, there exists a child $w''
    \neq w$ of $w'$ such that $\fend(w'') > \fbeg(w'')$.  This yields
    an occurrence of $\str(w')$ preceded in $\T$ with $X[1 \dd
    \deltatext]$. The same holds for $\str(w)$, since $\fbeg(w') \leq
    \fbeg(u) < \fend(u) \leq \fend(w')$.  In other words, for $c =
    \str(w)[|Q'| + 1]$ and $c' = \str(w'')[|Q'| + 1]$ we have $c \neq
    c'$, $\Occ(Qc, \T) \neq \emptyset$, and $\Occ(Qc', \T) \neq
    \emptyset$.  This concludes the proof of $\str(v') = Q$.
  \item Finally, recall that by definition, the node $\mapToTSSS(v')$
    satisfies $\str(v') = X[1 \dd \deltatext] \cdot
    \str(\mapToTSSS(v'))$. Thus, by $\str(v') = Q = X[1 \dd
    \deltatext] \cdot Q'$ and $\str(w') = Q'$, we must have
    $\mapToTSSS(v') = w'$. This implies $w' \in \mathcal{U}'$.
  \end{itemize}

  We are now ready to prove the main claim.  Let $v' = \WA(v, d)$ and
  $v'' = \parent(v')$.  We then have $\sdepth(v'') < d \leq
  \sdepth(v')$. Moreover, by $|X| \leq 3\tau - 1 \leq d$, we have $v'
  \in \mathcal{V}$. Let $u' = \mapToTSSS(v')$.  Then, $u' \in
  \mathcal{U}'$. By the above discussion, we also have $d - \deltatext
  \leq \sdepth(\widehat{u}) \leq \sdepth(u')$.  By $3\tau - 1 \leq d$
  this implies $2\tau = |X| - \deltatext \leq 3\tau - 1 - \deltatext
  \leq d - \deltatext \leq \sdepth(\widehat{u})$, i.e., $\widehat{u}
  \in \mathcal{U}$.  Let $k \in \Zz$ be such that
  $\widehat{u} = \parent^{(k)}(u')$.  This implies that
  $\parent^{(i)}(u')\not\in \mathcal{U}'$ holds for $i \in [1 \dd k]$,
  since otherwise it would contradict $v' = \WA(v, d)$. If $k = 0$
  then we trivially have $(\fbeg(u'), \fend(u')) =
  (\fbeg(\widehat{u}), \fend(\widehat{u}))$. Otherwise, by (the
  contraposition of) the above property of $\mathcal{U}'$ we have
  \begin{align*}
    (\fbeg(u'), \fend(u'))
      &= (\fbeg(\parent(u')), \fend(\parent(u')))\\
      &= \dots\\
      &= (\fbeg(\parent^{(k)}(u')), \fend(\parent^{(k)}(u')))\\
      &= (\fbeg(\widehat{u}), \fend(\widehat{u})).
  \end{align*}
  By \cref{lm:st-nonperiodic-invmap}, we obtain $\repr(\WA(v, d)) =
  \repr(v') = \pseudoInvTSSS{X}{u'} = (\fbeg(u'), \fend(u')) =
  (\fbeg(\widehat{u}), \fend(\widehat{u})) =
  \pseudoInvTSSS{X}{\widehat{u}}$.
\end{proof}

\begin{proposition}\label{pr:st-nonperiodic-wa}
  Let $v$ be an explicit nonperiodic node of $\ST$ satisfying $3\tau -
  1 \leq |\str(v)|$. Given the data structure from
  \cref{sec:st-nonperiodic-ds}, $\repr(v)$, and an integer $d$
  satisfying $3\tau - 1 \leq d \leq |\str(v)|$, in
  $\bigO(\log^{\epsilon} n)$ time we can compute $\repr(\WA(v, d))$.
\end{proposition}
\begin{proof}
  First, using \cref{pr:st-nonperiodic-map}, in $\bigO(\log^{\epsilon}
  n)$ time we compute a pointer to $u = \mapToTSSS(v)$.  Next, using
  \cref{pr:st-core-map} in $\bigO(1)$ time we compute a pointer to $u'
  = \mapToShort(v)$. By
  \cref{lm:st-core-map}\eqref{lm:st-core-map-it-1}, $\sdepth(v) \geq
  3\tau - 1$ implies $\sdepth(u') = 3\tau - 1$ and $\str(u') =
  \str(v)[1 \dd 3\tau {-} 1]$.  We obtain $Y = \str(u')$ (stored with
  $u'$) in $\bigO(1)$ time.  Using $\LTD$ on $Y$, in $\bigO(1)$ time
  we then compute a prefix $X \in \D$ of $Y$ (see
  \cref{sec:st-nonperiodic-nav}).  Let $\deltatext = |X| - 2\tau$.
  Finally, using the representation of $\TSSS$ stored as part of the
  structure in \cref{sec:st-nonperiodic-ds}, and
  \cref{pr:compact-trie}, in $\bigO(\log \log n)$ time we compute a
  pointer to $\widehat{u} = \WA(u, d - \deltatext)$.  Using
  \cref{pr:st-nonperiodic-invmap}, in $\bigO(\log^{\epsilon} n)$ time
  we then compute and return $\pseudoInvTSSS{X}{\widehat{u}}$, which
  by \cref{lm:st-nonperiodic-wa} is equal to $\repr(\WA(v, d))$.
\end{proof}

\subsubsection{Construction Algorithm}\label{sec:st-nonperiodic-construction}

\begin{proposition}\label{pr:st-nonperiodic-construction}
  Given $\STCore(\T)$, we can in $\bigO(n \min(1, \log \sigma /
  \sqrt{\log n}))$ time and $\bigO(n / \log_{\sigma} n)$ working space
  augment it into a data structure from \cref{sec:st-nonperiodic-ds}.
\end{proposition}
\begin{proof}
  First, we combine
  \cref{pr:sa-core-construction,pr:sa-nonperiodic-construction}
  (recall that the packed representation of $\T$ is a component of
  $\STCore(\T)$) to construct in $\bigO(n \min(1, \log \sigma /
  \sqrt{\log n}))$ time and using $\bigO(n / \log_{\sigma} n)$ working
  space the data structure from \cref{sec:sa-nonperiodic-ds}. In
  particular, this constructs $(\slex_i)_{i \in [1 \dd n']}$. We
  then initialize $\ARRslex[i] = \slex_i$ for $i \in [1 \dd n']$
  and in $\bigO(n/\log_{\sigma} n)$ time construct $\TSSS$
  represented using \cref{pr:compact-trie}.
\end{proof}

\subsection{The Periodic Nodes}\label{sec:st-periodic}

In this section, we describe a data structure used to perform
operations on periodic nodes (see \cref{def:node-periodicity}) in
$\bigO(\log \log n)$ time.

The section is organized as follows. First, we introduce the
components of the data structure (\cref{sec:st-periodic-ds}).  We then
show how using this structure to implement some basic navigational
routines (\cref{sec:st-periodic-nav}). Next, we describe the query
algorithms for the fundamental operations
(\cref{sec:st-periodic-lca,sec:st-periodic-child,%
sec:st-periodic-pred,sec:st-periodic-wa}).  Finally, we show the
construction algorithm (\cref{sec:st-periodic-construction}).

\subsubsection{The Data Structure}\label{sec:st-periodic-ds}

\bfparagraph{Definitions}

Let $v$ be a periodic node of $\ST$. We define $\Lroot(v) :=
\Lroot(\str(v))$, $\pend{v} := \pend{\str(v)}$, $\Lhead(v) :=
\Lhead(\str(v))$, $\Lexp(v) := \Lexp(\str(v))$, $\Ltail(v) :=
\Ltail(\str(v))$, $\pendfull{v} := \pendfull{\str(v)}$, and $\type(v)
:= \type(\str(v))$.  Let $q = |\R'^{-}|$. Recall
(\cref{sec:sa-periodic-ds}) that $(\rlexm_i)_{i \in [1 \dd q]}$ is a
sequence containing all elements $k \in \R'^{-}$ sorted first
according to $\Lroot(k)$ and in case of ties, by $\T[\rendfull{k} \dd
n]$. Recall also (\cref{sec:pm-periodic-ds}) that $\Zset =
\{\rendfull{j} - |\Pow(\Lroot(j))|: j \in \R'^{-}\}$ and $\ARRzlex[1
\dd q]$ is an array defined by $\ARRzlex[i] = \rendfull{\rlexm_i} -
|\Pow(H_i)|$, where $H_i = \Lroot(\rlexm_i)$ and $\Pow(H_i) =
H_i^\infty[1\dd |H_i|\lceil\tfrac{\tau}{|H_i|}\rceil]$.  Let $\TZ$
denote the compact trie of the set $\{\T[i \dd n] : i \in \Zset\}$.

\bfparagraph{Components}

The data structure consists of two parts. The first part consists of
the following three components:

\begin{enumerate}
\item The index core $\STCore(\T)$ (\cref{sec:st-core}). It takes
  $\bigO(n / \log_{\sigma} n)$ space.
\item The first part of the structure from \cref{sec:sa-periodic-ds}
  using $\bigO(n / \log_{\sigma} n)$ space.
\item The compact trie $\TZ$ represented as in \cref{pr:compact-trie}
  (i.e., for the array $\ARRzlex[1 \dd q]$ defined as above). By $q =
  \bigO(n / \log_{\sigma} n)$ and \cref{pr:compact-trie}, it needs
  $\bigO(n/ \log_{\sigma} n)$ space.
\end{enumerate}

The second part of the structure consists of the symmetric
counterparts of the above components adapted according to
\cref{lm:lce} (see also \cref{sec:sa-periodic-ds})

In total, the data structure takes $\bigO(n / \log_{\sigma} n)$ space.

\subsubsection{Navigation Primitives}\label{sec:st-periodic-nav}

\bfparagraph{Mapping from $\ST$ to $\TZ$}

For any periodic explicit node $v$ of $\ST$ satisfying $\rend{v} \leq
|\str(v)|$ and $\type(v) = -1$, we define $\mapToTZ(v) = u$ as a node
of $\TZ$ satisfying $\str(u) = \Pow(H) \cdot \str(v)[\rendfull{v} \dd
|\str(v)|]$, where $H = \Lroot(v)$.

\begin{lemma}\label{lm:st-periodic-map}
  Let $v$ be a periodic explicit node $v$ of $\ST$ such that $\rend{v}
  \leq |\str(v)|$ and $\type(v) = -1$.
  \begin{enumerate}
  \item The node $\mapToTZ(v)$ is well-defined.
  \item Let $i_1 = \lrank(v) + 1$, $i_2 = \rrank(v)$, $y_1$ and $y_2$
    be such that $\rendfull{\rlexm_{y_1}} = \rendfull{\SA[i_1]}$ and
    $\rendfull{\rlexm_{y_2}} = \rendfull{\SA[i_2]}$ (respectively),
    $u_1$ and $u_2$ be the $y_1$th and $y_2$th leftmost leaf of $\TZ$
    (respectively), and $u = \LCA(u_1, u_2)$. Then, $\mapToTZ(v) = u$.
  \end{enumerate}
\end{lemma}
\begin{proof}
  Denote $s = \Lhead(v)$, $H = \Lroot(v)$, $p = |H|$, $Q = \str(v)$,
  and $\Qsuf = Q[\rendfull{Q} \dd |Q|]$. Note that by $\rendfull{Q}
  \leq \rend{Q} \leq |Q|$, it holds $\Qsuf \neq \emptystring$.

  1. We start by observing that by \cref{lm:pm-lce} and
  \cref{lm:pm-lce-3}\eqref{lm:pm-lce-3-it-2}, for every $i \in \Occ(Q,
  \T)$, it holds $i \in \R_{s,H}^{-}$ and $\rendfull{i}-i =
  \rendfull{Q}-1$. Note that this implies $\rendfull{i} \in
  \Occ(\Qsuf, \T)$.  To show that $\TZ$ has a node $u$ satisfying
  $\str(u) = \Pow(H) \cdot \Qsuf$, consider two cases:
  \begin{itemize}
  \item Assume that $v$ is a leaf. Let $i \in \Occ(\str(v), \T)$.
    By the above, $i \in \R_{H}^{-}$.  Let $j$
    be the smallest integer such that $[j \dd i] \sub \R$. It holds $j
    \in \R'$ and moreover, by \cref{lm:R-block}, $j \in \R_{H}^{-}$
    and $\rendfull{j} = \rendfull{i}$. Thus, by $\rendfull{i} \in
    \Occ(\Qsuf, \T)$ (see above), we have $\rendfull{j} \in
    \Occ(\Qsuf, \T)$. Finally, by $i \in \Occ(Q, \T)$ and $|Q| = n - i
    + 1$, we have $|\Qsuf| = |Q| - \rendfull{Q} + 1 = |Q| -
    (\rendfull{Q} - 1) = (n - i + 1) - (\rendfull{i} - i) = n -
    \rendfull{i} + 1 = n - \rendfull{j} + 1$.  We have thus shown that
    there exists $j \in \R'^{-}_{H}$ such that $\rendfull{j} \in
    \Occ(\Qsuf, \T)$ and $n-\rendfull{j}+1 = |\Qsuf|$. By definition
    of $\ARRzlex[1 \dd q]$ (see \cref{sec:pm-periodic-ds}) this
    implies that there exists a leaf $u$ of $\TZ$ such that $\str(u) =
    \Pow(H) \cdot \T[\rendfull{j} \dd n] = \Pow(H) \cdot \Qsuf$.
  \item Assume that $v$ is an internal node.  Consider the leftmost
    and the rightmost leaves $v_1$ and $v_2$ (respectively) in the
    subtree rooted in $v$.  Letting $i_1 \in \Occ(\str(v_1), \T)$ and
    $i_2 \in \Occ(\str(v_2), \T)$,
    we have $i_1 \neq i_2$ and $i_1, i_2 \in \Occ(Q,
    \T)$. Thus, $i_1, i_2 \in \R_{H}^{-}$ and $\rendfull{i_1}-i_1 =
    \rendfull{Q}-1 = \rendfull{i_2}-i_2$. Therefore, $\rendfull{i_1}
    \neq \rendfull{i_2}$ and, by \cref{lm:R-block}, $i_1$ and $i_2$
    are in different maximal contiguous blocks of positions from $\R$,
    i.e., letting $j_1$ (resp. $j_2$) be the smallest integer such
    that $[j_1 \dd i_1] \sub \R$ (resp.\ $[j_2 \dd i_2] \sub \R$), we
    have $j_1, j_2 \in \R'$ and $j_1 \neq j_2$.  By \cref{lm:R-block},
    it then holds $j_1, j_2 \in \R_{H}^{-}$, $\rendfull{j_1} =
    \rendfull{i_1}$, and $\rendfull{j_2} = \rendfull{i_2}$.  Thus, by
    $\rendfull{i_1}, \rendfull{i_2} \in \Occ(\Qsuf, \T)$ (following
    from $i_1, i_2 \in \Occ(Q, \T)$), we obtain $\rendfull{j_1},
    \rendfull{j_2} \in \Occ(\Qsuf, \T)$.  Next, we show
    $\LCE(\rendfull{j_1}, \rendfull{j_2}) = |\Qsuf|$. As noted
    earlier, $\rendfull{i_1} - i_1 = \rendfull{i_2} - i_2 =
    \rendfull{Q} - 1$. Thus, by $\LCE(i_1, i_2) = |\str(\LCA(v_1,
    v_2))| = |\str(v)| = |Q|$ and $\rendfull{Q} - 1 \leq \rend{Q} - 1
    < |Q|$, we have $|Q| = \LCE(i_1, i_2) = \rendfull{Q} - 1 +
    \LCE(\rendfull{i_1}, \rendfull{i_2})$.  Equivalently,
    $\LCE(\rendfull{i_1}, \rendfull{i_2}) = |Q| - \rendfull{Q} + 1 =
    |\Qsuf|$, which by $\rendfull{j_1} = \rendfull{i_1}$ and
    $\rendfull{j_2} = \rendfull{i_2}$ yields $\LCE(\rendfull{j_1},
    \rendfull{j_2}) = |\Qsuf|$.  We have thus shown that there
    exist distinct positions $j_1,j_2 \in \R'^{-}_{H}$ satisfying
    $\rendfull{j_1},\rendfull{j_2} \in \Occ(\Qsuf, \T)$ and
    $\LCE(\rendfull{j_1}, \rendfull{j_1}) = |\Qsuf|$.  By definition
    of $\ARRzlex[1 \dd q]$, this implies that there exists leaves $u_1$
    and $u_2$ of $\TZ$ such that $\str(u_1) = \Pow(H) \cdot
    \T[\rendfull{j_1} \dd n]$ and $\str(u_2) = \Pow(H) \cdot
    \T[\rendfull{j_2} \dd n]$ (see the proof of \cref{pr:pm-occ-s})
    and consequently, by \cref{ob:lca}, the node $u = \LCA(u_1,u_2)$
    satisfies $\str(u) = \Pow(H)\cdot \Qsuf$.
  \end{itemize}

  2. To show that $y_1$ and $y_2$ are well-defined note that for every
  $i \in \R^{-}$, letting $j$ be the smallest integer satisfying $[j
  \dd i] \sub \R$, we have $j \in \R'^{-}$ and (by \cref{lm:R-block})
  $\rendfull{j} = \rendfull{i}$. Consequently, since $\{\rlexm_i\}_{i
  \in [1 \dd q]} = \R'^{-}$, taking $y \in [1 \dd q]$ such that
  $\rlexm_y = j$, it holds $\rendfull{\rlexm_y} = \rendfull{i}$.
  Therefore, by $\SA[i_1], \SA[i_2] \in \Occ(Q, \T) \sub \R^{-}$,
  $y_1, y_2 \in [1 \dd q]$ are (uniquely) defined.

  We start by showing that $\str(u_1) = \Pow(H) \cdot
  \T[\rendfull{\SA[i_1]} \dd n]$ and $\str(u_2) = \Pow(H) \cdot
  \T[\rendfull{\SA[i_2]} \dd n]$. As noted in \cref{sec:pm-periodic-ds},
  the sequence $(\T[\ARRzlex[i] \dd n])_{i \in [1 \dd
  q]}$ is lexicographically sorted.  Thus, by definition of $u_1$
  and $u_2$, we have $\str(u_1) = \T[\ARRzlex[y_1] \dd n]$ and $\str(u_2) =
  \T[\ARRzlex[y_2] \dd n]$.  As also noted in \cref{sec:pm-periodic-ds},
  for every $i \in [1 \dd q]$, $\T[\ARRzlex[i] \dd n] =
  \Pow(H_i)\cdot \T[\rendfull{\rlexm_i} \dd n]$, where $H_i =
  \Lroot(\rlexm_i)$. Combining that with the assumptions
  $\rendfull{\rlexm_{y_1}} = \rendfull{\SA[i_1]}$ and
  $\rendfull{\rlexm_{y_2}} = \rendfull{\SA[i_2]}$, we therefore obtain
  \begin{align*}
    \str(u_1) &= \Pow(H_{y_1}) \cdot \T[\rendfull{\SA[i_1]} \dd n],\\
    \str(u_2) &= \Pow(H_{y_2}) \cdot \T[\rendfull{\SA[i_2]} \dd n].
  \end{align*}
  To obtain $\str(u_1) = \Pow(H) \cdot \T[\rendfull{\SA[i_1]} \dd n]$
  and $\str(u_2) = \Pow(H) \cdot \T[\rendfull{\SA[i_2]} \dd n]$ it
  thus remains to show $H_{y_1} = H_{y_2} = H$. To this end, we first
  note that by $\rendfull{\rlexm_{y_1}} = \rendfull{\SA[i_1]}$,
  (resp.\ $\rendfull{\rlexm_{y_2}} = \rendfull{\SA[i_2]}$), and
  \cref{lm:R-block,lm:efull}, the positions $\rendfull{\rlexm_{y_1}}$
  and $\SA[i_1]$ (resp.\ $\rendfull{\rlexm_{y_2}}$ and $\SA[i_2]$)
  belong to the same contiguous block of elements from $\R$.  Next, by
  $|Q| \geq 3\tau - 1$ and $\SA[i_1] \in \Occ(Q, \T)$, we obtain
  $\lcp(\T[\SA[i_1] \dd n], Q) \geq 3\tau - 1$. Thus, by combining
  \cref{lm:R-block} and \cref{lm:pm-lce}, we have $H_{y_1} =
  \Lroot(\rlexm_{y_1}) = \Lroot(\SA[i_1]) = \Lroot(Q) = \Lroot(v) =
  H$. Analogously, $H_{y_2} = H$.

  By the above and \cref{ob:lca}, we thus have $\str(u) =
  \str(\LCA(u_1,u_2)) = \Pow(H) \cdot \T[\rendfull{\SA[i_1]} \dd
  \rendfull{\SA[i_1]}+\ell)$, where $\ell = \LCE(\rendfull{\SA[i_1]},
  \rendfull{\SA[i_2]})$. Moreover, by $\SA[i_1] \in \Occ(Q, \T)$, we
  have $\rendfull{\SA[i_1]} \in \Occ(\Qsuf, \T)$. Thus, it remains to
  show that $\ell = |\Qsuf|$. For this, recall that by $\SA[i_1],
  \SA[i_2] \in \Occ(Q, \T)$ we also have $\rendfull{\SA[i_1]} -
  \SA[i_1] = \rendfull{\SA[i_2]} - \SA[i_2] =
  \rendfull{Q}-1$. Therefore, by $\rendfull{Q} - 1 \leq \rend{Q} - 1 <
  |Q|$, we have $|Q| = \LCE(\SA[i_1], \SA[i_2]) = \rendfull{Q} - 1 +
  \LCE(\rendfull{\SA[i_1]}, \rendfull{\SA[i_2]})$. Equivalently,
  $\LCE(\rendfull{\SA[i_1]}, \rendfull{\SA[i_2]}) = |Q| - \rendfull{Q}
  + 1 = |\Qsuf|$.  We therefore obtained $\str(u) = \Pow(H) \cdot
  \Qsuf$, i.e., $\mapToTZ(v) = u$.
\end{proof}

\begin{proposition}\label{pr:st-periodic-map}
  Let $v$ be a periodic explicit node $v$ of $\ST$ satisfying
  $\rend{v} \leq |\str(v)|$ and $\type(v) = -1$.  Given the data
  structure from \cref{sec:st-periodic-ds} and $\repr(v)$, we can in
  $\bigO(\log \log n)$ time compute the pointer to the node
  $\mapToTZ(v)$.
\end{proposition}
\begin{proof}
  Denote $(b, e) = \repr(v)$, $i_1 = b + 1$, and $i_2 = e$. First,
  using \cref{pr:sa-delta-a} in $\bigO(1)$ time we compute the
  $\Lexp(\SA[i_1])$ and $\deltas(\SA[i_1])$. Next, as explained in the
  proof of \cref{pr:sa-delta-s}, given $i_1$, $\Lexp(\SA[i_1])$, and
  $\deltas(\SA[i_1])$, in $\bigO(\log \log n)$ time we compute $y_1
  \in [1 \dd q]$ satisfying $\rendfull{\rlexm_{y_1}} =
  \rendfull{\SA[i_1]}$. Note that \cref{pr:sa-delta-s} requires that
  $\SA[i_1] \in \R^{-}$, which holds by $\SA[i_1] \in \Occ(Q, \T)$,
  where $Q = \str(v)$ (see the proof of \cref{lm:st-periodic-map}).
  Analogously we compute $y_2 \in [1 \dd q]$ satisfying
  $\rendfull{\rlexm_{y_2}} = \rendfull{\SA[i_2]}$.  Then, using
  \cref{pr:compact-trie} in $\bigO(1)$ time we compute the pointers to
  the $y_1$th and $y_2$th leftmost leaves $u_1$ and $u_2$
  (respectively) of $\TZ$.  Then, again using \cref{pr:compact-trie},
  in $\bigO(1)$ time we compute and return the pointer to $u =
  \LCA(u_1, u_2)$. By \cref{lm:st-periodic-map}, it holds $\mapToTZ(v)
  = u$.
\end{proof}

\bfparagraph{Mapping from $\TZ$ to $\ST$}

Let $u$ be a node of $\TZ$ such that there exists $H \in \Lroots$ for
which $\Pow(H)$ is a prefix of $\str(u)$.  For any $\ell \geq 0$, we
define $\pseudoInvTZ{\ell}{u}$ as follows.  If, letting $s = \ell
\bmod |H|$ and $k = \lfloor \frac{\ell}{|H|} \rfloor$, it holds
$\mathsf{P} := \{j \in \R_{s,H}^{-} : \Lexp(j) = k\} = \emptyset$,
then $\pseudoInvTZ{\ell}{u} := (0, 0)$. Otherwise (i.e., $\mathsf{P}
\neq \emptyset$), letting $b_{\mathsf{P}}, e_{\mathsf{P}} \in [0 \dd
n]$ be such that $\{\SA[i]\}_{i \in (b_{\mathsf{P}} \dd
e_{\mathsf{P}}]} = \mathsf{P}$, $b_H, e_H \in [0 \dd q]$ be such that
$\{\rlexm_i\}_{i \in (b_H \dd e_H]} = \R'^{-}_{H}$, and letting $z_1 =
\lrank(u)$ and $z_2 = \rrank(u)$, we define $\pseudoInvTZ{\ell}{u} :=
(e_{\mathsf{P}} - c_1, e_{\mathsf{P}} - c_2)$, where
\begin{align*}
  c_1 &:= \rcount{\ARRnontail}{\ell}{e_H} -
          \rcount{\ARRnontail}{\ell}{z_1},\\
  c_2 &:= \rcount{\ARRnontail}{\ell}{e_H} -
          \rcount{\ARRnontail}{\ell}{z_2}.
\end{align*}

\begin{remark}
  To see that $H$ is well-defined, recall (see the proof of
  \cref{pr:pm-occ-s}) that $\{\Pow(H)\}_{H \in \Lroots}$ is
  prefix-free. Thus, at most one element of $\{\Pow(H)\}_{H \in
  \Lroots}$ can be a prefix of $\str(u)$. Furthermore, since $X \neq
  Y$ implies $\Pow(X) \neq \Pow(Y)$, $\Pow(H)$ uniquely identifies
  $H$.

  To see that $b_{\mathsf{P}}$ and $e_{\mathsf{P}}$ are well-defined,
  recall that by \cref{lm:lce}, if $\mathsf{P} \neq \emptyset$, then
  all positions in $\mathsf{P}$ occupy a contiguous block in $\SA$
  (see also the proof of \cref{pr:isa-delta-a}).

  Finally, to show that $b_H$ and $e_H$ are well-defined, note that
  $\Pow(H)$ being a prefix of $\str(u)$ implies, by definition of
  $\TZ$, that there exists $i \in [1 \dd q]$ such that $H =
  \Lroot(\rlexm_i)$. Recall (see the proof of \cref{pr:isa-delta-s}),
  that for any $i, i' \in [1 \dd q]$, $i < i'$ implies
  $\Lroot(\rlexm_i) \preceq \Lroot(\rlexm_{i'})$.  Thus, there exists
  a unique $(b_H, e_H)$ (with $0 \leq b_H < e_H \leq q$) such that
  $\{\rlexm_i\}_{i \in (b_H \dd e_H]} = \R'^{-}_H$.
\end{remark}

\begin{remark}
  Note that similarly as for $\mapToTSSS$ (see
  \cref{sec:st-nonperiodic-nav}), the mapping $\mapToTZ$ is not
  necessarily injective, and hence it may not have an inverse (see
  also \cref{rm:st-nonperiodic-invmap}). To perform the mapping from
  $\TZ$ to $\ST$, we will use the above function. Although it is
  well-defined for every $\ell$ and $u$ (specified as above), its
  value is not always meaningful. Below we show a simple but useful
  condition where it is, and in the following sections we show the
  more subtle uses.
\end{remark}

\begin{lemma}\label{lm:st-periodic-nav-aux}
  Let $\Pat \in \Alphabet^m$ be a periodic pattern satisfying
  $\rend{\Pat} \leq m$ and $\type(\Pat) = -1$.  Denote $H =
  \Lroot(\Pat)$, $s = \Lhead(\Pat)$, $k = \Lexp(\Pat)$, $\ell =
  \rendfull{\Pat} - 1$, and $\Pat' = \Pow(H) \cdot \Pat(\ell \dd m]$.
  Assume that $\mathsf{P} := \{j \in \R^{-}_{s,H} : \Lexp(j) = k\}
  \neq \emptyset$ and let $b_{\mathsf P}, e_{\mathsf P} \in [0 \dd n]$
  be such that $\{\SA[i]\}_{i \in (b_{\mathsf{P}} \dd e_{\mathsf{P}}]}
  = \mathsf{P}$ and $b_{H}, e_{H} \in [0 \dd q]$ be such that
  $\{\rlexm_i\}_{i \in (b_H \dd e_H]} = \R'^{-}_{H}$. Finally, let
  $(\bpre, \epre)$ be such that $\bpre = |\{i \in [1 \dd q] :
  \T[\ARRzlex[i] \dd n] \prec \Pat'\}|$ and $(\bpre \dd \epre] = \{i
  \in [1 \dd q] : \Pat'\text{ is a prefix of }\T[\ARRzlex[i] \dd
  n]\}$.  Then, it holds
  \[
    (\LB(\Pat, \T), \UB(\Pat, \T)) = (e_{\mathsf{P}} - c_1,
      e_{\mathsf{P}} - c_2),
  \]
  where $c_1 = \rcount{\ARRnontail}{\ell}{e_H} -
  \rcount{\ARRnontail}{\ell}{\bpre}$ and $c_2 =
  \rcount{\ARRnontail}{\ell}{e_H} -
  \rcount{\ARRnontail}{\ell}{\epre}$.
\end{lemma}
\begin{proof}
  The proof consists of two steps:
  \begin{enumerate}
  \item First, we prove that $|\Occ(\Pat, \T)| = c_1 - c_2 =
    \rcount{\ARRnontail}{\ell}{\epre} -
    \rcount{\ARRnontail}{\ell}{\bpre}$. By \cref{lm:occ}, $\Occ(\Pat,
    \T)$ is a disjoint union of $\Occa(\Pat, \T)$ and $\Occs(\Pat,
    \T)$ (see the beginning of \cref{sec:pm-periodic-pm} for
    definitions). Moreover, since $\rend{\Pat} \leq m$,
    \cref{lm:occ-a} and its symmetric version (adapted according to
    \cref{lm:pm-lce}) imply that $\Occa(\Pat, \T) = \emptyset$. Thus,
    we need to prove $|\Occs(\Pat, \T)| =
    \rcount{\ARRnontail}{\ell}{\epre} -
    \rcount{\ARRnontail}{\ell}{\bpre}$. By
    \cref{lm:pm-lce-3}\eqref{lm:pm-lce-3-it-2}, it holds $\Occ(\Pat,
    \T) \sub \R^{-}$. Thus, $\Occs(\Pat, \T) = \Occsm(\Pat, \T)$.
    Recall now that $\ARRzlex[i] = \rendfull{\rlexm_i} -
    |\Pow(\Lroot(\rlexm_i))|$.  Since the set $\{\Pow(H) : H \in
    \Lroots\}$ is prefix-free, it follows, letting $H_j = \Lroot(j)$
    (where $j \in \R$), that
    \begin{align*}
      \{\rlexm_i\}_{i \in (\bpre \dd \epre]}
        &= \{j \,{\in}\, \R'^{-} : \Pow(H) \cdot
             \Pat(\ell \dd m]\text{ is a prefix of }
             \T[\rendfull{j} - |\Pow(H_j)| \dd n]\}\\
        &= \{j \,{\in}\, \R'^{-} : \Pow(H) \cdot
             \Pat(\ell \dd m]\text{ is a prefix of }
             \Pow(H_j) \cdot \T[\rendfull{j} \dd n]\}\\
        &= \{j \,{\in}\, \R'^{-}_{H} :
             \Pat(\ell \dd m]\text{ is a prefix of }
             \T[\rendfull{j} \dd n]\}
    \end{align*}
    By \cref{lm:occ-s}, we thus have $|\Occsm(\Pat, \T)| = |\{i \in
    (\bpre \dd \epre] : \rendfull{\rlexm_i} - \rlexm_i \geq
    \rendfull{\Pat} - 1\}| = \rcount{\ARRnontail}{\ell}{\epre} -
    \rcount{\ARRnontail}{\ell}{\bpre}$ (recall, that $\ARRnontail[i] =
    \rendfull{\rlexm_i} - \rlexm_i$; see \cref{sec:sa-periodic-ds}).
  \item Second, we prove that $\LB(\Pat, \T) = e_{\mathsf{P}} - c_1$.
    We start by observing that since $\Pat$ is periodic and satisfies
    $\type(\Pat) = -1$, it follows from
    \cref{lm:pm-delta-1,lm:pm-delta-2} that $\LB(\Pat, \T) = \LB(X,
    \T) + \delta(\Pat, \T) = \LB(X, \T) + \deltaa(\Pat, \T) -
    \deltas(\Pat, \T)$, where $X = \Pat[1 \dd 3\tau - 1]$. On the
    other hand, combining the equalities $\Lhead(\Pat) = s$,
    $\Lroot(\Pat) = H$, $\Lexp(\Pat) = k$, and $\type(\Pat) = -1$ with
    the definition of $\mathsf{P}$ yields $\LB(X, \T) + \deltaa(\Pat,
    \T) = e_{\mathsf{P}}$. Consequently, we obtain $\LB(\Pat, \T) =
    e_{\mathsf{P}} - \deltas(\Pat, \T)$. It thus remains to show
    $\deltas(\Pat, \T) = c_1$. By utilizing that by definition of the
    sequence $(\rlexm_i)_{i \in [1 \dd q]}$, for every $i, i' \in (b_H
    \dd e_H]$, $i < i'$ implies $\T[\rendfull{\rlexm_i} \dd n] \preceq
    \T[\rendfull{\rlexm_{i'}} \dd n]$, it follows by the above formula
    for $\{\rlexm_i\}_{i \in (\bpre \dd \epre]}$ that
    \[
      \{\rlexm_i\}_{i \in (\bpre \dd e_H]} = \{j \in \R'^{-}_{H} :
        \Pat(\ell \dd m] \preceq \T[\rendfull{j} \dd n]\}.
    \]
    By \cref{lm:pm-delta-s}, we thus have $\deltas(\Pat, \T) = |\{i
    \in (\bpre \dd e_H] : \rendfull{\rlexm_i} - \rlexm_i \geq
    \rendfull{\Pat} - 1\}| = \rcount{\ARRnontail}{\ell}{e_H} -
    \rcount{\ARRnontail}{\ell}{\bpre} = c_1$. \qedhere
  \end{enumerate}
\end{proof}

\begin{remark}
  Note that since the range $(\bpre \dd \epre]$ is well-defined
  even if $\epre - \bpre = 0$, the above lemma holds even if
  $|\Occ(\Pat, \T)| = 0$.
\end{remark}

\begin{lemma}\label{lm:st-periodic-invmap}
  Let $v$ be an explicit periodic node of $\ST$ such that $\rend{v}
  \leq |\str(v)|$ and $\type(v) = -1$. Let $u = \mapToTZ(v)$ and $\ell
  = \rendfull{v} - 1$. Then, it holds $\repr(v) =
  \pseudoInvTZ{\ell}{u}$.
\end{lemma}
\begin{proof}
  Denote $H = \Lroot(v)$, $\Pat = \str(v)$, and $(b, e) =
  \pseudoInvTZ{\ell}{u}$. Let $s = \ell \bmod |H|$, $k = \lfloor
  \tfrac{\ell}{|H|} \rfloor$, and $\mathsf{P} = \{j \in \R^{-}_{s,H} :
  \Lexp(j) = k\}$. Note that we then have $\Lhead(\Pat) =
  (\rendfull{\Pat} - 1) \bmod |H| = s$ and $\Lexp(\Pat) = \lfloor
  \tfrac{\rendfull{\Pat} - 1}{|H|} \rfloor = k$.  Observe that this
  implies $\mathsf{P} \neq \emptyset$. To see this, consider any
  $j \in \Occ(\str(v), \T) = \Occ(\Pat, \T)$. By \cref{lm:pm-lce},
  it follows that $j \in \R$, $\Lroot(j) = \Lroot(\Pat) = H$,
  and $\Lhead(j) = \Lhead(\Pat) = s$, i.e., $j \in
  \R_{s,H}$. Furthermore, by $\rend{\Pat} \leq |\Pat|$ and
  $\type(\Pat) = -1$ we obtain from
  \cref{lm:pm-lce-3}\eqref{lm:pm-lce-3-it-2} that $\Lexp(j) =
  \Lexp(\Pat) = k$ and $\type(j) = \type(\Pat) = -1$. Thus, $j \in
  \mathsf{P}$ and consequently $\mathsf{P} \neq \emptyset$. By
  definition of $\pseudoInvTZ{\ell}{u}$, we thus obtain that $(b, e) =
  (e_{\mathsf{P}} - c_1, e_{\mathsf{P}} - c_2)$, where
  $b_{\mathsf{P}}, e_{\mathsf{P}} \in [0 \dd n]$ are such that
  $\{\SA[i]\}_{i \in (b_{\mathsf{P}} \dd e_{\mathsf{P}}]} =
  \mathsf{P}$, $b_H, e_H \in [0 \dd q]$ are such that $\{\rlexm_i\}_{i
  \in (b_H \dd e_H]} = \R'^{-}_{H}$, $z_1 = \lrank(u)$, $z_2 =
  \rrank(u)$, and
  \begin{align*}
    c_1 &= \rcount{\ARRnontail}{\ell}{e_H} -
           \rcount{\ARRnontail}{\ell}{z_1},\\
    c_2 &= \rcount{\ARRnontail}{\ell}{e_H} -
           \rcount{\ARRnontail}{\ell}{z_2}.
  \end{align*}

  By definition of $\mapToTZ(v)$, we have $\str(u) = \Pow(H) \cdot
  \Pat[\rendfull{\Pat} \dd |\Pat|]$. Thus, denoting $\Pat' = \Pow(H)
  \cdot \Pat[\rendfull{\Pat} \dd |\Pat|]$, by definition of $\TZ$, we
  have $\lrank(u) = |\{i \in [1 \dd q] : \T[\ARRzlex[i] \dd n] \prec
  \Pat'\}|$ and $(\lrank(u) \dd \rrank(u)] = \{i \in [1 \dd q] :
  \Pat'\text{ is a prefix of }\T[\ARRzlex[i] \dd n]\}$.  By
  \cref{lm:st-periodic-nav-aux}, this implies that $(\LB(\Pat, \T),
  \UB(\Pat, \T)) = (e_{\mathsf{P}} - (\rcount{\ARRnontail}{\ell}{e_H}
  - \rcount{\ARRnontail}{\ell}{\lrank(u)}), e_{\mathsf{P}} -
  (\rcount{\ARRnontail}{\ell}{e_H} -
  \rcount{\ARRnontail}{\ell}{\rrank(u)})) = (e_{\mathsf{P}} - c_1,
  e_{\mathsf{P}} - c_2) = (b, e)$. This immediately implies $\repr(v)
  = \pseudoInvTZ{\ell}{u}$.
\end{proof}

\begin{proposition}\label{pr:st-periodic-invmap}
  Let $H \in \Lroots$ and let $u$ be a node of $\TZ$ such that
  $\Pow(H)$ is a prefix of $\str(u)$.  Given the data structure from
  \cref{sec:st-periodic-ds}, a pointer to $u$, and integers $\Int(H)$
  and $\ell \geq 0$, we can in $\bigO(\log \log n)$ time compute the
  pair $\pseudoInvTZ{\ell}{u}$.
\end{proposition}
\begin{proof}
  Let $p := |H|$. We first compute $s := \ell \bmod p$ and $k =
  \lfloor \frac{\ell}{p} \rfloor$. Next, using the lookup tables
  $\LTpref$ and $\LTrange$, we compute in $\bigO(1)$ time the pair
  $(b_{s,H}, e_{s,H}) = (\LB(X, \T),\allowbreak \UB(X, \T))$, where $X
  = \Pref_{3\tau - 1}(s, H)$. By \cref{lm:pm-lce}, we then have that
  $b_{s,H} = e_{s,H}$ holds if and only if $\R_{s,H} = \emptyset$, and
  if $b_{s,H} \neq e_{s,H}$ then $\{\SA[i] : i \in (b_{s,H} \dd
  e_{s,H}]\} = \R_{s,H}$.  If $b_{s,H} = e_{s,H}$, we return
  $\pseudoInvTZ{\ell}{u} = (0, 0)$. Let us now assume $b_{s,H} \neq
  e_{s,H}$.

  Next, using the data structure from \cref{sec:sa-periodic-ds}, as
  explained in the proof of \cref{pr:isa-delta-a}, in $\bigO(1)$ time
  we compute the pair $(b_{\mathsf{P}}, e_{\mathsf{P}})$ satisfying
  $\{\SA[i]\}_{i \in (b_{\mathsf{P}} \dd e_{\mathsf{P}}]} =
  \mathsf{P}$, where $\mathsf{P} = \{j \in \R^{-}_{s,H} : \Lexp(j) =
  k\}$.  More precisely, first, in $\bigO(1)$ time we compute $d =
  \rank{\BVexp}{1}{e_{s,H}} - \rank{\BVexp}{1}{b_{s,H}}$. If $d = 0$,
  then $\R^{-}_{s,H} = \emptyset$, and hence we return
  $\pseudoInvTZ{\ell}{u} = (0, 0)$. Otherwise, in $\bigO(1)$ time we
  retrieve $k_{\min} = \LTminexp[\Int(X)]$. Then, letting $k_{\max} =
  k_{\min} + d - 1$, we have $[k_{\min} \dd k_{\max}] = \{\Lexp(j) : j
  \in \R^{-}_{s,H}\}$.  If $k \not\in [k_{\min} \dd k_{\max}]$, then
  $\mathsf{P} = \emptyset$, and thus we return $\pseudoInvTZ{\ell}{u}
  = (0, 0)$. Otherwise, we have two cases. Let $p =
  \rank{\BVexp}{1}{b_{s,H}}$. If $k = k_{\min}$, then in $\bigO(1)$
  time we compute $(b_{\mathsf{P}}, e_{\mathsf{P}}) = (b_{s,H},
  \select{\BVexp}{1}{p+1})$.  If $k > k_{\min}$, in $\bigO(1)$ time we
  compute $(b_{\mathsf{P}}, e_{\mathsf{P}}) =
  (\select{\BVexp}{1}{p+k-k_{\min}},
  \select{\BVexp}{1}{p+k+1-k_{\min}})$.

  For the final step, we first in $\bigO(1)$ time compute $e_H =
  \sum_{H' \preceq H} |\R'^{-}_{H'}|$ using the lookup table $\LTruns$
  stored as part of the structure from \cref{sec:sa-periodic-ds}.
  Then, it holds that there exists $b_H < e_H$ such that
  $\{\rlexm_i\}_{i \in (b_H \dd e_H]} = \R'^{-}_{H}$. Then, in
  $\bigO(1)$ time we obtain $z_1 = \lrank(u)$ and $z_2 = \rrank(u)$
  (\cref{pr:compact-trie}). Finally, in $\bigO(\log \log n)$ time we
  compute $c_1 = \rcount{\ARRnontail}{\ell}{e_H} -
  \rcount{\ARRnontail}{\ell}{z_1}$ and $c_2 =
  \rcount{\ARRnontail}{\ell}{e_H} - \rcount{\ARRnontail}{\ell}{z_2}$
  and return $\pseudoInvTZ{\ell}{u} = (e_{\mathsf{P}} - c_1,
  e_{\mathsf{P}} - c_2)$.  The range counting queries are implemented
  using the structure from \cref{pr:range-queries} for the array $A$,
  which is stored as part of the structure from
  \cref{sec:sa-periodic-ds}.
\end{proof}

\bfparagraph{Handling Nodes Satisfying $\rend{v} > |\str(v)|$}

Next, we present a combinatorial result describing how to compute the
value $\rend{v}$, and to check if it holds $\rend{v} > |\str(v)|$. We
then show how to compute $(\LB(\Pat, \T),\allowbreak \UB(\Pat, \T))$
and $(\LB(\Pat c, \T), \UB(\Pat c, \T))$
for any periodic pattern $\Pat \in \Alphabet^{+}$ satisfying
$\rend{\Pat} > |\Pat|$. We will use it to efficiently perform queries
on periodic nodes $v$ of $\ST$ satisfying $\rend{v} > |\str(v)|$.

\begin{lemma}\label{lm:st-periodic-rend}
  Let $v$ be an explicit periodic node of $\ST$. Let $i_1 = \lrank(v)
  + 1$ and $i_2 = \rrank(v)$.
  \begin{enumerate}
  \item\label{lm:st-periodic-rend-it-1} It holds $\SA[i_1], \SA[i_2]
    \in \R$ and $\rend{v} = 1 + \min(\rend{\SA[i_1]} - \SA[i_1],
    \rend{\SA[i_2]} - \SA[i_2])$.
  \item\label{lm:st-periodic-rend-it-2} $\rend{v} \leq |\str(v)|$
    holds if and only if $\T[\SA[i_1] + \rend{v} - 1] = \T[\SA[i_2] +
    \rend{v} - 1]$.
  \end{enumerate}
\end{lemma}
\begin{proof}
  Denote $\ell = \sdepth(v)$, $b = \lrank(v)$, and $e = \rrank(v)$.

  1. Let $s = \Lhead(v)$, $H = \Lroot(v)$, $p = |H|$, and $Q =
  \str(v)$.  By definition, we have $b < e$ and $\{\SA[i]\}_{i \in (b
  \dd e]} = \Occ(Q, \T)$. On the other hand, by $|Q| \geq 3\tau - 1$
  and \cref{lm:pm-lce}, for every $j \in \Occ(Q, \T)$ it holds $j \in
  \R_{s,H}$.  In particular, we thus obtain $\SA[i_1], \SA[i_2] \in
  \R$.

  Next, we prove the following two facts.
  \begin{itemize}
  \item First, we show that there exists $t \in \{1, 2\}$ satisfying
    $\rend{v} - 1 = \rend{\SA[i_t]} - \SA[i_t]$.  By definition, it
    holds $\ell = \LCE(\SA[i_1], \SA[i_2])$ and $\T[\SA[i_1] \dd
    \SA[i_1] + \ell) = \T[\SA[i_2] \dd \SA[i_2] + \ell)$.  If $e - b =
    1$, then any $t \in \{1, 2\}$ satisfies the claim, since then
    $\ell = n - \SA[i_t] + 1$, and thus it follows from $\SA[i_t] \in
    \Occ(Q, \T)$ that
    \begin{align*}
      \rend{v} - 1
        &= p + \lcp(Q(0 \dd \ell-p], Q(p \dd \ell])\\
        &= p + \lcp(\T[\SA[i_t] \dd \SA[i_t] + \ell - p),
           \T[\SA[i_t] + p \dd \SA[i_t] + \ell))\\
        &= p + \lcp(\T[\SA[i_t] \dd n - p],
           \T[\SA[i_t] + p \dd n])\\
        &= p + \LCE(\SA[i_t], \SA[i_t] + p)\\
        &= \rend{\SA[i_t]} - \SA[i_t].
    \end{align*}
    Assume now $e - b > 1$. Then, $\T[\SA[i_1] + \ell] \neq
    \T[\SA[i_2] + \ell]$.~\footnote{To see that symbols $\T[\SA[i_1] +
    \ell]$ and $\T[\SA[i_2] + \ell]$ are well-defined, observe that by
    $\ell > 0$ and $b + 1 < e$, it follows that $\SA[i_1] + \ell - 1
    \neq \SA[i_2] + \ell - 1$. On the other hand, we have $\T[\SA[i_1]
    + \ell - 1] = \T[\SA[i_2] + \ell - 1]$. Thus, by the uniqueness of
    $\T[n]$ we must have $\SA[i_1] + \ell - 1 < n$ and $\SA[i_2] +
    \ell - 1 < n$} Thus, by $\T[\SA[i_1] + \ell - p] = \T[\SA[i_2] +
    \ell - p]$ there exists $t \in \{1, 2\}$ such that $\T[\SA[i_t] +
    \ell] \neq \T[\SA[i_t] + \ell - p]$.  For such $t$, we have
    $\LCE(\SA[i_t], p + \SA[i_t]) \leq \ell - p$ and hence
    $\rend{\SA[i_t]} - \SA[i_t] = p + \LCE(\SA[i_t], \SA[i_t] + p)
    \leq \ell$.  We therefore obtain $\rend{\SA[i_t]} - \SA[i_t] = p +
    \lcp(\T[\SA[i_t] \dd \SA[i_t] + \ell - p), \T[\SA[i_t] + p \dd
    \SA[i_t] + \ell))$.  On the other hand, by $Q = \T[\SA[i_t] \dd
    \SA[i_t] + \ell)$ we have $\rend{v} - 1 = p + \lcp(Q(0 \dd \ell -
    p], Q(p \dd \ell]) = p + \lcp(\T[\SA[i_t] \dd \SA[i_t] + \ell -
    p), \T[\SA[i_t] + p \dd \SA[i_t] + \ell))$.  Therefore, $\rend{v}
    - 1 = \rend{\SA[i_t]} - \SA[i_t]$.
  \item Second, we show that for every $i \in (b \dd e]$, it holds
    $\rend{v} - 1 \leq \rend{\SA[i]} - \SA[i]$. For this, recall that
    $\rend{\SA[i]} - \SA[i] = p + \LCE(\SA[i], \SA[i]+p)$.  Therefore,
    by $\SA[i] \in \Occ(Q, \T)$, we obtain $\rend{v} - 1 = \rend{Q} -
    1 = p + \lcp(Q(0 \dd \ell-p], Q(p \dd \ell]) = p + \lcp(\T[\SA[i]
    \dd \SA[i]+\ell-p), \T[\SA[i]+p \dd \SA[i]+\ell)) \leq p +
    \LCE(\SA[i],\SA[i]+p) = \rend{\SA[i]} - \SA[i]$.
  \end{itemize}
  By the above two facts, we obtain $\min(\rend{\SA[i_1]} - \SA[i_1],
  \rend{\SA[i_2]} - \SA[i_2]) = \min(\rend{\SA[i_t]} - \SA[i_t],
  \rend{\SA[i_{3-t}]} - \SA[i_{3-t}]) = \rend{v} - 1$.

  2. We start by showing that $\SA[i_1] + \rend{v} - 1, \SA[i_2] +
  \rend{v} - 1 \leq n$.  Observe that for every $j \in \R$, by the
  uniqueness of $\T[n]$, it holds $\rend{j} \leq n$.  Consider any $i
  \in (b \dd e]$. Above, we proved $\rend{v} - 1 \leq \rend{\SA[i]} -
  \SA[i]$. Thus, we obtain $\SA[i] + \rend{v} - 1 \leq \rend{\SA[i]}
  \leq n$.  In particular, $\SA[i_1] + \rend{v} - 1, \SA[i_2] +
  \rend{v} - 1 \leq n$.  We now prove the equivalence. Recall, that
  $|\str(v)| = \ell = \LCE(\SA[i_1], \SA[i_2])$ holds by
  definition. Let us first assume $\rend{v} \leq \ell$. By the
  assumption $\T[\SA[i_1] \dd \SA[i_1] + \ell) = \T[\SA[i_2] \dd
  \SA[i_2] + \ell)$, this immediately implies $\T[\SA[i_1] + \rend{v}
  - 1] = \T[\SA[i_2] + \rend{v} - 1]$.  To show the opposite
  implication, assume by contraposition that $\rend{v} > \ell$.  Since
  by definition we have $\rend{v} \leq |\str(v)| + 1$, we must have
  $\rend{v} = \ell + 1$.  Then, by definition of $\LCE$, we have
  $\T[\SA[i_1] + \rend{v} - 1] = \T[\SA[i_1] + \ell] \neq \T[\SA[i_2]
  + \ell] = \T[\SA[i_2] + \rend{v} - 1]$.
\end{proof}

\begin{proposition}\label{pr:st-periodic-range}
  Let $\Pat \in \Alphabet^{+}$ be a periodic pattern satisfying
  $\rend{\Pat} > |\Pat|$. Given the structure from
  \cref{sec:st-periodic-ds}, and the values $\Lhead(\Pat)$,
  $\Lroot(\Pat)$, and $|\Pat|$, we can in $\bigO(\log \log n)$ time
  compute the pair $(\LB(\Pat, \T), \UB(\Pat, \T))$.
\end{proposition}
\begin{proof}

  Denote $s = \Lhead(\Pat)$, $H = \Lroot(\Pat)$, and $m =
  |\Pat|$. First, in $\bigO(1)$ time we compute $k := \Lexp(\Pat) =
  \lfloor \tfrac{m - s}{|H|} \rfloor$ and $t := \Ltail(\Pat) = m - s -
  k|H|$. Next, using the lookup table $\LTpref$, in $\bigO(1)$ time
  we compute $X := \Pref_{3\tau - 1}(s, H) = \Pat[1 \dd 3\tau {-} 1]$.

  Next, we compute $|\Occ(\Pat, \T)|$.  Recall that by \cref{lm:occ},
  $|\Occ(\Pat, \T)| = |\Occa(\Pat, \T)| + \allowbreak |\Occs(\Pat,
  \T)| = |\Occam(\Pat, \T)| + |\Occap(\Pat, \T)| + |\Occsm(\Pat, \T)|
  + |\Occsp(\Pat, \T)|$ (see \cref{sec:pm-periodic-pm}).
  \begin{itemize}
  \item To compute $|\Occam(\Pat, \T)|$, we proceed as in the proof of
    \cref{pr:pm-occ-a}, except for one modification: Since we already
    have $\Lhead(\Pat)$, $\Lroot(\Pat)$, and $\Lexp(\Pat)$ (note that
    we do not need $\Ltail(\Pat)$ here since we assumed $\rend{\Pat} =
    |\Pat| + 1$), we can skip the first step which takes $\bigO(1 + m
    / \log_{\sigma} n)$ time. Note that after such modification, we no
    longer need the packed representation of the whole pattern $\Pat$,
    but only $\Pat[1 \dd 3\tau - 1]$, which we computed above. The
    rest of the algorithm in \cref{pr:pm-occ-a} takes $\bigO(1)$ time.
    The structures from \cref{pr:pm-occ-a} that we used (augmented
    bitvector $\BVexp$ and lookup tables $\LTminexp$ and $\LTrange$) are
    components of the structure from \cref{sec:st-periodic-ds}.
  \item To compute $|\Occsm(\Pat, \T)|$, we proceed as in
    \cref{pr:pm-occ-s}, except for two modifications. First, we again
    already have $\Lhead(\Pat)$, $\Lroot(\Pat)$, and $\Lexp(\Pat)$,
    which lets us skip the first step taking $\bigO(1 + m /
    \log_{\sigma} n)$ time.  Second, rather than computing $\bpre$ and
    $\epre$ in $\bigO(m / \log_{\sigma} n + \log \log n)$ time, we use
    the lookup table $\LTruns$ stored in the structure from
    \cref{sec:st-periodic-ds}.  More precisely, $\bpre$ and $\epre$
    are obtained in $\bigO(1)$ time by looking up in $\LTruns$ the
    pair associated with the key $(H,H')$, where $H'$ is a length-$t$
    prefix of $H$ (note that $\Pat[\pendfull{\Pat} \dd m] = H'$).  The
    rest of the algorithm in \cref{pr:pm-occ-s} takes $\bigO(\log \log
    n)$ time. Again, the components used in \cref{pr:pm-occ-s} are
    present in structure from \cref{sec:st-periodic-ds}.
  \end{itemize}
  The values $|\Occap(\Pat, \T)|$ and $|\Occsp(\Pat, \T)|$ are
  computed analogously (see the proof of \cref{pr:pm-occ}) using the
  symmetric components of the structure from
  \cref{sec:st-periodic-ds}. We can thus compute $|\Occ(\Pat, \T)|$ in
  $\bigO(\log \log n)$ time.

  The next step of the algorithm is to compute $\delta(\Pat, \T)$
  (\cref{sec:pm-periodic-pm}).  Observe (see
  \cref{sec:pm-periodic-prelim}) that $\rend{\Pat} > |\Pat|$ implies
  $\type(\Pat) = -1$. Recall that for such $\Pat$, by
  \cref{lm:pm-delta-2}, $\delta(\Pat, \T) = \deltaa(\Pat, \T) -
  \deltas(\Pat, \T)$.
  \begin{itemize}
  \item To compute $\deltaa(\Pat, \T)$, we proceed as in the proof of
    \cref{pr:pm-delta-a}, employing the same modification as when
    computing $|\Occam(\Pat, \T)|$ above. Thus, the computation takes
    $\bigO(1)$ time.  \cref{pr:pm-delta-a} uses the structure from
    \cref{pr:pm-occ-a} and, as above, the used components are already
    present in the structure from \cref{sec:st-periodic-ds}.
  \item To compute $\deltas(\Pat, \T)$, we observe that for a periodic
    pattern $\Pat$ satisfying $\pend{\Pat} > |\Pat|$, it holds by
    \cref{lm:pm-lce}\eqref{lm:pm-lce-it-2} that $\Poss(\Pat, \T) =
    \Occsm(\Pat, \T)$.  Consequently, we can compute $\deltas(\Pat,
    \T) = |\Occsm(\Pat, \T)|$ as above in $\bigO(\log \log n)$ time.
  \end{itemize}
  Combining the above two steps, the computation of $\delta(\Pat, \T)$
  takes $\bigO(\log \log n)$ time.

  We use the above values to obtain $(\LB(\Pat, \T), \UB(\Pat, \T))$
  as follows. By \cref{lm:pm-delta-1}, $\LB(\Pat, \T) = \LB(X, \T) +
  \delta(\Pat, \T)$, where $X = \Pat[1 \dd 3\tau {-} 1]$.  The value
  $\LB(X, \T)$ is obtained in $\bigO(1)$ time using the lookup table
  $\LTrange$.  We thus obtain $\LB(\Pat, \T)$. By definition, we then
  compute $\UB(\Pat, \T) = \UB(\Pat, \T) + |\Occ(\Pat, \T)|$.  In
  total, the query takes $\bigO(\log \log n)$ time.
\end{proof}

\begin{remark}\label{rm:st-periodic-range}
  Note that the above result holds even if $\Occ(\Pat, \T) =
  \emptyset$. Thus, it is more general than the result needed to
  support efficient processing of periodic nodes $v$ of $\ST$
  satisfying $\rend{v} > |\str(v)|$, since for such nodes we have
  $\Occ(\str(v), \T) \neq \emptyset$.
\end{remark}

\begin{proposition}\label{pr:st-periodic-range-2}
  Let $\Pat \in \Alphabet^{+}$ be a periodic pattern satisfying
  $\rend{\Pat} > |\Pat|$. Given any $c \in \Alphabet$, the structure
  from \cref{sec:st-periodic-ds}, and the values $\Lhead(\Pat)$,
  $\Lroot(\Pat)$, and $|\Pat|$, we can in $\bigO(\log \log n)$ time
  compute the pair $(\LB(\Pat c, \T), \UB(\Pat c, \T))$.
\end{proposition}
\begin{proof}
  Denote $\Pat' = \Pat c$ and $m = |\Pat| + 1 = |\Pat'|$.  Observe,
  that since $\Pat$ is periodic and it is a prefix of $\Pat'$, by
  \cref{lm:pm-lce-2}, $\Pat'$ is also periodic and it holds
  $\Lhead(\Pat) = \Lhead(\Pat')$ and $\Lroot(\Pat) = \Lroot(\Pat')$.
  Let us denote $s = \Lhead(\Pat) = \Lhead(\Pat')$ and $H =
  \Lroot(\Pat) = \Lroot(\Pat')$. By the assumption, we have
  $\rend{\Pat} = m$.  First, in $\bigO(1)$ time we compute $t :=
  \Ltail(\Pat) = (m - 1 - s) \bmod |H|$. We then check if
  $\rend{\Pat'} \leq |\Pat'|$ by comparing $c$ to $H[t + 1]$. If $c =
  H[t + 1]$, then we have $\rend{\Pat'} > |\Pat'|$. Since we have
  $\Lhead(\Pat') = s$, $\Lroot(\Pat') = H$, and $|\Pat'| = m$, in
  $\bigO(\log \log n)$ time we thus compute and return $(\LB(\Pat',
  \T), \UB(\Pat', \T))$ using \cref{pr:st-periodic-range}. Let us thus
  assume $c \neq H[t + 1]$, i.e., $\rend{\Pat'} \leq |\Pat'|$. We then
  compute $\type(\Pat')$ by comparing $c$ with $H[t + 1]$. Let us
  assume that $c \prec H[t + 1]$, i.e., $\type(\Pat) = -1$ (the case
  $\type(\Pat) = +1$ is handled symmetrically). We now execute the
  modified algorithm from \cref{pr:pm-periodic-range} for $\Pat'$.
  The modification is to replace implementation of operations taking
  $\Theta(m / \log_{\sigma} n)$ time with faster alternatives,
  exploiting the fact that by $\rend{\Pat'} = \rend{\Pat}$,
  $\Lhead(\Pat') = \Lhead(\Pat)$, and $\Lroot(\Pat') = \Lroot(\Pat)$
  it follows that $\rendfull{\Pat'} = \rendfull{\Pat} = \rend{\Pat} -
  \Ltail(\Pat) = m - t$ and thus $\Pat'[\rendfull{\Pat'} \dd m]$ is of
  length $t+1 \leq \tau$ (importantly, the modified algorithm will not
  use the components of the data structures in \cref{sec:pm-periodic}
  which are not part of the structure from \cref{sec:st-periodic-ds}).
  More precisely:
  \begin{itemize}
  \item First, using $\LTpref$, in $\bigO(1)$ time we compute $X =
    \Pref_{3\tau - 1}(s, H) = \Pat'[1 \dd 3\tau {-} 1]$.
  \item We then compute $|\Occ(\Pat', \T)|$.  First, note that since
    $\rend{\Pat'} \leq |\Pat'|$ and $\type(\Pat') = -1$, it follows by
    \cref{lm:pm-lce-3}\eqref{lm:pm-lce-3-it-2} that $\Occ(\Pat', T)
    \sub \R^{-}$, and that for every $j \in \Occ(\Pat', \T)$ it holds
    $\Lexp(j) = \Lexp(\Pat')$.  Thus, the sets $\Occa(\Pat', \T)$ and
    $\Occsp(\Pat', \T)$ are empty, and hence it remains to explain the
    computation of $|\Occsm(\Pat', \T)|$ (\cref{pr:pm-occ-s}).  Observe,
    that the expensive operations are the computation of
    $\Lhead(\Pat')$, $\Lroot(\Pat')$, $\Lexp(\Pat')$, and the pair
    $(\bpre, \epre)$.  Observe, however, that here we already have $s
    = \Lhead(\Pat')$, $H = \Lroot(\Pat')$, and $\rend{\Pat'} = m$.
    This lets us deduce $k := \Lexp(\Pat') = \lfloor \tfrac{m - 1 -
    s}{|H|} \rfloor $ in $\bigO(1)$ time.  As for the computation of
    $(\bpre, \epre)$, we first in $\bigO(1)$ time compute $H' :=
    \Pat'[\rendfull{\Pat'} \dd m] = H[1 \dd t + 1]$, and then obtain
    $(\bpre, \epre)$ by looking up the pair associated with the key
    $(H,H')$ in the lookup table $\LTruns$. The rest of the algorithm
    in \cref{pr:pm-occ-s} takes $\bigO(\log \log n)$ time.
  \item Finally, we compute $\delta(\Pat', \T)$. By $\type(\Pat') =
    -1$ and \cref{lm:pm-delta-2}, it holds $\delta(\Pat', \T) =
    \deltaa(\Pat', \T) - \deltas(\Pat', \T)$.  To compute
    $\deltaa(\Pat', \T)$, we proceed as in the proof of
    \cref{pr:pm-delta-a}. The string $X$ was already obtained above.
    The expensive step in \cref{pr:pm-delta-a} is the computation of
    $\Lroot(\Pat')$ and $\Lexp(\Pat')$. As noted above, here we
    already have $\Lroot(\Pat') = H$, and in $\bigO(1)$ time we can
    compute $\Lexp(\Pat') = \lfloor \tfrac{\rend{\Pat'} - 1 -
    \Lhead(\Pat')}{|\Lroot(\Pat')|} \rfloor = \lfloor \tfrac{m - 1 -
    s}{|H|} \rfloor$. The rest of the algorithm in
    \cref{pr:pm-delta-a} takes $\bigO(1)$ time.  We then compute
    $\deltas(\Pat', \T)$ using a modified \cref{pr:pm-delta-s}. The
    expensive part is the computation of $x$ and $x'$. After those are
    computed, the rest takes $\bigO(\log \log n)$ time. Here, we
    obtain $x$ by observing that it is equal to $\bpre$ (which was
    computed above), and then obtain $x'$ using $\LTruns$ (this only
    requires knowing $\Lroot(\Pat')$, which we already have).
  \end{itemize}
  Note that all components of the structure from
  \cref{pr:pm-occ-s,pr:pm-delta-a,pr:pm-delta-s} that we used are also
  components of the structure from \cref{sec:st-periodic-ds}.  Using
  the above values, we now obtain $(\LB(\Pat', \T), \UB(\Pat', \T))$
  as follows. By \cref{lm:pm-delta-1}, it holds $\LB(\Pat', \T) =
  \LB(X, \T) + \delta(\Pat', \T)$, where $X = \Pat'[1 \dd 3\tau {-}
  1]$.  The value $\LB(X, \T)$ is obtained in $\bigO(1)$ time using
  the lookup table $\LTrange$.  We thus obtain $\LB(\Pat', \T)$. By
  definition, we then compute $\UB(\Pat', \T) = \UB(\Pat', \T) +
  |\Occ(\Pat', \T)|$.
\end{proof}

\begin{remark}
  Note that, analogously to \cref{pr:st-periodic-range} (see
  \cref{rm:st-periodic-range}), the above result holds even if
  $\Occ(\Pat c, \T) = \emptyset$.
\end{remark}

\subsubsection{Implementation of
  \texorpdfstring{$\LCA(u, v)$}{LCA(u, v)}}\label{sec:st-periodic-lca}

\begin{lemma}\label{lm:st-periodic-lca}
  Let $v_1$ and $v_2$ be explicit nodes of $\ST$ such that $\LCA(v_1,
  v_2)$ is periodic and it holds $\rend{\LCA(v_1, v_2)} \leq
  |\str(\LCA(v_1, v_2))|$ and $\type(\LCA(v_1, v_2)) = -1$.  Then,
  $v_1$ and $v_2$ are periodic and it holds $\rend{v_1} \leq
  |\str(v_1)|$, $\rend{v_2} \leq |\str(v_2)|$, and $\type(v_1) =
  \type(v_2) = -1$.  Moreover,
  \[
    \mapToTZ(\LCA(v_1, v_2)) = \LCA(\mapToTZ(v_1), \mapToTZ(v_2)).
  \]
\end{lemma}
\begin{proof}
  Denote $v = \LCA(v_1, v_2)$, $Q = \str(v)$, $H = \Lroot(v)$, and $s
  = \Lhead(v)$.  Let also $Q_1 = \str(v_1)$.  By $|Q| \geq 3\tau - 1$
  and since $v$ is an ancestor of $v_1$, we have $\lcp(Q, Q_1) \geq
  3\tau - 1$. Consequently, by \cref{lm:pm-lce-2}, the node $v_1$ is
  periodic and it holds $\Lroot(v_1) = \Lroot(Q_1) = \Lroot(Q) =
  \Lroot(v) = H$ and $\Lhead(v_1) = \Lhead(Q_1) = \Lhead(Q) =
  \Lhead(v) = s$.  Furthermore, by $\rend{Q} \leq |Q|$ and $\type(Q) =
  -1$, it holds $Q[\rend{Q}] \prec Q[\rend{Q} - p]$. Since $Q$ is a
  prefix of $Q_1$, this immediately implies $\rend{Q_1} = \rend{Q}
  \leq |Q| \leq |Q_1|$ and $Q_1[\rend{Q_1}] = Q[\rend{Q}] \prec
  Q[\rend{Q} - p] = Q_1[\rend{Q_1} - p]$, i.e., $\type(Q_1) = -1$. We
  have thus shown $\rend{v_1} \leq |\str(v_1)|$ and $\type(v_1) = -1$.
  Analogously, we obtain that $v_2$ is periodic and it holds
  $\rend{v_2} = \rend{v}$, $\Lroot(v_2) = H$, $\Lhead(v_2) = s$,
  $\rend{v_2} \leq |\str(v_2)|$, and $\type(v_2) = -1$. We have thus
  shown that $u_1 = \mapToTZ(v_1)$ and $u_2 = \mapToTZ(v_2)$ are
  well-defined (see \cref{sec:st-periodic-nav}).

  Let $u = \LCA(u_1, u_2)$, $\ell' = \sdepth(u)$, and $\ell =
  \sdepth(v)$.  By \cref{ob:lca}, we have $\ell = \lcp(\str(v_1),
  \str(v_2))$, $\ell' = \lcp(\str(u_1), \str(u_2))$, $\str(v) =
  \str(v_1)[1 \dd \ell]$, and $\str(u) = \str(u_1)[1 \dd \ell']$.
  Denote $\delta = \rendfull{v}$. As observed above, $\rend{v_1} =
  \rend{v}$, $\Lhead(v_1) = \Lhead(v)$, and $\Lroot(v_1) =
  \Lroot(v)$. Thus, $\rendfull{v_1} = 1 + \Lhead(v_1) + |\Lroot(v_1)|
  \cdot \lfloor \tfrac{\rend{v_1} - 1 - \Lhead(v_1)}{|\Lroot(v_1)|}
  \rfloor = 1 + \Lhead(v) + |\Lroot(v)|\cdot \lfloor \tfrac{\rend{v} -
  1 - \Lhead(v)}{|\Lroot(v)|} \rfloor = \rendfull{v} = \delta$.
  Analogously, $\rendfull{v_2} = \delta$.  By definition of
  $\mapToTZ(v_1)$ and $\mapToTZ(v_2)$, we have $\str(u_1) = \Pow(H)
  \cdot \str(v_1)[\rendfull{v_1} \dd |\str(v_1)|] = \Pow(H) \cdot
  \str(v_1)[\delta \dd |\str(v_1)|]$ and $\str(u_2) = \Pow(H) \cdot
  \str(v_2)[\rendfull{v_2} \dd |\str(v_2)|] = \Pow(H) \cdot
  \str(v_2)[\delta \dd |\str(v_2)|]$.  Thus, $\ell' = \lcp(\str(u_1),
  \str(u_2)) = |\Pow(H)| + (\lcp(\str(v_1), \str(v_2)) - \delta + 1) =
  |\Pow(H)| + \ell - \delta + 1$.  Consequently, $\str(u) = \str(u_1)[1
  \dd \ell'] = \Pow(H) \cdot \str(v_1)[\delta \dd \delta + \ell' -
  |\Pow(H)| - 1] = \Pow(H) \cdot \str(v_1)[\delta \dd \ell] = \Pow(H)
  \cdot \str(v)[\delta \dd \ell] = \Pow(H) \cdot \str(v)[\rendfull{v}
  \dd |\str(v)|]$.  Thus, by definition of $\mapToTZ(u)$, and since no
  two nodes of $\TZ$ have the same value of $\str$, we therefore
  obtain $\mapToTZ(v) = u$.
\end{proof}

\begin{lemma}\label{lm:st-periodic-lca-2}
  Let $v_1$ and $v_2$ be explicit nodes of $\ST$ such that $\LCA(v_1,
  v_2)$ is periodic. Denote $v = \LCA(v_1, v_2)$, $i_{\min} =
  \min(\lrank(v_1), \lrank(v_2)) + 1$, and $i_{\max} =
  \max(\rrank(v_1), \rrank(v_2))$. Then, it holds:
  \begin{enumerate}
  \item\label{lm:st-periodic-lca-2-it-1} $\SA[i_{\min}], \SA[i_{\max}]
    \in \R$ and $\rend{v} = 1 + \min(\rend{\SA[i_{\min}]} -
    \SA[i_{\min}], \rend{\SA[i_{\max}]} - \SA[i_{\max}])$.
  \item\label{lm:st-periodic-lca-2-it-2} $\rend{v} \leq |\str(v)|$
    holds if and only if $\T[\SA[i_{\min}] + \rend{v} - 1] =
    \T[\SA[i_{\max}] + \rend{v} - 1]$.
  \end{enumerate}
\end{lemma}
\begin{proof}
  Denote $i_1 = \lrank(v_1) + 1$, $i_2 = \rrank(v_1)$, $i_3 =
  \lrank(v_2) + 1$, and $i_4 = \rrank(v_2)$.

  1. Let $H = \Lroot(v)$, and $p = |H|$. We start by noting that
  $\SA[i_{\min}], \SA[i_{\max}] \in \R$ follows by
  \cref{lm:st-periodic-rend}\eqref{lm:st-periodic-rend-it-1}.  Next,
  we prove the formula for $\rend{v}$.

  First, we show that $\rend{v} = \min(\rend{v_1},
  \rend{v_2})$. Observe that if $P$ is a prefix of $S$, both $P$ and
  $S$ and periodic, and $\Lroot(S) = H$, then, $\pend{S} = 1 + p +
  \lcp(S(0 \dd |S|-p], S(p \dd |S|]) \geq 1 + p + \lcp(S(0 \dd |P|-p],
  S(p \dd |P|]) = 1 + p + \lcp(P(0 \dd |P|-p], P(p \dd |P|]) =
  \pend{P}$.  Since $\str(v)$ is a prefix of $\str(v_1)$ and
  $\str(v_2)$, we thus obtain $\rend{v_1} \geq \rend{v}$ and
  $\rend{v_2} \geq \rend{v}$.  It remains to show that there exists $t
  \in \{1,2\}$ such that $\rend{v_t} = \rend{v}$.  Consider two cases:
  \begin{itemize}
  \item If $\rend{v} = |\str(v)| + 1$, then there are two
    possibilities. Either for some $t \in \{1,2\}$, we have
    $|\str(v_t)| = |\str(v)|$, in which case $\str(v_t) = \str(v)$ and
    thus $\rend{v_t} = \rend{v}$ follows.  The other possibility is
    that $|\str(v_1)| > |\str(v)|$ and $|\str(v_2)| > |\str(v)|$.
    Since $\str(v)$ is the longest common prefix of $\str(v_1)$ and
    $\str(v_2)$, we then have $\str(v_1)[\rend{v}] =
    \str(v_1)[|\str(v)|+1] \neq \str(v_2)[|\str(v)+1] =
    \str(v_2)[\rend{v}]$.  Thus, there exists $t \in \{1,2\}$ such
    that $\str(v_t)[\rend{v}] \neq \str(v)[\rend{v}-p]$.  By
    definition, for such $t$ we have $\rend{v_t} = \rend{v}$.
  \item Let us now assume $\rend{v} \leq |\str(v)|$.  This implies
    that $\str(v)[\rend{v}] = \str(v_1)[\rend{v}] =
    \str(v_2)[\rend{v}]$ and $\str(v)[\rend{v}] \neq
    \str(v)[\rend{v}-p]$.  Thus, by $\str(v)[1 \dd \rend{v}] =
    \str(v_1)[1 \dd \rend{v}] \allowbreak = \str(v_2)[1 \dd \rend{v}]$
    we obtain $\rend{v_1} = \rend{v_2} = \rend{v}$.
  \end{itemize}
  We have thus shown that there exists $t \in \{1,2\}$ such that
  $\rend{v} = \rend{v_t}$. Combined with $\rend{v_1} \geq \rend{v}$
  and $\rend{v_2} \geq \rend{v}$, this yields $\min(\rend{v_1},
  \rend{v_2}) = \min(\rend{v_t}, \rend{v_{3-t}}) = \rend{v}$.

  By the above and
  \cref{lm:st-periodic-rend}\eqref{lm:st-periodic-rend-it-1} for $v_1$
  and $v_2$, it holds $\rend{v} = 1 + \min_{t \in [1 \dd
  4]}\{\rend{\SA[i_t]}-\SA[i_t]\}$.  To show that this is equal to the
  expression for $\rend{v}$ from the claim, we first observe that
  letting $X = \str(v)[1 \dd 3\tau - 1]$ and $(b,e) = (\LB(X,\T),
  \UB(X,\T))$, we have $i_t \in (b \dd e]$ for all $t \in [1 \dd
  4]$. Observe that by \cref{lm:lce}, the sequence $(\rend{\SA[i]} -
  \SA[i])_{i=b+1}^{e}$ is bitonic, i.e., there exists $m \in (b \dd
  e]$ such that $\rend{\SA[b+1]}-\SA[b+1] \leq
  \rend{\SA[b+2]}-\SA[b+2] \leq \dots \leq \rend{\SA[m]} - \SA[m]$ and
  $\rend{\SA[m]} - \SA[m] \geq \rend{\SA[m+1]}-\SA[m+1] \geq \dots
  \geq \rend{\SA[e]}-\SA[e]$. This implies that for every triple $k_1,
  k_2, k_3 \in (b \dd e]$, the inequalities $k_1 \leq k_2 \leq k_3$
  imply $\min(\rend{\SA[k_1]}-\SA[k_1], \rend{\SA[k_3]}-\SA[k_3]) =
  \min_{t \in [1 \dd 3]}\{\rend{\SA[k_t]}-\SA[k_t]\}$.  For a proof,
  consider two cases:
  \begin{itemize}
  \item If $k_2 < m$, then by the bitonic property, we have
    $\rend{\SA[k_2]}-\SA[k_2] \geq \rend{\SA[k_1]}-\SA[k_1]$. Thus,
    the expression $\rend{\SA[k_2]}-\SA[k_2]$ has no effect on the
    minimum.
  \item If $k_2 \geq m$, then by the bitonic property, we have
    $\rend{\SA[k_2]}-\SA[k_2] \geq \rend{\SA[k_3]}-\SA[k_3]$. Thus,
    the expression $\rend{\SA[k_2]}-\SA[k_2]$ can again be excluded in
    the minimum.
  \end{itemize}
  By the above, letting $i'_{\min} = \min_{t \in [1 \dd 4]}\{i_t\}$
  and $i'_{\max} = \max_{i \in [1 \dd 4]}\{i_t\}$, we thus have
  $\rend{v} = 1 + \min_{t \in [1 \dd 4]}\{\rend{\SA[i_t]}-\SA[i_t]\} =
  1 + \min(\rend{\SA[i'_{\min}]}-\SA[i'_{\min}],
  \rend{\SA[i'_{\max}]}-\SA[i'_{\max}])$.

  It remains to show that $i'_{\min} = i_{\min}$ and $i'_{\max} =
  i_{\max}$. For this, it suffices to note that by definition, we have
  $i_1 \leq i_2$ and $i_3 \leq i_4$, thus, $i'_{\min} = \min_{t \in [1
  \dd 4]}\{i_t\} = \min(i_1,i_3) = i_{\min}$ and analogously,
  $i'_{\max} = \max_{t \in [1 \dd 4]}\{i_t\} = \max(i_2, i_4) =
  i_{\max}$.

  2. As observed in the proof of
  \cref{lm:st-periodic-rend}\eqref{lm:st-periodic-rend-it-2}, it holds
  $\SA[i_1] + \rend{v_1} - 1 \leq n$, $\SA[i_2] + \rend{v_1} - 1 \leq
  n$, $\SA[i_3] + \rend{v_2} - 1 \leq n$, and $\SA[i_4] + \rend{v_2} -
  1 \leq n$.  Thus, by $\rend{v} = \min(\rend{v_1}, \rend{v_2})$, for
  every $t \in [1 \dd 4]$, we have $\SA[i_t] + \rend{v} - 1 \leq n$.
  In particular, $\SA[i_{\min}] + \rend{v} - 1 \leq n$ and
  $\SA[i_{\max}] + \rend{v} - 1 \leq n$.  We now prove the
  equivalence.  Let us first assume $\rend{v} \leq |\str(v)|$.  Then,
  since $\str(v)$ is a prefix of both $\str(v_1)$ and $\str(v_2)$ and
  $\SA[i_{\min}], \SA[i_{\max}] \in \Occ(\str(v_1), \T) \cup
  \Occ(\str(v_2), \T)$, it follows that $\SA[i_{\min}], \SA[i_{\max}]
  \in \Occ(\str(v), \T)$.  Therefore, $\T[\SA[i_{\min}] + \rend{v} -
  1] = \T[\SA[i_{\max}] + \rend{v} - 1]$ follows immediately. To show
  the opposite implication, assume by contraposition that $\rend{v} =
  |\str(v)| + 1$.  Then, there are two possibilities. Either for some
  $t \in \{1, 2\}$ we have $\str(v_t) = \str(v)$, in which case
  $\str(v_t)$ is a prefix of $\str(v_{3-t})$, which in turn implies
  $i_{\min} = \lrank(v_t) + 1$ and $i_{\max} = \rrank(v_t)$. Then,
  $\rend{v} = \rend{v_t}$ and by applying
  \cref{lm:st-periodic-rend}\eqref{lm:st-periodic-rend-it-2} to $v_t$,
  we obtain $\T[\SA[i_{\min}] + \rend{v} - 1] = \T[\SA[\lrank(v_t)+1]
  + \rend{v_t} - 1] \neq \T[\SA[\rrank(v_t)] + \rend{v_t} - 1] =
  \T[\SA[i_{\max}] + \rend{v} - 1]$.  The other possibility is that
  $|\str(v_1)| > |\str(v)|$ and $|\str(v_2)| > |\str(v)|$.  Then,
  since $\str(v)$ is the longest common prefix of $\str(v_1)$ and
  $\str(v_2)$, neither of $v_1$ or $v_2$ is an ancestor of the other,
  and hence either it holds $i_{\min} = i_1 \leq i_2 < i_3 \leq i_4 =
  i_{\max}$ or $i_{\min} = i_3 \leq i_4 < i_1 \leq i_2 = i_{\max}$.
  In the first case $\SA[i_{\min}] \in \Occ(\str(v_1), \T)$ and
  $\SA[i_{\max}] \in \Occ(\str(v_2), \T)$, and in the second case
  $\SA[i_{\min}] \in \Occ(\str(v_2), \T)$ and $\SA[i_{\max}] \in
  \Occ(\str(v_1), \T)$.  Therefore, in both cases we have
  $\T[\SA[i_{\min}] + \rend{v} - 1] = \T[\SA[i_{\min}] + |\str(v)|]
  \neq \T[\SA[i_{\max}] + |\str(v)|] = \T[\SA[i_{\max}] + \rend{v} -
  1]$.
\end{proof}

\begin{remark}
  Observe that in \cref{lm:st-periodic-lca-2}, it does not
  necessarily hold that $\lrank(\LCA(v_1, v_2)) = i_{\min}$ or
  $\rrank(\LCA(v_1, v_2)) = i_{\max}$. Thus, the lemma does not
  immediately follow as a corollary from \cref{lm:st-periodic-rend}.
\end{remark}

\begin{proposition}\label{pr:st-periodic-lca}
  Let $v_1$ and $v_2$ be explicit nodes of $\ST$ such that $\LCA(v_1,
  v_2)$ is periodic.  Given the data structure from
  \cref{sec:st-periodic-ds} and the pairs $\repr(v_1)$ and
  $\repr(v_2)$, we can in $\bigO(\log \log n)$ time compute
  $\repr(\LCA(v_1, v_2))$.
\end{proposition}
\begin{proof}
  Denote $v = \LCA(v_1, v_2)$, $\repr(v_1) = (b_1, e_1)$ and
  $\repr(v_2) = (b_2, e_2)$ (recall that for $i \in \{1, 2\}$, we have
  $b_i = \lrank(v_i)$ and $e_i = \rrank(v_i)$).

  First, in $\bigO(1)$ time we compute $i_{\min} = \min(b_1, b_2) + 1$
  and $i_{\max} = \max(e_1, e_2)$. By
  \cref{lm:st-periodic-lca-2}\eqref{lm:st-periodic-lca-2-it-1}, we
  have $\SA[i_{\min}], \SA[i_{\max}] \in \R$. Using
  \cref{pr:sa-periodic-sa}, in $\bigO(\log \log n)$ time we compute
  $j_{\min} = \SA[i_{\min}]$ and $j_{\max} = \SA[i_{\max}]$.  Next,
  using \cref{pr:isa-root} in $\bigO(1)$ time we compute $H :=
  \Lroot(j_{\min})$, $s := \Lhead(j_{\min})$, $k_{\min} =
  \Lexp(j_{\min})$, $k_{\max} = \Lexp(j_{\max})$, $t_{\min} =
  \Ltail(j_{\min})$, and $t_{\max} = \Ltail(j_{\max})$.  Observe that
  since $v$ is periodic, and $j_{\min}, j_{\max} \in \Occ(\str(v_1),
  \T) \cup \Occ(\str(v_2), \T) \subseteq \Occ(\str(v), \T)$, it
  follows by \cref{lm:lce,lm:pm-lce}, that $\Lroot(v) =
  \Lroot(j_{\max}) = H$ and $\Lhead(v) = \Lhead(j_{\max}) = s$.  In
  $\bigO(1)$ time we thus compute $e_{\min} := \rend{j_{\min}} =
  j_{\min} + s + k_{\min}|H| + t_{\min}$ and $e_{\max} :=
  \rend{j_{\max}} = j_{\max} + s + k_{\max}|H| + t_{\max}$.  Next, in
  $\bigO(1)$ time we compute $e_v := \rend{v} = 1 + \min(e_{\min} -
  j_{\min}, e_{\max} - j_{\max})$ (see
  \cref{lm:st-periodic-lca-2}\eqref{lm:st-periodic-lca-2-it-1}).
  Using \cref{lm:st-periodic-lca-2}\eqref{lm:st-periodic-lca-2-it-2},
  we then in $\bigO(1)$ time check if it holds $\rend{v} \leq |\str(v)|$
  by comparing $\T[j_{\min} + e_v - 1]$ with $\T[j_{\max} + e_v -
  1]$. Consider two cases:
  \begin{itemize}
  \item Let $\T[j_{\min} + e_v - 1] = \T[j_{\max} + e_v - 1]$, i.e.,
    $\rend{v} \leq |\str(v)|$. Recall now that $j_{\min} \in
    \Occ(\str(v), \T)$.  In $\bigO(1)$ time we thus compute $\type(v)$
    by comparing $\T[j_{\min} + e_v - 1]$ with $\T[j_{\min} + e_v - 1
    - |H|]$.  Let us assume that $\T[j_{\min} + e_v - 1] \prec
    \T[j_{\min} + e_v - 1 - |H|]$, i.e., $\type(v) = -1$ (the case
    $\type(v) = +1$ is handled symmetrically, using the part of the
    structure from \cref{sec:st-periodic-ds} adapted according to
    \cref{lm:lce}).  By \cref{lm:st-periodic-lca}, we now have that
    $v_1$ and $v_2$ are periodic and it holds $\rend{v_1} \leq
    |\str(v_1)|$, $\rend{v_2} \leq |\str(v_2)|$, and $\type(v_1) =
    \type(v_2) = -1$. Using \cref{pr:st-periodic-map}, in $\bigO(\log
    \log n)$ time we compute pointers to $u_1 = \mapToTZ(v_1)$ and
    $u_2 = \mapToTZ(v_2)$.  Using the representation of $\TZ$ stored
    as part of the structure in \cref{sec:st-periodic-ds}, and
    \cref{pr:compact-trie}, in $\bigO(1)$ time we compute a pointer to
    $u = \LCA(u_1, u_2)$. By \cref{lm:st-periodic-lca}, it holds
    $\mapToTZ(v) = u$. Our goal is to exploit this connection to
    compute $\repr(v)$.  In $\bigO(1)$ time we compute $k := \Lexp(v)
    = \lfloor \tfrac{e_v - 1 - s}{|H|} \rfloor$ and $\ell :=
    \rendfull{v} - 1 = s + k|H|$.  Using \cref{pr:st-periodic-invmap},
    in $\bigO(\log \log n)$ time we then compute the pair $(b, e) =
    \pseudoInvTZ{\ell}{u}$. As noted above, it holds $\mapToTZ(v) =
    u$. Thus, by \cref{lm:st-periodic-invmap}, we have $\repr(v) = (b,
    e)$.
  \item Let $\T[j_{\min} + e_v - 1] \neq \T[j_{\max} + e_v - 1]$,
    i.e., $\rend{v} > |\str(v)|$. Letting $\Pat = \str(v)$, we then
    have $\pend{\Pat} > |\Pat|$, $\Lhead(\Pat) = s$, $\Lroot(\Pat) =
    H$, and $|\Pat| = e_v - 1$. Using \cref{pr:st-periodic-range-2},
    we thus compute $(b, e) = (\LB(\Pat, \T), \UB(\Pat, \T))$ in
    $\bigO(\log \log n)$ time, and return $\repr(v) = (b,
    e)$. \qedhere
  \end{itemize}
\end{proof}

\subsubsection{Implementation of
  \texorpdfstring{$\child(v, c)$}{child(v, c)}}\label{sec:st-periodic-child}

\begin{lemma}\label{lm:st-periodic-child}
  Let $c \in \Alphabet$ and $v$ be an explicit periodic internal node
  of $\ST$ satisfying $\rend{v} \leq |\str(v)|$ and $\type(v) =
  -1$. Let $u = \mapToTZ(v)$. If $\child(u,c) = \nil$ then
  $\child(v,c) = \nil$. Otherwise, letting $u' = \child(u,c)$, it
  holds
  \vspace{1.5ex}
  \[
    \repr(\child(v,c)) =
    \begin{cases}
      (b, e) & \text{if $b \neq e$},\\
      (0, 0)   & \text{otherwise},
    \end{cases}
    \vspace{1.5ex}
  \]
  where $(b,e) = \pseudoInvTZ{\ell}{u'}$ and $\ell = \rendfull{v} -
  1$.
\end{lemma}
\begin{proof}
  Let $H = \Lroot(v)$, $s = \Lhead(v)$, $k = \Lexp(v)$, and
  $\mathsf{P} = \{j \in \R^{-}_{s,H} : \Lexp(j) = k\}$. We first show
  that $\mathsf{P} \neq \emptyset$. Consider any $j \in \Occ(\str(v), \T)$.
  By \cref{lm:pm-lce}, $j \in \R$, $\Lroot(j) = \Lroot(v) = H$, and
  $\Lhead(j) = \Lhead(v) = s$, i.e., $j \in \R_{s,H}$. Furthermore,
  by $\rend{v} \leq |\str(v)|$ and $\type(v) = -1$ we obtain from
  \cref{lm:pm-lce-3}\eqref{lm:pm-lce-3-it-2} that $\Lexp(j) = \Lexp(v)
  = k$ and $\type(j) = \type(v) = -1$. Thus, $j \in \mathsf{P}$, and hence
  $\mathsf{P} \neq \emptyset$. Let $b_{\mathsf{P}}, e_{\mathsf{P}} \in
  [0 \dd n]$ be such that $\{\SA[i]\}_{i \in (b_{\mathsf{P}} \dd
  e_{\mathsf{P}}]} = \mathsf{P}$, and $b_H, e_H \in [0 \dd q]$ be
  such that $\{\rlexm_i\}_{i \in (b_H \dd e_H]} = \R'^{-}_H$.

  Denote $\Pat = \str(v)c$ and $\Pat' = \Pow(H) \cdot
  \Pat[\rendfull{\Pat} \dd |\Pat|]$.  Using the above notation, we now
  establish the characterization of $(\LB(\Pat, \T), \UB(\Pat, \T))$
  with the help of \cref{lm:st-periodic-nav-aux}. First, we observe
  that since $\str(v)$ is periodic, it follows by \cref{lm:pm-lce-2}
  that $\Pat$ is periodic and it holds $\Lroot(\Pat) = \Lroot(v) = H$
  and $\Lhead(\Pat) = \Lhead(v) = s$. Moreover, since $\rend{v} \leq
  |\str(v)|$ and $\type(v) = -1$, it follows by
  \cref{lm:pm-lce-3}\eqref{lm:pm-lce-3-it-1}, that $\rend{\Pat} =
  \rend{v}$, $\rendfull{\Pat} - 1 = \rendfull{v} - 1 = \ell$,
  $\Lexp(\Pat) = \Lexp(v) = k$, and $\type(\Pat) = \type(v) = -1$. In
  particular, this implies that the assumptions of
  \cref{lm:st-periodic-nav-aux} as satisfied. More precisely,
  $\rend{\Pat} = \rend{v} \leq |\str(v)| \leq |\Pat|$.  On the other
  hand, as shown above, $\{j \in \R^{-}_{s,H} : \Lexp(j) = k\} \neq
  \emptyset$. Observe also that by $\rendfull{\Pat} - 1 = \ell$, we
  have $\Pat' = \Pow(H) \cdot \Pat(\ell \dd |\Pat|]$. Putting all this
  together, by \cref{lm:st-periodic-nav-aux} we obtain that
  $(\LB(\Pat, \T), \UB(\Pat, \T)) = (e_{\mathsf{P}} - c_1,
  e_{\mathsf{P}} - c_2)$, where $c_1 = \rcount{\ARRnontail}{\ell}{e_H}
  - \rcount{\ARRnontail}{\ell}{\bpre}$, $c_2 =
  \rcount{\ARRnontail}{\ell}{e_H} -
  \rcount{\ARRnontail}{\ell}{\epre}$, $\bpre = |\{i \in [1 \dd q] :
  \T[\ARRzlex[i] \dd n] \prec \Pat'\}|$, and $(\bpre \dd \epre] = \{i
  \in [1 \dd q] : \Pat'\text{ is a prefix of }\T[\ARRzlex[i] \dd n]\}$.

  We are now ready to show the first claim. Recall, that by definition
  of $\mapToTZ(v)$, we have $\str(u) = \Pow(H) \cdot \str(v)(\ell \dd
  |\str(v)|]$. Thus, it holds $\str(u)c = \Pat'$. By definition of
  $\TZ$ and $\child(u, c)$, we thus obtain that $\child(u, c) = \nil$
  implies $\epre - \bpre = 0$. Consequently, by the above
  characterization, it holds
  \begin{align*}
    |\Occ(\Pat, \T)|
      &= \UB(\Pat, \T) - \LB(\Pat, \T)\\
      &= (e_{\mathsf{P}} - c_2) - (e_{\mathsf{P}} - c_1)\\
      &= \rcount{\ARRnontail}{\ell}{\epre} -
         \rcount{\ARRnontail}{\ell}{\bpre}\\
      &= 0.
  \end{align*}
  Thus, $\child(v, c) = \nil$.

  Let us now assume $\child(u, c) = u' \neq \nil$. Using the above
  notation, we first show the characterization of
  $\pseudoInvTZ{\ell}{u'}$. First, note that $\Pow(H)$ is a prefix of
  $\str(u')$ (since it is a prefix of $\str(u)$).  Next, note that
  $\ell \bmod |H| = (\rendfull{v} - 1) \bmod |H| = \Lhead(v) = s$ and
  $\lfloor \tfrac{\ell}{|H|} \rfloor = \lfloor \tfrac{\rendfull{v} -
  1}{|H|} \rfloor = \Lexp(v) = k$. As shown above, the set $\{j \in
  \R^{-}_{s,H} : \Lexp(j) = k\}$ is nonempty.  This implies that
  $\pseudoInvTZ{\ell}{u} = (e_{\mathsf{P}} - c'_1, e_{\mathsf{P}} -
  c'_2)$, where $c'_1 = \rcount{\ARRnontail}{\ell}{e_H} -
  \rcount{\ARRnontail}{\ell}{\lrank(u')}$ and $c'_2 =
  \rcount{\ARRnontail}{\ell}{e_H} -
  \rcount{\ARRnontail}{\ell}{\rrank(u')}$.  It remains to observe that
  by definition of $\TZ$ and the facts that $\child(u, c) = u'$ and
  $\str(u)c = \Pat'$, we have $\lrank(u') = \bpre$ and $\rrank(u') =
  \epre$. Thus, we have $c'_1 = c_1$ and $c'_2 = c_2$, and
  consequently
  \begin{align*}
    (\LB(\Pat, \T), \UB(\Pat, \T))
      &= (e_{\mathsf{P}} - c_1, e_{\mathsf{P}} - c_2)\\
      &= (e_{\mathsf{P}} - c'_1, e_{\mathsf{P}} - c'_2)\\
      &= \pseudoInvTZ{\ell}{u'}\\
      &= (b, e).
  \end{align*}
  By the above, if $b \neq e$, then $\Occ(\Pat, \T) \neq
  \emptyset$. This implies $\child(v, c) \neq \nil$ and
  $\repr(\child(v, c)) = (\LB(\Pat, \T), \UB(\Pat, \T))$. We thus
  indeed have $\repr(\child(v, c)) = (b, e)$.  Otherwise (i.e., if $b
  = e$), by the above we have $\Occ(\Pat, \T) = \emptyset$. This
  implies $\child(v, c) = \nil$ and hence indeed we also have
  $\repr(\child(v, c)) = (0, 0)$.
\end{proof}

\begin{remark}
  Note that, similarly as in \cref{lm:st-nonperiodic-child} (see
  \cref{rm:st-nonperiodic-child}), even though in the above result we
  have $\mapToTZ(v) = u$ and $\child(u, c)$ contains information used
  to determine $\child(v, c)$, it does not necessarily hold that
  $\mapToTZ(\child(v, c)) = \child(u, c)$.
\end{remark}

\begin{proposition}\label{pr:st-periodic-child}
  Let $v$ be an explicit periodic internal node of $\ST$. Given the
  data structure from \cref{sec:st-periodic-ds}, $\repr(v)$, and $c
  \in \Alphabet$, in $\bigO(\log \log n)$ time we can compute
  $\repr(\child(v, c))$.
\end{proposition}
\begin{proof}
  Denote $i_1 = \lrank(v) + 1$ and $i_2 = \rrank(v)$.  By
  \cref{lm:st-periodic-rend}\eqref{lm:st-periodic-rend-it-1}, it holds
  $\SA[i_1], \SA[i_2] \in \R$. Using \cref{pr:sa-periodic-sa}, in
  $\bigO(\log \log n)$ time we compute $j_1 = \SA[i_1]$ and $j_2 =
  \SA[i_2]$.  Next, using \cref{pr:isa-root} in $\bigO(1)$ time we
  compute $H = \Lroot(j_1)$, $s = \Lhead(j_1)$, $k_1 = \Lexp(j_1)$,
  $k_2 = \Lexp(j_2)$, $t_1 = \Ltail(j_1)$, and $t_2 = \Ltail(j_2)$.
  Observe that since $v$ is periodic, and $j_1, j_2 \in \Occ(\str(v),
  \T)$, it follows by \cref{lm:lce,lm:pm-lce} that $\Lroot(v) =
  \Lroot(j_2) = H$ and $\Lhead(v) = \Lhead(j_2) = s$. In $\bigO(1)$
  time we thus compute $e_1 := \rend{j_1} = j_1 + s + k_1|H| + t_1$
  and $e_2 := \rend{j_2} = j_2 + s + k_2|H| + t_2$.  Next, in
  $\bigO(1)$ time we compute $e_v := \rend{v} = 1 + \min(e_1 - j_1,
  e_2 - j_2)$ (see
  \cref{lm:st-periodic-rend}\eqref{lm:st-periodic-rend-it-1}).  Using
  \cref{lm:st-periodic-rend}\eqref{lm:st-periodic-rend-it-2}, we then
  in $\bigO(1)$ time check if it holds $\rend{v} \leq |\str(v)|$ by
  comparing $\T[j_1 + e_v - 1]$ with $\T[j_2 + e_v - 1]$. Consider two
  cases:
  \begin{itemize}
  \item Let $\T[j_1 + e_v - 1] = \T[j_2 + e_v - 1]$, i.e., $\rend{v}
    \leq |\str(v)|$. In $\bigO(1)$ time we compute $\type(v)$ by
    comparing $\T[j_1 + e_v - 1]$ with $\T[j_1 + e_v - 1 - |H|]$.  Let
    us assume that $\T[j_1 + e_v - 1] \prec \T[j_1 + e_v - 1 - |H|]$,
    i.e., $\type(v) = -1$ (the case $\type(v) = +1$ it handled
    symmetrically, using the part of the structure from
    \cref{sec:st-periodic-ds} adapted according to
    \cref{lm:lce}). Using \cref{pr:st-periodic-map}, in $\bigO(\log
    \log n)$ time we compute a pointer to $u = \mapToTZ(v)$. Using the
    representation of $\TZ$ stored as part of the structure in
    \cref{sec:st-periodic-ds}, and \cref{pr:compact-trie}, in
    $\bigO(\log \log n)$ time we check if $\child(u, c) = \nil$. If
    so, by \cref{lm:st-periodic-child} we have $\child(v, c) = \nil$,
    and thus we return $\repr(\child(v, c)) = (0, 0)$. Otherwise
    ($\child(u, c) \neq \nil$), we obtain a pointer to $u' = \child(u,
    c)$. In $\bigO(1)$ time we now compute $k := \Lexp(v) = \lfloor
    \tfrac{e_v - 1 - s}{|H|} \rfloor$ and $\ell := \rendfull{v} - 1 =
    s + k|H|$. Using \cref{pr:st-periodic-invmap}, in $\bigO(\log \log
    n)$ time we then compute the pair $(b, e) =
    \pseudoInvTZ{\ell}{u'}$. If $b = e$ then by
    \cref{lm:st-periodic-child} it holds $\child(v, c) = \nil$ and
    hence we return $\repr(\child(v, c)) = (0, 0)$. Otherwise, by
    \cref{lm:st-periodic-child}, it holds $\repr(\child(v, c)) = (b,
    e)$. We thus return $(b, e)$.
  \item Let $\T[j_1 + e_v - 1] \neq \T[j_2 + e_v - 1]$, i.e.,
    $\rend{v} > |\str(v)|$. Denote $\Pat = \str(v)$. We then have
    $\rend{\Pat} > |\Pat|$, $\Lhead(\Pat) = s$, $\Lroot(\Pat) = H$,
    and $|\Pat| = e_v - 1$.  Using \cref{pr:st-periodic-range-2}, in
    $\bigO(\log \log n)$ time we compute $(b, e) = (\LB(\Pat c, \T),
    \UB(\Pat c, \T))$.  If $b = e$, then $\Occ(\Pat, \T) =
    \Occ(\str(v)c, \T) = \emptyset$, and hence $\child(v, c) =
    \nil$. We thus return $\repr(\child(v, c)) = (0, 0)$.  Otherwise,
    we return that $\repr(\child(v, c)) = (b, e)$. \qedhere
  \end{itemize}
\end{proof}

\subsubsection{Implementation of
  \texorpdfstring{$\pred(v, c)$}{pred(v, c)}}\label{sec:st-periodic-pred}

\begin{lemma}\label{lm:st-periodic-pred}
  Let $c \in \Alphabet$ and $v$ be an explicit periodic internal node
  of $\ST$ satisfying $\rend{v} \leq |\str(v)|$ and $\type(v) =
  -1$. Let $u = \mapToTZ(v)$. If $\pred(u,c) = \nil$ then
  $\LB(\str(v)c, \T) = \LB(\str(v), \T)$. Otherwise, letting $u' =
  \pred(u,c)$, it holds
  \vspace{0.5ex}
  \[
    \LB(\str(v)c, \T) = e,
    \vspace{0.5ex}
  \]
  where $(b,e) = \pseudoInvTZ{\ell}{u'}$ and $\ell =
  \rendfull{v} - 1$.
\end{lemma}
\begin{proof}
  We start by characterizing $\LB(\str(v), \T)$.  Let $H = \Lroot(v)$,
  $s = \Lhead(v)$, $k = \Lexp(v)$, and $\mathsf{P} = \{j \in
  \R^{-}_{s,H} : \Lexp(j) = k\}$.  In the proof of
  \cref{lm:st-periodic-child}, we showed that $\mathsf{P} \neq
  \emptyset$. Let $b_{\mathsf{P}}, e_{\mathsf{P}} \in [0 \dd n]$ be
  such that $\{\SA[i]\}_{i \in (b_{\mathsf{P}} \dd e_{\mathsf{P}}]} =
  \mathsf{P}$, and $b_H, e_H \in [0 \dd q]$ be such that
  $\{\rlexm_i\}_{i \in (b_H \dd e_H]} = \R'^{-}_H$.  We now
  additionally note that by definition of $\mapToTZ(v)$, we have
  $\str(u) = \Pow(H) \cdot \str(v)(\ell \dd |\str(v)|]$. Thus, by
  definition of $\TZ$, it holds $|\{i \in [1 \dd q] : \T[\ARRzlex[i]
  \dd n] \prec \Pow(H) \cdot \str(v)(\ell \dd |\str(v)|]\}| =
  \lrank(u)$.  By \cref{lm:st-periodic-nav-aux} for pattern $\str(v)$,
  we thus obtain $\LB(\str(v), \T) = e_{\mathsf{P}} -
  (\rcount{\ARRnontail}{\ell}{e_H} -
  \rcount{\ARRnontail}{\ell}{\lrank(u)})$.

  Next, we characterize $\LB(\str(v)c, \T)$.  Denote $\Pat = \str(v)c$
  and $\Pat' = \Pow(H) \cdot \Pat[\rendfull{\Pat} \dd |\Pat|]$. In the
  proof of \cref{lm:st-periodic-child}, we observed that $\Pat$ is
  periodic and it holds $\Lroot(\Pat) = \Lroot(v) = H$, $\Lhead(\Pat)
  = \Lhead(v) = s$, $\rend{\Pat} = \rend{v}$, $\rendfull{\Pat} - 1 =
  \rendfull{v} - 1 = \ell$, $\Lexp(\Pat) = \Lexp(v) = k$, and
  $\type(\Pat) = \type(v) = -1$. Moreover, we noted that $\rend{\Pat}
  \leq |\Pat|$ and $\Pat' = \Pow(H) \cdot \Pat(\ell \dd
  |\Pat|]$. Finally, putting all this together, we observed that
  $(\LB(\Pat, \T), \UB(\Pat, \T)) = (e_{\mathsf{P}} - c_1,
  e_{\mathsf{P}} - c_2)$, where $c_1 = \rcount{\ARRnontail}{\ell}{e_H}
  - \rcount{\ARRnontail}{\ell}{\bpre}$, $c_2 =
  \rcount{\ARRnontail}{\ell}{e_H} -
  \rcount{\ARRnontail}{\ell}{\epre}$, $\bpre = |\{i \in [1 \dd q] :
  \T[\ARRzlex[i] \dd n] \prec \Pat'\}|$, and $(\bpre \dd \epre] = \{i
  \in [1 \dd q] : \Pat'\text{ is a prefix of }\T[\ARRzlex[i] \dd
  n]\}$.

  Let us first assume $\pred(u, c) = \nil$. By definition, this
  implies $|\{i \in [1 \dd q] : \T[\ARRzlex[i] \dd n] \prec
  \str(u)c\}| = \lrank(u)$.  Recall, however, that $\str(u) = \Pow(H)
  \cdot \str(v)(\ell \dd |\str(v)|]$. Thus, $\str(u)c = \Pat'$ and
  consequently $\bpre = \lrank(u)$. Using the above characterization,
  we thus have $\LB(\str(v)c, \T) = e_{\mathsf{P}} -
  (\rcount{\ARRnontail}{\ell}{e_H} -
  \rcount{\ARRnontail}{\ell}{\lrank(u)})$.  Since above we also
  established that $\LB(\str(v), \T) = e_{\mathsf{P}} -
  (\rcount{\ARRnontail}{\ell}{e_H} -
  \rcount{\ARRnontail}{\ell}{\lrank(u)})$, we have thus proved that
  $\pred(u, c) = \nil$ implies $\LB(\str(v)c, \T) = \LB(\str(v), \T)$.

  Let us now assume $\pred(u, c) = u' \neq \nil$. By definition of
  $\pred(u, c)$, this implies $|\{i \in [1 \dd q] : \T[\ARRzlex[i] \dd
  n] \prec \str(u)c\}| = \rrank(u')$.  By recalling again that
  $\str(u)c = \Pat'$, we thus have $\bpre = \rrank(u')$. By the above
  characterization, we thus have $\LB(\str(v)c, \T) = e_{\mathsf{P}} -
  (\rcount{\ARRnontail}{\ell}{e_H} -
  \rcount{\ARRnontail}{\ell}{\rrank(u')})$.  On the other hand, by
  definition of $(b, e) = \pseudoInvTZ{\ell}{u'}$, we have $e =
  e_{\mathsf{P}} - (\rcount{\ARRnontail}{\ell}{e_H} -
  \rcount{\ARRnontail}{\ell}{\rrank(u')})$. We thus obtain
  $\LB(\str(v)c, \T) = e$.
\end{proof}

\begin{proposition}\label{pr:st-periodic-pred}
  Let $v$ be an explicit periodic internal node of $\ST$.  Given the
  data structure from \cref{sec:st-periodic-ds}, $\repr(v)$, and $c
  \in \Alphabet$, in $\bigO(\log \log n)$ time we can compute
  $\LB(\str(v)c, \T)$.
\end{proposition}
\begin{proof}
  Denote $i_1 = \lrank(v) + 1$ and $i_2 = \rrank(v)$.  By
  \cref{lm:st-periodic-rend}\eqref{lm:st-periodic-rend-it-1}, it holds
  $\SA[i_1], \SA[i_2] \in \R$. Using \cref{pr:sa-periodic-sa}, in
  $\bigO(\log \log n)$ time we compute $j_1 = \SA[i_1]$ and $j_2 =
  \SA[i_2]$.  Next, using \cref{pr:isa-root} in $\bigO(1)$ time we
  compute $H = \Lroot(j_1)$, $s = \Lhead(j_1)$, $k_1 = \Lexp(j_1)$,
  $k_2 = \Lexp(j_2)$, $t_1 = \Ltail(j_1)$, and $t_2 = \Ltail(j_2)$.
  Observe that since $v$ is periodic, and $j_1, j_2 \in \Occ(\str(v),
  \T)$, it follows by \cref{lm:lce,lm:pm-lce} that $\Lroot(v) =
  \Lroot(j_2) = H$ and $\Lhead(v) = \Lhead(j_2) = s$. In $\bigO(1)$
  time we thus compute $e_1 := \rend{j_1} = j_1 + s + k_1|H| + t_1$
  and $e_2 := \rend{j_2} = j_2 + s + k_2|H| + t_2$.  Next, in
  $\bigO(1)$ time we compute $e_v := \rend{v} = 1 + \min(e_1 - j_1,
  e_2 - j_2)$ (see
  \cref{lm:st-periodic-rend}\eqref{lm:st-periodic-rend-it-1}).  Using
  \cref{lm:st-periodic-rend}\eqref{lm:st-periodic-rend-it-2}, we then
  in $\bigO(1)$ time check if it holds $\rend{v} \leq |\str(v)|$ by
  comparing $\T[j_1 + e_v - 1]$ with $\T[j_2 + e_v - 1]$. Consider two
  cases:
  \begin{itemize}
  \item Let $\T[j_1 + e_v - 1] = \T[j_2 + e_v - 1]$, i.e., $\rend{v}
    \leq |\str(v)|$. In $\bigO(1)$ time we compute $\type(v)$ by
    comparing $\T[j_1 + e_v - 1]$ with $\T[j_1 + e_v - 1 - |H|]$.  Let
    us assume that $\T[j_1 + e_v - 1] \prec \T[j_1 + e_v - 1 - |H|]$,
    i.e., $\type(v) = -1$ (the case $\type(v) = +1$ it handled
    symmetrically, using the part of the structure from
    \cref{sec:st-periodic-ds} adapted according to
    \cref{lm:lce}). Using \cref{pr:st-periodic-map}, in $\bigO(\log
    \log n)$ time we compute a pointer to $u = \mapToTZ(v)$. Using the
    representation of $\TZ$ stored as part of the structure in
    \cref{sec:st-periodic-ds}, and \cref{pr:compact-trie}, in
    $\bigO(\log \log n)$ time we check if $\pred(u, c) = \nil$. If so,
    by \cref{lm:st-periodic-pred} we have $\LB(\str(v)c, \T) =
    \LB(\str(v), \T)$, and thus we return $\rrank(v)$ as the
    answer. Otherwise ($\pred(u, c) \neq \nil$), we obtain a pointer
    to $u' = \pred(u, c)$. In $\bigO(1)$ time we now compute $k :=
    \Lexp(v) = \lfloor \tfrac{e_v - 1 - s}{|H|} \rfloor$ and $\ell :=
    \rendfull{v} - 1 = s + k|H|$. Using \cref{pr:st-periodic-invmap},
    in $\bigO(\log \log n)$ time we then compute the pair $(b, e) =
    \pseudoInvTZ{\ell}{u'}$. By \cref{lm:st-periodic-pred}, we then
    have $\LB(\str(v)c, \T) = e$. Thus, we return $e$ as the answer.
  \item Let $\T[j_1 + e_v - 1] \neq \T[j_2 + e_v - 1]$, i.e.,
    $\rend{v} > |\str(v)|$. Denote $\Pat = \str(v)$. We then have
    $\rend{\Pat} > |\Pat|$, $\Lhead(\Pat) = s$, $\Lroot(\Pat) = H$,
    and $|\Pat| = e_v - 1$.  Using \cref{pr:st-periodic-range-2}, in
    $\bigO(\log \log n)$ time we compute $(b, e) = (\LB(\Pat c, \T),
    \UB(\Pat c, \T))$. We then return $b$ as the answer. \qedhere
  \end{itemize}
\end{proof}

\subsubsection{Implementation of
  \texorpdfstring{$\WA(v,d)$}{WA(v, d)}}\label{sec:st-periodic-wa}

\begin{lemma}\label{lm:st-periodic-wa}
  Let $v$ be an explicit periodic node of $\ST$ satisfying $\type(v) =
  -1$ and $d$ be an integer satisfying $\rend{v} \leq d \leq
  |\str(v)|$. Then, letting $u = \mapToTZ(v)$, it holds
  \[
    \repr(\WA(v, d)) = \pseudoInvTZ{\ell}{\widehat{u}},
  \]
  where $\ell = \rendfull{v} - 1$, $H = \Lroot(v)$, and $\widehat{u} =
  \WA(u, d - \ell + |\Pow(H)|)$.
\end{lemma}
\begin{proof}
  As in the proof of \cref{lm:st-nonperiodic-wa}, let us denote
  $f^{(0)}(x) = x$ and $f^{(i)}(x) = f(f^{(i-1)}(x))$ for $i \in
  \Zp$. Let
  \vspace{-1ex}
  \begin{align*}
    \mathcal{V} &:=
      \{\parent^{(i)}(v) : i \in\Zz\text{ and }
      \sdepth(\parent^{(i)}(v)) \geq \rend{v}\}\text{ and }\\
    \mathcal{U} &:=
      \{\parent^{(i)}(u) : i \in \Zz\text{ and }
      \sdepth(\parent^{(i)}(u)) \geq |\Pow(H)|{+}\Ltail(v){+}1\}
  \end{align*}
  By $\rend{v} \leq |\str(v)|$, $\type(v) = -1$, and
  \cref{lm:pm-lce-3}\eqref{lm:pm-lce-3-it-1}, for every $v' \in
  \mathcal{V}$ it holds that $\str(v')$ is periodic, and we have
  $\rend{v'} = \rend{v} \leq |\str(v')|$, $\type(v') = \type(v) = -1$,
  and $\rendfull{v'} = \rendfull{v} = \ell + 1$.  Thus, for every $v'
  \in \mathcal{V}$, the node $u' = \mapToTZ(v')$ is well-defined and
  satisfies $\str(u') = \Pow(H) \cdot \str(v')[\rendfull{v'} \dd
  |\str(v')|] = \Pow(H) \cdot \str(v')(\ell \dd |\str(v')|]$. In
  particular, $\str(u) = \Pow(H) \cdot \str(v)(\ell \dd |\str(v)|]$.
  Since for any $v' \in \mathcal{V}$, $\str(v') = \str(v)[1 \dd
  |\str(v')|]$, we thus obtain that for $u' = \mapToTZ(v')$ it holds
  $\str(u') = \Pow(H) \cdot \str(v')(\ell \dd |\str(v')|] = \Pow(H)
  \cdot \str(v)(\ell \dd |\str(v')|] = \str(u)[1 \dd |\str(u')|]$.
  i.e., $u'$ is an ancestor of $u$.  Moreover, $\sdepth(u') =
  |\Pow(H)| + |\str(v')| - \rendfull{v'} + 1 = |\Pow(H)| + |\str(v')|
  - \rendfull{v} + 1 \geq |\Pow(H)| + \rend{v} - \rendfull{v} + 1 =
  |\Pow(H)| + \Ltail(v) + 1$.  Consequently, $\mathcal{U}' :=
  \{\mapToTZ(v') : v' \in \mathcal{V}\}$ satisfies $\mathcal{U}' \sub
  \mathcal{U}$. Note also, that $v' \neq v''$ implies $\mapToTZ(v')
  \neq \mapToTZ(v'')$.

  For any $u' \in \mathcal{U}$, denote $(\fbeg(u'), \fend(u')) =
  \pseudoInvTZ{\ell}{u'}$.  We prove the following property of
  $\mathcal{U}'$. Let $w, w' \in \mathcal{U}$ be such that $w =
  \parent(w')$. We claim, that $(\fbeg(w), \fend(w)) \neq (\fbeg(w'),
  \fend(w'))$ implies $w \in \mathcal{U}'$. The proof consists of five
  steps:
  \begin{enumerate}

  \item\label{lm:st-periodic-wa-it-1}
    For any node $y$ of $\TZ$ such that $\Pow(H)$ is a prefix of
    $\str(y)$, by $S_y$ we denote a string such that $\str(y) =
    \Pow(H) \cdot S_y$.  Let $P$ and $P'$ be such that $PP'$ is a
    prefix of $\str(v)$, and it holds $|P| = \ell$ and $|P'| =
    \Ltail(v) + 1$ (which is defined by $\rend{v} \leq |\str(v)|$).
    Consider any node $y$ of $\TZ$ such that $\Pow(H) \cdot P'$ is a
    prefix of $\str(y)$ (note that although this includes all nodes in
    $\mathcal{U}$, it is possible that $y \not\in \mathcal{U}$).  We
    prove that for any such $y$, it holds $|\Occ(PS_y, \T)| =
    \rcount{\ARRnontail}{\ell}{\rrank(y)} -
    \rcount{\ARRnontail}{\ell}{\lrank(y)}$.  First, observe that since
    $\rend{\str(v)} = |PP'|$, we obtain by $\lcp(PS_y, \str(v)) \geq
    |PP'|$ and \cref{lm:pm-lce-3}\eqref{lm:pm-lce-3-it-1}, that $PS_y$
    is periodic, $\rend{PS_y} = \rend{\str(v)} = |PP'| \leq |PS_y|$,
    $\rendfull{PS_y} = \rendfull{\str(v)} = \ell + 1$, and
    $\type(PS_y) = \type(\str(v)) = -1$.  By \cref{lm:occ},
    $\Occ(PS_{y}, \T)$ is thus a disjoint union of $\Occa(PS_{y}, \T)$
    and $\Occs(PS_{y}, \T)$ (see the beginning of
    \cref{sec:pm-periodic-pm} for definitions).  By $\rend{PS_y} \leq
    |PS_y|$, \cref{lm:occ-a} and its symmetric version (adapted
    according to \cref{lm:pm-lce}) moreover imply that $\Occa(PS_y,
    \T) = \emptyset$. Finally, by $\type(PS_y) = -1$ and
    \cref{lm:pm-lce-3}\eqref{lm:pm-lce-3-it-2}, it follows that
    $\Occ(PS_y, \T) \sub \R^{-}$.  Thus, $\Occs(PS_y, \T) =
    \Occsm(PS_y, \T)$ and consequently $\Occ(PS_y, \T) = \Occsm(PS_y,
    \T)$.  It thus remains to prove $|\Occsm(PS_y, \T)| =
    \rcount{\ARRnontail}{\ell}{\rrank(y)} -
    \rcount{\ARRnontail}{\ell}{\lrank(y)}$. Recall that the set
    $\{\Pow(H) : H \in \Lroots\}$ is prefix-free. Letting $H_j =
    \Lroot(j)$ (where $j \in \R$), it follows by definition of $\TZ$
    that:
    \begin{align*}
      \{\rlexm_i\}_{i \in (\lrank(y) \dd \rrank(y)]}
        &= \{j \,{\in}\, \R'^{-} :
             \Pow(H) \cdot S_y\text{ is a prefix of }
             \T[\rendfull{j} \,{-}\, |\Pow(H_j)| \dd n]\}\\
        &= \{j \,{\in}\, \R'^{-} :
             \Pow(H) \cdot S_y\text{ is a prefix of }
             \Pow(H_j) \cdot \T[\rendfull{j} \dd n]\}\\
        &= \{j \,{\in}\, \R'^{-}_{H} :
             S_y\text{ is a prefix of }
             \T[\rendfull{j} \dd n]\}.
    \end{align*}
    Finally, note that by $\rendfull{PS_y} = \ell + 1$, we have
    $(PS_y)[\rendfull{PS_y} \dd |PS_y|] = S_y$.  Thus, by the above
    and \cref{lm:occ-s}, $|\Occsm(PS_y, \T)| = |\{i \in (\lrank(y) \dd
    \rrank(y)] : \ARRnontail[i] \geq \rendfull{PS_y} - 1\}| = |\{i \in
    (\lrank(y) \dd \rrank(y)] : \ARRnontail[i] \geq \ell\}| =
    \rcount{\ARRnontail}{\ell}{\rrank(y)} -
    \rcount{\ARRnontail}{\ell}{\lrank(y)}$.

  \item\label{lm:st-periodic-wa-it-2}
    We prove that there exists $c' \in \Alphabet$ such that
    $|\Occ(PS_w c', \T)| > 0$ (where $S_{w}$ is defined as above).
    First, note that by $u = \mapToTZ(v)$, it holds $\str(u) = \Pow(H)
    \cdot \str(v)[\rendfull{v} \dd |\str(v)|] = \Pow(H) \cdot
    \str(v)(\ell \dd |\str(v)|]$. On the other hand, by definition of
    $S_u$, we have $\str(u) = \Pow(H) \cdot S_u$. Thus, $S_u =
    \str(v)(\ell \dd |\str(v)|]$. By $P = \str(v)[1 \dd \ell]$, we
    thus obtain $\str(v) = PS_u$. Consequently, since $v$ is a node of
    $\ST$, we have $|\Occ(PS_u, \T)| = |\Occ(\str(v), \T)| > 0$.
    Observe now that since $w'$ is an ancestor of $u$, the string
    $\str(w') = \Pow(H) \cdot S_{w'}$ is a prefix of $\str(u) =
    \Pow(H) \cdot S_u$. This implies that $S_{w'}$ is a prefix of
    $S_u$, and hence $PS_{w'}$ is a prefix of $PS_u$. Consequently,
    $|\Occ(PS_{w'}, \T)| \geq |\Occ(PS_u, \T)| > 0$. In particular,
    since $PS_w$ is a prefix of $PS_{w'}$, letting $c' \in \Alphabet$
    be such that $\child(w, c') = w'$, we have $|\Occ(PS_w c', \T)| >
    0$.

  \item\label{lm:st-periodic-wa-it-3}
    Let $s = \ell \bmod |H|$, $k = \lfloor \tfrac{\ell}{|H|} \rfloor$.
    Let also $\mathsf{P} := \{j \in \R^{-}_{s,H} : \Lexp(j) = k\}$,
    $b_{\mathsf{P}}, e_{\mathsf{P}} \in [0 \dd n]$ be such that
    $\{\SA[i]\}_{i \in (b_{\mathsf{P}} \dd e_{\mathsf{P}}]} =
    \mathsf{P}$, and $b_H, e_H \in [0 \dd q]$ be such that
    $\{\rlexm_i\}_{i \in (b_H \dd e_H]} = \R'^{-}_{H}$.  Note that
    $\Lhead(v) = s$, $\Lroot(v) = H$, $\Lexp(v) = k$, $\rend{v} \leq
    |\str(v)|$, and $\type(v) = -1$ imply that $\mathsf{P} \neq
    \emptyset$ (it suffices to take $j = \SA[i]$ for any $i \in
    (\lrank(v) \dd \rrank(v)]$, and apply \cref{lm:pm-lce} and
    \cref{lm:pm-lce-3}\eqref{lm:pm-lce-3-it-2}). Therefore,
    $(b_{\mathsf{P}}, e_{\mathsf{P}})$ and $(b_H, e_H)$ are
    well-defined.  Denote $\delta = e_{\mathsf{P}} -
    \rcount{\ARRnontail}{\ell}{e_H}$. By definition of
    $\pseudoInvTZ{\ell}{w}$ and $\pseudoInvTZ{\ell}{w'}$, we then have
    $(\fbeg(w), \fend(w)) = (\delta +
    \rcount{\ARRnontail}{\ell}{\lrank(w)}, \delta +
    \rcount{\ARRnontail}{\ell}{\rrank(w)})$ and $(\fbeg(w'),
    \fend(w')) = (\delta + \rcount{\ARRnontail}{\ell}{\lrank(w')},
    \delta + \rcount{\ARRnontail}{\ell}{\rrank(w')})$.  Thus, the
    assumption $(\fbeg(w), \fend(w)) \neq (\fbeg(w'), \fend(w'))$, or
    equivalently, $\fbeg(w) \neq \fbeg(w')$ or $\fend(w) \neq
    \fend(w')$, implies
    \begin{align*}
      \rcount{\ARRnontail}{\ell}{\lrank(w)} &\neq
        \rcount{\ARRnontail}{\ell}{\lrank(w')}\text{ or}\\
      \rcount{\ARRnontail}{\ell}{\rrank(w)} &\neq
        \rcount{\ARRnontail}{\ell}{\rrank(w')}.
    \end{align*}

  \item\label{lm:st-periodic-wa-it-4}
    By definition, the values
    $\rcount{\ARRnontail}{\ell}{\rrank(\widehat{w})} -
    \rcount{\ARRnontail}{\ell}{\lrank(\widehat{w})}$ over all children
    $\widehat{w}$ of $w$ sum up to
    $\rcount{\ARRnontail}{\ell}{\rrank(w)} -
    \rcount{\ARRnontail}{\ell}{\lrank(w)}$.  Thus, it follows by
    Step~\ref{lm:st-periodic-wa-it-3} that there exists a child $w''
    \neq w'$ of $w$ such that $\rcount{\ARRnontail}{\ell}{\rrank(w'')}
    - \rcount{\ARRnontail}{\ell}{\lrank(w'')} > 0$.  By
    Step~\ref{lm:st-periodic-wa-it-1}, for such $w''$, we thus have
    $|\Occ(PS_{w''}, \T)| = \rcount{\ARRnontail}{\ell}{\rrank(w'')} -
    \rcount{\ARRnontail}{\ell}{\lrank(w'')} > 0$.  In particular,
    letting $c'' \in \Alphabet$ be such that $\child(w, c'') = w''$,
    it holds $|\Occ(PS_w c'', \T)| > 0$. Note that $w'' \neq w'$
    implies $c'' \neq c'$.

  \item
    We have thus proved (Steps~\ref{lm:st-periodic-wa-it-2}
    and~\ref{lm:st-periodic-wa-it-4}) that there exist $c', c'' \in
    \Alphabet$ such that $c' \neq c''$, $|\Occ(PS_{w} c', \T)| > 0$,
    and $|\Occ(PS_{w} c'', \T)| > 0$. This implies that there exists a
    node $v'$ in $\ST$ such that $\str(v') = PS_{w}$. As observed in
    Step~\ref{lm:st-periodic-wa-it-1}, $PS_w$ is periodic, and it
    holds $\rend{PS_w} \leq |PS_{w}|$, $\type(PS_{w}) = -1$, and
    $\rendfull{PS_w} = |P| + 1$. Thus, the node $u' = \mapToTZ(v')$ is
    defined and satisfies $\str(u') = \Pow(H) \cdot S_w$.  This
    implies $u' = w$, and consequently, $w \in \mathcal{U}'$.
  \end{enumerate}

  We are now ready to prove the main claim.  Let $v' = \WA(v, d)$ and
  $v'' = \parent(v')$.  We then have $\sdepth(v'') < d \leq
  \sdepth(v')$. Moreover, by $\rend{v} \leq d$, we have $v' \in
  \mathcal{V}$. Let $u' = \mapToTZ(v')$.  As observed earlier, we
  then have $u' \in \mathcal{U}'$, and hence $d - \ell + |\Pow(H)|
  \leq \sdepth(u')$.  This implies that $\widehat{u} = \WA(u', d -
  \ell + |\Pow(H)|)$ satisfies $d - \ell + |\Pow(H)| \leq
  \sdepth(\widehat{u}) \leq \sdepth(u')$.  By $d - \ell + |\Pow(H)|
  \geq \rend{v} - \ell + |\Pow(H)| = \rend{v} - (\rendfull{v} - 1) +
  |\Pow(H)| = \Ltail(v) + 1 + |\Pow(H)|$, this implies that
  $\widehat{u} \in \mathcal{U}$.  Let $k \in \Zz$ be
  such that $\widehat{u} = \parent^{(k)}(u')$.  This implies that
  $\parent^{(i)}(u')\not\in \mathcal{U}'$ holds for $i \in [1 \dd k]$,
  since otherwise it would contradict $v' = \WA(v, d)$. If $k = 0$
  then we trivially have $(\fbeg(u'), \fend(u')) =
  (\fbeg(\widehat{u}), \fend(\widehat{u}))$. Otherwise, by (the
  contraposition of) the above property of $\mathcal{U}'$ we have
  \begin{align*}
    (\fbeg(u'), \fend(u'))
      &= (\fbeg(\parent(u')), \fend(\parent(u')))\\
      &= \dots\\
      &= (\fbeg(\parent^{(k)}(u')), \fend(\parent^{(k)}(u')))\\
      &= (\fbeg(\widehat{u}), \fend(\widehat{u})).
  \end{align*}
  Recall now that $\rend{v'} \leq |\str(v')|$, $\type(v') = -1$, and
  $\rendfull{v'} = \rendfull{v} = \ell + 1$. Thus, by
  \cref{lm:st-periodic-invmap}, we have $\repr(v') =
  \pseudoInvTZ{\rendfull{v'} - 1}{u'} = \pseudoInvTZ{\ell}{u'}$.
  Consequently, $\repr(\WA(v, d)) = \pseudoInvTZ{\ell}{u'} =
  (\fbeg(u'), \fend(u')) = (\fbeg(\widehat{u}), \fend(\widehat{u})) =
  \pseudoInvTZ{\ell}{\widehat{u}}$.
\end{proof}

\begin{proposition}\label{pr:st-periodic-wa}
  Let $v$ be an explicit periodic node of $\ST$. Given the data
  structure from \cref{sec:st-periodic-ds}, $\repr(v)$, and an integer
  $d$ satisfying $3\tau - 1 \leq d \leq |\str(v)|$, in $\bigO(\log
  \log n)$ time we can compute $\repr(\WA(v, d))$.
\end{proposition}
\begin{proof}
  Denote $i_1 = \lrank(v) + 1$ and $i_2 = \rrank(v)$.  By
  \cref{lm:st-periodic-rend}\eqref{lm:st-periodic-rend-it-1}, it holds
  $\SA[i_1], \SA[i_2] \in \R$. Using \cref{pr:sa-periodic-sa}, in
  $\bigO(\log \log n)$ time we first compute $j_1 = \SA[i_1]$ and $j_2
  = \SA[i_2]$. Next, using \cref{pr:isa-root}, in $\bigO(1)$ time we
  compute $H = \Lroot(j_1)$, $s = \Lhead(j_1)$, $k_1 = \Lexp(j_1)$,
  $k_2 = \Lexp(j_2)$, $t_1 = \Ltail(j_1)$, and $t_2 =
  \Ltail(j_2)$. Since $v$ is periodic, and $j_1, j_2 \in \Occ(\str(v),
  \T)$, it follows by \cref{lm:lce,lm:pm-lce} that $\Lroot(v) =
  \Lroot(j_2) = H$ and $\Lhead(v) = \Lhead(j_2) = s$. In $\bigO(1)$
  time we thus compute $e_1 := \rend{j_1} = j_1 + s + k_1|H| + t_1$
  and $e_2 := \rend{j_2} = j_2 + s + k_2|H| + t_2$.  Next, in
  $\bigO(1)$ time we compute $e_v := \rend{v} = 1 + \min(e_1 - j_1,
  e_2 - j_2)$ (see
  \cref{lm:st-periodic-rend}\eqref{lm:st-periodic-rend-it-1}).  We
  then consider two cases:
  \begin{itemize}
  \item Assume $e_v \leq d$. Then, to obtain $\repr(\WA(v, d))$ we
    follow \cref{lm:st-periodic-wa}. First, in $\bigO(1)$ time we
    compute $\type(v)$ by comparing $\T[j_1 + e_v - 1]$ with $\T[j_1 +
    e_v - 1 - |H|]$.  Let us assume that $\T[j_1 + e_v - 1] \prec
    \T[j_1 + e_v - 1 - |H|]$, i.e., $\type(v) = -1$ (the case
    $\type(v) = +1$ it handled symmetrically, using the part of the
    structure from \cref{sec:st-periodic-ds} adapted according to
    \cref{lm:lce}). Using \cref{pr:st-periodic-map}, in $\bigO(\log
    \log n)$ time we compute a pointer to $u = \mapToTZ(v)$. In
    $\bigO(1)$ time we also calculate $|\Pow(H)| = |H| \lceil
    \tfrac{\tau}{|H|} \rceil$.  Using the representation of $\TZ$
    stored as part of the structure in \cref{sec:st-periodic-ds}, and
    \cref{pr:compact-trie}, in $\bigO(\log \log n)$ time we compute a
    pointer to $\widehat{u} = \WA(u, d - \ell + |\Pow(H)|)$. In
    $\bigO(1)$ time we now compute $k := \Lexp(v) = \lfloor \tfrac{e_v
    - 1 - s}{|H|} \rfloor$ and $\ell := \rendfull{v} - 1 = s +
    k|H|$. Using \cref{pr:st-periodic-invmap}, in $\bigO(\log \log n)$
    time we then compute the pair $(b, e) =
    \pseudoInvTZ{\ell}{\widehat{u}}$. By \cref{lm:st-periodic-wa}, it
    holds $\repr(\WA(v, d)) = (b, e)$. We thus return $(b, e)$.
  \item Assume $e_v > d$. Let $v' = \WA(v, d)$, $S = \str(v')$, and
    $S' = S[1 \dd d]$. Since, by definition, $v'$ does not have an
    ancestor $v''$ in $\ST$ satisfying $\sdepth(v'') \geq d$, it holds
    $\repr(v') = (\LB(S, \T), \UB(S, \T)) = (\LB(S', \T), \UB(S',
    \T))$.  We thus focus on computing the latter pair. First, we
    observe that since $v'$ is an ancestor $v$, we have $S' =
    \str(v)[1 \dd d]$.  Therefore, since $\str(v)$ is periodic, and it
    holds $3\tau - 1 \leq d$, we obtain by \cref{lm:pm-lce-2} that
    $S'$ is periodic, and it holds $\Lroot(S') = \Lroot(v) = H$ and
    $\Lhead(S') = \Lhead(v) = s$. To show $\rend{S'} > |S'|$, let us
    denote $Q = \str(v)[1 \dd \rend{v})$. By definition, we have
    $\rend{Q} = 1 + p + \lcp(Q, Q(p \dd |Q|]) = |Q| + 1$. Thus, we
    must have $\lcp(Q, Q(p \dd |Q|]) = |Q| - p$. Consequently, since
    by $e_v = \rend{v} > d$ the string $S'$ is a prefix of $Q$, we
    have $\lcp(S', S'(p \dd |S'|]) = |S'| - p$, and hence $\rend{S'} =
    1 + p + \lcp(S', S'(p \dd |S'|]) = |S'| + 1$.  Considering all
    the above properties of $S'$, the next step of the algorithm is
    therefore to compute and return the pair $(b, e) = (\LB(S', \T),
    \UB(S', \T))$ in $\bigO(\log \log n)$ time using
    \cref{pr:st-periodic-range}. As observed above, it holds
    $\repr(\WA(v, d)) = (b, e)$. \qedhere
  \end{itemize}
\end{proof}

\subsubsection{Construction Algorithm}\label{sec:st-periodic-construction}

\begin{proposition}\label{pr:st-periodic-construction}
  Given $\STCore(\T)$, we can in $\bigO(n / \log_{\sigma} n)$
  time we can augment it into a data structure from
  \cref{sec:st-periodic-ds}.
\end{proposition}
\begin{proof}
  First, we combine
  \cref{pr:sa-core-construction,pr:sa-periodic-construction} (recall
  that the packed representation of $\T$ is a component of
  $\STCore(\T)$) to construct the data structure from
  \cref{sec:sa-periodic-ds} in $\bigO(n / \log_{\sigma} n)$ time.  In
  particular, this constructs $(\rlexm_i)_{i \in [1 \dd q]}$.  Using
  \cref{pr:isa-root}, we can now compute $\ARRzlex[i]$ for any $i \in
  [1 \dd q]$ in $\bigO(1)$ time.  Then, in $\bigO(n / \log_{\sigma}
  n)$ time we construct the data structure $\TZ$ using
  \cref{pr:compact-trie}.

  After the above components are constructed, we then analogously
  construct their symmetric counterparts (adapted according to
  \cref{lm:lce}).
\end{proof}

\subsection{The Final Data Structure}\label{sec:st-final}

In this section, we put together
\cref{sec:st-core,sec:st-nonperiodic,sec:st-periodic} to obtain a data
structure that performs suffix tree operations in
$\bigO(\log^{\epsilon} n)$ time.

The section is organized as follows. First, we introduce the
components of the data structure (\cref{sec:st-ds}).  We then describe
the query algorithms for all operations in \cref{tab:st-operations}
(\cref{sec:st-sdepth,sec:st-lca,sec:st-child,sec:st-pred,%
sec:st-wa,sec:st-wlink,sec:st-slink,sec:st-slink,sec:st-slink-iter,%
sec:st-parent,sec:st-firstchild,sec:st-lastchild,sec:st-rightsibling,%
sec:st-leftsibling,sec:st-isleaf,sec:st-index,sec:st-count,sec:st-letter,%
sec:st-isancestor,sec:st-findleaf}).  Finally, we show the
construction algorithm (\cref{sec:st-construction}).

\subsubsection{The Data Structure}\label{sec:st-ds}

\bfparagraph{Definitions}

Recall (\cref{sec:prelim}), that we assumed $\T[n] = 0$, and that $0$
that not appear anywhere else in $\T$. We define $\Trev$ as a text
obtained by first reversing $\T$, and then moving the symbol $0$ from
the beginning to the end. Formally, for every $i \in [1 \dd n]$:
\vspace{1.5ex}
\[
  \Trev[i] =
  \begin{cases}
    \T[n-i] & \text{if $i \neq n$},\\
    \T[n]   & \text{if $i = n$}.
  \end{cases}
  \vspace{1.5ex}
\]
Observe that for every $\Pat$ not containing the symbol $0$, $j \in
\Occ(\Pat, \T)$ holds if and only if $j' \in \Occ(\revstr{\Pat},
\Trev)$, where $j' = n - (j + |\Pat| - 1)$.

\begin{remark}
  The motivation for defining $\Trev$ is that the standard reverse
  operation on $\T$ (denoted $\revstr{T}$) does not preserve a unique
  sentinel at the end.
\end{remark}

\bfparagraph{Components}

The data structure consists of two parts. The first part is
constructed for $\T$ and consists of the following two components:

\begin{enumerate}
\item The structure from \cref{sec:st-nonperiodic-ds} (used to handle
  nonperiodic nodes).
\item The structure from \cref{sec:st-periodic-ds} (used to handle
  periodic nodes). Note that similarly as the first component it also
  includes $\STCore(\T)$. It suffices, however, to only store one
  copy.
\end{enumerate}
The second part contains the analogous two components for the text
$\Trev$. In this section, unless specified otherwise, we refer to the
part of the structure for text $\T$.

In total, the data structure takes $\bigO(n / \log_{\sigma} n)$ space.

\subsubsection{Implementation of
  \texorpdfstring{$\sdepth(v)$}{sdepth(v)}}\label{sec:st-sdepth}

\begin{proposition}\label{pr:st-sdepth}
  Let $v$ be an explicit node of $\ST$.  Given the data structure from
  \cref{sec:st-ds} and $\repr(v)$, we can in $\bigO(\log^{\epsilon}
  n)$ time compute $\sdepth(v)$.
\end{proposition}
\begin{proof}
  Denote $i_1 = \lrank(v) + 1$ and $i_2 = \rrank(v)$. Let $v_1$ and
  $v_2$ be the $i_1$th and $i_2$th (respectively) leftmost leaf of
  $\ST$. Then, $v = \LCA(v_1, v_2)$. By \cref{ob:lca}, we thus have
  $\sdepth(v) = \lcp(\str(v_1), \str(v_2)) = \LCE(\SA[i_1],
  \SA[i_2])$.  Consequently, to compute $\sdepth(v)$ we proceed as
  follows.  First, in $\bigO(\log^{\epsilon} n)$ time we compute $j_1
  = \SA[i_1]$ and $j_2 = \SA[i_2]$ using \cref{pr:sa-sa}. Then, using
  the structure to answer LCE queries (stored as part of the structure
  in \cref{sec:sa-periodic-ds}), in $\bigO(1)$ time we compute and
  return $\sdepth(v) = \LCE(j_1, j_2)$.
\end{proof}

\subsubsection{Implementation of
  \texorpdfstring{$\LCA(u, v)$}{LCA(u, v)}}\label{sec:st-lca}

\begin{proposition}\label{pr:st-lca}
  Let $v_1$ and $v_2$ be explicit nodes of $\ST$.  Given the data
  structure from \cref{sec:st-ds} and the pairs $\repr(v_1)$ and
  $\repr(v_2)$, we can in $\bigO(\log^{\epsilon} n)$ time compute the
  pair $\repr(\LCA(v_1, v_2))$.
\end{proposition}
\begin{proof}
  First, using \cref{pr:st-core-lca}, in $\bigO(1)$ time we check if
  $\sdepth(\LCA(v_1, v_2)) < 3\tau - 1$. If so, in $\bigO(1)$ time we
  additionally obtain $\repr(\LCA(v_1, v_2))$. Let us thus assume
  $\sdepth(\LCA(v_1, v_2)) \geq 3\tau - 1$. Then,
  \cref{pr:st-core-lca} additionally indicates whether $\LCA(v_1,
  v_2)$ is periodic. If not, we use \cref{pr:st-nonperiodic-lca} to
  compute $\repr(\LCA(v_1, v_2))$ in $\bigO(\log^{\epsilon} n)$
  time. Otherwise, we obtain $\repr(\LCA(v_1, v_2))$ in $\bigO(\log
  \log n)$ time using \cref{pr:st-periodic-lca}.
\end{proof}

\subsubsection{Implementation of
  \texorpdfstring{$\child(v, c)$}{child(v, c)}}\label{sec:st-child}

\begin{proposition}\label{pr:st-child}
  Let $v$ be an explicit internal node of $\ST$.  Given the data
  structure from \cref{sec:st-ds}, $\repr(v)$, and $c \in \Alphabet$,
  in $\bigO(\log^{\epsilon} n)$ time we can compute $\repr(\child(v,
  c))$.
\end{proposition}
\begin{proof}
  First, using \cref{pr:st-core-periodicity}, in $\bigO(1)$ time we
  check if $v$ is periodic. If so, we obtain $\repr(\child(v, c))$ in
  $\bigO(\log \log n)$ time using
  \cref{pr:st-periodic-child}. Otherwise (i.e., if $v$ is not
  periodic), \cref{pr:st-core-periodicity} additionally return the
  information on whether it holds $\sdepth(v) < 3\tau - 1$. If so,
  then we obtain $\repr(\child(v, c))$ in $\bigO(1)$ time using
  \cref{pr:st-core-child}. Otherwise, we obtain $\repr(\child(v, c))$
  in $\bigO(\log^{\epsilon} n)$ time using
  \cref{pr:st-nonperiodic-child}.
\end{proof}

\subsubsection{Implementation of
  \texorpdfstring{$\pred(v, c)$}{pred(v, c)}}\label{sec:st-pred}

\begin{proposition}\label{pr:st-pred-aux}
  Let $v$ be an explicit internal node of $\ST$.  Given the data
  structure from \cref{sec:st-ds}, $\repr(v)$, and $c \in \Alphabet$,
  in $\bigO(\log^{\epsilon} n)$ time we can compute
  $\LB(\str(v)c, \T)$.
\end{proposition}
\begin{proof}
  First, using \cref{pr:st-core-periodicity}, in $\bigO(1)$ time we
  check if $v$ is periodic. If so, we obtain $\LB(\str(v)c, \T)$ in
  $\bigO(\log \log n)$ time using
  \cref{pr:st-periodic-pred}. Otherwise (i.e., if $v$ is not
  periodic), \cref{pr:st-core-periodicity} additionally return the
  information on whether it holds $\sdepth(v) < 3\tau - 1$. If so,
  then we obtain $\LB(\str(v)c, \T)$ in $\bigO(1)$ time using
  \cref{pr:st-core-pred}. Otherwise, we obtain $\LB(\str(v)c, \T)$
  in $\bigO(\log^{\epsilon} n)$ time using
  \cref{pr:st-nonperiodic-pred}.
\end{proof}

\begin{proposition}\label{pr:st-pred}
  Let $v$ be an explicit internal node of $\ST$.  Given the data
  structure from \cref{sec:st-ds}, $\repr(v)$, and $c \in \Alphabet$,
  in $\bigO(\log^{\epsilon} n)$ time we can compute $\repr(\pred(v,
  c))$.
\end{proposition}
\begin{proof}
  Denote $(b, e) = \repr(v)$. First, using \cref{pr:st-pred-aux}, in
  $\bigO(\log^{\epsilon} n)$ time we compute $i = \LB(\str(v)c,
  \T)$. Observe that by definition of $\pred(v, c)$ we then have
  $\pred(v, c) = \nil$ if and only if $i = b$. If $i = b$, we thus
  return $\repr(\pred(v, c)) = (0, 0)$. Let us thus assume $i \neq
  b$. Observe that we then have $\SA[i] \in \Occ(\str(\pred(v, c)),
  \T)$, and moreover, $\pred(v, c) = \child(v, c')$, where $c' =
  \T[\SA[i] + \sdepth(v)]$. We thus proceed as follows.  First, using
  \cref{pr:sa-sa}, in $\bigO(\log^{\epsilon} n)$ time we compute $j =
  \SA[i]$. Next, using \cref{pr:st-sdepth}, in $\bigO(\log^{\epsilon}
  n)$ time we compute $\ell = \sdepth(v)$.  In $\bigO(1)$ time we then
  obtain $c' = \T[j + \ell]$. Finally, using \cref{pr:st-child}, in
  $\bigO(\log^{\epsilon} n)$ time we compute and return
  $\repr(\child(v, c')) = \repr(\pred(v, c))$.
\end{proof}

\subsubsection{Implementation of
  \texorpdfstring{$\WA(v,d)$}{WA(v, d)}}\label{sec:st-wa}

\begin{proposition}\label{pr:st-wa}
  Let $v$ be an explicit node of $\ST$.  Given the data structure from
  \cref{sec:st-ds}, $\repr(v)$, and an integer $d$ satisfying $0 \leq
  d \leq |\str(v)|$, in $\bigO(\log^{\epsilon} n)$ time we can compute
  $\repr(\WA(v, d))$.
\end{proposition}
\begin{proof}
  If $d < 3\tau - 1$, we obtain $\repr(\WA(v, d))$ in $\bigO(1)$ time
  using \cref{pr:st-core-wa}. Let us thus assume $d \geq 3\tau -
  1$. This implies $\sdepth(v) \geq 3\tau - 1$.  First, using
  \cref{pr:st-core-periodicity}, in $\bigO(1)$ time we determine
  whether $v$ is periodic. If not, then in $\bigO(\log^{\epsilon} n)$
  time we compute $\repr(\WA(v, d))$ using
  \cref{pr:st-nonperiodic-wa}. Otherwise, we obtain $\repr(\WA(v, d))$
  using \cref{pr:st-periodic-wa} in $\bigO(\log \log n)$ time.
\end{proof}

\subsubsection{Implementation of
  \texorpdfstring{$\wlink(v, c)$}{wlink(v, c)}}\label{sec:st-wlink}

\begin{proposition}\label{pr:st-wlink-occ}
  Let $\Pat \in \Alphabet^m$. Given the data structure from
  \cref{sec:st-ds}, the value $|\Pat|$, any $j \in \Occ(\Pat, \T)$,
  and any $c \in \Alphabet$, in $\bigO(\log^{\epsilon} n)$ time we can
  check if $\Occ(\Pat c, \T) \neq \emptyset$, and if so, return some
  position $j' \in \Occ(\Pat c, \T)$.
\end{proposition}
\begin{proof}
  We start by checking if $\Pat$ contains the symbol $0$. For this, we
  simply check if $j + |\Pat| = n + 1$. If so, we return that
  $\Occ(\Pat c, \T) = \emptyset$.  Let us thus assume $j + |\Pat| \leq
  n$.

  Using \cref{pr:sa-isa}, in $\bigO(\log^{\epsilon} n)$ time we
  compute $i = \ISA[j]$. Let $(b, e) = (i - 1, i)$, and observe that
  we then have $(b, e) = \repr(v)$, where $v$ is a leaf of $\ST$
  satisfying $\str(v) = \T[j \dd n]$. Next, using \cref{pr:st-wa}, in
  $\bigO(\log^{\epsilon} n)$ time we compute the pair $(b', e') =
  \repr(\WA(v, |\Pat|))$ (we can use it, since $|\Pat| \leq n - j + 1
  = \sdepth(v)$).  We then have:
  \begin{align*}
    (b', e')
      &= (\LB(\str(v)[1 \dd |\Pat|], \T),
          \UB(\str(v)[1 \dd |\Pat|], \T))\\
      &= (\LB(\T[j \dd j + |\Pat|), \T),
          \UB(\T[j \dd j + |\Pat|), \T))\\
      &= (\LB(\Pat, \T),
          \UB(\Pat, \T)).
  \end{align*}

  Next, note that it holds $(b', e') = \repr(v')$ for some node $v'$
  such that $\sdepth(v') \geq |\Pat|$.  To check if $\sdepth(v') =
  |\Pat|$, in $\bigO(\log^{\epsilon} n)$ time we compute $j_1 = \SA[b'
  + 1]$ and $j_2 = \SA[e']$ using \cref{pr:sa-sa}. As explained in the
  proof of \cref{pr:st-sdepth}, we then have $\sdepth(v') = |\Pat|$ if
  and only if $\T[j_1 + |\Pat|] \neq \T[j_2 + |\Pat|]$, which we can
  check in $\bigO(1)$ time (note that $\T[j_1 + |\Pat|]$ and $\T[j_2 +
  |\Pat|]$ are well-defined, since $j_1, j_2 \in \Occ(\Pat, \T)$ and
  we assumed that $\Pat$ does not contain symbol $\T[n] =
  0$). Consider two cases:
  \begin{itemize}
  \item If $\T[j_1 + |\Pat|] = \T[j_2 + |\Pat|]$, then $v'$ satisfies
    $\sdepth(v') > |\Pat|$. In that case we check if $c = \T[j_1 +
    |\Pat|]$. If so, we have $j_1 \in \Occ(\Pat c, \T)$ and hence we
    return $j_1$. Otherwise, we return that $\Occ(\Pat c, \T) =
    \emptyset$.
  \item Otherwise (i.e., if $\T[j_1 + |\Pat|] \neq \T[j_2 + |\Pat|]$),
    we have $\sdepth(v') = |\Pat|$.  Using \cref{pr:st-child}, we then
    compute the pair $(b'', e'') = \repr(\child(v', c))$ in
    $\bigO(\log^{\epsilon} n)$ time. If $(b'', e'') = (0, 0)$, then we
    return that $\Occ(\Pat c, \T) = \emptyset$. Otherwise, we have
    $\Occ(\Pat c, \T) \neq \emptyset$.  We then use \cref{pr:sa-sa} to
    compute $j' = \SA[e''] \in \Occ(\Pat c, \T)$ in
    $\bigO(\log^{\epsilon} n)$ time.  \qedhere
  \end{itemize}
\end{proof}

\begin{proposition}\label{pr:st-wlinkprim}
  Let $v$ be an explicit node of $\ST$.  Given the data structure from
  \cref{sec:st-ds}, $\repr(v)$, and $c \in \Alphabet$, in
  $\bigO(\log^{\epsilon} n)$ time we can compute $\repr(\wlinkprim(v,
  c))$.
\end{proposition}
\begin{proof}
  Denote $(b, e) = \repr(v)$ and $\Pat = \str(v)$.  The algorithm
  consists of two steps:
  \begin{enumerate}
  \item The first step is to determine if $\Occ(c\Pat, \T) \neq
    \emptyset$, and if so, to compute some $j' \in \Occ(c\Pat, \T)$.
    First, using \cref{pr:st-sdepth}, in $\bigO(\log^{\epsilon} n)$
    time we compute $\ell := \sdepth(v) = |\Pat|$. Using
    \cref{pr:sa-sa}, in $\bigO(\log^{\epsilon} n)$ time we also
    compute $j = \SA[e]$. We then have $j \in \Occ(\Pat, \T)$.  We now
    check if $j + \ell - 1 = n$. If so, then by the uniqueness of
    $\T[n]$, we have $\Occ(\Pat, \T) = \{j\}$. In that case, we have
    $\Occ(c\Pat, \T) \neq \emptyset$ if and only if $\T[j - 1] = c$,
    which we can check in $\bigO(1)$ time.  If $\T[j - 1] = c$, in
    $\bigO(1)$ time we then obtain $j' \in \Occ(c\Pat, \T)$, where $j'
    = j - 1$.  Let us now assume that $j + \ell - 1 \neq n$.  We now
    check if $c = 0$. If so, then $\Occ(c\Pat, \T) \neq \emptyset$
    holds if and only if $\ell = 0$. We can again check this condition
    in $\bigO(1)$ time.  Moreover, if $\ell = 0$, then we have $j' \in
    \Occ(c\Pat, \T)$, where $j' = n$.  Let us thus assume that $c \neq
    0$.  Observe that then, letting $j^{\rm rev} := n - (j + \ell -
    1)$, it holds $j^{\rm rev} \in \Occ(\revstr{\Pat}, \Trev)$.
    Denote $\ell' = \ell + 1$.  Using \cref{pr:st-wlink-occ} for the
    text $\Trev$, in $\bigO(\log^{\epsilon} n)$ time we check if
    $\Occ(\revstr{\Pat}c, \Trev) = \emptyset$ (note that we have
    $|\revstr{\Pat}c| = \ell'$).  If so, we have $\Occ(c\Pat, \T) =
    \emptyset$, since $\revstr{c\Pat} = \revstr{\Pat}c$ and hence
    $\Occ(\revstr{\Pat}c, \Trev) = \emptyset$ holds if and only if
    $\Occ(c\Pat, \T) = \emptyset$. Otherwise (i.e., if
    $\Occ(\revstr{\Pat}c, \Trev) \neq \emptyset$),
    \cref{pr:st-wlink-occ} returns some position $j_{c}^{\rm rev} \in
    \Occ(\revstr{\Pat}c, \Trev)$.  Letting $j' := n - (j_{c}^{\rm rev}
    + \ell' - 1)$, we then have $j' \in \Occ(c\Pat, \T)$.
  \item If in the first step we found that $\Occ(c\Pat, \T) =
    \emptyset$, then by definition it holds $\wlinkprim(v, c) = \nil$,
    and hence we return $\repr(\wlinkprim(v, c)) = (0, 0)$. Let us
    thus assume that $\Occ(c\Pat, \T) \neq \emptyset$ and $j' \in
    \Occ(c\Pat, \T)$.  We now compute the $\SA$ range containing all
    elements of $\Occ(c\Pat, \T)$. For this, we first compute $i =
    \ISA[j']$ using \cref{pr:sa-isa} in $\bigO(\log^{\epsilon} n)$
    time.  Letting $(b', e') = (i-1, i)$, we then have $(b', e') =
    \repr(v')$, where $v'$ is a leaf of $\ST$ satisfying $\str(v') =
    \T[j' \dd n]$.  Using \cref{pr:st-wa}, in $\bigO(\log^{\epsilon}
    n)$ time, we compute $(b'', e'') = \repr(\WA(v', \ell'))$. We then
    have
    \begin{align*}
      (b'', e'')
        &= (\LB(\str(v')[1 \dd \ell'], \T),
            \UB(\str(v')[1 \dd \ell'], \T))\\
        &= (\LB(c\Pat, \T), \UB(c\Pat, \T))\\
        &= \repr(\wlinkprim(v, c)).
    \end{align*}
  \end{enumerate}
  In total, the query takes $\bigO(\log^{\epsilon} n)$ time.
\end{proof}

\begin{proposition}\label{pr:st-wlink}
  Let $v$ be an explicit node of $\ST$.  Given the data structure from
  \cref{sec:st-ds}, $\repr(v)$, and $c \in \Alphabet$, in
  $\bigO(\log^{\epsilon} n)$ time we can compute $\repr(\wlink(v,
  c))$.
\end{proposition}
\begin{proof}
  As observed at the beginning of \cref{sec:st}, $\wlink(v, c) \neq
  \nil$ holds if and only if $\wlinkprim(v, c) \neq \nil$ and
  $\sdepth(\wlinkprim(v, c)) = \sdepth(v) + 1$. Therefore, we can use
  $\wlinkprim(v, c)$ to compute $\wlink(v, c)$. First, using
  \cref{pr:st-wlinkprim}, in $\bigO(\log^{\epsilon} n)$ time we
  compute $(b, e) = \repr(\wlinkprim(v, c))$.  If $(b, e) = (0, 0)$,
  then by the above we have $\wlink(v, c) = \nil$, and hence return
  $\repr(\wlink(v, c)) = (0, 0)$. Otherwise, using
  \cref{pr:st-sdepth}, in $\bigO(\log^{\epsilon} n)$ time we compute
  $\ell = \sdepth(v)$ and $\ell' = \sdepth(\wlinkprim(v, c))$. If
  $\ell' = \ell + 1$, then we return that $\repr(\wlink(v, c)) = (b,
  e)$. Otherwise, we have $\wlink(v, c) = \nil$ and we return
  $\repr(\wlink(v, c)) = (0, 0)$.
\end{proof}

\subsubsection{Implementation of
  \texorpdfstring{$\slink(v)$}{slink(v)}}\label{sec:st-slink}

\begin{proposition}\label{pr:st-slink}
  Let $v \neq \Root(\ST)$ be an explicit node of $\ST$.  Given the
  data structure from \cref{sec:st-ds} and $\repr(v)$, in
  $\bigO(\log^{\epsilon} n)$ time we can compute $\repr(\slink(v))$.
\end{proposition}
\begin{proof}
  Denote $(b, e) = \repr(v)$, $\Pat = \str(v)$, and $\Pat' = \Pat[2
  \dd |\Pat|]$.  Recall, that for every $v \neq \Root(\ST)$,
  $\slink(v)$ is an explicit node of $\ST$. Thus, to compute
  $\repr(\slink(v))$, we need to determine $(\LB(\Pat', \T),
  \UB(\Pat', \T))$.

  First, using \cref{pr:st-sdepth}, in $\bigO(\log^{\epsilon} n)$ time
  we compute $\ell := \sdepth(v) = |\Pat|$. Next, using
  \cref{pr:sa-sa}, in $\bigO(\log^{\epsilon} n)$ time we compute $j =
  \SA[e]$. We then have $j \in \Occ(\Pat, \T)$. Then, $j' := j + 1$
  satisfies $j' \in \Occ(\Pat', \T)$. Using \cref{pr:sa-isa}, in
  $\bigO(\log^{\epsilon} n)$ time we compute $i = \ISA[j']$. Letting
  $(b', e') = (i-1, i)$, we then have $(b', e') = \repr(v')$, where
  $v'$ is a leaf of $\ST$ satisfying $\str(v') = \T[j' \dd n]$.  Using
  \cref{pr:st-wa}, in $\bigO(\log^{\epsilon} n)$ time, we compute
  $(b'', e'') = \repr(\WA(v', \ell - 1))$. We then have
  \begin{align*}
    (b'', e'')
      &= (\LB(\str(v')[1 \dd \ell - 1], \T),
          \UB(\str(v')[1 \dd \ell - 1], \T))\\
      &= (\LB(\T[j' \dd j' + \ell - 1), \T),
          \UB(\T[j' \dd j' + \ell - 1), \T))\\
      &= (\LB(\Pat', \T), \UB(\Pat', \T))\\
      &= \repr(\slink(v)).
  \end{align*}
  In total, the query takes $\bigO(\log^{\epsilon} n)$ time.
\end{proof}

\subsubsection{Implementation of
  \texorpdfstring{$\slink(v, i)$}{slink(v, i)}}\label{sec:st-slink-iter}

\begin{proposition}\label{pr:st-slink-iter}
  Let $i \in \Zp$ and let $v$ be an explicit node of $\ST$ satisfying
  $\sdepth(v) \geq i$.  Given the data structure from
  \cref{sec:st-ds}, $\repr(v)$, and the value $i$, in
  $\bigO(\log^{\epsilon} n)$ time we can compute $\repr(\slink(v,
  i))$.
\end{proposition}
\begin{proof}
  Denote $(b, e) = \repr(v)$, $\Pat = \str(v)$, and $\Pat' =
  \Pat[i{+}1 \dd |\Pat|]$. Note that since for every $v \neq
  \Root(\ST)$, $\slink(v)$ is an explicit node of $\ST$
  (\cref{pr:st-slink}), it follows that for every explicit node $v$ of
  $\ST$ that satisfies $\sdepth(v) \geq i$, $\slink(v, i)$ is an
  explicit node of $\ST$.  Thus, to compute $\repr(\slink(v))$, we
  need to determine $(\LB(\Pat', \T), \UB(\Pat', \T))$.

  The procedure is a generalization of the one explained in the proof
  of \cref{pr:st-slink-iter}. First, using \cref{pr:st-sdepth}, in
  $\bigO(\log^{\epsilon} n)$ time we compute $\ell := \sdepth(v) =
  |\Pat|$. Next, using \cref{pr:sa-sa}, in $\bigO(\log^{\epsilon} n)$
  time we compute $j = \SA[e]$. We then have $j \in \Occ(\Pat, \T)$.
  Then, $j' := j + i$ satisfies $j' \in \Occ(\Pat', \T)$. Using
  \cref{pr:sa-isa}, in $\bigO(\log^{\epsilon} n)$ time we compute $i'
  = \ISA[j']$. Letting $(b', e') = (i'-1, i')$, we then have $(b', e')
  = \repr(v')$, where $v'$ is a leaf of $\ST$ satisfying $\str(v') =
  \T[j' \dd n]$.  Using \cref{pr:st-wa}, in $\bigO(\log^{\epsilon} n)$
  time, we compute $(b'', e'') = \repr(\WA(v', \ell - i))$. We then
  have $ (b'', e'') = (\LB(\str(v')[1 \dd \ell - i], \T),\allowbreak
  \UB(\str(v')[1 \dd \ell - i], \T)) = (\LB(\T[j' \dd j' + \ell - i),
  \T),\allowbreak \UB(\T[j' \dd j' + \ell - i), \T)) = (\LB(\Pat',
  \T), \UB(\Pat', \T)) = \repr(\slink(v, i))$.
\end{proof}

\subsubsection{Implementation of
  \texorpdfstring{$\parent(v)$}{parent(v)}}\label{sec:st-parent}

\begin{lemma}\label{lm:st-parent}
  Let $v \neq \Root(\ST)$ be an explicit node of $\ST$. Let $\repr(v)
  = (b, e)$. If $b \neq 0$ (resp.\ $e \neq n$) then, letting $v_1$ and
  $v_2$ be the $b$th and $(b+1)$st (resp.\ $e$th and $(e+1)$st)
  leftmost leaves of $\ST$, the following conditions are equivalent:
  \begin{enumerate}
  \item $\leftsibling(v) \neq \nil$ (resp.\ $\rightsibling(v) \neq
    \nil$),
  \item $\parent(v) = \LCA(v_1,v_2)$.
  \end{enumerate}
\end{lemma}
\begin{proof}
  Assume $\leftsibling(v) = v_s \neq \nil$ (resp.\ $\rightsibling(v) =
  v_s \neq \nil$). By $\repr(v) = (b, e)$, we have $\repr(v_1) = (b-1,
  b)$ (resp.\ $\repr(v_1) = (e-1,e)$), and $\repr(v_2) = (b,b+1)$
  (resp.\ $\repr(v_2) = (e,e+1)$). This implies that $v_1$ is in the
  subtree rooted in $v_s$ (resp.\ $v$) and $v_2$ in the subtree rooted
  in $v$ (resp.\ $v_s$). Consequently, $\LCA(v_1,v_2) =
  \LCA(v,v_s)$. On the other hand, since $v_s$ is a sibling of $v$, we
  have $\LCA(v,v_s) = \parent(v)$. Thus, $\parent(v) = \LCA(v_1,
  v_2)$.

  We show that $\parent(v) = \LCA(v_1, v_2)$ implies $\leftsibling(v)
  \neq \nil$ (resp.\ $\rightsibling(v) \neq \nil$) by contraposition.
  Assume $\leftsibling(v) = \nil$ (resp.\ $\rightsibling(v) = \nil$)
  and denote $v_p = \parent(v)$. Observe that then $\repr(v_p) = (b,
  e_p)$ for some $e_p > b$ (resp.\ $\repr(v_p) = (b_p, e)$ for some
  $b_p < e$). By $\repr(v_1) = (b-1,b)$ (resp.\ $\repr(v_2) =
  (e,e+1)$), the node $v_1$ (resp.\ $v_2$) is thus not in the subtree
  rooted in $v_p$. On the other hand, $\repr(v_2) = (b,b+1)$
  (resp.\ $\repr(v_1) = (e-1,e)$) implies that $v_p$ is an ancestor of
  $v_2$ (resp.\ $v_1$). Therefore, the node $\LCA(v_1,v_2) =
  \LCA(v_1,v_p)$ (resp.\ $\LCA(v_1,v_2) = \LCA(v_p,v_2)$) is a proper
  ancestor of $v_p$. In particular, $\LCA(v_1,v_2) \neq v_p$.
\end{proof}

\begin{proposition}\label{pr:st-parent}
  Let $v \neq \Root(\ST)$ be an explicit node of $\ST$.  Given the
  data structure from \cref{sec:st-ds} and $\repr(v)$, in
  $\bigO(\log^{\epsilon} n)$ time we can compute $\repr(\parent(v))$.
\end{proposition}
\begin{proof}
  Let $\repr(v) = (b, e)$. We first construct a set of pairs
  $\mathcal{P}$ as follows.
  \begin{itemize}
  \item If $b = 0$, we skip this step. Otherwise, let $v_1$ and $v_2$
    be the leftmost $b$th and $(b+1)$st leaves of $\ST$, $(b_1,e_1) =
    (b-1,b)$, and $(b_2,e_2) = (b,b+1)$. We then have $\repr(v_1) =
    (b_1,e_1)$ and $\repr(v_2)=(b_2,e_2)$. Using \cref{pr:st-lca}, in
    $\bigO(\log^{\epsilon} n)$ we obtain $\repr(v')$, where $v' =
    \LCA(v_1,v_2)$. Note that $v'$ is an ancestor of $v$.  We add
    $\repr(v')$ to $\mathcal{P}$.
  \item If $e = n$, we skip this step. Otherwise, let $v_1'$ and
    $v_2'$ be the leftmost $e$th and $(e+1)$st leaves of $\ST$,
    $(b_1',e_1') = (e-1,e)$, and $(b_2',e_2') = (e,e+1)$. We repeat
    the same procedure as above, again adding $\repr(v'')$ (where $v''
    = \LCA(v_1',v_2')$) to $\mathcal{P}$.
  \end{itemize}
  Recall now that we assumed $|\T| \geq 2$ and that $\T[n]$ is unique
  in $\T$.  This implies that the root of $\ST$ has at least two
  children.  On the other hand, any other non-leaf node has at least
  two children by definition. This implies that for every explicit
  node $v \neq \Root(\ST)$, it holds that either $\leftsibling(v) \neq
  \nil$ or $\rightsibling(v) \neq \nil$. Therefore, by
  \cref{lm:st-parent}, there exists $(b_p,e_p) \in \mathcal{P}$ such
  that $(b_p, e_p) = \repr(\parent(p))$.  Since each of the nodes $u$
  corresponding to an element in $\mathcal{P}$ is an ancestor of $v$,
  to compute $\parent(v)$, it suffices to compute $\sdepth(u)$ for all
  candidates $u$ and return the pair $\repr(u)$ corresponding to $u$
  with the largest value.  We obtain $\sdepth(u)$ using
  \cref{pr:st-sdepth} in $\bigO(\log^{\epsilon} n)$ time.  By
  $|\mathcal{P}| \leq 2$, the whole procedure takes
  $\bigO(\log^{\epsilon} n)$ time.
\end{proof}

\begin{remark}
  It might appear that the computation of $\parent(v)$ could be
  implemented by modifying the definition of the $\WA(v, d)$ to
  instead return the deepest ancestor $v'$ of $v$ satisfying
  $\sdepth(v') \leq d$ (rather than the most shallow ancestor $v'$ of
  $v$ satisfying $\sdepth(v') \geq d$). Observe, however that as shown
  in the proof of \cref{lm:st-nonperiodic-wa}
  (resp.\ \cref{lm:st-periodic-wa}), $\mapToTSSS(v)$
  (resp.\ $\mapToTZ(v)$) always returns the \emph{lowest} of all nodes
  $u'$ of $\TSSS$ (resp.\ $\TZ$) satisfying $(\fbeg(u'), \fend(u')) =
  \repr(v)$. This enforces the current definition of $\WA(v, d)$ and
  implies that the implementation of $\parent(v)$ with $\WA(v, d)$
  would require a binary search. Thus, to achieve faster time,
  $\parent(v)$ is implemented as above.
\end{remark}

\subsubsection{Implementation of
  \texorpdfstring{$\firstchild(v)$}{firstchild(v)}}\label{sec:st-firstchild}

\begin{proposition}\label{pr:st-firstchild}
  Let $v$ be an explicit node of $\ST$. Given the data structure from
  \cref{sec:st-ds} and $\repr(v)$, in $\bigO(\log^{\epsilon} n)$ time
  we can compute $\repr(\firstchild(v))$.
\end{proposition}
\begin{proof}
  Denote $(b, e) = \repr(v)$ and $\Pat = \str(v)$.  First, we check if
  $b + 1 = e$. If so, then $v$ is a leaf and hence we return
  $\repr(\firstchild(v)) = (0, 0)$ (note that here we used that $|\T|
  \geq 2$ and that $\T[n]$ is unique in $\T$, since this implies that
  every non-leaf node of $\ST$, including the root, has at least two
  children).

  Let us thus assume $b + 1 \neq e$. Denote $v' = \firstchild(v)$ and
  $\Pat' = \str(v')$. We then have $v' \neq \nil$. Observe that
  letting $(b', e') = (b, b + 1)$, it holds $(b', e') = \repr(v'')$,
  where $v''$ is a leaf of $\ST$ such that $\str(v')$ is a prefix of
  $\str(v'')$. On the other hand, by definition, we have $\sdepth(v')
  \geq \sdepth(v) + 1$, and there is no ancestor of $v'$ at depth $d
  \in (\sdepth(v) \dd \sdepth(v'))$. Therefore, we must have $v' =
  \WA(v'', \sdepth(v) + 1)$. We thus proceed as follows. First, using
  \cref{pr:st-sdepth}, in $\bigO(\log^{\epsilon} n)$ time we compute
  $\ell := \sdepth(v) = |\Pat|$. Next, using \cref{pr:st-wa}, in
  $\bigO(\log^{\epsilon} n)$ time we compute $(b'', e'') =
  \repr(\WA(v'', \ell + 1))$. We then have $\repr(\firstchild(v)) =
  (b'', e'')$.  In total, the query takes $\bigO(\log^{\epsilon} n)$
  time.
\end{proof}

\subsubsection{Implementation of
  \texorpdfstring{$\lastchild(v)$}{lastchild(v)}}\label{sec:st-lastchild}

\begin{proposition}\label{pr:st-lastchild}
  Let $v$ be an explicit node of $\ST$. Given the data structure from
  \cref{sec:st-ds} and $\repr(v)$, in $\bigO(\log^{\epsilon} n)$ time
  we can compute $\repr(\lastchild(v))$.
\end{proposition}
\begin{proof}
  Denote $(b, e) = \repr(v)$ and $\Pat = \str(v)$. The algorithm is
  symmetrical to the one presented in the proof of
  \cref{pr:st-firstchild}, i.e., rather than setting $(b', e') = (b, b
  + 1)$, we set $(b', e') = (e - 1, e)$. For such pair, it holds $(b',
  e') = \repr(v'')$, where $v''$ is a leaf of $\ST$ such that, letting
  $v' = \lastchild(v)$, the string $\str(v')$ is a prefix of
  $\str(v'')$.
\end{proof}

\subsubsection{Implementation of
  \texorpdfstring{$\rightsibling(v)$}{rightsibling(v)}}\label{sec:st-rightsibling}

\begin{proposition}\label{pr:st-rightsibling}
  Let $v$ be an explicit node of $\ST$. Given the data structure from
  \cref{sec:st-ds} and $\repr(v)$, in $\bigO(\log^{\epsilon} n)$ time
  we can compute $\repr(\rightsibling(v))$.
\end{proposition}
\begin{proof}
  Denote $(b, e) = \repr(v)$.  We start by checking if $(b, e) = (0,
  n)$.  If so, then $v = \Root(\ST)$. In that case, we have
  $\rightsibling(v) = \nil$ and hence we return
  $\repr(\rightsibling(v)) = (0, 0)$.

  Let us thus assume that $(b, e) \neq (0, n)$, i.e., $v \neq
  \Root(\ST)$.  Using \cref{pr:st-parent}, in $\bigO(\log^{\epsilon}
  n)$ time we compute $(b', e') = \repr(\parent(v))$.  We then have
  $b' \leq b < e \leq e'$. Next, we compare $e$ and $e'$.  If $e = e'$
  then, by definition, $v$ is the rightmost child of its parent and
  hence we return $\repr(\rightsibling(v)) = (0, 0)$. Let us thus
  assume $e < e'$. We then have $\rightsibling(v) \neq
  \nil$. Moreover, letting $v' = \rightsibling(v)$ and $\Pat' =
  \str(v')$, it then holds $\SA[e + 1] \in \Occ(\Pat', \T)$. This
  implies that, letting $(b'', e'') = (e, e + 1)$, we have $(b'', e'')
  = \repr(v'')$, where $v''$ is a leaf of $\ST$ such that $\Pat'$ is a
  prefix of $\str(v'')$.  Moreover, it holds $\sdepth(v') \geq
  \sdepth(\parent(v)) + 1$, and the node $v'$ does not have any
  ancestors at depth $d \in (\sdepth(\parent(v)) \dd \sdepth(v'))$.
  Consequently, $v' = \WA(v'', \sdepth(\parent(v)) + 1)$.  We thus
  proceed as follows. First, using \cref{pr:st-sdepth}, in
  $\bigO(\log^{\epsilon} n)$ time we compute $\ell :=
  \sdepth(\parent(v))$.  Then, using \cref{pr:st-wa}, in
  $\bigO(\log^{\epsilon} n)$ time we compute $(b''', e''') =
  \repr(\WA(v'', \ell + 1))$. By the above discussion, we have $(b''',
  e''') = \repr(\rightsibling(v))$.
\end{proof}

\subsubsection{Implementation of
  \texorpdfstring{$\leftsibling(v)$}{leftsibling(v)}}\label{sec:st-leftsibling}

\begin{proposition}\label{pr:st-leftsibling}
  Let $v$ be an explicit node of $\ST$. Given the data structure from
  \cref{sec:st-ds} and $\repr(v)$, in $\bigO(\log^{\epsilon} n)$ time
  we can compute $\repr(\leftsibling(v))$.
\end{proposition}
\begin{proof}
  Denote $(b, e) = \repr(v)$. The algorithm is symmetrical to the one
  presented in the proof of \cref{pr:st-rightsibling}. More precisely,
  we replace the check $e = e'$ with $b = b'$. We also set $(b'', e'')
  = (b - 1, b)$ instead of $(b'', e'') = (e, e + 1)$.
\end{proof}

\subsubsection{Implementation of
  \texorpdfstring{$\isleaf(v)$}{isleaf(v)}}\label{sec:st-isleaf}

\begin{proposition}\label{pr:st-isleaf}
  Let $v$ be an explicit node of $\ST$. Given $\repr(v)$, we can check
  if $v$ is a leaf in $\bigO(1)$ time.
\end{proposition}
\begin{proof}
  Recall, that $|\T| \geq 2$ and that $\T[n]$ is unique in $\T$.  This
  implies that the root of $\ST$ has at least two children.  On the
  other hand, any other non-leaf node has at least two children by
  definition. Thus, letting $(b, e) = \repr(v)$, and recalling that
  $\repr(v) = (\lrank(v), \rrank(v))$, the node $v$ is a leaf if and
  only if $b + 1 = e$, which we can check in $\bigO(1)$ time.
\end{proof}

\subsubsection{Implementation of
  \texorpdfstring{$\ind(v)$}{index(v)}}\label{sec:st-index}

\begin{proposition}\label{pr:st-index}
  Let $v$ be an explicit node of $\ST$. Given the data structure from
  \cref{sec:st-ds} and $\repr(v)$, in $\bigO(\log^{\epsilon} n)$ time
  we can compute $\ind(v)$.
\end{proposition}
\begin{proof}
  Denote $(b, e) = \repr(v)$ and $\Pat = \str(v)$. Recall, that it
  holds $\repr(v) = (\LB(\Pat, \T),\allowbreak \UB(\Pat, \T))$, and
  hence $\Occ(\Pat, \T) = \{\SA[i]\}_{i \in (b \dd e]}$. Thus, to
  obtain $\ind(v)$, it suffices to compute $j = \SA[i]$ for any $i \in
  (b \dd e]$. Using \cref{pr:sa-sa}, this takes $\bigO(\log^{\epsilon}
  n)$ time.
\end{proof}

\subsubsection{Implementation of
  \texorpdfstring{$\cnt(v)$}{count(v)}}\label{sec:st-count}

\begin{proposition}\label{pr:st-count}
  Let $v$ be an explicit node of $\ST$. Given $\repr(v)$, we can
  compute $\cnt(v)$ in $\bigO(1)$ time.
\end{proposition}
\begin{proof}
  Denote $(b, e) = \repr(v)$. Since by definition we have
  $\Occ(\str(v), \T) = \{\SA[i]\}_{i \in (b \dd e]}$, in $\bigO(1)$
  time we return $\cnt(v) = |\Occ(\str(v), \T)| = e - b$.
\end{proof}

\subsubsection{Implementation of
  \texorpdfstring{$\letter(v, i)$}{letter(v, i)}}\label{sec:st-letter}

\begin{proposition}\label{pr:st-letter}
  Let $v$ be an explicit node of $\ST$ and $i \in [1 \dd |\str(v)|]$.
  Given the data structure from \cref{sec:st-ds}, $\repr(v)$, and the
  value $i$, we can compute $\letter(v, i)$ in $\bigO(\log^{\epsilon}
  n)$ time.
\end{proposition}
\begin{proof}
  It suffices to find any $j \in \Occ(\str(v), \T)$ and return $\T[j +
  i - 1]$. Using \cref{pr:st-index}, we find $j$ in
  $\bigO(\log^{\epsilon} n)$ time. We then return the output symbol in
  $\bigO(1)$ time (recall, that the packed representation of $\T$ is
  stored as part of $\STCore(\T)$).
\end{proof}

\subsubsection{Implementation of
  \texorpdfstring{$\isancestor(u, v)$}{isancestor(u, v)}}\label{sec:st-isancestor}

\begin{proposition}\label{pr:st-isancestor}
  Let $u$ and $v$ be explicit nodes of $\ST$. Given $\repr(u)$ and
  $\repr(v)$, we can check if $u$ is an ancestor of $v$ in $\bigO(1)$
  time.
\end{proposition}
\begin{proof}
  Denote $(b, e) = \repr(u)$, $(b', e') = \repr(v)$. The node $u$ is
  an ancestor of $v$ if and only if $\Occ(\str(v), \T) \sub
  \Occ(\str(u), \T)$, which (by definition) holds if and only if $b
  \leq b' < e' \leq e$.
\end{proof}

\subsubsection{Implementation of
  \texorpdfstring{$\findleaf(j)$}{findleaf(j)}}\label{sec:st-findleaf}

\begin{proposition}\label{pr:st-findleaf}
  Let $j \in [1 \dd n]$. Given the data structure from
  \cref{sec:st-ds} and the position $j$, in $\bigO(\log^{\epsilon} n)$
  time we can return $\repr(v)$, where $v$ is a leaf of $\ST$
  satisfying $\str(v) = \T[j \dd n]$.
\end{proposition}
\begin{proof}
  Let $(b, e) = \repr(v)$. Observe that by definition of $v$, we have
  $\Occ(\str(v), \T) = \{j\}$. Thus, it must hold $\{\SA[i]\}_{i \in
  (b \dd e]} = \{j\}$. This implies, that it suffices to compute $i
  = \ISA[j]$ and then return that $\repr(v) = (i-1, i)$.  Using
  \cref{pr:sa-isa}, the computation of $\ISA[j]$ takes
  $\bigO(\log^{\epsilon} n)$ time.
\end{proof}

\subsubsection{Construction Algorithm}\label{sec:st-construction}

\begin{proposition}\label{pr:st-construction}
  Given the packed representation of a text $\T \in \Alphabet^n$, we
  can construct the data structure from \cref{sec:st-ds} in $\bigO(n
  \min(1, \log \sigma / \sqrt{\log n}))$ time and $\bigO(n /
  \log_{\sigma} n)$ working space.
\end{proposition}
\begin{proof}
  The first part of the structure is constructed as follows.  First,
  from a packed representation of $\T$, we construct $\STCore(\T)$ in
  $\bigO(n / \log_{\sigma} n)$ time using
  \cref{pr:st-core-construction}. Then, using
  \cref{pr:st-nonperiodic-construction,pr:st-periodic-construction},
  we augment $\STCore(\T)$ in $\bigO(n \min(1, \log \sigma /
  \sqrt{\log n}))$ and $\bigO(n / \log_{\sigma} n)$ time (respectively)
  and using $\bigO(n / \log_{\sigma} n)$ working space into the two
  components of the structure from \cref{sec:st-ds}.  The overall
  runtime is thus $\bigO(n \min(1, \log \sigma / \sqrt{\log n}))$.

  Next, we compute $\Trev$. With the help of the lookup table
  $\LTrev$, we first compute $\Trev[1 \dd n) = \revstr{\T[1 \dd n)}$
  in $\bigO(n / \log_{\sigma} n)$ time. In $\bigO(1)$ time we then
  append the sentinel $\Trev[n] := 0$. After that, analogously as
  above, we construct the structures from
  \cref{sec:st-nonperiodic-ds,sec:st-periodic-ds} for $\Trev$, i.e.,
  the second part of the structure from \cref{sec:st-ds}.
\end{proof}

\subsection{Summary}\label{sec:st-summary}

By combining \cref{pr:st-sdepth,pr:st-lca,pr:st-child,pr:st-pred,%
pr:st-wa,pr:st-wlink,pr:st-slink,pr:st-slink,pr:st-slink-iter,%
pr:st-parent,pr:st-firstchild,pr:st-lastchild,pr:st-rightsibling,%
pr:st-leftsibling,pr:st-isleaf,pr:st-index,pr:st-count,pr:st-letter,%
pr:st-isancestor,pr:st-findleaf,pr:st-construction} we obtain the
following main result of this section.

\begin{theorem}\label{th:st}
  Given any constant $\epsilon \in (0,1)$ and the packed
  representation of a text $\T \in \Alphabet^n$ with $2 \leq \sigma <
  n^{1/7}$, in $\bigO(n \min(1, \log \sigma / \sqrt{\log n}))$ time
  and $\bigO(n / \log_{\sigma} n)$ working space we can construct a
  representation of the suffix tree of $\T$ occupying $\bigO(n /
  \log_{\sigma} n)$ space and supporting all standard operations (see
  \cref{tab:st-operations}) in $\bigO(\log^{\epsilon} n)$ time.
\end{theorem}

We also immediately obtain the following general reduction.

\begin{theorem}\label{th:st-general}
  Consider a data structure answering prefix rank and selection
  queries that, for any string of length $m$ over alphabet
  $\Alphabet^\ell$, achieves the following complexities:
  \begin{enumerate}
  \item Space usage $S(m, \ell, \sigma)$,
  \item Preprocessing time $P_t(m, \ell, \sigma)$,
  \item Preprocessing space $P_s(m, \ell, \sigma)$,
  \item Query time $Q(m, \ell, \sigma)$.
  \end{enumerate}
  For every $\T \in \Alphabet^n$ with $2 \leq \sigma < n^{1/7}$, there
  exist $m = \bigO(n/\log_{\sigma} n)$ and $\ell = \bigO(\log_{\sigma}
  n)$ such that, given the packed representation of $\T$, we can in
  $\bigO(n / \log_{\sigma} n + P_t(m, \ell, \sigma))$ time and
  $\bigO(n / \log_{\sigma} n + P_s(m, \ell,\sigma))$ working space
  construct a representation of the suffix tree of $\T$ occupying
  $\bigO(n/\log_{\sigma} n + S(m, \ell, \sigma))$ space and supporting
  all standard operations (see \cref{tab:st-operations}) in
  $\bigO(\log \log n + Q(m, \ell, \sigma))$ time.
\end{theorem}

\bibliographystyle{plainurl}
\bibliography{paper}

\begin{thebibliography}{10}

\bibitem{bwtbook}
Donald Adjeroh, Tim Bell, and Amar Mukherjee.
\newblock {\em The {B}urrows-{W}heeler Transform: Data Compression, Suffix
  Arrays, and Pattern Matching}.
\newblock Springer, Boston, MA, USA, 2008.
\newblock \href {https://doi.org/10.1007/978-0-387-78909-5}
  {\path{doi:10.1007/978-0-387-78909-5}}.

\bibitem{AmirLLS07}
Amihood Amir, Gad~M. Landau, Moshe Lewenstein, and Dina Sokol.
\newblock Dynamic text and static pattern matching.
\newblock {\em {ACM} Trans. Algorithms}, 3(2):19, 2007.
\newblock \href {https://doi.org/10.1145/1240233.1240242}
  {\path{doi:10.1145/1240233.1240242}}.

\bibitem{WaveletSuffixTree}
Maxim Babenko, Pawe{\l} Gawrychowski, Tomasz Kociumaka, and Tatiana
  Starikovskaya.
\newblock Wavelet trees meet suffix trees.
\newblock In {\em 26th Annual {ACM}-{SIAM} Symposium on Discrete Algorithms
  (SODA)}, pages 572--591, 2015.
\newblock \href {https://doi.org/10.1137/1.9781611973730.39}
  {\path{doi:10.1137/1.9781611973730.39}}.

\bibitem{BarbayCGNN14}
J{\'{e}}r{\'{e}}my Barbay, Francisco Claude, Travis Gagie, Gonzalo Navarro, and
  Yakov Nekrich.
\newblock Efficient fully-compressed sequence representations.
\newblock {\em Algorithmica}, 69(1):232--268, 2014.
\newblock \href {https://doi.org/10.1007/s00453-012-9726-3}
  {\path{doi:10.1007/s00453-012-9726-3}}.

\bibitem{Belazzougui14}
Djamal Belazzougui.
\newblock Linear time construction of compressed text indices in compact space.
\newblock In {\em 46th Annual {ACM} {SIGACT} Symposium on Theory of Computing
  ({STOC})}, pages 148--193, 2014.
\newblock \href {https://doi.org/10.1145/2591796.2591885}
  {\path{doi:10.1145/2591796.2591885}}.

\bibitem{DCC2015}
Djamal Belazzougui, Travis Gagie, Pawe{\l} Gawrychowski, Juha
  K{\"{a}}rkk{\"{a}}inen, Alberto~Ord{\'{o}}{\~{n}}ez Pereira, Simon~J.
  Puglisi, and Yasuo Tabei.
\newblock Queries on {LZ}-bounded encodings.
\newblock In {\em Data Compression Conference (DCC)}, pages 83--92, 2015.
\newblock \href {https://doi.org/10.1109/DCC.2015.69}
  {\path{doi:10.1109/DCC.2015.69}}.

\bibitem{BelazzouguiN14}
Djamal Belazzougui and Gonzalo Navarro.
\newblock Alphabet-independent compressed text indexing.
\newblock {\em {ACM} Trans. Algorithms}, 10(4):23:1--23:19, 2014.
\newblock \href {https://doi.org/10.1145/2635816} {\path{doi:10.1145/2635816}}.

\bibitem{BelazzouguiN15}
Djamal Belazzougui and Gonzalo Navarro.
\newblock Optimal lower and upper bounds for representing sequences.
\newblock {\em {ACM} Trans. Algorithms}, 11(4):31:1--31:21, 2015.
\newblock \href {https://doi.org/10.1145/2629339} {\path{doi:10.1145/2629339}}.

\bibitem{Belazzougui2016}
Djamal Belazzougui and Simon~J. Puglisi.
\newblock Range predecessor and {L}empel-{Z}iv parsing.
\newblock In {\em 27th Annual {ACM-SIAM} Symposium on Discrete Algorithms
  (SODA)}, pages 2053--2071, 2016.
\newblock \href {https://doi.org/10.1137/1.9781611974331.ch143}
  {\path{doi:10.1137/1.9781611974331.ch143}}.

\bibitem{BenderF00}
Michael~A. Bender and Martin Farach{-}Colton.
\newblock The {LCA} problem revisited.
\newblock In {\em 4th Latin American Symposium on Theoretical Informatics
  (LATIN)}, pages 88--94, 2000.
\newblock \href {https://doi.org/10.1007/10719839_9}
  {\path{doi:10.1007/10719839_9}}.

\bibitem{BilleEGV18}
Philip Bille, Mikko~Berggren Ettienne, Inge~Li G{\o}rtz, and Hjalte~Wedel
  Vildh{\o}j.
\newblock Time-space trade-offs for {L}empel-{Z}iv compressed indexing.
\newblock {\em Theor. Comput. Sci.}, 713:66--77, 2018.
\newblock \href {https://doi.org/10.1016/j.tcs.2017.12.021}
  {\path{doi:10.1016/j.tcs.2017.12.021}}.

\bibitem{BilleGS17}
Philip Bille, Inge~Li G{\o}rtz, and Frederik~Rye Skjoldjensen.
\newblock Deterministic indexing for packed strings.
\newblock In {\em 28th Annual Symposium on Combinatorial Pattern Matching
  (CPM)}, pages 6:1--6:11, 2017.
\newblock \href {https://doi.org/10.4230/LIPIcs.CPM.2017.6}
  {\path{doi:10.4230/LIPIcs.CPM.2017.6}}.

\bibitem{BLRSRW15}
Philip Bille, Gad~M. Landau, Rajeev Raman, Kunihiko Sadakane, Srinivasa~Rao
  Satti, and Oren Weimann.
\newblock Random access to grammar-compressed strings and trees.
\newblock {\em {SIAM} J. Comput.}, 44(3):513--539, 2015.
\newblock \href {https://doi.org/10.1137/130936889}
  {\path{doi:10.1137/130936889}}.

\bibitem{BoucherCGHMNR21}
Christina Boucher, Ondrej Cvacho, Travis Gagie, Jan Holub, Giovanni Manzini,
  Gonzalo Navarro, and Massimiliano Rossi.
\newblock {PFP} compressed suffix trees.
\newblock In {\em 24th Symposium on Algorithm Engineering and Experiments
  (ALENEX)}, pages 60--72, 2021.
\newblock \href {https://doi.org/10.1137/1.9781611976472.5}
  {\path{doi:10.1137/1.9781611976472.5}}.

\bibitem{bwt}
Michael Burrows and David~J. Wheeler.
\newblock A block-sorting lossless data compression algorithm.
\newblock Technical Report 124, Digital Equipment Corporation, 1994.
\newblock URL:
  \url{https://www.hpl.hp.com/techreports/Compaq-DEC/SRC-RR-124.pdf}.

\bibitem{CaceresN22}
Manuel C{\'{a}}ceres and Gonzalo Navarro.
\newblock Faster repetition-aware compressed suffix trees based on block trees.
\newblock {\em Inf. Comput.}, 285(Part):104749, 2022.
\newblock \href {https://doi.org/10.1016/j.ic.2021.104749}
  {\path{doi:10.1016/j.ic.2021.104749}}.

\bibitem{ChanHLS07}
Ho{-}Leung Chan, Wing{-}Kai Hon, Tak~Wah Lam, and Kunihiko Sadakane.
\newblock Compressed indexes for dynamic text collections.
\newblock {\em {ACM} Trans. Algorithms}, 3(2):21, 2007.
\newblock \href {https://doi.org/10.1145/1240233.1240244}
  {\path{doi:10.1145/1240233.1240244}}.

\bibitem{ChanP10}
Timothy~M. Chan and Mihai P\u{a}tra\c{s}cu.
\newblock Counting inversions, offline orthogonal range counting, and related
  problems.
\newblock In {\em 21st Annual {ACM}-{SIAM} Symposium on Discrete Algorithms
  (SODA)}, pages 161--173, 2010.
\newblock \href {https://doi.org/10.1137/1.9781611973075.15}
  {\path{doi:10.1137/1.9781611973075.15}}.

\bibitem{chazelle}
Bernard Chazelle.
\newblock A functional approach to data structures and its use in
  multidimensional searching.
\newblock {\em {SIAM} J. Comput.}, 17(3):427--462, 1988.
\newblock \href {https://doi.org/10.1137/0217026} {\path{doi:10.1137/0217026}}.

\bibitem{ChristiansenEKN21}
Anders~Roy Christiansen, Mikko~Berggren Ettienne, Tomasz Kociumaka, Gonzalo
  Navarro, and Nicola Prezza.
\newblock Optimal-time dictionary-compressed indexes.
\newblock {\em {ACM} Trans. Algorithms}, 17(1):8:1--8:39, 2021.
\newblock \href {https://doi.org/10.1145/3426473} {\path{doi:10.1145/3426473}}.

\bibitem{Clark98}
David~R. Clark.
\newblock {\em Compact Pat Trees}.
\newblock PhD thesis, University of Waterloo, 1998.
\newblock URL:
  \url{https://uwspace.uwaterloo.ca/bitstream/handle/10012/64/nq21335.pdf}.

\bibitem{CKL15}
Richard Cole, Tsvi Kopelowitz, and Moshe Lewenstein.
\newblock Suffix trays and suffix trists: Structures for faster text indexing.
\newblock {\em Algorithmica}, 72(2):450--466, 2015.
\newblock \href {https://doi.org/10.1007/s00453-013-9860-6}
  {\path{doi:10.1007/s00453-013-9860-6}}.

\bibitem{cover2006elements}
T.~M. Cover and J.~A. Thomas.
\newblock {\em Elements of information theory}.
\newblock Wiley, 2nd edition, 2006.
\newblock \href {https://doi.org/10.1002/047174882X}
  {\path{doi:10.1002/047174882X}}.

\bibitem{AlgorithmsOnStrings}
Maxime Crochemore, Christophe Hancart, and Thierry Lecroq.
\newblock {\em Algorithms on strings}.
\newblock Cambridge University Press, Cambridge, UK, 2007.
\newblock \href {https://doi.org/10.1017/cbo9780511546853}
  {\path{doi:10.1017/cbo9780511546853}}.

\bibitem{FarachM96}
Martin Farach and S.~Muthukrishnan.
\newblock Perfect hashing for strings: Formalization and algorithms.
\newblock In {\em 7th Annual Symposium on Combinatorial Pattern Matching
  (CPM)}, pages 130--140, 1996.
\newblock \href {https://doi.org/10.1007/3-540-61258-0\_11}
  {\path{doi:10.1007/3-540-61258-0\_11}}.

\bibitem{FarachFM00}
Martin Farach{-}Colton, Paolo Ferragina, and S.~Muthukrishnan.
\newblock On the sorting-complexity of suffix tree construction.
\newblock {\em J. {ACM}}, 47(6):987--1011, 2000.
\newblock \href {https://doi.org/10.1145/355541.355547}
  {\path{doi:10.1145/355541.355547}}.

\bibitem{fm}
Paolo Ferragina and Giovanni Manzini.
\newblock Opportunistic data structures with applications.
\newblock In {\em 41st {IEEE} Annual Symposium on Foundations of Computer
  Science (FOCS)}, pages 390--398, 2000.
\newblock \href {https://doi.org/10.1109/SFCS.2000.892127}
  {\path{doi:10.1109/SFCS.2000.892127}}.

\bibitem{FerraginaM05}
Paolo Ferragina and Giovanni Manzini.
\newblock Indexing compressed text.
\newblock {\em J. {ACM}}, 52(4):552--581, 2005.
\newblock \href {https://doi.org/10.1145/1082036.1082039}
  {\path{doi:10.1145/1082036.1082039}}.

\bibitem{fine1965uniqueness}
Nathan~J. Fine and Herbert~S. Wilf.
\newblock Uniqueness theorems for periodic functions.
\newblock {\em Proc. Am. Math. Soc.}, 16(1):109--114, 1965.
\newblock \href {https://doi.org/10.2307/2034009} {\path{doi:10.2307/2034009}}.

\bibitem{wexp}
Johannes Fischer and Pawe{\l} Gawrychowski.
\newblock Alphabet-dependent string searching with {W}exponential search trees.
\newblock In {\em 26th Annual Symposium on Combinatorial Pattern Matching
  (CPM)}, pages 160--171, 2015.
\newblock Full version: \url{https://arxiv.org/abs/1302.3347}.
\newblock \href {https://doi.org/10.1007/978-3-319-19929-0_14}
  {\path{doi:10.1007/978-3-319-19929-0_14}}.

\bibitem{FischerMN09}
Johannes Fischer, Veli M{\"{a}}kinen, and Gonzalo Navarro.
\newblock Faster entropy-bounded compressed suffix trees.
\newblock {\em Theor. Comput. Sci.}, 410(51):5354--5364, 2009.
\newblock \href {https://doi.org/10.1016/j.tcs.2009.09.012}
  {\path{doi:10.1016/j.tcs.2009.09.012}}.

\bibitem{GagieGKNP12}
Travis Gagie, Pawe{\l} Gawrychowski, Juha K{\"{a}}rkk{\"{a}}inen, Yakov
  Nekrich, and Simon~J. Puglisi.
\newblock A faster grammar-based self-index.
\newblock In {\em 6th International Conference on Language and Automata Theory
  and Applications (LATA)}, pages 240--251, 2012.
\newblock \href {https://doi.org/10.1007/978-3-642-28332-1_21}
  {\path{doi:10.1007/978-3-642-28332-1_21}}.

\bibitem{GagieGKNP14}
Travis Gagie, Pawe{\l} Gawrychowski, Juha K{\"{a}}rkk{\"{a}}inen, Yakov
  Nekrich, and Simon~J. Puglisi.
\newblock {LZ77}-based self-indexing with faster pattern matching.
\newblock In {\em 11th Latin American Symposium on Theoretical Informatics
  (LATIN)}, pages 731--742, 2014.
\newblock \href {https://doi.org/10.1007/978-3-642-54423-1_63}
  {\path{doi:10.1007/978-3-642-54423-1_63}}.

\bibitem{Gagie2020}
Travis Gagie, Gonzalo Navarro, and Nicola Prezza.
\newblock Fully functional suffix trees and optimal text searching in
  {BWT}-runs bounded space.
\newblock {\em J. {ACM}}, 67(1):1--54, apr 2020.
\newblock \href {https://doi.org/10.1145/3375890} {\path{doi:10.1145/3375890}}.

\bibitem{Gao0N20}
Younan Gao, Meng He, and Yakov Nekrich.
\newblock Fast preprocessing for optimal orthogonal range reporting and range
  successor with applications to text indexing.
\newblock In {\em 28th Annual European Symposium on Algorithms (ESA)}, pages
  54:1--54:18, 2020.
\newblock \href {https://doi.org/10.4230/LIPIcs.ESA.2020.54}
  {\path{doi:10.4230/LIPIcs.ESA.2020.54}}.

\bibitem{Gawrychowski2015}
Pawel Gawrychowski, Adam Karczmarz, Tomasz Kociumaka, Jakub Lacki, and Piotr
  Sankowski.
\newblock Optimal dynamic strings.
\newblock In {\em 29th Annual {ACM-SIAM} Symposium on Discrete Algorithms
  (SODA)}, pages 1509--1528, 2018.
\newblock Full version: \url{https://arxiv.org/abs/1511.02612}.
\newblock \href {https://doi.org/10.1137/1.9781611975031.99}
  {\path{doi:10.1137/1.9781611975031.99}}.

\bibitem{Gog11}
Simon Gog.
\newblock {\em Compressed suffix trees: design, construction, and
  applications}.
\newblock PhD thesis, University of Ulm, 2011.
\newblock URL: \url{http://vts.uni-ulm.de/docs/2011/7786/vts\_7786\_11228.pdf}.

\bibitem{sdsl}
Simon Gog, Timo Beller, Alistair Moffat, and Matthias Petri.
\newblock From theory to practice: Plug and play with succinct data structures.
\newblock In {\em 13th International Symposium on Experimental Algorithms
  (SEA)}, pages 326--337, 2014.
\newblock \href {https://doi.org/10.1007/978-3-319-07959-2_28}
  {\path{doi:10.1007/978-3-319-07959-2_28}}.

\bibitem{GogKKPP19}
Simon Gog, Juha K{\"{a}}rkk{\"{a}}inen, Dominik Kempa, Matthias Petri, and
  Simon~J. Puglisi.
\newblock Fixed block compression boosting in {FM}-indexes: Theory and
  practice.
\newblock {\em Algorithmica}, 81(4):1370--1391, 2019.
\newblock \href {https://doi.org/10.1007/s00453-018-0475-9}
  {\path{doi:10.1007/s00453-018-0475-9}}.

\bibitem{GogMP17}
Simon Gog, Alistair Moffat, and Matthias Petri.
\newblock {CSA++}: Fast pattern search for large alphabets.
\newblock In {\em 19th Workshop on Algorithm Engineering and Experiments
  (ALENEX)}, pages 73--82, 2017.
\newblock \href {https://doi.org/10.1137/1.9781611974768.6}
  {\path{doi:10.1137/1.9781611974768.6}}.

\bibitem{wt}
Roberto Grossi, Ankur Gupta, and Jeffrey~Scott Vitter.
\newblock High-order entropy-compressed text indexes.
\newblock In {\em 14th Annual {ACM-SIAM} Symposium on Discrete Algorithms
  (SODA)}, pages 841--850, 2003.
\newblock URL: \url{https://dl.acm.org/doi/10.5555/644108.644250}.

\bibitem{csa}
Roberto Grossi and Jeffrey~Scott Vitter.
\newblock Compressed suffix arrays and suffix trees with applications to text
  indexing and string matching (extended abstract).
\newblock In {\em 32nd Annual {ACM} Symposium on Theory of Computing (STOC)},
  pages 397--406, 2000.
\newblock \href {https://doi.org/10.1145/335305.335351}
  {\path{doi:10.1145/335305.335351}}.

\bibitem{GrossiV05}
Roberto Grossi and Jeffrey~Scott Vitter.
\newblock Compressed suffix arrays and suffix trees with applications to text
  indexing and string matching.
\newblock {\em {SIAM} J. Comput.}, 35(2):378--407, 2005.
\newblock \href {https://doi.org/10.1137/S0097539702402354}
  {\path{doi:10.1137/S0097539702402354}}.

\bibitem{gusfield}
Dan Gusfield.
\newblock {\em Algorithms on Strings, Trees, and Sequences - {C}omputer Science
  and Computational Biology}.
\newblock Cambridge University Press, 1997.
\newblock \href {https://doi.org/10.1017/cbo9780511574931}
  {\path{doi:10.1017/cbo9780511574931}}.

\bibitem{Hagerup98}
Torben Hagerup.
\newblock Sorting and searching on the word {RAM}.
\newblock In {\em 15th Annual Symposium on Theoretical Aspects of Computer
  Science (STACS)}, pages 366--398, 1998.
\newblock \href {https://doi.org/10.1007/BFb0028575}
  {\path{doi:10.1007/BFb0028575}}.

\bibitem{HonSS03}
Wing{-}Kai Hon, Kunihiko Sadakane, and Wing{-}Kin Sung.
\newblock Breaking a time-and-space barrier in constructing full-text indices.
\newblock In {\em 44th {IEEE} Symposium on Foundations of Computer Science
  (FOCS)}, pages 251--260, 2003.
\newblock \href {https://doi.org/10.1109/SFCS.2003.1238199}
  {\path{doi:10.1109/SFCS.2003.1238199}}.

\bibitem{Jac89}
Guy Jacobson.
\newblock Space-efficient static trees and graphs.
\newblock In {\em 30th {IEEE} Symposium on Foundations of Computer Science
  (FOCS)}, pages 549--554, 1989.
\newblock \href {https://doi.org/10.1109/SFCS.1989.63533}
  {\path{doi:10.1109/SFCS.1989.63533}}.

\bibitem{KarkkainenKP14a}
Juha K{\"{a}}rkk{\"{a}}inen, Dominik Kempa, and Simon~J. Puglisi.
\newblock Hybrid compression of bitvectors for the {FM}-index.
\newblock In {\em Data Compression Conference (DCC)}, pages 302--311, 2014.
\newblock \href {https://doi.org/10.1109/DCC.2014.87}
  {\path{doi:10.1109/DCC.2014.87}}.

\bibitem{KarkkainenSB06}
Juha K{\"{a}}rkk{\"{a}}inen, Peter Sanders, and Stefan Burkhardt.
\newblock Linear work suffix array construction.
\newblock {\em J. {ACM}}, 53(6):918--936, 2006.
\newblock \href {https://doi.org/10.1145/1217856.1217858}
  {\path{doi:10.1145/1217856.1217858}}.

\bibitem{sss}
Dominik Kempa and Tomasz Kociumaka.
\newblock String synchronizing sets: Sublinear-time {BWT} construction and
  optimal {LCE} data structure.
\newblock In {\em 51st Annual {ACM} {SIGACT} Symposium on Theory of Computing
  (STOC)}, pages 756--767, 2019.
\newblock \href {https://doi.org/10.1145/3313276.3316368}
  {\path{doi:10.1145/3313276.3316368}}.

\bibitem{resolution}
Dominik Kempa and Tomasz Kociumaka.
\newblock Resolution of the {B}urrows-{W}heeler {T}ransform conjecture.
\newblock In {\em 61st {IEEE} Annual Symposium on Foundations of Computer
  Science (FOCS)}, pages 1002--1013, 2020.
\newblock \href {https://doi.org/10.1109/FOCS46700.2020.00097}
  {\path{doi:10.1109/FOCS46700.2020.00097}}.

\bibitem{dynsa}
Dominik Kempa and Tomasz Kociumaka.
\newblock Dynamic suffix array with polylogarithmic queries and updates.
\newblock In {\em 54th Annual {ACM} {SIGACT} Symposium on Theory of Computing
  (STOC)}, pages 1657--1670, 2022.
\newblock \href {https://doi.org/10.1145/3519935.3520061}
  {\path{doi:10.1145/3519935.3520061}}.

\bibitem{attractors}
Dominik Kempa and Nicola Prezza.
\newblock At the roots of dictionary compression: String attractors.
\newblock In {\em 50th Annual {ACM} {SIGACT} Symposium on Theory of Computing
  (STOC)}, pages 827--840, 2018.
\newblock \href {https://doi.org/10.1145/3188745.3188814}
  {\path{doi:10.1145/3188745.3188814}}.

\bibitem{phdtomek}
Tomasz Kociumaka.
\newblock {\em Efficient Data Structures for Internal Queries in Texts}.
\newblock PhD thesis, University of Warsaw, 2018.
\newblock URL:
  \url{https://depotuw.ceon.pl/bitstream/handle/item/3614/1000-DR-INF-170341.pdf}.

\bibitem{bowtie}
Ben Langmead, Cole Trapnell, Mihai Pop, and Steven~L Salzberg.
\newblock Ultrafast and memory-efficient alignment of short {DNA} sequences to
  the human genome.
\newblock {\em Genome Biol.}, 10(3):R25, 2009.
\newblock \href {https://doi.org/10.1186/gb-2009-10-3-r25}
  {\path{doi:10.1186/gb-2009-10-3-r25}}.

\bibitem{bwa}
Heng Li and Richard Durbin.
\newblock Fast and accurate short read alignment with {B}urrows-{W}heeler
  transform.
\newblock {\em Bioinform.}, 25(14):1754--1760, 2009.
\newblock \href {https://doi.org/10.1093/bioinformatics/btp324}
  {\path{doi:10.1093/bioinformatics/btp324}}.

\bibitem{soap2}
Ruiqiang Li, Chang Yu, Yingrui Li, Tak~Wah Lam, Siu{-}Ming Yiu, Karsten
  Kristiansen, and Jun Wang.
\newblock {SOAP2}: An improved ultrafast tool for short read alignment.
\newblock {\em Bioinform.}, 25(15):1966--1967, 2009.
\newblock \href {https://doi.org/10.1093/bioinformatics/btp336}
  {\path{doi:10.1093/bioinformatics/btp336}}.

\bibitem{MBCT2015}
Veli M{\"{a}}kinen, Djamal Belazzougui, Fabio Cunial, and Alexandru~I. Tomescu.
\newblock {\em Genome-scale algorithm design: Biological sequence analysis in
  the era of high-throughput sequencing}.
\newblock Cambridge University Press, Cambridge, UK, 2015.
\newblock \href {https://doi.org/10.1017/cbo9781139940023}
  {\path{doi:10.1017/cbo9781139940023}}.

\bibitem{MakinenN08}
Veli M{\"{a}}kinen and Gonzalo Navarro.
\newblock Dynamic entropy-compressed sequences and full-text indexes.
\newblock {\em {ACM} Trans. Algorithms}, 4(3):32:1--32:38, 2008.
\newblock \href {https://doi.org/10.1145/1367064.1367072}
  {\path{doi:10.1145/1367064.1367072}}.

\bibitem{mm1993}
Udi Manber and Eugene~W. Myers.
\newblock Suffix arrays: A new method for on-line string searches.
\newblock {\em {SIAM} J. Comput.}, 22(5):935--948, 1993.
\newblock \href {https://doi.org/10.1137/0222058} {\path{doi:10.1137/0222058}}.

\bibitem{MunroNN17}
J.~Ian Munro, Gonzalo Navarro, and Yakov Nekrich.
\newblock Space-efficient construction of compressed indexes in deterministic
  linear time.
\newblock In {\em 28th Annual {ACM-SIAM} Symposium on Discrete Algorithms
  (SODA)}, pages 408--424, 2017.
\newblock \href {https://doi.org/10.1137/1.9781611974782.26}
  {\path{doi:10.1137/1.9781611974782.26}}.

\bibitem{MunroNN20a}
J.~Ian Munro, Gonzalo Navarro, and Yakov Nekrich.
\newblock Fast compressed self-indexes with deterministic linear-time
  construction.
\newblock {\em Algorithmica}, 82(2):316--337, 2020.
\newblock \href {https://doi.org/10.1007/s00453-019-00637-x}
  {\path{doi:10.1007/s00453-019-00637-x}}.

\bibitem{MunroNN20b}
J.~Ian Munro, Gonzalo Navarro, and Yakov Nekrich.
\newblock Text indexing and searching in sublinear time.
\newblock In {\em 31st Annual Symposium on Combinatorial Pattern Matching
  (CPM)}, pages 24:1--24:15, 2020.
\newblock \href {https://doi.org/10.4230/LIPIcs.CPM.2020.24}
  {\path{doi:10.4230/LIPIcs.CPM.2020.24}}.

\bibitem{MunroNV16}
J.~Ian Munro, Yakov Nekrich, and Jeffrey~Scott Vitter.
\newblock Fast construction of wavelet trees.
\newblock {\em Theor. Comput. Sci.}, 638:91--97, 2016.
\newblock \href {https://doi.org/10.1016/j.tcs.2015.11.011}
  {\path{doi:10.1016/j.tcs.2015.11.011}}.

\bibitem{navarrobook}
Gonzalo Navarro.
\newblock {\em Compact data structures: A practical approach}.
\newblock Cambridge University Press, Cambridge, UK, 2016.
\newblock \href {https://doi.org/10.1017/cbo9781316588284}
  {\path{doi:10.1017/cbo9781316588284}}.

\bibitem{NavarroMeasures}
Gonzalo Navarro.
\newblock Indexing highly repetitive string collections, part {I}:
  Repetitiveness measures.
\newblock {\em {ACM} Comput. Surv.}, 54(2), 2021.
\newblock \href {https://doi.org/10.1145/3434399} {\path{doi:10.1145/3434399}}.

\bibitem{NavarroIndexes}
Gonzalo Navarro.
\newblock Indexing highly repetitive string collections, part {II}: Compressed
  indexes.
\newblock {\em {ACM} Comput. Surv.}, 54(2), 2021.
\newblock \href {https://doi.org/10.1145/3432999} {\path{doi:10.1145/3432999}}.

\bibitem{NavarroM07}
Gonzalo Navarro and Veli M{\"{a}}kinen.
\newblock Compressed full-text indexes.
\newblock {\em {ACM} Comput. Surv.}, 39(1):2, 2007.
\newblock \href {https://doi.org/10.1145/1216370.1216372}
  {\path{doi:10.1145/1216370.1216372}}.

\bibitem{NavarroN17}
Gonzalo Navarro and Yakov Nekrich.
\newblock Time-optimal top-k document retrieval.
\newblock {\em {SIAM} J. Comput.}, 46(1):80--113, 2017.
\newblock \href {https://doi.org/10.1137/140998949}
  {\path{doi:10.1137/140998949}}.

\bibitem{navarro201941}
Gonzalo Navarro and Nicola Prezza.
\newblock Universal compressed text indexing.
\newblock {\em Theor. Comput. Sci.}, 762:41--50, 2019.
\newblock \href {https://doi.org/10.1016/j.tcs.2018.09.007}
  {\path{doi:10.1016/j.tcs.2018.09.007}}.

\bibitem{NishimotoDAM}
Takaaki Nishimoto, Tomohiro I, Shunsuke Inenaga, Hideo Bannai, and Masayuki
  Takeda.
\newblock Dynamic index and {LZ} factorization in compressed space.
\newblock {\em Discret. Appl. Math.}, 274:116--129, 2020.
\newblock \href {https://doi.org/10.1016/j.dam.2019.01.014}
  {\path{doi:10.1016/j.dam.2019.01.014}}.

\bibitem{ohl2013}
Enno Ohlebusch.
\newblock {\em Bioinformatics algorithms: Sequence analysis, genome
  rearrangements, and phylogenetic reconstruction}.
\newblock Oldenbusch Verlag, Ulm, Germany, 2013.

\bibitem{OhlebuschFG10}
Enno Ohlebusch, Johannes Fischer, and Simon Gog.
\newblock {CST++}.
\newblock In {\em 17th International Symposium on String Processing and
  Information Retrieval (SPIRE)}, pages 322--333, 2010.
\newblock \href {https://doi.org/10.1007/978-3-642-16321-0_34}
  {\path{doi:10.1007/978-3-642-16321-0_34}}.

\bibitem{Prezza17}
Nicola Prezza.
\newblock A framework of dynamic data structures for string processing.
\newblock In {\em 16th International Symposium on Experimental Algorithms
  (SEA)}, pages 11:1--11:15, 2017.
\newblock \href {https://doi.org/10.4230/LIPIcs.SEA.2017.11}
  {\path{doi:10.4230/LIPIcs.SEA.2017.11}}.

\bibitem{PrezzaR21}
Nicola Prezza and Giovanna Rosone.
\newblock Space-efficient construction of compressed suffix trees.
\newblock {\em Theor. Comput. Sci.}, 852:138--156, 2021.
\newblock \href {https://doi.org/10.1016/j.tcs.2020.11.024}
  {\path{doi:10.1016/j.tcs.2020.11.024}}.

\bibitem{Patrascu07}
Mihai P\u{a}tra\c{s}cu.
\newblock Lower bounds for 2-dimensional range counting.
\newblock In {\em 39th Annual {ACM} Symposium on Theory of Computing (STOC)},
  pages 40--46, 2007.
\newblock \href {https://doi.org/10.1145/1250790.1250797}
  {\path{doi:10.1145/1250790.1250797}}.

\bibitem{RussoNO11}
Lu{\'{\i}}s M.~S. Russo, Gonzalo Navarro, and Arlindo~L. Oliveira.
\newblock Fully compressed suffix trees.
\newblock {\em {ACM} Trans. Algorithms}, 7(4):53:1--53:34, 2011.
\newblock \href {https://doi.org/10.1145/2000807.2000821}
  {\path{doi:10.1145/2000807.2000821}}.

\bibitem{Ruzic08}
Milan R{\v{u}}zi{\'{c}}.
\newblock Constructing efficient dictionaries in close to sorting time.
\newblock In {\em 35th International Colloquium on Automata, Languages, and
  Programming (ICALP)}, pages 84--95, 2008.
\newblock \href {https://doi.org/10.1007/978-3-540-70575-8_8}
  {\path{doi:10.1007/978-3-540-70575-8_8}}.

\bibitem{cst}
Kunihiko Sadakane.
\newblock Succinct representations of lcp information and improvements in the
  compressed suffix arrays.
\newblock In {\em 13th Annual {ACM-SIAM} Symposium on Discrete Algorithms
  (SODA)}, pages 225--232, 2002.
\newblock URL: \url{http://dl.acm.org/citation.cfm?id=545381.545410}.

\bibitem{Sadakane07}
Kunihiko Sadakane.
\newblock Compressed suffix trees with full functionality.
\newblock {\em Theory Comput. Syst.}, 41(4):589--607, 2007.
\newblock \href {https://doi.org/10.1007/s00224-006-1198-x}
  {\path{doi:10.1007/s00224-006-1198-x}}.

\bibitem{Weiner73}
Peter Weiner.
\newblock Linear pattern matching algorithms.
\newblock In {\em 14th Annual Symposium on Switching and Automata Theory
  ({SWAT}/{FOCS})}, pages 1--11, 1973.
\newblock \href {https://doi.org/10.1109/SWAT.1973.13}
  {\path{doi:10.1109/SWAT.1973.13}}.

\bibitem{LZ77}
Jacob Ziv and Abraham Lempel.
\newblock A universal algorithm for sequential data compression.
\newblock {\em Trans. Inf. Theory}, 23(3):337--343, 1977.
\newblock \href {https://doi.org/10.1109/TIT.1977.1055714}
  {\path{doi:10.1109/TIT.1977.1055714}}.

\end{thebibliography}

\end{document}